\pdfoutput=1
\documentclass[journal]{IEEEtran}

\usepackage{bm}
\usepackage{amsmath}
\usepackage{epsf}
\usepackage{graphics}
\usepackage{ amssymb }
\usepackage[dvips]{graphicx}
\usepackage{mathtools}% Loads amsmath
\usepackage{epsfig}
\usepackage{cite}
\usepackage[linesnumbered,ruled,vlined]{algorithm2e}
\usepackage{colortbl}
\usepackage{color}
\usepackage{enumitem}
\usepackage{soul,xcolor}

\usepackage{bm}
\usepackage{amsmath}
\usepackage{epsf}
\usepackage{graphics}
\usepackage{ amssymb }
\usepackage[dvips]{graphicx}
\usepackage{epsfig}
\usepackage{cite}
\usepackage[linesnumbered,ruled,vlined]{algorithm2e}
\usepackage{graphicx}
\usepackage{epsfig}
\usepackage{latexsym}
\usepackage{amsfonts}
\usepackage{here}
\usepackage{rawfonts}
\usepackage[protrusion=true,expansion=true]{microtype}
\pdfoutput=1
\usepackage[utf8]{inputenc}
\usepackage[english]{babel}
\usepackage{amsmath}
\usepackage{amsfonts}
\usepackage{amssymb}
\usepackage{color} %May be necessary if you want to color links
\usepackage{bm}
\usepackage{listings}
\usepackage{caption}
\usepackage{amssymb}
\usepackage{amsthm}
\usepackage{graphicx}
\usepackage{epstopdf}
\usepackage{listings}
\usepackage{float}
\usepackage{amsmath}
\usepackage{amssymb}
\usepackage{amsfonts}
\usepackage{epstopdf}
\usepackage[overload]{empheq}

\usepackage{multirow}
\usepackage{amscd}
\usepackage{mathrsfs}
\usepackage{graphicx}
\usepackage{color}
\usepackage{url}
\usepackage{bm}
\usepackage{setspace}
\usepackage{footnote}
\usepackage{xcolor}
\DeclareMathOperator*{\argmax}{arg\,max}

\DeclareMathOperator*{\argmin}{arg\,min}
\newtheorem{theorem}{Theorem}

\newtheorem{definition}{Definition}
\newtheorem{proposition}{Proposition}
\newtheorem{corollary}{Corollary}

\newtheorem{remark}{Remark}

\newtheorem{assumption}{Assumption}

\usepackage[]{algorithm2e}

\addto\captionsenglish{}
\usepackage{bbm}

\newcommand\norm[1]{\left\lVert#1\right\rVert}
\usepackage{multicol}
\usepackage{mathtools}

\usepackage{lipsum,graphicx,subcaption}

\usepackage{diagbox}
\usepackage[font=footnotesize]{caption}
\let\emptyset\varnothing

\usepackage{graphicx}
\usepackage{grffile}
\usepackage{tabularx}
\makeatletter
\newcommand{\vast}{\bBigg@{4}}

\newcommand{\Vast}{\bBigg@{5}}
\makeatother
\begin{document}

\title{Multi-Stage Hybrid Federated Learning over Large-Scale D2D-Enabled Fog Networks}
% \author{Seyyedali Hosseinalipour, Sheikh Shams Azam, Christopher G. Brinton, \\ Nicolo Michelusi, Vaneet Aggarwal, David~J.~Love, and Huaiyu~Dai
% %Seyyedali Hosseinalipour, Sheikh Shams Azam, Christopher G. Brinton, \\ Nicolo Michelusi, Vaneet Aggarwal, David~J.~Love, and Huaiyu~Dai
\author{Seyyedali Hosseinalipour,~\IEEEmembership{Member,~IEEE},~Sheikh Shams Azam, Christopher G. Brinton,~\IEEEmembership{Senior Member,~IEEE}, Nicol\`{o} Michelusi, \IEEEmembership{Senior Member,~IEEE}, Vaneet Aggarwal,~\IEEEmembership{Senior Member,~IEEE},\\ David~J.~Love,~\IEEEmembership{Fellow,~IEEE}, and Huaiyu~Dai,~\IEEEmembership{Fellow,~IEEE}
\thanks{\scriptsize S. Hosseinalipour, S. Azam, C. Brinton, V.~Aggarwal, and D. Love are with Purdue University: \{hosseina,azam1,cgb,vaneet,djlove\}@purdue.edu. 
N. Michelusi is with Arizona State University: nicolo.michelusi@asu.edu.
H. Dai is with NC State University: hdai@ncsu.edu.}
\thanks{\scriptsize {C. Brinton was supported in part by ONR under grant N00014-21-1-2472, and NSC grant W15QKN-15-9-1004. Part of Michelusi's research has been funded by NSF under grants CNS-1642982 and CNS-2129015. D. Love was supported in part by the NSF under grant EEC1941529. H. Dai was supported by NSF CNS-1824518.}}
}
%\date{\vspace{-2in}}
\maketitle

\begin{abstract}
Federated learning has generated significant interest, with nearly all works focused on a ``star" topology where nodes/devices are each connected to a central server. We migrate away from this architecture and extend it through the \textit{network} dimension to the case where there are multiple layers of nodes between the end devices and the server. Specifically, we develop multi-stage hybrid federated learning ({\tt MH-FL}), a hybrid of intra- and inter-layer model learning that considers the network as a \textit{multi-layer cluster-based structure.} {\tt MH-FL} considers the  \textit{topology structures} among the nodes in the clusters, including local networks formed via device-to-device (D2D) communications, and {presumes a \textit{semi-decentralized architecture} for federated learning}. It orchestrates the devices at different network layers in a collaborative/cooperative manner (i.e., using D2D interactions) to form \textit{local consensus} on the model parameters and combines it with multi-stage parameter relaying between layers of the tree-shaped hierarchy. 
We derive the upper bound of convergence for {\tt MH-FL} with respect to parameters of the network topology (e.g., the spectral radius) and the learning algorithm (e.g., the number of D2D rounds in different clusters). We obtain a set of policies for the D2D rounds at different clusters to guarantee either a finite optimality gap or convergence to the global optimum. We then develop a distributed control algorithm for {\tt MH-FL} to tune the D2D rounds in each cluster over time to meet specific convergence criteria. Our experiments on real-world datasets verify our analytical results and demonstrate the advantages of {\tt MH-FL} in terms of resource utilization metrics.

\end{abstract}
\vspace{-1mm}

\vspace{-4mm}
\section{Introduction}
\noindent 
% The number of Internet-connected devices continues to rise dramatically. It is currently estimated that the Internet of Things (IoT) will reach 31 billion devices in 2020 and 75 billion by 2025~\cite{onlineNumber}. 
%Smartphones, smart cars, and unmanned aerial vehicles (UAVs) have driven next generation cellular systems, vehicular networks, and UAV-assisted networks, which has in turn increased the number of devices capable of connecting to the Internet. 
%Studying the large-scale networks consisting of millions/billions of connected devices has formed the framework of Internet of Things (IoT)~\cite{whitmore2015internet,li20185g}.
%It is expected that there will be around 31 billion IoT devices by the end of year 2020.
%This number is expected to get more than double and reach 75 billion by year 2025~\cite{onlineNumber}.
%Utilization of the computation resources of the IoT devices along with the cloud networks to compute the emerging processing applications on the devices has motivated the framework of Fog computing~\cite{7498684}.
Machine learning (ML) has produced automated solutions to problems ranging from natural language processing to object detection/tracking~\cite{ciregan2012multi,collobert2008unified}.
Traditionally, ML model training has been carried out at a central node (e.g., a server). In many contemporary applications of ML, however, the relevant data is generated at the end user devices. As these devices generate larger volumes of data, transferring it to a central server for model training has several drawbacks: (i) it may require significant energy consumption from battery-powered devices; (ii) the round-trip-times between data generation and model training may incur prohibitive delays; and (iii) in privacy-sensitive applications, end users may not be willing to transmit their raw data in the first place.

Federated learning has emerged as a technique for distributing model training across devices while keeping the devices' datasets local~\cite{mcmahan2017communication}. Its conventional architecture consists of a main server connected to multiple devices in a \textit{star} topology (see Fig.~\ref{fig:simpleFL}). Each round of model training consists of two steps: (i) \textit{local updating}, where each device updates its local model based on its  dataset and the global model, e.g., using gradient descent, and (ii) \textit{global aggregation}, where the server gathers devices' local models  and computes a new global model, which is then synchronized across the  devices to begin the next round.

\begin{figure}[t]
\centering
\includegraphics[width=.47\textwidth]{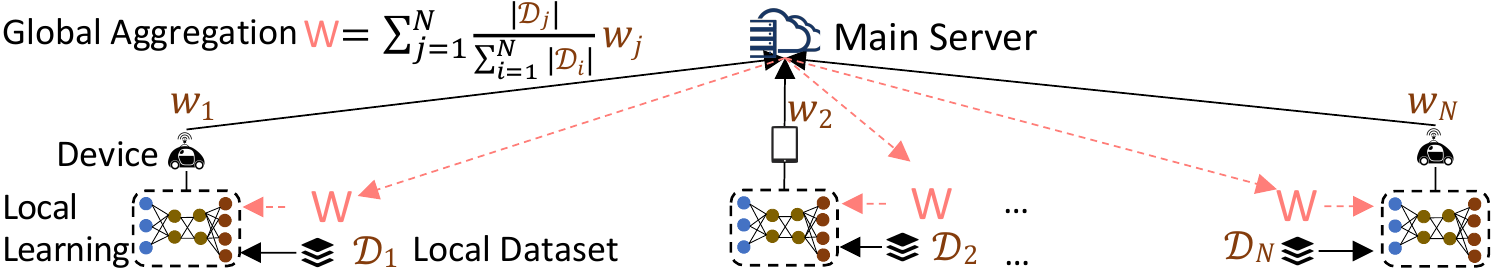}
\caption{Conventional \textit{star} topology architecture of federated learning.}
% User devices perform local model updates using their datasets, and send their learned parameters to the main server. The main server aggregates the received parameters into a new global model broadcast to the devices.}
\label{fig:simpleFL}
\vspace{-1mm}
\end{figure}

In conventional federated learning, only device-to-server (in step (i)) and server-to-device (in step (ii)) communications occur. This is  limiting and prohibitive in contemporary large-scale network scenarios, where there are several layers of nodes between the end devices and the cloud (see Fig.~\ref{diag:scenario22}(\subref{fig:multi_stage})). In particular, it can lead to long delays, large bandwidth utilization, and high power consumption for aggregations~\cite{hosseinalipour2020federated}. We migrate federated learning from its star structure to a more distributed structure that accounts for the multi-layer network dimension and leverages \textit{topology structures} among the devices.

\begin{figure*}[t]
\centering
\subcaptionbox{A schematic of model transfer stages for a large-scale ML task in a fog network. The parameters of the end devices are carried through multiple layers of the network consisting of base stations (BSs), road side units (RSUs), unmanned aerial vehicles (UAVs), high altitude platforms (HAPs), edge servers, and cloud servers before reaching the main server. Devices located at different layers of the network can engage in direct communications via mobile-mobile (M2M), vehicle-vehicle (V2V), UAV-UAV (U2U), inter-edge, and inter-cloud links.\label{fig:multi_stage}}[.45\linewidth][c]{		\includegraphics[width=2.8in,height=1.4in]{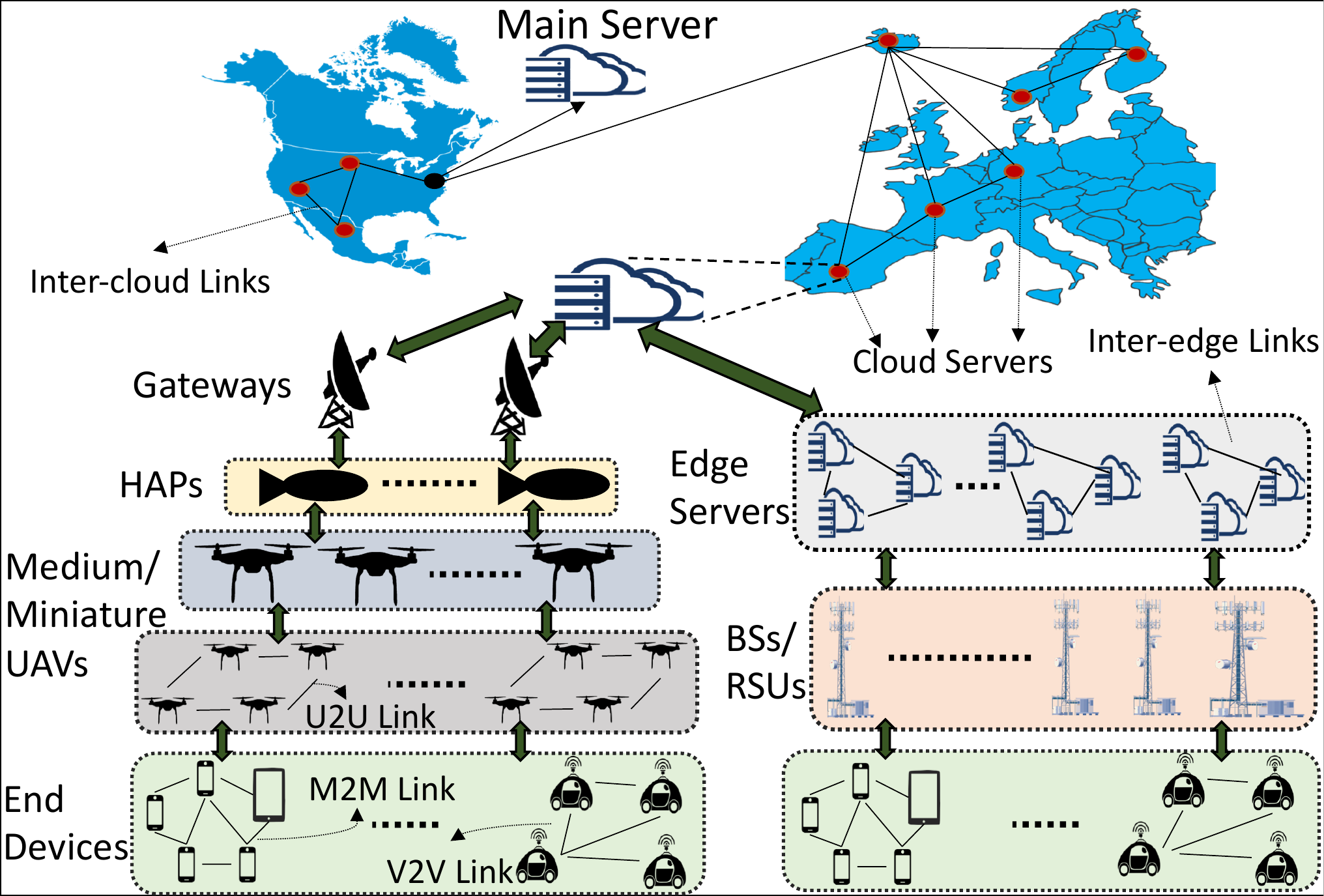}}\qquad
\subcaptionbox{Partitioning the network layers into multiple LUT/EUT clusters for FogL, introducing a hybrid model training framework consisting of both horizontal and vertical parameter transfer. Inside each LUT cluster, the devices engage in collaborative/cooperative D2D communications through a time varying network topology and exchange their parameters.
The parents of LUT clusters obtain the consensus of their children parameters by sampling the model parameter of one node, while the parents of EUT clusters receive all the parameters of their children.\label{fig:multi_stage_Clusters}}[.45\linewidth][c]{%
\includegraphics[width=2.8in,height=1.4in]{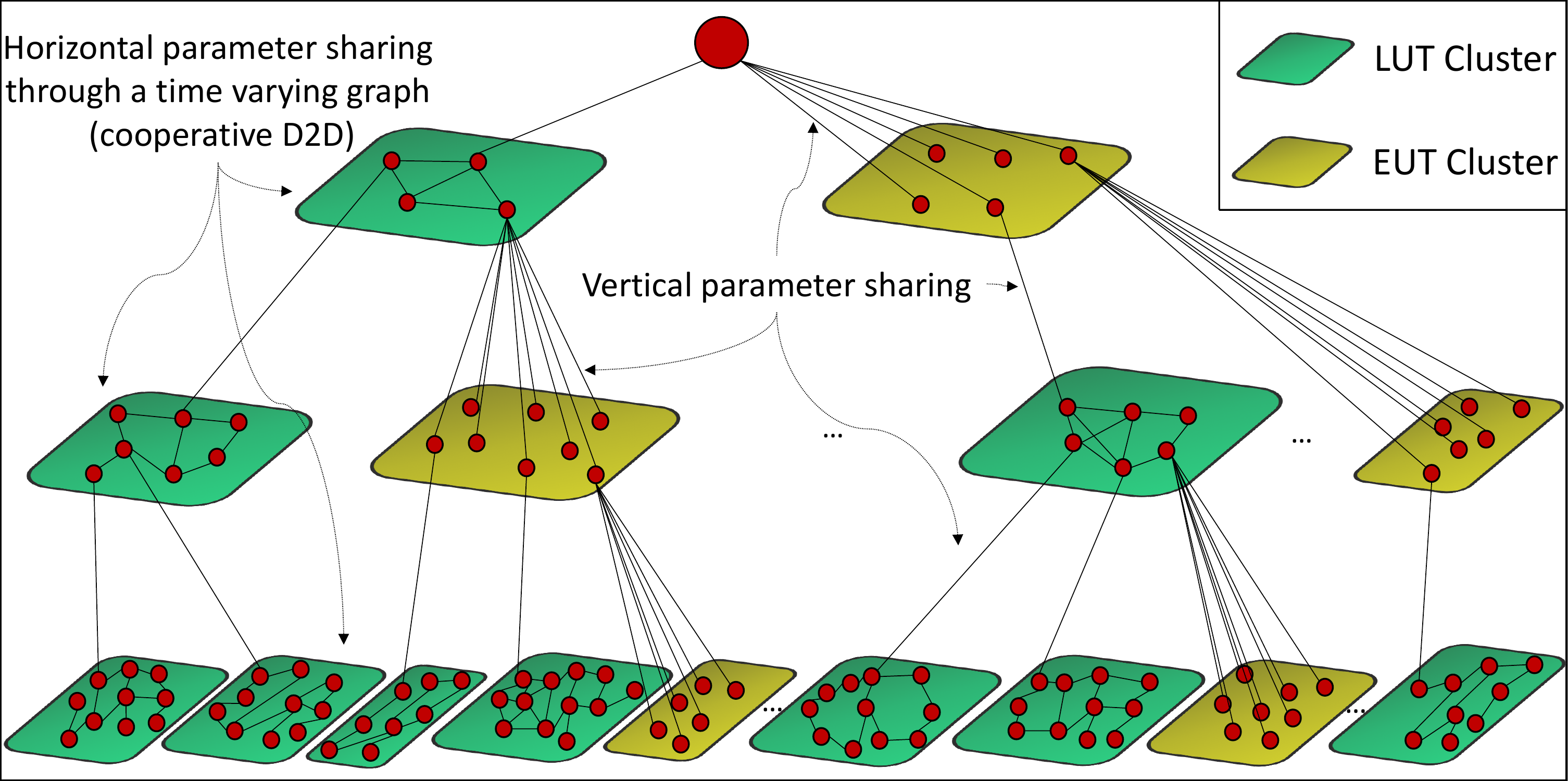}}
\vspace{-1.5mm}
\caption{Architecture of (a) the multi-layer network structure of fog computing systems and (b) the network layers and parameters transfer for FogL.}
\label{diag:scenario22}
\vspace{-7.5mm}
\end{figure*}

\vspace{-4mm}
\subsection{Fog Learning: Federated Learning in Fog Environments}
\label{ssec:fogL}
% \ali{These numbers, combined with increasing computational capabilities, have motivated \textit{fog computing}.}
\vspace{-.5mm}
Fog computing is an emerging technology which aims to manage computation resources across the cloud-to-things continuum, encompassing the cloud, core, edge, metro, clients, and things~\cite{7901470}. We recently introduced the fog learning (FogL) paradigm~\cite{hosseinalipour2020federated}, which advocates leveraging the fog computing architecture to handle ML tasks. Specifically, FogL requires extending federated learning to (i) incorporate fog network structures, (ii) account for device computation heterogeneity, and (iii) manage the proximity between resource-abundant and resource-constrained nodes. Our focus in this paper is (i), i.e., extending federated learning along the \textit{network} dimension.

We consider the sample network architecture of FogL in Fig.~\ref{diag:scenario22}(\subref{fig:multi_stage}). There are multiple layers between user devices and the cloud, including base stations (BSs) and edge servers. Compared with federated learning, FogL has two key characteristics: (i) It assumes a multi-layer cluster-based structure with local parameter aggregations at different layers. (ii) In addition to inter-layer communications, it includes intra-layer communications via device-to-device (D2D) connectivity, which is promoted in 5G and IoT~\cite{6815897}. There exists a literature on D2D communication protocols for ad-hoc and sensor networks including vehicular ad-hoc networks (VANET), mobile ad-hoc networks (MANET),  and flying ad-hoc networks (FANET)~\cite{abolhasan2004review,zeadally2012vehicular,bekmezci2013flying,akkaya2005survey}. Exploiting D2D communications has also been promoted in agriculture and rural use cases~\cite{zhang2021challenges}, making FogL a promising model training strategy in such environments. FogL  also considers server-to-server interactions~\cite{maheshwari2018scalability} and other types of peer-to-peer (P2P) interactions under the umbrella of D2D. 

Characteristic (ii) mentioned above orchestrates the devices at each layer in a cooperative framework, introducing a set of \textit{local networks} to the learning paradigm. This motivates studying the learning performance with explicit consideration of \textit{topology} structure among the devices. To do this, we partition network layers into clusters of two types as depicted in Fig.~\ref{diag:scenario22}(\subref{fig:multi_stage_Clusters}): (i) limited uplink transmission (LUT), where D2D communications are enabled, and (ii) extensive uplink transmission (EUT), where, similar to the conventional federated learning, 
all nodes only communicate with their upper layer.

Accommodating both inter- and intra-layer communications introduces a \textit{semi-decentralized learning architecture}, which is a \textit{hybrid} model for training that considers conventional server-device interactions (i.e., centralized ``star" topology) in conjunction with collaborative/cooperative D2D communications (i.e., fully decentralized ``mesh" topology). Thus, the methodology we develop in this paper is called \textit{multi-stage hybrid federated learning} ({\tt MH-FL}) and considers both intra-cluster consensus formation and inter-cluster aggregations for distributed ML. In developing {\tt MH-FL}, we incorporate the time-varying local network topologies among the devices as a dimension of federated learning, and demonstrate how it impacts ML model convergence and accuracy.

\vspace{-3mm}
\subsection{Related Work}
\vspace{-.5mm}
% Federated learning has received significant research attention. 
% The existing literature on federated learning can be categorized into three aspects: works on studying/improving (i) network communication demands, (ii) device computation requirements, and (iii) privacy/security guarantees. In the first category, 
Researchers have considered the effects of limited and imperfect communication capabilities in wireless networks -- such as channel fading, packet loss, and limited bandwidth -- on the operation of federated learning~\cite{chen2019joint,8737464,9014530}. Also, communication techniques such as quantization~\cite{9054168,lee2021finite} and sparsification of model updates (i.e., when only a fraction of the model parameters are shared during model training)~\cite{renggli2019sparcml} have been studied. Recently, \cite{8664630} analyzed the convergence bounds in the presence of edge network resource constraints. 

% Focusing on the case of FogL, our work generalizes the communication strategy of federated learning to include both distributed average consensus formation via D2D and multi-layer aggregations.

Research has also considered the computation aspects of federated learning in wireless networks~\cite{8737464,8964354,tu2020network,9488906}. Part of this literature has focused on learning in the presence of \textit{stragglers}, i.e., when a node has significantly lower computation capabilities than others~\cite{8737464,8964354}. Another emphasis has been reducing the computation requirements through intelligent  raw data offloading between devices~\cite{tu2020network} and judicious selection of device  participation~\cite{ji2020dynamic,9488906}. Other techniques for mitigating compute limitations, e.g., through model compression, have also been applied to distributed ML~\cite{wang2018wide}. 

% In our work, we study the computation requirements of devices in FogL in terms of the number of required training rounds.

%Given the limitations of processing equipment and varying availabilities of devices at the network edge, 

% Recent works have also studied privacy and security in federated learning~\cite{geyer2017differentially,hardy2017private}. Even though source data is never transmitted over the network in federated learning, it can be possible to extract sensitive information from transmitted model parameters.

% Recent works have also studied adapting security techniques such as differential privacy~\cite{geyer2017differentially} and homomorphic encryption~\cite{hardy2017private} to federated learning. Our focus is instead on the communication and computation aspects.

There exist recent works on hierarchical federated learning~\cite{TobeAdded4,TobeAdded3,TobeAdded2}. These works are mainly focused on specific use cases of two-tiered network structures above wireless cellular devices, e.g., edge clouds connected to a main server~\cite{TobeAdded4,TobeAdded3} or small cell and macro cell base stations~\cite{TobeAdded2}. As compared to all the aforementioned works, which consider the star model training topology (or \textit{tree} in case of hierarchical considerations), our work is distinct for several reasons, including that: (i) we introduce a \textit{multi-layer cluster-based structure} with an arbitrary height that encompasses all IoT elements between the end devices and the main server, which generalizes all the prior models; and (ii) more importantly, we explicitly consider the \textit{network} dimension and topology structure among the devices at each network layer formed via cooperative/collaborative D2D communications. This migrates us from prior models and enables new analysis and considerations for hybrid intra- and inter-layer model learning over large-scale fog networks.

There also exist recent works on \textit{fully decentralized (server-less)} federated learning~\cite{8950073,hu2019decentralized,lee2021finite}.
% \nm{why are you not comparing with any of these works by simulation?}
{These architectures require a well-connected D2D communication graph among all the devices in the network, which becomes less feasible to maintain as the geographical span of the devices increases (e.g., the end devices across multiple regions in Fig.~\ref{diag:scenario22}(\subref{fig:multi_stage})).}
We establish an intermediate learning architecture that couples the star topology {assumed in conventional federated learning} with fully decentralized {architectures} to provide a \textit{scalable model training}. In particular, we propose a {novel} semi-decentralized learning architecture that (i) uses a cluster-based representation of devices with local communications only  among the D2D-enabled devices inside the same cluster; (ii) reduces the reliance on resource-intensive uplink model transmissions via sampling only one device from D2D-enabled clusters; and (iii) is based on a layered coordination of global aggregations across the fog learning hierarchy facilitated by a main server.

Beyond federated learning, there is a well developed literature on other distributed ML techniques (e.g.,~\cite{elgabli2019gadmm,smith2017cocoa,richtarik2016distributed}). Our proposed framework for FogL inherits its model aggregation rule from federated learning, i.e., local gradient descent at the devices and weighted averaging to obtain the global model. We choose this due to specific characteristics that make it better suited for fog: keeping the user data local, handling non-iid datasets across devices, and handling imbalances between sizes of local datasets~\cite{mcmahan2017communication}. These capabilities have made federated learning the most widely acknowledged distributed learning framework for future wireless systems~\cite{niknam2019federated,8970161}. The multi-stage hybrid architecture of FogL could also be studied in the context of other distributed ML techniques, e.g., ADMM~\cite{elgabli2019gadmm}.

Finally, there exist a literature on distributed  consensus with applications in multi-agent systems~\cite{6733336,li2010consensus}, sensor networks~\cite{4663899,manfredi2013design}, and  optimization~\cite{4749425,6945888,6748910,johansson2008subgradient}. Our scenario is unique given its multi-layer network structure and focus on a hybrid ML model training, where the goal is to propagate an expectation of the nodes' parameters through the hierarchy to train an ML model. The results we obtain have thus not yet appeared in either consensus-related or ML-related literature. 

%Thus, we are interested in analyzing a new distributed ML paradigm that is the result of intertwining the multi-stage distributed consensus over time varying cluster topologies, local gradient descent updates, and global aggregations. Although the aforementioned literature has inspired us, to the best of our knowledge, our framework, proposed bounds, and the subsequent results that explicitly study the described twisted ML paradigm have not yet been appeared in neither consensus-related nor  machine learning-related literature.

\vspace{-3.5mm}
\subsection{Summary of Contributions}\label{ssec:contribution}
\vspace{-.1mm}
Our contributions in this work can be summarized as follows:
\begin{itemize}[leftmargin=4mm]
     \item We formalize multi-stage hybrid federated learning ({\tt MH-FL}), a new methodology for distributed ML. {\tt MH-FL} extends federated learning along the network dimension, relaying the local updates of end devices through the network hierarchy via a novel multi-stage, cluster-based parameter aggregation technique. This paradigm introduces local aggregations achieved by an interplay between cooperative D2D communications and distributed consensus formation. 
     
    %  As part of our investigation, we build the network architecture of FogL and propose a corresponding \textit{FogL augmented graph}.
     
    %  Different from federated learning, {\tt MH-FL} incorporates (i) distributed average consensus formation within node clusters, and (ii) model aggregations between layers of the network hierarchy, to improve resource utilization in large-scale wireless networks.

     %We build the network architecture of FogL and propose a corresponding \textit{FogL augmented graph} that opens a door to conducting tractable mathematical analysis. We exploit D2D communications and demonstrate that the hybrid learning characteristic of FogL migrates the learning architecture from the strict star learning topology of federated learning to a collaborative learning architecture.
     %\item We utilize the well-known \textit{distributed average consensus} technique and realize FogL via multi-stage consensus among the devices and demonstrate its perfect match with the inherent updating mechanism of the main server in FogL. Also, we break down the global aggregations to aggregations conducted at different network layers and show its significant impact on reducing the network traffic.
     
     \item We analytically characterize an upper bound of convergence of {\tt MH-FL}. We demonstrate how this bound depends on characteristics of the ML model, the network topology, and the learning algorithm, including the number of model parameters, the communication graph structure, and the number of D2D rounds at different device clusters.
     
     \item We demonstrate that the model loss achieved by {\tt MH-FL} under unlimited D2D rounds coincides with that of federated learning. Under the finite D2D rounds regime, we obtain a condition under which a constant optimality gap can be achieved asymptotically. We further show that under limited finely-tuned D2D rounds, {\tt MH-FL} converges linearly to the optimal ML model. We further introduce a practical cluster sampling technique and investigate its convergence behavior.

     \item We obtain analytical relationships for tuning (i) the number of D2D rounds in different clusters at different layers and (ii) the number of global iterations to meet certain convergence criteria. We use these relationships to develop distributed control algorithms that each cluster can employ individually to adapt the number of D2D rounds it uses over time.
    
    \item Our experimental results on real-world datasets verify our theoretical findings and show that {\tt MH-MT} can improve network resource utilization significantly with negligible impact on model training convergence speed and accuracy.

    % \item We obtain two relationships for design considerations in {\tt MH-ML} implementations: (i) the number of consensus rounds required (at different clusters and layers of the network) to obtain a desired accuracy at a particular global aggregation round, and (ii) the number of global aggregations required to achieve a desired accuracy for a fixed number of consensus rounds.
     %\item  We introduce the idea of tapering consensus rounds through time and space. We propose a distributed algorithm to achieve this by tuning the consensus rounds at different network clusters over time. We also develop a version of this algorithm that can be used for non-convex ML models like neural networks. \ali{we sometimes boost it as well.}
     %one for smooth and strongly convex loss functions and another for non-convex models like neural networks. These algorithms tune the consensus rounds at different network clusters over time. Given the fact that the inherent assumptions are smoothness and strong convexity of the loss function, we further propose a version of the algorithm that can be used for training (deep) neural networks that do not obey such assumptions.
 \end{itemize}
%  \nic{overall, the intro reads pretty well}

\vspace{-3mm}
\section{System Model and Problem Formulation}\label{sec:FogArchitecture}
\vspace{-.5mm}
\noindent In this section, we formalize the FogL network architecture (Sec.~\ref{sec:augm}) and the ML problem (Sec.~\ref{ssec:ml}). Then, we formalize the hybrid learning paradigm via intra- and inter-cluster communications (Sec.~\ref{ssec:intra-inter}) and parameter sharing (Sec.~\ref{ssec:sharing}).
\vspace{-8mm}
\subsection{Network Architecture and Graph Model}
\label{sec:augm}
We consider the network architecture of FogL as depicted on the left in Fig.~\ref{fig:treeD2D}. In FogL, both inter-layer and intra-layer communications  take place to conduct ML model training. The inter-layer communications are captured via a tree graph, with the main server at the root and end devices as the leaves. Each layer is partitioned into multiple clusters, with each device in a cluster sharing the same parent node in the next layer up. In general, each node may have a parent node located multiple layers above (e.g., an edge device can be directly connected to an edge server), and multiple clusters can share the same parent node (e.g., multiple groups of cellular devices share the same BS). Except at the bottom layer, each node in the hierarchy is the \textit{parent} for a subset of the nodes, and is responsible for gathering the model parameters of its \textit{children} nodes.
%, e.g., close vehicles in a highway that all have access to the same road side unit (RSU), or a set of cellular D2D pairs that have access to a cellular base station (BS)

%To aid our analysis, from the FogL network representation on the left in Fig.~\ref{fig:treeD2D}, we construct the \textit{FogL augmented network graph} shown on the right, where several virtual nodes have been added. In practice, multiple clusters can share the same parent node\nic{this statement seems in contradiction with your later statement that "The nodes belonging to the same layer that have the same parent node form clusters" (i.e., this second one might induce to think that they form one blig cluster)} (see the two highlighted nodes\nic{I am not sure which ones are highlighted} on the left in Fig.~\ref{fig:treeD2D}), e.g., a BS will generally have three sectors of end user devices. We represent parent nodes with multiple clusters as virtual nodes, which themselves form a virtual cluster. To preserve the hierarchy, the set of virtual nodes form a new layer that serves as an intermediate between the physical parent node and its children (see the added virtual layer on the right in Fig.~\ref{fig:treeD2D}). Throughout, we do not distinguish between virtual and non-virtual nodes, clusters, and layers unless specified.

From this FogL network representation, we construct the \textit{FogL augmented network graph} shown on the right of Fig.~\ref{fig:treeD2D}, where several virtual nodes and clusters have been added. Virtual nodes are added in such a way that each node has a single parent node in its immediate upper layer. Also, when multiple clusters share the same parent node, a layer is added to the FogL augmented graph that consists of multiple intermediate virtual nodes forming a virtual cluster, such that there is always a one-to-one mapping between the clusters and parent nodes. If necessary, the nodes in the highest layer before the main server will form a virtual cluster to preserve this one-to-one mapping. A node without any neighbors in its layer is also assumed to form a virtual singleton cluster. For convenience, we refer to the structure of Fig.~\ref{fig:treeD2D} as a \textit{tree} since in macro-scale it resembles a tree structure. However, in the micro-scale, nodes inside the clusters form \textit{connected} graphs through which D2D communications are performed, differentiating the structure from  a tree graph. 

\begin{figure}[t]
\vspace{-.05mm}
\includegraphics[width=3.5in,height=1.22in]{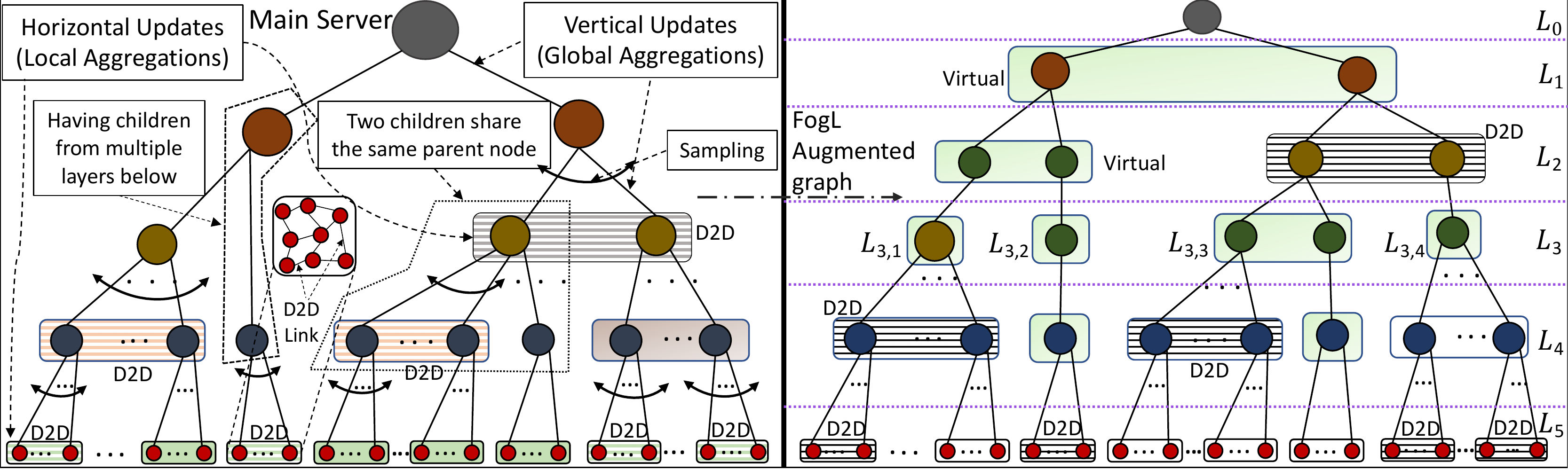}
\centering
\caption{Left: An example of FogL Network representation. The root corresponds to the main server, the leaves are the end devices, and the nodes in-between are intermediate fog nodes.  In D2D-enabled clusters, the nodes form a certain topology over which the devices communicate with their neighbors for distributed model consensus. Right: The corresponding FogL augmented graph for analysis. Virtual nodes and clusters are highlighted in green.}
\label{fig:treeD2D}
\vspace{-1mm}
\end{figure}

The FogL augmented graph has the following properties: (i) each parent node has a single cluster associated with it in the layer immediately below it;
(ii) the length of the paths from the root to each of the leaf nodes are the same; and (iii) the layer below the root always consists of one cluster. The network consists of $|\mathcal{L}| + 1$ layers, where $\mathcal{L}\triangleq \{ L_1,\cdots,L_{|\mathcal{L}|}\}$ denotes the set of layers below the server. The root/server is located at ${L}_{0}$ and the end devices are contained in ${L}_{|\mathcal{L}|}$. Inside layer ${L}_{j}$, $1\leq j \leq |\mathcal{L}|$, there exists a set of clusters indexed by ${L}_{j,1}, {L}_{j,2}, \cdots$ (see Fig.~\ref{fig:treeD2D}). 
We let the nodes move between the clusters in the same layer. To capture these dynamics, we use $\mathcal{L}^{(k)}_{j,i}$ (i.e., calligraphic font) to refer to the set of nodes inside cluster ${L}_{j,i}$ at learning iteration $k$ (described in Sec.~\ref{ssec:ml}). For ease of presentation, we assume that the  number of clusters at each layer is time invariant, i.e., each cluster always contains at least one node and nodes do not form new clusters. We let $\mathcal{N}_j$ denote the set of nodes in layer ${L}_j$ and $N_j\triangleq |\mathcal{N}_j|$.
% is not time varying since nodes do not join and leave the layers.

For convenience, we will sometimes use ${C}$ to denote an arbitrary cluster (i.e., any ${L}_{j,i}$) and  $\mathcal{C}^{(k)}$ (i.e., calligraphic font) to refer to its set of nodes at iteration $k$.  
We also sometimes use $n$ to denote an arbitrary node.
% \footnote{{\color{blue}As a spacial case, occasional existence of clusters consisting of no nodes in the bottom-most layer of the network, which inherently possesses the highest dynamics (e.g., cellular phones), is tolerable and does not affect our analysis.}}
For each node $n$ located in layers ${L}_{|\mathcal{L}|-1},\cdots,{L}_0$, we let ${Q}(n)$ denote the index of its child cluster and $\mathcal{Q}^{(k)}(n)$ denote the set of its children nodes (i.e., the nodes in $Q(n)$) at global iteration $k$.

% Nodes are indexed such that $\mathcal{L}^i_1=\{L^i_1, L^i_2,\cdots,L^i_{ |\mathcal{L}^{i}_{1}|}\}$, $\mathcal{L}^i_2=\{L^i_{|\mathcal{L}^{i}_{1}|+1}, L^i_{|\mathcal{L}^{i}_{1}|+2},\cdots,L^i_{|\mathcal{L}^{i}_{1}|+ |\mathcal{L}^{i}_{2}|} \}$, and so on. 

% and each cluster $c$ belongs to a layer, denoted $l(c)$. We write $\mathcal{N}(c)$ as the set of nodes belonging to cluster $c$, and $\mathcal{C}(l)$ as the set of clusters belonging to layer $l$. Except at the top layer $l = 0$, each node $i$ has a single parent $p(i)$, where $l(i) = l(p(i)) + 1$. Also, except at the bottom layer $l = L$, each node $i$ has a set of children $\mathcal{H}(i)$, where for each $j \in \mathcal{H}(i)$, $p(j) = i$.
\vspace{-3mm}
\subsection{Machine Learning Task}
\label{ssec:ml}
Each end device $n$ is associated with a dataset $\mathcal{D}_{n}$.
% , where $\mathcal{D}_{n} \cap \mathcal{D}_{{n'}} = \emptyset$, $n \neq n'$\nm{why do you need to assume this? are you saying that two nodes may not have the same data points? IT is quite restrictive..I'd remove}.\ali{Originally, I thought that would make the overall loss function not biased toward a certain data point. But we can remove it and I removed it!}
Each element $d_i \in \mathcal{D}_{n}$ of a dataset, called a training sample, is represented via a feature vector $\textbf{x}_i$ and a label $y_i$ for the ML task of interest. For example, in image classification, $\mathbf{x}_i$ may be the RGB colors of all pixels in the image, and $y_i$ may be the identity of the person in the image sample. The goal of the ML task is to learn the parameters $\mathbf{w} \in \mathbb{R}^M$ of a particular $M$-dimensional model (e.g., an SVM or a neural network) that are expected to maximize the accuracy in mapping from $\mathbf{x}_i$ to $y_i$ across any sample in the network. The model is associated with a loss $\widetilde{f}(\textbf{w},\textbf{x}_i,y_i)$, referred to as $\widetilde{f}_i(\textbf{w})$
% \nm{sholdnt it be $\widetilde{f}_{n,i}$ for user n?} \ali{The model structures are the same at different nodes (the same NN or SVM, etc.) If you mean we add the index n to the data point as well, like d_{i,n}, that can be done but may make the notations more complicated.} 
for brevity, that quantifies the error of parameter realization $\mathbf{w}$ on $d_i$. We refer to Table 1 in~\cite{8664630} for a list of common ML loss functions. 
%To measure the accuracy of the mapping upon realization $\textbf{w}$ for the parameters for each feature-label tuple $(\textbf{x}_i,y_i)$, a loss function is defined as: $\widetilde{f}(\textbf{w},\textbf{x}_i,y_i)$, which in short we refer to as $\widetilde{f}_i(\textbf{w})$.  A lower value of loss function implies a model with a higher accuracy.
%find a high accuracy mapping (e.g., achieved via an SVM, logistic regression, or a deep neural network, see Table 1 of~\cite{8664630} for more information about the expressions of different loss functions), identified by a set of (learning) parameters, from the entire nodes feature set to their label set.  To measure the accuracy of the mapping upon realization $\textbf{w}$ for the parameters for each feature-label tuple $(\textbf{x}_i,y_i)$, a loss function is defined as: $\widetilde{f}(\textbf{w},\textbf{x}_i,y_i)$, which in short we refer to as $\widetilde{f}_i(\textbf{w})$.  A lower value of loss function implies a model with a higher accuracy.  

The global loss of the ML model is formulated as
\vspace{-1.2mm}
\begin{equation}\label{eq:globlossinit}
    F(\mathbf{w})= \frac{1}{D} \sum_{n \in \mathcal{N}_{|\mathcal{L}|}} |\mathcal{D}_n| f_{n}(\textbf{w}),~~~D=\sum_{n \in \mathcal{N}_{|\mathcal{L}|}} |\mathcal{D}_n|,
\end{equation}
    \vspace{-4mm}
    
\noindent where $f_{n}$ is the local loss at node $n$, i.e., $f_{n}(\textbf{w}) = \frac{1}{|\mathcal{D}_n|} \sum_{d_i \in \mathcal{D}_n} \widetilde{f}_i(\textbf{w})$.
The goal of model training is to identify the optimal parameter $\mathbf{w}^*$ that minimizes the global loss:
\vspace{-2.5mm}
\begin{equation}
\mathbf{w}^*= \underset{\mathbf{w}\in \mathbb{R}^M}{\argmin} \; F(\mathbf{w}).
\end{equation}
  \vspace{-4.8mm}
    
\noindent To achieve this in a distributed manner, training is conducted through consecutive global iterations. At the start of global iteration $k\in \mathbb{N}$, the main server possesses a parameter vector $\mathbf{w}^{(k-1)} \in \mathbb{R}^M$, which propagates downstream through the hierarchy to the end devices.
Each end device $n$ overrides its current local model parameter vector $\mathbf{w}^{(k-1)}_{n}$ according to
  \vspace{-2.1mm}
\begin{equation}\label{eq:localupdateOverrride}
    \mathbf{w}^{(k-1)}_{n}=\mathbf{w}^{(k-1)},
\end{equation} 
\vspace{-5.2mm}
    
\noindent
and then performs a local update using gradient descent~\cite{mcmahan2017communication} as
  \vspace{-2.5mm}
\begin{equation}\label{eq:localupdate}
     \mathbf{w}^{(k)}_{n}=\mathbf{w}^{(k-1)}_{n} - \beta \nabla f_{n} (\mathbf{w}^{(k-1)}_{n}),~n\in\mathcal{N}_{|\mathcal{L}|},
 \end{equation}
  where $\beta$ is the step-size.
%  \nm{Id remove parantheses around the gradient, may be confusing}
The main server wishes to obtain the global model used for initiating the next global iteration, which is defined as a weighted average of end devices' parameters:
\vspace{-2mm}
\begin{equation}\label{eq:weightBaisc}
   \mathbf{w}^{(k)}= \frac{\sum_{n\in \mathcal{N}_{|\mathcal{L}|}}|\mathcal D_n| \mathbf{w}^{(k)}_{n}}{D}.
\end{equation}

The value of $D$ in \eqref{eq:weightBaisc} is assumed to be known at the main server, since it only requires uplink transmission of the scalar $|\mathcal{D}_n|$ by each end device $n \in \mathcal{N}_{|\mathcal{L}|}$, which can be aggregated and propagated upstream by each parent node in the hierarchy. As we will see, our method does not require knowledge of each individual $|\mathcal{D}_n|$ at the server to conduct the parameter averaging given in \eqref{eq:weightBaisc}, since the number of data points at each end device will be encapsulated in a scaled model parameter vector shared to its parent node (see Sec.~\ref{ssec:traversing}).
% \ali{Similar to the literature on federated learning, the parameter $D$ in \eqref{eq:weightBaisc} is assumed to be known at the server. As we will see, our method does not require knowledge of the size of each users' dataset at the server to conduct the parameter averaging given in~\eqref{eq:weightBaisc}, since the number of data points at each edge device will be encapsulated in the shared parameter to its parent node (see Sec.~\ref{ssec:traversing}).}
On the downlink from the server to the edge devices, we assume that the (common) global parameter $\mathbf{w}^{(k-1)}$ can be readily shared through the hierarchy to reach to the end devices, e.g., through a broadcasting protocol. These devices will then conduct~\eqref{eq:localupdateOverrride} and~\eqref{eq:localupdate} locally. The challenge, then, is computing~\eqref{eq:weightBaisc} at the main server. To do this, in federated learning, the devices will directly upload~\eqref{eq:localupdate} to the server. However, this is prohibitive in a large-scale fog computing system. First, it may require \textit{high energy consumption}: uplink transmissions from battery-limited devices to nodes at a higher layer typically correspond to long physical distances, and can deplete individual device batteries. Second, it may induce \textit{high network traffic and long latencies}: a neural network with even hundreds of parameters, which would be small by today's standards~\cite{wang2018wide}, would require transmission of billions of parameters in the upper layers during each iteration over a network with millions of nodes. Additionally, it may \textit{overload current cellular and vehicular architectures}: these infrastructures are not designed to handle large jumps in the number of active users \cite{clarke2014expanding}, which would be the case with simultaneous uplink transmissions at the bottom-most layer. These issues require a novel approach to parameter aggregations, which we develop in this paper.

\vspace{-4mm}
\subsection{Hybrid Learning via Intra- and Inter-Layer Communications}
\label{ssec:intra-inter}
\vspace{-0.5mm}
The main server in FogL is only interested in the weighted average of the local parameters~\eqref{eq:weightBaisc}. Consequently, we propose \textit{local aggregations} at each network layer. To achieve this, each cluster in Fig.~\ref{fig:treeD2D} follows one of two mechanisms:
\vspace{-.5mm}
\begin{enumerate}[leftmargin=5mm]
    \item \textit{Distributed aggregation}: The nodes engage in a cooperative scheme facilitated by D2D communications to realize the consensus/average of their local model parameters. The parent node then samples parameters of one of the children and scales it by the number of children to calculate an approximate sum of the children nodes' parameters.
    \item \textit{Instant aggregation}: Each node instead uploads its local model to the parent node. The parent computes the aggregation directly as a sum of the children nodes' parameters.
\end{enumerate}
\vspace{-.5mm}
The communication requirement of the instant aggregation is often significantly higher than the distributed aggregation since D2D communications generally occur over much shorter distances, which makes them less power/energy consuming.
% although the distributed aggregation may require multiple D2D communication rounds, it generally occurs over much shorter distances, which makes it less power/energy consuming.
 We refer to mechanisms 1 and 2 mentioned above as limited uplink transmission (LUT) and extensive uplink transmission (EUT), respectively.\footnote{We assume that each node belongs either to an EUT or to an LUT cluster. If some portion of the nodes in a cluster are capable of D2D communications while the rest are not, the cluster can be broken down into two clusters (i.e., an LUT and an EUT cluster) with the same parent node, based on which the fogL augmented network graph in Sec.~\ref{sec:augm} is then constructed.} Not all clusters can operate in LUT mode, since not all are D2D enabled, e.g., due to sparse connections between devices. Still, we can expect significant advantages in terms of the volume of parameters uploaded through the system with a combination of EUT and LUT clusters, as depicted in Fig.~\ref{fig:treeDimension}. At the bottom layers where communication is mostly over the air, D2D can also be implemented through the out-band mode \cite{kar2018overview}. This has the additional advantage of not occupying the licensed spectrum, which results in bandwidth savings that can lead to an improved quality of service.%\ali{Decreasing the amount of data transfer between different network layers is of particular interest since those links mostly correspond to data transfer over long distances and in turn possess higher delays. I think we want to motivate why we are interested in that since D2D requires data transmission inside layers. ]}

We allow a cluster to switch between EUT and LUT over time. For instance, the connectivity between a fleet of miniature UAVs will vary as they travel, necessitating EUT when D2D is not feasible. To capture this dynamic, for each cluster ${C}$, at global iteration $k$, $\mathbbm{1}^{(k)}_{\{{C}\}}$ captures the operating mode, which is $1$ if the cluster operates in LUT mode, and $0$ otherwise.
In each LUT cluster, a node will only communicate with its neighboring devices, which may not include the whole cluster. Further, each node's neighborhood may evolve over time. For example, when the communications are conducted over the air as in Fig.~\ref{diag:scenario22}(\subref{fig:multi_stage}), the neighbors in one aggregation interval are identified based on the distances among the nodes and their transmit powers. We will explicitly consider such evolutions in cluster topology in our distributed consensus model in Sec.~\ref{sec:propCons}. 
% \ali{Though individual nodes may move between clusters, we assume that the total number of nodes inside each cluster does not evolve over time.}
%To simplify the presentation\nic{are you trying to simplify the presentation or the proof, or both?}, we will assume that the total number of nodes inside each cluster does not evolve over time\nic{what about their ID? Can a node move from a cluster to a different one?}, though our analysis can readily\nic{it is not clear to me how it can be "readily" extended. I would be careful making these statements. For instance, in this case, what assumptions do you need if you allow devices to switch clusters? Can this be done in a completely arbitrary way?} be extended to cases where devices switch between clusters.

%In each cluster that conducts D2D communications, each node only communicate with its neighboring devices, that may not include all the devices inside the cluster. The neighbors of a node are defined with respect to the characteristics of the encompassing layer. 

%  \begin{figure}[t]
% \includegraphics[width=3in,height=2in]{3Dfig.pdf}
% \centering
% \caption{.}
% \label{fig:3D}
% \vspace{-0.15in}
% \end{figure}
%The network architecture of FogL consists of multi-stages of data transfer among a set of clusters, where there is a network topology inside a portion of cluster using which the devices conduct D2D communications.

In the simple case where there is only one layer below the server, i.e., $|\mathcal{L}|=1$, consisting of a single EUT cluster, the FogL architecture reduces to federated learning. If instead there is just one layer of one LUT cluster with no server, FogL resembles fully distributed learning~\cite{hu2019decentralized,8950073,lee2021finite}. One of our contributions is developing and analyzing this generalized cluster-based multi-layer hybrid learning paradigm for FogL.
%However, even the simple case of $|\mathcal{L}|=1$ upon having multiple EUT and LUT clusters in the bottom layer (e.g., an edge  server handling multiple cellular BSs with local D2D communications inside the cellular cells), is an extension to the current literature, which to the best of our knowledge we are the first to investigate.

\begin{figure}[t]
\vspace{-.05mm}
\includegraphics[width=3.5in,height=1.2in]{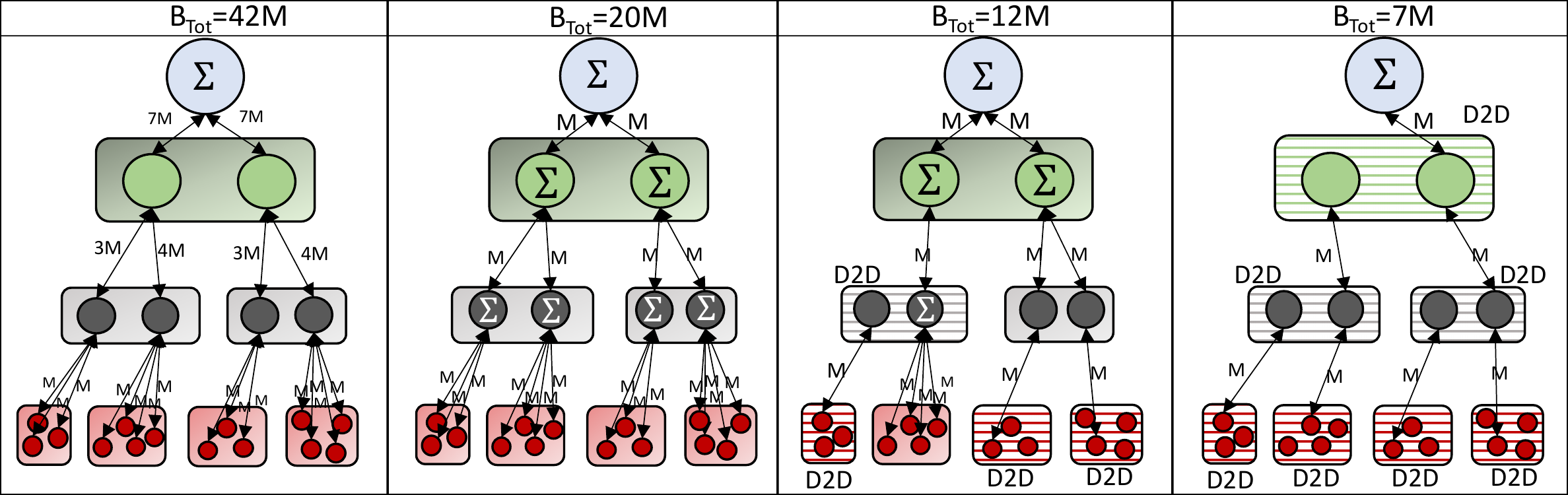}
\centering
\caption{Example of the network traffic reduction provided by multi-layer aggregations in FogL, where each device trains a model with parameter length $M$. The sum of the length of parameters transferred among the layers ($\textrm{B}_{\textrm{Tot}}$) is depicted at the top. (a) Network consisting of all EUT clusters, where the parent nodes upload all received parameters from their children. (b) Network consisting of all EUT clusters, where the parent nodes sum all received parameters and upload to the next layer. (c) Network with a portion of clusters in LUT mode, where each parent node only samples the parameters of one device after consensus formation. (d) Network consisting of all LUT clusters.
}\label{fig:treeDimension}
\end{figure}

%High-level illustration of the dimensionality reduction concept provided by multi-layer aggregations in FogL. The number of bits used to represent the original learning parameter vectors at each end device is assumed to be $G$. The total number of transferred bits among the layers is depicted on top of each plot. (Left): Network consists of all EUT clusters and the parent nodes offload all the received parameters from their children. (Middle): Network consists of all EUT clusters and the parent nodes sum all the received parameters together and then offload it to the higher layer. Regardless of the number of input to the middle nodes, the size of data transmitted upstream from them is $G$, reduced from the total input size to the node by a factor of the number of inputs owed to instant aggregations.
%(Right): A portion of clusters are assumed to work in LUT mode, where the devices obtain the average of their parameters in a distributed manner and the parent node only samples the parameter of one device, which is scaled by the number of devices at the parent node.}
% \label{fig:treeDimension}
% \vspace{-.1mm}
% \end{figure}
 \vspace{-2mm}
\subsection{Parameter Sharing vs. Gradient Sharing}
 \label{ssec:sharing}
 Note that the parameter update in~\eqref{eq:weightBaisc} can be written as
 \begin{equation}\label{eq:gradBasic}
 \begin{aligned}
        \mathbf{w}^{(k)}\hspace{-.2mm}&=\hspace{-.2mm} \frac{\sum_{n\in \mathcal{N}_{|\mathcal{L}|}} |\mathcal{D}_n|\left( \mathbf{w}^{(k-1)}_{n}- \beta \left(\nabla f_{n} (\mathbf{w}^{(k-1)}_{n})\right)\right)}{D} \hspace{-.1mm}\\&=\hspace{-.1mm}
         \mathbf{w}^{(k-1)}-\beta\frac{\sum_{n\in \mathcal{N}_{|\mathcal{L}|}} |\mathcal{D}_n| \nabla f_{n} (\mathbf{w}^{(k-1)}_{n})}{D}.
          \end{aligned}
 \end{equation}
%  \nm{and what about $\mathbf w_n^{(k-1)}=\mathbf w^{(k-1)}$?} \ali{It is used to get the equality.}
 This asserts that the global parameters can also be obtained via the gradients of the devices, implying that the devices can either share their gradients or their parameters during training. This equivalence arises from the one-step update in~\eqref{eq:localupdate}, which is a common assumption in federated learning~\cite{chen2019joint,8737464,zeng2020federated}. However, recent work~\cite{8664630,tu2020network} has advocated conducting multiple rounds of local updates between global aggregations to reduce communication costs; in this case, the parameters are required for each aggregation. In this paper, we focus on the more general case of parameter sharing, although we obtain one theoretical result (Proposition~\ref{prop:conv_0_dem_step}) based on gradient sharing.
 %, we make our study general and mainly focus on the sharing of the parameters unless specified, and discuss the differences in our subsequent analysis. We further leave the  analysis of FogL upon considering multiple rounds of local updates before one aggregation as a future work. 

In LUT clusters, devices leverage D2D communications to obtain an approximate value of the average of their parameters. A basic approach would be to implement a message passing algorithm where nodes exchange parameters with their neighbors until each node in the cluster has all parameters stored locally. Each node can then readily calculate the aggregated value, and one of them can be sampled by the parent node. Collecting a table of parameters at each node may not be feasible, however, given how large these vectors can be for contemporary ML models, as discussed in Sec.~\ref{ssec:ml}. Instead, we desire a technique that (i) does not require any node to store a table of all model parameters in the cluster, (ii) can be implemented in a distributed manner via D2D, and (iii) is generalizable across different network layers. In the following, we leverage \textit{distributed average consensus} methods for this.

\vspace{-2mm}
\section{\hspace{-1.27mm}Multi-Stage Hybrid Federated Learning \hspace{-1.1mm}({\tt MH-FL})}\label{sec:propCons}
\vspace{-0.7mm}
\noindent In this section, we develop our {\tt MH-FL} methodology  (Sec.~\ref{ssec:algo}\&\ref{ssec:traversing}). Then, we conduct a detailed performance analysis of our method (Sec.~\ref{ssec:analysis}). Finally, based on this analysis, we develop online control algorithms for tuning the number of D2D rounds in each cluster over time to guarantee the convergence properties for our method (Sec.~\ref{ssec:control}).
%Through the rest, unless specified, time instant is defined with respect to the count of global aggregations, where $k$-th time instant refers to the $k$-th global aggregations.

% All proofs have been deferred to the supplementary material.

\vspace{-3mm}
 \subsection{Distributed Average Consensus within Clusters}
 \label{ssec:algo}
 \vspace{-.5mm}
 Referring to Fig.~\ref{fig:treeD2D}, during each iteration $k$ of training, the LUT clusters engage in D2D communications, where each node desires estimating the average value of the parameters inside its cluster. For each cluster ${C}$ that is LUT during the $k$th global aggregation iteration (i.e., for which $\mathbbm{1}^{(k)}_{\{{C}\}}=1$), we let $G^{(k)}_{{C}}= \big(\mathcal{C}^{(k)}, \mathcal{E}^{(k)}_{{C}} \big)$ denote its communication graph. In this graph, $\mathcal{C}^{(k)}$ is the set of nodes belonging to the cluster, and there is an edge $(m,n) \in \mathcal{E}^{(k)}_{{C}}$ between nodes $m,n \in \mathcal{C}^{(k)}$ iff they can communicate via D2D during $k$. We assume that $G^{(k)}_{{C}}$ is undirected, connected, and static for the duration of one iteration $k$, although it may vary in different iterations. 
 Note that we have abstracted the physical layer wireless/wired communication medium of each LUT cluster $C$ to a  general graph topology $G_C^{(k)}$. In a wireless cluster, different channel conditions and communication configurations among the devices will manifest as different topologies for $G_C^{(k)}$.

%  \nm{what kind of topological knowledge of the graph do you need? Reviewers will argue that since the graph is time-varying across k, it may be difficult to acquire such knowledge and track changes over time..}\ali{In the algorithm, we will argue that we only need an upper bound on the spectral radius of consensus matrix. The consensus matrix may rely on the graph, e.g., it may need the maximum degree of the nodes.}

 %between the set of nodes in cluster $\mathcal{L}^{(k)}_{j,i}$ at the $k$-th global aggregation in which the existence of an edge between two nodes ${L}^{i}_m\in \mathcal{L}^{(k)}_{j,i}, {L}^{i}_n \in \mathcal{L}^{(k)}_{j,i}$ is assumed, i.e., $({L}^{i}_m, {L}^{i}_n )\in \mathcal{E}^{(k)}_{\mathcal{L}^{(k)}_{j,i}}$  if, at the time of $k$-th global aggregations, they can communicate using D2D communications. It is assumed that the network changes slowly so that during the $k$-th iteration, $G^{(k)}_{\mathcal{L}^{(k)}_{j,i}$ is fixed, $\forall i,j$.
 
 \begin{table*}[t]
\begin{minipage}{0.99\textwidth}
{\footnotesize
\begin{equation}\label{eq:firstWeightUpdate}
 \widehat{\mathbf{w}}^{(k)}_{n'}=\frac{\displaystyle \sum_{j
    \in \mathcal{Q}^{(k)}(n)}|\mathcal D_j| \mathbf{w}^{(k-1)}_{j}} {{|\mathcal{Q}^{(k)}(n)|}}-   \frac{\displaystyle \sum_{j
    \in \mathcal{Q}^{(k)}(n)}\beta |\mathcal D_j|\nabla f_{j}(\mathbf{w}^{(k-1)}_{j})} {|\mathcal{Q}^{(k)}(n)|}  +\mathbbm{1}^{(k)}_{\left\{{Q}(n)\right\}}  \mathbf{c}^{(k)}_{n'},~n' \in \mathcal{Q}^{(k)}(n)
    % \nm{why dont you combine w_{k-1}-beta*nabla f(w_{k-1}) as w_k?} \ali{I think this is more clear since we can easily see the weight and gradient transfer. Also it lets me go with a faster pace in the proofs by combining the weights.}
 \end{equation}
 \vspace{-1mm}
 \hrulefill
 \begin{equation}\label{eq:secondWeightUpdate}
  \hspace{-11mm}
 \begin{aligned}
 \widehat{\mathbf{w}}^{(k)}_{n'}\hspace{-1mm}=\frac{\displaystyle \sum_{i
    \in \mathcal{Q}^{(k)}(n)} \sum_{j
    \in \mathcal{Q}^{(k)}(i)} \hspace{-2.5mm}|\mathcal D_j| \mathbf{w}^{(k-1)}_{j}} {|\mathcal{Q}^{(k)}(n)|} - \frac{\displaystyle \sum_{i
    \in \mathcal{Q}^{(k)}(n)} \sum_{j
    \in \mathcal{Q}^{(k)}(i)} \hspace{-2.5mm}\beta |\mathcal D_j|\nabla f_{j}(\mathbf{w}^{(k-1)}_{j})} {|\mathcal{Q}^{(k)}(n)|} + \frac{\displaystyle\sum_{i
    \in \mathcal{Q}^{(k)}(n)}\hspace{-2.5mm}\mathbbm{1}^{(k)}_{\left\{{{Q}(i)}\right\}}|\mathcal{Q}^{(k)}(i)|  \mathbf{c}^{(k)}_{i'}}{|\mathcal{Q}^{(k)}(n)|}+\mathbbm{1}^{(k)}_{\left\{{Q}(n)\right\}}  \mathbf{c}^{(k)}_{n'}, n'\hspace{-.5mm} \in\hspace{-.5mm} \mathcal{Q}^{(k)}(n)
     \end{aligned}
     \hspace{-7mm}
     \vspace{-2mm}
 \end{equation}
 \vspace{-1mm}
\hrulefill
}
\end{minipage}\vspace{-6mm}
\end{table*}

We employ the family of \textit{linear distributed consensus} algorithms~\cite{xiao2004fast}, where during the global iteration $k$, the nodes inside LUT cluster ${C}$ conduct $\theta^{(k)}_{{C}}\in \mathbb{N}$
% \nm{who coordinates the number of rounds?}\ali{A footnote is added!}
\textit{rounds} of D2D (number of D2D rounds is a design parameter obtained in Sec.~\ref{ssec:analysis}). Each round of D2D consists of parameter transfers between neighboring nodes. Formally, during global iteration $k$, each node $n \in \mathcal{C}^{(k)}$  engages in the following rounds of iterative updates for $t=0,...,\theta^{(k)}_{{C}}-1$: 
\vspace{-2mm}
    \begin{equation}\label{eq:ConsCenter}
       \textbf{z}_{n}^{(t+1)}= v^{(k)}_{n,n} \textbf{z}_{n}^{(t)}+\hspace{-2mm} \sum_{m\in \mathcal{\zeta}^{(k)}(n)} v^{(k)}_{n,m}\textbf{z}_{m}^{(t)},
  \end{equation}
  
  \vspace{-3mm}
%   \begin{equation}\label{eq:ConsCenter}
%       \textbf{z}_{n}^{(t+1)}= \textbf{z}_{n}^{(t)}+d^{(k)}_{\mathcal{C}}\hspace{-2mm} \sum_{m\in \mathcal{N}_{(k)}(n)} (\textbf{z}_{m}^{(t)}-\textbf{z}_{n}^{(t)}),
%   \end{equation} 
  \noindent where $\textbf{z}_{n}^{(0)} = \mathbf{w}_{n}^{(k)}$ corresponds to node $n$'s initial parameter, and $\textbf{z}_{n}^{\big(\theta^{(k)}_{{C}}\big)}$ denotes the parameter after the D2D consensus process. $\mathcal{\zeta}^{(k)}(n)$ denotes the set of neighbors of node $n$ during global iteration $k$, and $v^{(k)}_{n,p}$, $p\in \{n\} \cup \mathcal{\zeta}^{(k)}(n)$
%   \nm{for convenience you can simply assume that $\mathcal{\zeta}^{(k)}(n)$ contains n as well} \ali{Yes. But writing as 7  may get a better perspective to the reader. Finally, we will work with the matrix form that encapsulates the neighbors in its elements.}
are the \textit{consensus weights} associated with node $n$  during $k$. 
There are several potential choices for these weights that guarantee convergence of the distributed consensus iterations, provided that the cluster graph $G^{(k)}_{{C}}$ is connected~\cite{xiao2004fast} (if it is not, we can partition the cluster into multiple connected subgraphs with the same parent node). We will detail the conditions required for convergence of local model aggregations in Assumption~\ref{assump:cons} of Sec.~\ref{ssec:analysis}. One choice that satisfies the conditions is {\small$\textbf{z}_{n}^{(t+1)} = \textbf{z}_{n}^{(t)}+d^{(k)}_{{C}}\sum_{m\in \mathcal{\zeta}^{(k)}(n)} (\textbf{z}_{m}^{(t)}-\textbf{z}_{n}^{(t)})$}, {\small$0 < d^{(k)}_{{C}} < 1 / D^{(k)}_{{C}}$}, where {\small$D^{(k)}_{{C}}$} is the maximum degree of the nodes in $G^{(k)}_{{C}}$~\cite{xiao2004fast}.
With this constant edge weight implementation, the nodes inside LUT cluster ${C}$ only need to have knowledge of the parameter $d^{(k)}_{{C}}$ to conduct local aggregations, which can be broadcast by the respective parent node. This choice of consensus weights is used in our simulations in Sec.~\ref{sec:num-res}.

% \ali{For instance, one common choice for these weights~\cite{xiao2004fast} gives $\textbf{z}_{n}^{(t+1)} = \textbf{z}_{n}^{(t)}+d^{(k)}_{{C}}\sum_{m\in \mathcal{\zeta}^{(k)}(n)} (\textbf{z}_{m}^{(t)}-\textbf{z}_{n}^{(t)})$, $0 < d^{(k)}_{{C}} < 1 / D^{(k)}_{{C}}$, where $D^{(k)}_{{C}}$ is the maximum degree of the nodes in $G^{(k)}_{{C}}$.
% With this implementation, the nodes inside LUT cluster ${C}$ only need to have the knowledge of the parameter $d^{(k)}_{{C}}$ to perform D2D communications, which can be broadcast by the respective parent node. This choice of consensus weights is used in our simulations in Sec.~\ref{sec:num-res}, and satisfies the desired properties that lead to the convergence of local model aggregations that we will describe in Assumption~\ref{assump:cons} in Sec.~\ref{ssec:analysis}}
Due to time constraints (i.e., depending on the required time between global aggregations), the number of D2D rounds cannot be arbitrarily large, and thus the nodes inside a cluster often do not have a perfect estimate of the average value of their parameters. In Sec.~\ref{ssec:analysis}, we analyze the effect of a finite number of D2D rounds on the {\tt MH-FL} performance.

% for general choices of {\small$v^{(k)}_{n,p}$}, {\small$p\in \{n\} \cup \mathcal{\zeta}^{(k)}(n)$}, $\forall n$. %, that satisfy a set of standard assumptions.
%\footnote{The choice of parameter $d^_{\mathcal{C}}^{(k)}$ can be further tuned, e.g., as in~\cite{xiao2004fast}, which is not the scope of this work.}
%(ideally infinity upon having a proper choice of consensus matrix to have the perfect average estimation of the cluster nodes parameters at each node)

\vspace{-4mm}
 \subsection{Local Aggregations and Parameters Propagation}
 \label{ssec:traversing}
In the following, we introduce a \textit{scaled parameter} for each node $n$ denoted by $\widetilde{\mathbf{w}}_n$ and develop an approach to perform global aggregations based on manipulation and relaying of these parameters across different layers of the network. The definition of $\widetilde{\mathbf{w}}_n$ depends on the layer where node $n$ is located.
 
 \subsubsection{Nodes' parameters in the bottom-most layer}
% We aim to approximate the summation in the numerator of~\eqref{eq:weightBaisc} using D2D communications in different network layers. To this end,
  Given $\mathbf{w}^{(k-1)}$, the end devices in layer ${L}_{|\mathcal{L}|}$ first perform the local update described in~\eqref{eq:localupdate}. Since the nodes' parameters go through multiple stages of aggregations (see Fig.~\ref{fig:treeDimension}), the server cannot recover the individual parameters from the aggregated ones to calculate~\eqref{eq:weightBaisc}. To overcome this issue, each device $n \in \mathcal{N}_{|\mathcal{L}|}$ obtains its scaled parameter as $\widetilde{\mathbf{w}}^{(k)}_{n}= |\mathcal D_n| \mathbf{w}^{(k)}_{n}$
%   \nm{I think you should use a different notation for this one.. when I see it I get confused with the equation $\mathbf{w}^{(k-1)}_{n}=\mathbf{w}^{(k-1)}$ and I forget that there is a local gradient update happening..reviewers will get confused as well} 
and shares it with its neighbors during the D2D process for global iteration $k$, i.e., its parameters weighted by its number of datapoints.\footnote{Nodes inside EUT clusters of layer ${L}_{|\mathcal{L}|}$ directly share their scaled parameters with their parents.} Using this weighting technique, the number of datapoints of the nodes is encoded in the multi-layer aggregations. Specifically, each node $n \in \mathcal{C}^{(k)}$ belonging to cluster ${C}$ located in ${L}_{|\mathcal{L}|}$ engages in the local iterations described by \eqref{eq:ConsCenter}, where $\textbf{z}_{n}^{(0)}=\widetilde{\mathbf{w}}^{(k)}_{n}$. Finally, the node stores {\small $\widehat{\mathbf{w}}_{n}^{(k)}=\textbf{z}_{n}^{\big(\theta^{(k)}_{{C}}\big)}$}, which corresponds to the final weighted local parameter value after the D2D process.
%   \nm{at this point, you should give some intuition on what $\widehat{\mathbf{w}}_{n}^{(k)}$ represents} \ali{Done!}

 \subsubsection{Nodes' parameters in the middle layers} Once D2D communications are finished in layer ${L}_{|\mathcal{L}|}$, each parent node $n\in \mathcal{N}_{|\mathcal{L}|-1}$ of a cluster that operated in LUT mode selects a cluster head $n'$ among its children $\mathcal{Q}^{(k)}(n)$ in layer ${L}_{|\mathcal{L}|}$ based on a selection/sampling distribution. This child $n'$ uploads its parameter vector $\widehat{\mathbf{w}}^{(k)}_{n'}$ to the parent node. The resulting sampled parameter vector at the parent node is given by~\eqref{eq:firstWeightUpdate}, where the first two terms correspond to the true average of the parameters of the nodes inside cluster ${Q}(n)$, and $\mathbf{c}^{(k)}_{n'} \in \mathbb{R}^{M}$ is the error arising from the consensus. This error is only applicable to the LUT clusters and is concerned with the deviation from the true cluster average (which would be obtained from an EUT cluster). The parent node $n \in \mathcal{N}_{|\mathcal{L}|-1}$  then computes its scaled parameter $\widetilde{\mathbf{w}}^{(k)}_{n}$ by scaling the received vector by the number of its children, and stores the corresponding vector:
 \vspace{-3mm}
 
 {\small
 \begin{equation}\label{eq:multiWeight}
    \hspace{-5mm}\widetilde{\mathbf{w}}^{(k)}_{n}\hspace{-.2mm}=|\mathcal{Q}^{(k)}(n)| \widehat{\mathbf{w}}^{(k)}_{n'}, n \hspace{-.3mm}\in\hspace{-.3mm} \mathcal{N}_{|\mathcal{L}|-1}, n' \hspace{-.3mm}\in\hspace{-.3mm} \mathcal{Q}^{(k)}(n), \mathbbm{1}^{(k)}_{\{{Q}(n)\}}=1.\hspace{-5mm}
 \end{equation}
 }
 
  \vspace{-4mm}
\noindent Also, in layer $\mathcal{L}_{|\mathcal{L}|-1}$, each parent node $n$ of a cluster that operated in EUT mode receives all the parameters of its children and computes $\widetilde{\mathbf{w}}^{(k)}_{n}=  \sum_{i \in \mathcal{Q}^{(k)}(n)} \widetilde{\mathbf{w}}^{(k)}_i$. Once this is completed, the LUT clusters located in layer $\mathcal{L}_{|\mathcal{L}|-1}$ engage in distributed consensus formation via cooperative D2D.

This procedure continues up the hierarchy at each layer ${L}_{j}$, $j = |\mathcal{L}|-2,...,1$. More precisely, for an LUT cluster ${C}$ located in one of the middle layers, each node $n \in \mathcal{C}^{(k)}$ performs the local iterations described by \eqref{eq:ConsCenter} with initialization $\textbf{z}_{n}^{(0)}={\widetilde{\mathbf{w}}}^{(k)}_{n}$.
% \nm{this is confusing... why so? IS ${\mathbf{w}}^{(k)}_{n}$ the actual ${\mathbf{w}}^{(k)}_{n}$ scaled by the number of nodes? IF that is the case, it is confusing and you need to use a proper notation to differentiate between scaled and unscaled versions} \ali{As defined in (8) it is in fact the scaled version.}
At the end of the D2D rounds, each of these nodes stores the last obtained parameter ${\widehat{\mathbf{w}}_{n}}^{(k)}=\textbf{z}_{n}^{\big(\theta^{(k)}_{{C}}\big)}$, and passes it up if sampled by its parent. At the same time, each node $n$ inside an EUT cluster located in one of the middle layers directly shares $\widetilde{\mathbf{w}}^{(k)}_n$ for the calculation of the local aggregation.
%, since those that are parent node of an EUT cluster store the error free summation of the parameters of their children that is captured by the indicator function. The same holds in the following, and we thus explain the procedure for the parent nodes of LUT clusters.
%
% Let $\left(\widehat{\mathbf{w}}_{L^{j}_p}^{(k)}\right)_m$ denote the $m$-th element of $\widehat{\mathbf{w}}_{L^{j}_p}^{(k)}$, $\forall p$,  the evolution of nodes parameters can be described by~\eqref{eq:consensus} with: 
%  \begin{equation}\label{eq:ConsInit2}
%       \left[\mathbf{Q}_{\mathcal{L}^{(k)}_{j,i}}^{(k)}\right]_m = \left[\left({\mathbf{w}}_{L^{j}_{\sum_{q=0}^{i-1} |\mathcal{L}_{j,q}| +1}}^{(k)}\right)_m,\cdots,\left({\mathbf{w}}_{L^{j}_{\sum_{q=0}^{i} |\mathcal{L}_{j,q}|}}^{(k)}\right)_m \right]^\top,
%   \end{equation}
%  and
%   \begin{equation}\label{eq:ConsAfter2}
%      \left[\widehat{\mathbf{Q}}_{\mathcal{L}^{(k)}_{j,i}}^{(k)}\right]_m=\left[\left(\widehat{\mathbf{w}}_{L^{j}_{\sum_{q=0}^{i-1} |\mathcal{L}_{j,q}| +1}}^{(k)}\right)_m,\cdots,\left(\widehat{\mathbf{w}}_{L^{j}_{\sum_{q=0}^{i} |\mathcal{L}_{j,q}|}}^{(k)}\right)_m \right]^\top.
%   \end{equation}

Traversing up the layers, the consensus  errors accumulate. For example,~\eqref{eq:secondWeightUpdate} gives the expression for the final consensus parameter vector of a node $n'\in \mathcal{Q}^{(k)}(n)$ inside an LUT cluster at layer ${L}_{|\mathcal{L}|-1}$, which will be sampled by a parent node $n\in \mathcal{N}_{|\mathcal{L}|-2}$ located in the upper layer. In this expression, $i'$ denotes the sampled node chosen by parent node $i$. %, using which  node $n$ stores the value  $\mathbf{w}^{(k)}_{n}=|\mathcal{Q}(n)| \widehat{\mathbf{w}}^{(k)}_{{n'}}$. As can be seen from this expression, the errors of approximations gets accumulated as traversing toward the higher layers of the hierarchy. 
%  \begin{equation}\label{eq:ConsInit1}
%       \left[\mathbf{Q}_{\mathcal{L}^{(k)}_{j,i}}^{(k)}\right]_m = \left[\left({\mathbf{w}}_{L^{j}_{\sum_{k=0}^{j-1} |\mathcal{L}_{j,k}| +1}}^{(k)}\right)_m,\cdots,\left({\mathbf{w}}_{L^{|\mathcal{L}|}_{\sum_{k=0}^{j} |\mathcal{L}_{j,k}|}}^{(k)}\right)_m \right]^\top,
%   \end{equation}
%  and
%   \begin{equation}\label{eq:ConsAfter1}
%      \left[\widehat{\mathbf{Q}}_{\mathcal{L}^{(k)}_{j,i}}^{(k)}\right]_m=\left[\left(\widehat{\mathbf{w}}_{L^{j}_{\sum_{k=0}^{j-1} |\mathcal{L}_{j,k}| +1}}^{(k)}\right)_m,\cdots,\left(\widehat{\mathbf{w}}_{L^{|\mathcal{L}|}_{\sum_{k=0}^{j} |\mathcal{L}_{j,k}|}}^{(k)}\right)_m \right]^\top.
%   \end{equation}

\subsubsection{Main server} Let $\widetilde{\mathbf{w}}^{(k)}_{0}$ denote the scaled parameter at the main server. If the cluster in layer ${L}_1$ operates in LUT mode, then $\widetilde{\mathbf{w}}^{(k)}_{0}= {|\mathcal{L}^{(k)}_{1,1}| \widehat{\mathbf{w}}^{(k)}_{m}}$,
% \nm{again, I think you need to differentiate scaled and unscaled versions with an appropriate notation; for instance w for unscaled and s for scaled} 
where $m \in \mathcal{L}^{(k)}_{1,1}$ denotes the node sampled by the main server with parameter $\widehat{\mathbf{w}}^{(k)}_{m}$ obtained after performing D2D rounds\footnote{According to FogL augmented graph properties, ${L}_1$ consists of one cluster,  and thus $|\mathcal{L}^{(k)}_{1,1}|=|\mathcal{N}_{1}|$ which is inherently time invariant. The super-index $k$ is added for consistency in calligraphic notations.}. Otherwise, the main server sums all the received parameters, i.e., $\widetilde{\mathbf{w}}^{(k)}_{0}= {\sum_{m \in \mathcal{L}^{(k)}_{1,1}} {\widetilde{\mathbf{w}}}^{(k)}_{m}}$. The main server then computes the global parameter vector as
\vspace{-1.5mm}
  \begin{equation}\label{eq:updateRootCons}
    \mathbf{w}^{(k)}=  {\widetilde{\mathbf{w}}^{(k)}_{{0}}}/{D},
 \end{equation}
which is then broadcast down the hierarchy to start the next global  round $k + 1$, beginning with local updates at the devices.

The {\tt MH-FL} methodology we developed throughout this section  is summarized in~Algorithm~\ref{alg:GenCons}. The lines beginning with ** and \#\# are enhancements for tuning the D2D rounds over time at different clusters, which we present in Sec.~\ref{ssec:control}.

\begin{algorithm}[t]
 	\caption{\small Multi-stage hybrid federated learning \hspace{-.6mm}({\tt MH-FL})}\label{alg:GenCons}
 	 {\scriptsize
 	\SetKwFunction{Union}{Union}\SetKwFunction{FindCompress}{FindCompress}
 	\SetKwInOut{Input}{input}\SetKwInOut{Output}{output}
    \Input{number of global aggregations $K$, default number of D2D rounds $\theta^{(k)}_{{L}_{j,i}}$ $\forall i,j,k$.}
    \Output{Final global model $\mathbf{w}^{(K)}$.}
    \textbf{Initial operations at main server:} Initialize the global parameter $\mathbf{w}^{(0)}$ and synchronize the edge devices with it.\\
    ** \textbf{Initial operations at main server:} If asymptotic convergence to optimal is desired: (i)  server randomly sets $\Vert\nabla\widetilde{F(\mathbf{w}^{0})}\Vert$ and broadcasts it, (ii) Server sets $\delta$ either arbitrarily or according to~\eqref{eq:deltaLinear} to guarantee a certain accuracy, and broadcasts it. Otherwise, the server sets D2D control parameters $\{\sigma_j\}_{j=1}^{|\mathcal{L}|}$ as described in Sec.~\ref{ssec:control} and broadcasts them. \label{alg:2innerstart} \\
    \For{$k=1$ to $K$}{
       \For{$l=|\mathcal{L}|$  down to $l=0$}{
       \If{$l={|\mathcal{L}|}$}{
       Given $\mathbf{w}^{(k-1)}$, each node $n$  obtains $\mathbf{w}^{(k)}_{n}$ using~\eqref{eq:localupdate}.\\
       Each node $n$ 
       obtains its scaled parameter $\widetilde{\mathbf{w}}^{(k)}_{n}=|\mathcal{D}_n|\mathbf{w}^{(k)}_{n}$.\label{alg:inner1} \\
  ** Each cluster operating in LUT mode runs Algorithm~\ref{alg:tuneCons}.\\
  \#\# Each cluster operating in LUT mode runs Algorithm~\ref{alg:tuneConsNN}.\\
  Nodes inside LUT clusters update their parameters using~\eqref{eq:ConsCenter}.\label{alg:inner2}\\
       }\ElseIf{$1\leq l \leq |\mathcal{L}|-1$}{
       Each parent node $n$ of an LUT cluster ${Q}(n)$ samples a child $n' \in \mathcal{Q}^{(k)}(n)$ and computes $\widetilde{\mathbf{w}}^{(k)}_{n}=|\mathcal{Q}^{(k)}(n)| \widehat{\mathbf{w}}^{(k)}_{n'}$. \\
          Each parent node $n$ of an LUT cluster ${Q}(n)$ uses the received parameters to compute $\widetilde{\mathbf{w}}^{(k)}_{n}=  \sum_{i \in \mathcal{Q}^{(k)}(n)} \widetilde{\mathbf{w}}^{(k)}_i$.\\
  ** Each cluster operating in LUT mode runs Algorithm~\ref{alg:tuneCons}.\\
  \#\# Each cluster operating in LUT mode runs Algorithm~\ref{alg:tuneConsNN}.\\
  Nodes inside LUT clusters update their parameters using~\eqref{eq:ConsCenter}.\label{alg:inner4}\\
       }
       \ElseIf{$l=0$}{
       \If{$\mathbbm{1}^{(k)}_{\left\{{{L}_{1,1}}\right\}}=1$}{
       \hspace{-3.5mm} The server  computes~$\widetilde{\mathbf{w}}^{(k)}_{0}\hspace{-.8mm}= \hspace{-.8mm}{|\mathcal{L}^{(k)}_{1,1}| \widehat{\mathbf{w}}^{(k)}_{m}}$, where $m \hspace{-.8mm}\in\hspace{-.8mm} \mathcal{L}^{(k)}_{1,1}$.
       }\Else{
        \hspace{-3mm}The server computes $\widetilde{\mathbf{w}}^{(k)}_{0}= {\sum_{m \in \mathcal{L}^{(k)}_{1,1}} {\widetilde{\mathbf{w}}}^{(k)}_{m}}$.
        \vspace{-1mm}
       }
      The server computes ${\mathbf{w}}^{(k)}$  using~\eqref{eq:updateRootCons} and broadcasts it.\label{alg:inner5}\\
          ** If asymptotic convergence to optimal is desired, the main server approximates the gradient of the loss function used for the next iteration as in Sec.~\ref{ssec:control} and broadcasts it.
        }
       }
    }
    }
 \end{algorithm}
 
 \vspace{-4mm}
\subsection{Theoretical Analysis of {\tt MH-FL}}\label{ssec:analysis}
One of the key contributions of {\tt MH-FL} is the integration of D2D communications with multi-layer parameter transfers. As discussed, D2D communications are conducted through time varying topology structures among the nodes, which introduce a network dimension to model training. Studying this effect under limited D2D rounds regime is the main theme of our theoretical analysis. Before presenting our main results,
we first introduce a few assumptions and a definition. Henceforth, $\Vert.\Vert$ denotes the 2-norm unless otherwise stated.

%:\footnote{The strong convexity can be achieved by adding a regularization term to convex loss functions, which is a common practice~\cite{roux2012stochastic}.}
 \vspace{-1mm}
\begin{assumption}\label{assum:genConv}
The global ML loss function \eqref{eq:globlossinit} has the following properties: (i) $\mu$-strong convexity, i.e., $F(\mathbf{y})\geq F(\mathbf{x})+(\mathbf{y-x})^\top \nabla F(\mathbf{x}) +\frac{\mu}{2}\norm{\mathbf{y}-\mathbf{x}}^2$ for some $\mu > 0,~\forall \mathbf{x},\mathbf{y}$, and (ii) $\eta$-smoothness, i.e, $\norm{\nabla F(\mathbf{x})-\nabla F(\mathbf{y})}\leq \eta \norm{\mathbf{x}-\mathbf{y}}$ for some $\eta > \mu,~\forall \mathbf{x},\mathbf{y}$.
\end{assumption}
\vspace{-1mm}
The above properties are common assumptions in federated learning and ML literature~\cite{zeng2020federated,chen2019joint,reisizadeh2019fedpaq,dinh2019federated,friedlander2012hybrid}. Commonly encountered ML models with convex loss functions are linear regression, logistic regression, squared SVM, and single layer neural networks with convex activation functions. In practice, these models are implemented with an additional regularization term to improve the convergence and avoid model overfitting, which makes them strongly convex~\cite{friedlander2012hybrid}. We will conduct our convergence analysis based on this assumption and design control algorithms in Sec.~\ref{ssec:control} for both the convex (Algorithm~\ref{alg:tuneCons}) and non-convex (Algorithm~\ref{alg:tuneConsNN}) cases.

We will also find it useful to write the consensus algorithm in matrix form. Letting $\widetilde{\mathbf{W}}^{(k)}_{{C}} \in \mathbb{R}^{|\mathcal{C}^{(k)}| \times M}$ denote the matrix of scaled  parameters across all nodes in an LUT cluster ${C}$ prior to consensus in iteration $k$, the evolution of the nodes' parameters described by~\eqref{eq:ConsCenter} can be written as
\vspace{-1.5mm}
  \begin{equation}\label{eq:consensus}
      \widehat{\mathbf{W}}^{(k)}_{{C}}= \left(\mathbf{V}^{(k)}_{{C}}\right)^{\theta^{(k)}_{{C}}} \widetilde{\mathbf{W}}^{(k)}_{{C}},
      \vspace{-1mm}
  \end{equation}
  \vspace{-3mm}
  
\noindent where 
% $\left[\widetilde{\mathbf{W}}^{(k)}_{{C}}\right]_m$ is the $m$th column of $\widetilde{\mathbf{W}}^{(k)}_{{C}}$, $m = 1,...,M$, corresponding to the $m$th scaled parameter across nodes in the cluster, and
$\widehat{\mathbf{W}}^{(k)}_{{C}}\hspace{-1.5mm} \in \hspace{-1mm} \mathbb{R}^{|\mathcal{C}^{(k)}| \times M}$ denotes the matrix of node parameters after the consensus, and $\mathbf{V}^{(k)}_{{C}}=[v^{(k)}_{n,m}]_{n,m\in \mathcal{C}^{(k)}}$ is the consensus matrix applied to the parameter vector to realize \eqref{eq:ConsCenter}.

\vspace{-1.5mm}
\begin{assumption}\label{assump:cons}
The consensus matrix $\mathbf{V}^{(k)}_{{C}}$ for each LUT cluster ${C}$ has the following properties~\cite{xiao2004fast,xiao2007distributed}: (i) $\left(\mathbf{V}^{(k)}_{{C}}\right)_{m,n}=0~~\textrm{if}~~ \left({m},{n}\right)\notin \mathcal{E}^{(k)}_{{C}}$, (ii) $\mathbf{V}^{(k)}_{{C}}\textbf{1} = \textbf{1}$, (iii) $\mathbf{V}^{(k)}_{{C}} = {\mathbf{V}^{(k)}_{{C}}}^\top$, and (iv)~$ \rho \left(\mathbf{V}^{(k)}_{{C}}-\frac{\textbf{1} \textbf{1}^\top}{|\mathcal{C}^{(k)}|}\right) \leq \lambda^{(k)}_{{C}} < 1$, where $\textbf{1}$ is the vector of 1s and $\rho(\mathbf{A})$ is the spectral radius of matrix $\mathbf{A}$.
\end{assumption}
%maximum magnitude of eigenvalues
\vspace{-1mm}
In Assumption~\ref{assump:cons}, $\lambda^{(k)}_{{C}}$ can be interpreted as an upper bound on the spectral radius, which plays a key role in our results. 
\vspace{-1.5mm}
\begin{definition}\label{def:clustDiv}
The divergence of parameters in cluster ${C}$ at iteration $k$, denoted by $\Upsilon^{(k)}_{{C}}$, is defined as an upper bound on the difference of its nodes' scaled parameters as follows:
\vspace{-1mm}\begin{equation}
  \big\Vert\widetilde{\mathbf{w}}^{(k)}_{{q}} -\widetilde{\mathbf{w}}^{(k)}_{{q'}} \big\Vert \leq \Upsilon^{(k)}_{{C}} , ~\forall {q},{q'}\in\mathcal{C}^{(k)}.
\end{equation} 
%   ${L}_{|\mathcal{L}|}$: $\Vert\widetilde{\mathbf{w}}^{(k)}_{{q}} -\widetilde{\mathbf{w}}^{(k)}_{{q'}} \Vert_2 \leq \Upsilon^{(k)}_{{C}} , \forall {q},{q'}\in\mathcal{C}^{(k)}$
% , (ii) if $C$ is located at ${L}_{j}$, $1\leq j\leq {|\mathcal{L}|}-1$, $\Vert\mathbf{w}^{(k)}_{{q}} -\mathbf{w}^{(k)}_{{q'}} \Vert_2 \leq \hspace{-.1mm}\Upsilon^{(k)}_{{C}} , \forall {q},{q'}\in \mathcal{C}^{(k)}$.
% \begin{equation}\label{eq:divergenceCluster}
%     \hspace{-5mm}
% \begin{cases}
%     \Vert\mathbf{w}^{(k)}_{{q}} -\mathbf{w}^{(k)}_{{q'}} \Vert_2 \leq \hspace{-.1mm}\Upsilon^{(k)}_{{C}} , \forall {q},{q'}\in \mathcal{C}^{(k)},\\ \hspace{30mm}\textrm{{$C$}~located at}~{L}_{j}, 1\leq j\leq {|\mathcal{L}|}-1,\\
%     \Vert\widetilde{\mathbf{w}}^{(k)}_{{q}} -\widetilde{\mathbf{w}}^{(k)}_{{q'}} \Vert_2 \leq \Upsilon^{(k)}_{{C}} , \forall {q},{q'}\in\mathcal{C}^{(k)}, {C}~\textrm{located at}~{L}_{|\mathcal{L}|}.  
%     \end{cases}
%      \hspace{-8mm}
% \end{equation}
%\nic{why do tou need to distinguish the two cases? Notation with and w/o tilde is confusing!} \ali{Defining it this way, the nodes at the last layer do grad descent and multiply their weights by their data points to get the tilde version. This is the parameter that they actually share and we actually care about. For the rest of the nodes in the upper layers it is the weight they receive and scale by their number of children and save from the bottom layer. If we eliminate tilde, it may look like the parameters that bottom nodes receive from the server.}
\end{definition}
\vspace{-.5mm}
\begin{table*}[t]
\vspace{-1mm}
\begin{minipage}{0.99\textwidth}
{\small
\begin{align}\label{eq:ConsTh1}
     \hspace{-15mm}
    %  \begin{aligned}
     &F(\mathbf{w}^{(k)})-F(\mathbf{w}^{*})\leq \underbrace{ \left(1-\frac{\mu}{\eta}\right)^{k}\left(F(\mathbf{w}^{(0)})- F(\mathbf{w}^*)\right) }_{\textrm{(a)}}+ \frac{\eta{\Phi}}{2D^2} \sum_{t=0}^{k-1}
     \left(1-\frac{\mu}{\eta}\right)^{t}\Xi^{(k-t)}
     \end{align}    \vspace{-2mm}
     \hrulefill
        %   \vspace{-5mm}
     \begin{equation}\label{eq:Xi}
    \begin{aligned}
    &\Xi^{(k-t)} = \underbrace{\sum_{{a_{1}}
    \in \mathcal{L}^{(k-t)}_{{1},{1}}} \sum_{a_{2}
    \in \mathcal{Q}^{(k-t)}(a_{1})}\cdots  \sum_{a_{|\mathcal{L}|-1}
    \in \mathcal{Q}^{(k-t)}({a_{|\mathcal{L}|-2}})}}_{(b)} \underbrace{\mathbbm{1}^{(k-t)}_{ \left\{{Q}(a_{|\mathcal{L}|-1})\right\}}|\mathcal{Q}^{(k-t)}(a_{|\mathcal{L}|-1})|^3
    \left(\lambda^{(k-t)}_{{Q}(a_{|\mathcal{L}|-1})}\right)^{2\theta^{(k-t)}_{{Q}(a_{|\mathcal{L}|-1})}} \left(\Upsilon^{(k-t)}_{{Q}(a_{|\mathcal{L}|-1})}\right)^2}_{(c)}
    \\&+\sum_{{a_{1}}
    \in \mathcal{L}^{(k-t)}_{{1},{1}}} \sum_{a_{2}
    \in \mathcal{Q}^{(k-t)}(a_{1})} \cdots  \sum_{a_{|\mathcal{L}|-2}
    \in \mathcal{Q}^{(k-t)}({a_{|\mathcal{L}|-3}})} \mathbbm{1}^{(k-t)}_{ \left\{{Q}(a_{|\mathcal{L}|-2})\right\}}|\mathcal{Q}^{(k-t)}(a_{|\mathcal{L}|-2})|^3  \left(\lambda^{(k-t)}_{{Q}(a_{|\mathcal{L}|-2})}\right)^{2\theta^{(k-t)}_{{Q}(a_{|\mathcal{L}|-2})}} \left(\Upsilon^{(k-t)}_{{Q}(a_{|\mathcal{L}|-2})}\right)^2 \\&+\cdots+\sum_{a_1
    \in \mathcal{L}^{(k-t)}_{1,1}} \mathbbm{1}^{(k-t)}_{ \left\{{Q}(a_1)\right\}}|\mathcal{Q}^{(k-t)}(a_1)|^3 \left(\lambda^{(k-t)}_{{Q}(a_1)}\right)^{2\theta^{(k-t)}_{{Q}(a_1)}} \left(\Upsilon^{(k-t)}_{{Q}(a_1)}\right)^2+\mathbbm{1}^{(k-t)}_{\left\{{{L}_{1,1}}\right\}}|\mathcal{L}^{(k-t)}_{1,1}|^3 \left(\lambda^{(k-t)}_{{L}_{1,1}}\right)^{2\theta^{(k-t)}_{\mathcal{L}_{1,1}}} \left(\Upsilon^{(k-t)}_{{L}_{1,1}}\right)^2
 \end{aligned}
      \hspace{-3mm}
  \vspace{-7mm}
 \end{equation}
}
\vspace{-4mm}
\hrulefill
\end{minipage}
\vspace{-3mm}
\end{table*}
% The above defined quantity has some similarity to the divergence of gradients in ML literature (e.g., see Definition 1 of~\cite{8664630}). 
At the bottom-most layer, the above defined quantity is indicative of the degree of data heterogeneity (i.e., the level of non-i.i.d) among the nodes in a cluster, whereas in the upper layers it captures the heterogeneity of data contained in sub-trees with their roots being the nodes in the cluster. We show in Theorem \ref{theo:consbase} how it impacts the convergence bound, and in the subsequent results how it dictates the number of D2D rounds. Then, in Sec.~\ref{ssec:control}, we develop control algorithms that approximate the divergence in a distributed manner at every cluster, and use it to control the convergence rate.
\subsubsection{General convergence bound}
In the following theorem, we study the convergence of {\tt MH-FL} (see Appendix~\ref{app:ConsGen}):

\vspace{-2mm}
\begin{theorem}\label{theo:consbase}
With a learning rate $\beta=1/\eta$, 
% for any sampling distribution among the nodes and an arbitrary number of consensus rounds performed at different clusters of different layers, 
after $k$ global iterations of any realization of {\tt MH-FL}, an upper bound on $F(\mathbf{w}^{(k)}) - F(\mathbf{w}^{\star})$ is given in~\eqref{eq:ConsTh1}, where $\Phi = N_{{|\mathcal{L}|-1}}+N_{{|\mathcal{L}|-2}}+\cdots+N_{{1}}+1$ is the total number of nodes in the network besides the bottom layer and $\Xi^{(k-t)}$ is given by~\eqref{eq:Xi}.
\end{theorem}
% in our online technical report~\cite{ourTechRep}.
\vspace{-3.5mm}
\begin{remark}
Each nested sum in~\eqref{eq:Xi}, e.g., (b),  encompasses the nodes located in a path starting from the main server and ending at a node in one of the layers. Different nested sums capture paths with different lengths. Also, each summand, e.g., (c), corresponds to the characteristics of the child cluster, e.g., $Q(a_{|\mathcal{L}|-1})$ in (c), of the node in the last index of its associated nested sum, e.g., $a_{|\mathcal{L}|-1}$ in (b). This contains the operating mode, number of nodes, upper bound of spectral radius, number of D2D rounds, and divergence of parameters.
\end{remark}
\vspace{-1.5mm}
\noindent \textbf{Main takeaways.} The bound in~\eqref{eq:ConsTh1},~\eqref{eq:Xi} quantifies how the convergence is dependent on several learning and system parameters. In particular, we see a dependence on (i) the characteristics of the loss function (i.e., $\eta, \mu$), (ii) the number of nodes and clusters at each network layer (through the $|\mathcal{Q}|$ terms), (iii) the topology and characteristics of the communication graph among the nodes inside the clusters, captured via the spectral radius bounds (i.e., $\lambda$), (iv) the number of D2D rounds performed at each cluster (i.e., the $\theta$), and (v) the divergence among the node parameters at each cluster (i.e., $\Upsilon$). Given a fixed set of parameters at iteration $k - 1$, we can observe that increasing the number of D2D rounds at each cluster in iteration $k$ results in a smaller bound (since $\theta$ is appeared as the exponent of $\lambda<1$), and thus a better expected model loss, as we would expect. Furthermore, for a fixed number of D2D rounds, a smaller spectral radius, corresponding to a better connected cluster, results in a smaller bound. On the other hand, larger parameter divergence results in a worse bound. 

Term (a) in \eqref{eq:ConsTh1} corresponds to the case with no consensus error in the system, i.e., when all LUT clusters have {\small $\theta^{(k)}_{{C}} \rightarrow \infty$} (infinite D2D rounds) or when the network consists of all EUT clusters. Since {\small $1-{\mu}/{\eta} < 1$}, the overall rate of convergence of {\tt MH-FL} is at best linear with rate {\small $1 - \mu/\eta$}.  However, achieving this would incur prohibitively long delays, motivating us to study the effects of the number of D2D rounds. Also, note that the terms {\small $1-\big({\mu}/{\eta}\big)^t$} inside the summation have a dampening effect: at global iteration $k$, the errors from global iteration $t < k$ are multiplied by {\small $1-\big({\mu}/{\eta}\big)^{k-t}$}, meaning the initial errors for $t \ll k$ are dampened by very small coefficients, while the final errors have a more pronounced effect on the bound. At first glance, this seems to suggest that at higher global iteration indices, more D2D rounds are needed to reduce the errors. However, especially upon having i.i.d datasets, we can expect the parameters of the end devices to become more similar to one another with increasing global iteration count, which in turn would decrease the divergence (i.e., $\Upsilon$) within clusters over time. Further, the connectivity of the clusters (captured by $\lambda$) will change at different layers of the hierarchy, causing the spectral radius to vary. This motivates us in Sec.~\ref{ssec:control} to consider adapting the D2D rounds over two dimensions: time (i.e., global iterations) and space (i.e., network layers).

\vspace{-.1mm}
\subsubsection{Asymptotic optimality}
  We now explicitly connect the number of D2D rounds performed at different clusters with the asymptotic optimality of {\tt MH-FL} (see Appendix~\ref{app:boundedConsConv}).
  
  %In the following, we aim to explicitly draw the connection between the number of consensus rounds performed at different network layers and the convergence behavior of FogL. To this end, we first propose a policy that dynamically tunes the number of consensus at different clusters of different layers that  results in a bounded upper bound of convergence in Proposition~\ref{prop:boundedConsConv}. Then, in Proposition~\ref{prop:FogConv}, we demonstrate that by further adjustments of the number of consensus, FogL can exhibit linear convergence to the optimal solution of the problem under finite number of consensus rounds performed at different clusters.
  \vspace{-1.5mm} 
  \begin{proposition}\label{prop:boundedConsConv}
  For any realization of {\tt MH-FL}, if the number of D2D rounds at different clusters in different layers of the network satisfies the following criterion ($ \forall k,i,j$):
  
  \vspace{-2mm}
  {\small
  \begin{equation}\label{eq:IterFogLCons}
  \hspace{-19mm}
  \begin{cases}
     \theta^{(k)}_{{L}_{j,i}} \hspace{-1mm}\geq\hspace{-1mm}\left\lceil{\hspace{-.5mm}\frac{\log\left({\sigma_{j}}\right)-2\log\left({\big\vert\mathcal{L}^{(k)}_{j,i}\big\vert^{\frac{3}{2}} \Upsilon^{(k)}_{{L}_{j,i}}} \right)}{2\log\left(\lambda^{(k)}_{{L}_{j,i}} \right)}\hspace{-.5mm}}\right\rceil\hspace{-1mm}, &\hspace{-2mm} \textrm{if}~ \sigma_{j}\leq {\big\vert\mathcal{L}^{(k)}_{j,i}\big\vert^3 \left(\Upsilon^{(k)}_{{L}_{j,i}}\right)^2}\\
      \theta^{(k)}_{{L}_{j,i}}\hspace{-1mm}\geq 0, &\hspace{-2mm} \textrm{otherwise}
     \end{cases}
     \hspace{-15mm}
  \end{equation} 
  }
  \vspace{-3mm}
  
  \noindent for non-negative constants $\sigma_1,\cdots,\sigma_{|\mathcal{L}|}$, then the asymptotic upper bound on the distance from the optimal is given by
  \vspace{-1.5mm}
  \begin{equation}\label{eq:asympGapFinite}
   \lim_{k\rightarrow \infty} F(\mathbf{w}^{(k)})-F(\mathbf{w}^{*})\leq \frac{\eta^2{\Phi}}{2\mu D^2} \sum_{j=0}^{|\mathcal{L}|-1}\sigma_{j+1} N_j.
  \end{equation}
  \end{proposition}
  \vspace{-1.5mm}
%   \begin{proof}
% . 
% %   in our online technical report~\cite{ourTechRep}.
%   \end{proof}
  
  %We first propose a policy that tunes the consensus rounds in different clusters to guarantee a specific asymptotic upper bound on convergence (Proposition~\ref{prop:boundedConsConv}). We then propose a stricter policy which
  Proposition~\ref{prop:boundedConsConv} gives a guideline for designing the number of D2D rounds at different network clusters over time to achieve a desired (finite) upper bound on the optimality gap. It can be seen that optimality is tied to our introduced auxiliary variables $\{\sigma_j\}_{j=1}^{|\mathcal{L}|}$, which we refer to as \textit{D2D control parameters}. 
  
  To eliminate the existence of a constant optimality gap in~\eqref{eq:asympGapFinite}, we obtain an extra condition on the tuning of these parameters and make them time-varying to guarantee a linear convergence to the optimal solution (see Appendix~\ref{app:ConLinCons}):

\vspace{-1mm}
 \begin{proposition} \label{prop:FogConv}
For any realization of {\tt MH-FL}, suppose that the number of D2D rounds  at different clusters of different network layers satisfies the following criterion ($\forall k,i,j$):
 \vspace{-2mm}
 
 {\small
 \begin{equation}\label{eq:IterFogLCons2}
  \hspace{-19mm}
  \begin{cases}  
   \hspace{-.5mm}  \theta^{(k)}_{{L}_{j,i}}\hspace{-.5mm}\hspace{-1mm} \geq\hspace{-1mm} \hspace{-.5mm}\left\lceil{\hspace{-.5mm} \frac{\log\left({\sigma^{(k)}_{j}}\right)-2\log\left({\big\vert\mathcal{L}^{(k)}_{j,i}\big\vert^{\frac{3}{2}} \Upsilon^{(k)}_{{L}_{j,i}}} \right)}{2\log\left(\lambda^{(k)}_{{L}_{j,i}} \right)} }\hspace{-.5mm}\right\rceil\hspace{-1mm}, & \hspace{-3mm}\textrm{if}~ \sigma^{(k)}_{j}\hspace{-.5mm}\leq\hspace{-.5mm} {\big\vert\mathcal{L}^{(k)}_{j,i}\big\vert^3\hspace{-1mm} \left(\hspace{-.2mm}\Upsilon^{(k)}_{{L}_{j,i}}\hspace{-.2mm}\right)^2}\\
         \hspace{-.5mm} \theta^{(k)}_{{L}_{j,i}}\hspace{-1mm}\geq 0, &\hspace{-2mm} \textrm{otherwise}
     \end{cases}
     \hspace{-15mm}
  \end{equation}
  }
  \vspace{-2mm}
  
  \noindent where the non-negative constants $\sigma^{(k)}_1,\cdots,\sigma^{(k)}_{|\mathcal{L}|}$ satisfy
  \vspace{-2.5mm}
\begin{equation}\label{eq:ineqWeightsofNodes}
  \sum_{j=0}^{|\mathcal{L}|-1}\sigma^{(k)}_{j+1} N_j\leq  \frac{D^2\mu  ({\mu}-\delta\eta)}{\eta^4{\Phi}} \norm{\nabla F(\mathbf{w}^{(k-1)})}^2
  \vspace{-2.5mm}
\end{equation}
\vspace{-3mm}

\noindent for $0<\delta\leq \mu/\eta$. Then, we have
\vspace{-1.2mm}
\begin{equation}\label{eq:FogLinCon}
   \hspace{-3mm} {F(\mathbf{w}^{(k+1)})-F(\mathbf{w}^*)}{} \leq \hspace{-.5mm}(1-\delta)\hspace{-.5mm}\left(\hspace{-.5mm}F(\mathbf{w}^{(k)})-F(\mathbf{w}^*)\hspace{-.5mm}\right)\hspace{-.7mm},~\hspace{-.5mm}\forall k,\hspace{-3mm}
   \vspace{-3mm}
\end{equation}
which implies a linear convergence of {\tt MH-FL} and $\lim_{k\rightarrow \infty}{F(\mathbf{w}^{(k)})-F(\mathbf{w}^*)}=0$.
\end{proposition}
\vspace{-1mm}
% \begin{proof} 
% % in our online technical report~\cite{ourTechRep}.
% \end{proof}

%Tuning with respect to the conditions given in Propositions~\ref{prop:boundedConsConv} or~\ref{prop:FogConv} result in an adaptation of consensus rounds over time and space. 
\begin{table*}[t]
\begin{minipage}{0.99\textwidth}

\end{minipage}
\noindent\begin{minipage}{.65\linewidth}
{\scriptsize
 \begin{equation}\label{eq:kGapCons}
     \hspace{-8mm}\kappa\geq \left\lceil{\frac{\log\left({\epsilon-\frac{\eta^2{\Phi}}{2 \mu D^2}\sum_{j=0}^{|\mathcal{L}|-1}\sigma_{j+1} N_j}\right) -\log\left( {F(\mathbf{w}^{(0)})-F(\mathbf{w}^*)-\frac{\eta^2{\Phi}}{2\mu D^2}\sum_{j=0}^{|\mathcal{L}|-1}\sigma_{j+1} N_j} \right)}{\log\left( 1-\mu/\eta\right)}}\right\rceil
     \hspace{-4mm}
\end{equation}
}
\end{minipage}%
\begin{minipage}{.33\linewidth}
\vspace{1mm}
 {\scriptsize
 \begin{equation}\label{cor:klieaner}
\kappa\geq \left\lceil{\frac{\log(\epsilon)-\log(F(\mathbf{w}^{(0)})-F(\mathbf{w}^{*}))}{\log(1-\delta)}}\right\rceil \hspace{-4mm}
\end{equation} }
\end{minipage}
\vspace{-.1mm}
\begin{minipage}{.99\linewidth}
 \hrulefill
\end{minipage}
\vspace{-8mm}
\end{table*}

Proposition~\ref{prop:FogConv} asserts that under a stricter tuning of the number of D2D rounds at different layers, i.e., \eqref{eq:IterFogLCons2} and~\eqref{eq:ineqWeightsofNodes}, convergence to the optimal solution can be guaranteed with a rate that is at most $1 - \mu/\eta$ according to~\eqref{eq:FogLinCon}. Furthermore, the number of D2D rounds are always finite when all the D2D tuning variables are greater than zero;  according to~\eqref{eq:ineqWeightsofNodes}, this can be always satisfied until reaching the optimal point (the gradient of loss becomes zero) where the algorithm will stop. Note that Proposition~\ref{prop:FogConv}'s condition boosts the required $\theta^{(k)}_{{L}_{j,i}}$ over time, since the norm of the gradient in \eqref{eq:ineqWeightsofNodes} decreases over time, in turn lowering the values of $\{\sigma^{(k)}_j\}_{j=1}^{|\mathcal{L}|}$. In Proposition~\ref{prop:boundedConsConv}, by contrast, the D2D rounds will become \textit{tapered} over time (i.e, they will diminish over time), since $\sigma_j$ is fixed over $k$ and the divergence of the parameters is expected to decrease over global iterations, especially when dealing with i.i.d data. Our experiments in Sec.~\ref{sec:num-res} verify these effects.

Proposition~\ref{prop:FogConv}'s result also assumes knowledge of the global loss gradient $\norm{\nabla F(\mathbf{w}^{(k-1)})}$, which is not known at the beginning of global iteration $k$, where $\mathbf{w}^{(k-1)}$ is just sent down through the  layers. In Sec.~\ref{ssec:control}, we will develop an approximation technique for implementing this result in practice. Finally, note that in both Propositions~\ref{prop:boundedConsConv}\&\ref{prop:FogConv}, a  smaller spectral radius (corresponding to well connected clusters) is tied to a lower number of D2D rounds (note that $\lambda<1$).

%The main difference between these two results is that tuning the consensus rounds with respect to Proposition~\ref{prop:boundedConsConv} results in a tapering of the consensus rounds over time, since the divergence of the parameters usually decreases over global iterations. However, tuning the consensus rounds with respect to Proposition~\ref{prop:FogConv} results in boosting the number of consensus over time since the norm of the gradient in \eqref{eq:ineqWeightsofNodes} decreases over time and goes to zero, which results in decrease in the values of $\{\sigma^k_j\}_{j=1}^{|\mathcal{L}|}$ over time, which in general results in increasing the number of consensus. Further illustrations are given through numerical simulations.
% Also, since $0\leq \delta\leq \mu/\eta$, the maximum rate of convergence to the optimal solution of the problem is achieved when $\delta= \mu/\eta$, resulting in the rate of convergence $1-\mu/\eta$. Considering~\eqref{eq:ineqWeightsofNodes}, this convergence rate results in $\sigma^{k}_j=0$, $\forall k,j$, according to~\eqref{eq:IterFogLCons} implying infinite rounds of consensus at all the LUT clusters. 

\subsubsection{Relationship between global iterations and D2D rounds}
The following two corollaries to Propositions~\ref{prop:boundedConsConv}\&\ref{prop:FogConv} investigate the impact of the number of global iterations on the required D2D rounds, and vice versa. First, we obtain the number of D2D rounds required at different clusters to reach a desired accuracy in a desired global iteration (see Appendix~\ref{app:numIterCertainAccur}):

%In the following two corollaries, we further investigate the convergence characteristics of FogL from two interrelated perspectives. First, we assume that a certain accuracy is desired for the FogL at a certain global iteration. In this case, we obtain the sufficient number of consensus rounds at different clusters of the network to satisfy the requirements. Second, we assume that a certain accuracy is desired for the FogL given a set of fixed consensus rounds at different network clusters. For this case,  we obtain the sufficient number of global iterations to achieve the desired accuracy.
\vspace{-1.5mm}
\begin{corollary}
\label{cor:numIterCertainAccur}
Let {\small$\epsilon \in \big[(1-\frac{\mu}{\eta})^{\kappa}\big(F(\mathbf{w}^{(0)})-F(\mathbf{w}^{*})\big), \; F(\mathbf{w}^{(0)})-F(\mathbf{w}^{*}) \big)$}. To guarantee that {\tt MH-FL} obtains a solution to within $\epsilon$ of the optimal by global iteration $\kappa$, i.e., $F(\mathbf{w}^{(\kappa)}) - F(\mathbf{w}^{*}) \leq \epsilon$, a sufficient number of D2D rounds in different clusters of the network is given by either of the following conditions:
\begin{enumerate}[leftmargin=5mm]
    \item $\theta^{(k)}_{{L}_{j,i}}$, $\forall i,j,k$, given by~\eqref{eq:IterFogLCons}, where the values of $\sigma_1,\cdots,\sigma_{|\mathcal{L}|}$ satisfy the following inequality:
    \vspace{-1.5mm}
    \begin{equation}\label{eq:coeffgap}
    \hspace{-3mm}\sum_{j=0}^{|\mathcal{L}|-1} \sigma_{j+1} N_j \leq  {\frac{\epsilon-
     \left({1-\mu/\eta}\right)^{\kappa}(F(\mathbf{w}^{(0)})- F(\mathbf{w}^*)) }{ \left({1-(1-\mu/\eta)^{\kappa}}\right)\frac{\eta^2{\Phi}}{2\mu D^2} }}.
     \hspace{-2mm}
    \end{equation}
 \item $\theta^{(k)}_{{L}_{j,i}}$, $\forall i,j,k$, given by~\eqref{eq:IterFogLCons2}, where the values of $\sigma^{(k)}_1,\cdots,\sigma^{(k)}_{|\mathcal{L}|}$ satisfy~\eqref{eq:ineqWeightsofNodes} with $\delta$ given by
 \vspace{-1.5mm}
\begin{equation}\label{eq:deltaLinear}
\delta\geq 1-\sqrt[\leftroot{-3}\uproot{3}\kappa]{\frac{\epsilon}{F(\mathbf{w}^{(0)})-F(\mathbf{w}^*)}}.
\end{equation}
\end{enumerate}
\end{corollary}
% \begin{proof}

% % in our online technical report~\cite{ourTechRep}.
% \end{proof}
\vspace{-.8mm}
Second, we obtain the number of global iterations required to reach a desired accuracy for a predetermined policy of determining the D2D rounds in different clusters (see Appendix~\ref{app:numGlobIterCertainAccur}):

\vspace{-1.5mm}
\begin{corollary}
\label{cor:numGlobIterCertainAccur}
With $\epsilon$ as in Corollary~\ref{cor:numIterCertainAccur}, either of the following two conditions give a sufficient number of global iterations $\kappa$ to achieve $F(\mathbf{w}^{(\kappa)}) - F(\mathbf{w}^{*}) \leq \epsilon$:
\begin{enumerate}[leftmargin=5mm]
    \item If the $\theta^{(k)}_{{L}_{j,i}}$, $\forall i,j,k$, satisfy~\eqref{eq:IterFogLCons} given $\sigma_1,\cdots,\sigma_{|\mathcal{L}|}$, and $\epsilon\geq \frac{\eta^2{\Phi}}{2\mu D^2} \sum_{j=0}^{|\mathcal{L}|-1}\sigma_{j+1} N_j$, then $\kappa$ follows~\eqref{eq:kGapCons}.

 \item If the $\theta^{(k)}_{{L}_{j,i}}$, $\forall i,j,k$, satisfy~\eqref{eq:IterFogLCons2} and~\eqref{eq:ineqWeightsofNodes} given  $\sigma^{(k)}_1,\cdots,\sigma^{(k)}_{|\mathcal{L}|}$ and $\delta$, then $\kappa$ follows~\eqref{cor:klieaner}.
\end{enumerate}
\end{corollary}
% \begin{proof}

% % in our online technical report~\cite{ourTechRep}.
% \end{proof}

% We will use these corollaries in Sec.~\ref{ssec:control} to design control algorithms for tuning {\tt MH-FL}.

\subsubsection{Varying gradient step size}
% The gradient update~\eqref{eq:localupdate} is based on a constant step size $\beta$. 
% So far, we have been working with a constant step size $\beta$ in~\eqref{eq:localupdate}. 
If we design a time-varying step size $\beta_k$ that is decreasing over time, we can sharpen the convergence result in Proposition~\ref{prop:boundedConsConv}, when devices share gradients instead of parameters (see Sec.~\ref{ssec:sharing}). In particular,  we can guarantee that {\tt MH-FL} converges to the optimal solution, rather than having a finite optimality gap (see Appendix~\ref{app:conv_0_dem_step}):
%In the following, we investigate the effect of decreasing the gradient step size and demonstrate that it can guarantee the convergence to the optimal solution of the problem when the constant step size of the gradient leads to a finite (non-zero) upper bound of convergence (Proposition~\ref{prop:boundedConsConv}). Note that this technique is particularly effective when the nodes share their gradients instead of their parameters \ali{and each parent node samples the children nodes uniformly at random}:
\vspace{-2mm}
\begin{proposition}\label{prop:conv_0_dem_step}
Suppose that the nodes share gradients using the same procedure described in Algorithm~\ref{alg:GenCons}, and that each parent node samples one of its children uniformly at random. Also, assume that end devices use a step size $\beta_k=\frac{\alpha}{k+\lambda}$, where $\lambda>1$ and $\alpha>1/\mu$ at global iteration $k$, with $\beta_0\leq 1/\eta$. If the number of D2D rounds are performed according to~\eqref{eq:IterFogLCons} with non-negative constants $\sigma_1,\cdots,\sigma_{|\mathcal{L}|}$, we have
\vspace{-1mm}
 \begin{equation}\label{eq:taperGradStep}
  \mathbb{E} [F(\mathbf{w}^{(k)})-F(\mathbf{w}^*)]\leq \frac{\Gamma}{k+\lambda},
 \end{equation}
 \vspace{-4mm}
 
\noindent where 
\vspace{-1mm}
\small{
\begin{equation}\label{eq:gammadef}
\hspace{-4mm}\Gamma\hspace{-.5mm} =\hspace{-.5mm} \max \Bigg\{\hspace{-.5mm}\lambda(F(\mathbf{w}^{(0)})-F(\mathbf{w}^*)),  \frac{\eta\alpha^2 {{\Phi}}\sum_{j=0}^{|\mathcal{L}|-1} \sigma_{j+1} N_j }{2D^2(\alpha \mu -1)}\hspace{-.5mm}\Bigg\}.\hspace{-2mm}
\end{equation}
}
\vspace{-3mm}

\normalsize
\noindent Consequently, under such conditions, {\tt MH-FL} converges to the optimal solution: $
  \lim_{k\rightarrow \infty} \mathbb{E}[ F(\mathbf{w}^{(k)})-F(\mathbf{w}^*)]=0.
$
\end{proposition}
% \begin{proof}

% % in our online technical report~\cite{ourTechRep}.
% \end{proof}
\vspace{-1.2mm}
The bound in~\eqref{eq:taperGradStep} implies a rate of convergence of $O(1/k)$, which is slower than the linear convergence obtained in Proposition~\ref{prop:FogConv}, but also allows tapering of the D2D rounds over time as in Proposition~\ref{prop:boundedConsConv}. %Using the bound~\eqref{eq:taperGradStep}, the results of Corollaries~\ref{cor:numIterCertainAccur} and~\ref{cor:numGlobIterCertainAccur} can be extended to the case where a decreasing gradient step size is used too.

\subsubsection{Cluster sampling}
 In a large-scale network with millions of nodes, it may be desirable to reduce upstream communications even further than what is provided by LUT clusters. We develop a \textit{cluster sampling} technique where a portion of the clusters are activated in model training at each global iteration in Appendix~\ref{app:cluster}, %in~\cite{ourTechRep}
 and extend Theorem~\ref{theo:consbase} to this case. 
 We leave further investigation of this technique to future work.

 \vspace{-4mm}
\subsection{Control Algorithms for Distributed Consensus Tuning}
\label{ssec:control}

With all else constant, fewer rounds of D2D results in lower power consumption and network load among the devices in LUT clusters. Motivated by this, we develop control algorithms for {\tt MH-FL} that tune the number of D2D rounds through time (global aggregations) and space (network layers). %to minimize the consensus rounds required while retaining specific convergence properties.

\subsubsection{Adaptive D2D for loss functions satisfying Assumption~\ref{assum:genConv}} 
We are motivated to realize the two D2D consensus round tuning policies that we obtained in Propositions~\ref{prop:boundedConsConv} and~\ref{prop:FogConv}, which we refer to as Policies A and B, respectively. Policy A will provide a finite optimality gap, with tapering of the D2D rounds through time, while Policy B will provide linear convergence to the optimal, with boosting of the D2D rounds through time. We are interested in realizing these two policies in a distributed manner, where the number of D2D rounds for different device clusters are tuned by the corresponding parent nodes in real-time. It is assumed that parent node of  ${C}$ has an estimate on the topology of the cluster, and thus an upper-bound on the spectral radius of its children cluster graph $\lambda^{(k)}_{{C}}$, $\forall k$.

According to \eqref{eq:IterFogLCons} and \eqref{eq:IterFogLCons2}, for both policies, tuning of the D2D rounds for cluster ${C}$ requires  knowledge of the divergence of parameters~$\Upsilon^{(k)}_{C}$. Also, Policy A requires a set of fixed D2D control parameters $\sigma_{j}$ for clusters located in layer ${L}_{j}$, while Policy B requires the global gradient of the broadcast weight $\norm{\nabla F(\mathbf{w}^{(k-1)})}$ and the real-time D2D control parameters $\sigma^{(k)}_{j}$. In the following, we first derive the divergence of parameters in a distributed manner. Then, we focus on realizing the other specific parameters for each policy.

Given Definition~\ref{def:clustDiv}, at cluster ${C}$ in layer ${L}_{j}$, $1\leq j\leq |\mathcal{L}|$, the divergence of the parameters can be approximated as\footnote{Here, for practical purposes, we use the lower bound of divergence $\big|~\Vert\textbf{a}\Vert-\Vert\textbf{b}\Vert~\big|\leq \Vert\textbf{a}-\textbf{b}\Vert$. The upper bound alternative $\Vert\textbf{a}-\textbf{b}\Vert\leq \Vert\textbf{a}\Vert+\Vert\textbf{b}\Vert$ can be arbitrarily large even when $\textbf{a}=\textbf{b}$.}
\vspace{-1.2mm}
\begin{equation}\label{eq:appDiv}
\hspace{-22mm}
\begin{aligned}
 &\Upsilon^{(k)}_{{C}} \approx \max_{q, q' \in\mathcal{C}^{(k)}} \big\{ \Vert \widetilde{\mathbf{w}}^{(k)}_{q}\Vert- \Vert \widetilde{\mathbf{w}}^{(k)}_{q'} \Vert \big\} \\& =\underbrace{\max_{q\in\mathcal{C}^{(k)}} \Vert\widetilde{\mathbf{w}}^{(k)}_{q}\Vert}_{(a)}-   \underbrace{\min_{q'\in\mathcal{C}^{(k)}} \Vert\widetilde{\mathbf{w}}^{(k)}_{q'} \Vert}_{(b)}.
 \end{aligned}
 \hspace{-15mm}
\end{equation}
\vspace{-3.1mm}

\noindent To obtain (a) and (b) in a distributed manner, at any given LUT cluster, each node $n\in\mathcal{C}^{(k)}$ first computes the scalar value $\Vert\widetilde{\mathbf{w}}^{(k)}_{n}\Vert$.
% \footnote{If ${C}$ is in ${L}_{|\mathcal{L}|}$, $\Upsilon^{(k)}_{{C}} \approx \max_{q \in \mathcal{C}^{(k)}} \Vert\widetilde{\mathbf{w}}^{(k)}_{q} \Vert - \min_{q' \in \mathcal{C}^{(k)}} \Vert\widetilde{\mathbf{w}}^{(k)}_{q'} \Vert$, and thus each node $n$ performs the described procedure with $\Vert\widetilde{\mathbf{w}}^{(k)}_{n}\Vert$.} 
Nodes then share these scalar values with their neighbors iteratively. In each iteration, each node saves two scalars: the (i) maximum and (ii) minimum values among the received values and the node's current value. It is easy to verify that for any given communication graph ${G}^{(k)}_{{C}}$ among the cluster nodes, once the number of iterations has exceeded the diameter of ${G}^{(k)}_{{C}}$, the saved values at each node will correspond to (a) and (b) for cluster ${C}$ in~\eqref{eq:appDiv}. The parent node can then sample the value of one of its children to compute~\eqref{eq:appDiv}.

\begin{algorithm}[t]
 	\caption{\small Adaptive D2D round tuning at each cluster}\label{alg:tuneCons}
 	 {\footnotesize
 	\SetKwFunction{Union}{Union}\SetKwFunction{FindCompress}{FindCompress}
 	\SetKwInOut{Input}{input}\SetKwInOut{Output}{output}
    \Input{Global aggregation count $k$, tuning parameter $\omega>1$, cluster index ${C} = {L}_{j,i}$.
    % the desired rate of convergence to the optimal solution of the problem $\delta$,
    }
    \Output{Number of D2D rounds $\theta^{(k)}_{{L}_{j,i}}$ for the cluster.}
    Nodes inside the cluster ${C}$ iteratively compute (a) and (b) of~\eqref{eq:appDiv}. \label{alg:2innerfirst}\\
    Parent node of cluster samples one child and computes~\eqref{eq:appDiv}.\\
        \If{ asymptotic convergence to optimal desired}{
         Parent node uses~\eqref{eq:ineqWeightsofNodesApp} with stored $\Vert\widetilde{\nabla F(\mathbf{w}^{(k-1)})}\Vert$ and $\delta$ received from the server to compute  $\sigma^{(k)}_j$. \\
         Parent node uses~\eqref{eq:IterFogLCons2} to compute  $\theta^{(k)}_{{L}_{j,i}}$. \\
         }\Else{
         Parent node uses the received consensus tuning parameter $\sigma^{(k)}_j$ from the server in~\eqref{eq:IterFogLCons} to compute $\theta^{(k)}_{{L}_{j,i}}$.
         }
    }
 \end{algorithm}

For Policy A, since the values of $\{\sigma_j\}_{j=1}^{|\mathcal{L}|}$ are fixed through time, one option is for the server to tune them once at the beginning of training and distribute them among all the nodes. If satisfaction of a given accuracy $\epsilon$ at a certain iteration $\kappa$ is desired, we use the result of Corollary~\ref{cor:numIterCertainAccur} and obtain the D2D control parameters as the solution of the following max-min optimization problem: $\argmax_{\{\sigma_{j}\}_{j=1}^{|\mathcal{L}|}} ~\min{\{ {N_{j-1}\sigma_{j}}\}}$ subject to \eqref{eq:coeffgap}. It can be verified that the solution is given by
\vspace{-1.2mm}
\begin{equation}\label{eq:coeffgapApp}
\hspace{-2mm}
 \sigma_{j}^{\star}  =  {\frac{\epsilon-
     \left({1-\mu/\eta}\right)^{\kappa}(F(\mathbf{w}^{(0)})- F(\mathbf{w}^*)) }{ \left({1-(1-\mu/\eta)^{\kappa}}\right)\frac{\eta^2{\Phi}}{2\mu D^2}N_{j-1}|\mathcal{L}| }}, ~1\leq j\leq |\mathcal{L}|,
     \hspace{-2mm}
\end{equation}
\vspace{-3mm}

\noindent which can be broadcast at the beginning of training among the nodes. 
The reason behind the choice of the aforementioned max-min problem is two-fold. First, according to~\eqref{eq:IterFogLCons}, for a given set of divergence of parameters $\Upsilon^{(k)}_{{C}}$ across ${C}$, fewer numbers of D2D rounds at each layer ${L}_{j}$ is associated with larger values of $\sigma_{j}$, so larger values of D2D control parameters are often desired. Second, this choice of objective function results in smaller values of D2D control parameters as we move down the layers (towards the end devices) and the number of nodes increases. This leads to larger D2D rounds and lower errors in the bottom layers, which is desired in practice given the discussion in Sec.~\ref{ssec:traversing} that the errors from the bottom layers are propagated and amplified as we move up the layers.

For Policy B, to obtain $\Vert\nabla F (\mathbf{w}^{(k-1)})\Vert$, we use~\eqref{eq:gradBasic} to approximate $\nabla F (\mathbf{w}^{(k-2)})$ as $\nabla F (\mathbf{w}^{(k-2)})\approx \left({\mathbf{w}^{(k-2)}-\mathbf{w}^{(k-1)}}\right)/{\beta}$. This is an approximation due to the error introduced in the consensus process. Using this, the main server estimates $\Vert\nabla F (\mathbf{w}^{(k-1)})\Vert$ via $\widetilde{\Vert\nabla F (\mathbf{w}^{(k-1)})\Vert}=\frac{1}{\omega} \Vert\nabla F (\mathbf{w}^{(k-2)})\Vert$, where we introduce tuning parameter $\omega>1$ based on the intuition that the norm of the gradient should be decreasing over $k$, and broadcasts it along with $\mathbf{w}^{(k-1)}$. The choice of $\omega$ can be viewed as a tradeoff between the number of global aggregations $k$ and the number of D2D rounds $\theta^{(k)}_{{C}}$ per aggregation: as $\omega$ increases, we tolerate less consensus error in~\eqref{eq:ineqWeightsofNodes}, requiring more D2D rounds $\theta^{(k)}_{{C}}$ and fewer global iterations $k$ to achieve an accuracy. 
Then, the cluster heads obtain the $\sigma^{(k)}_j$, $\forall {j},k$, according to the following max-min problem: $\argmax_{\{\sigma^{(k)}_{j}\}_{j=1}^{|\mathcal{L}|}} ~\min{\{ {N_{j-1}\sigma^{(k)}_{j}}\}}$ subject to \eqref{eq:ineqWeightsofNodes} for a given $\delta$. It can be verified that the solution is given by
%Furthermore, for Policy B, given the similar behaviour of~\eqref{eq:IterFogLCons2} to~\eqref{eq:IterFogLCons}, we obtain the real-time consensus control parameters for all the clusters belonging to layer $\mathcal{L}_{j}$, i.e., $\sigma^{(k)}_j$, via solving the following max-min problem: $\displaystyle\argmax_{\{\sigma^{(k)}_{j}\}_{j=1}^{|\mathcal{L}|}} ~\min{\{ {N_{j-1}\sigma^{(k)}_{j}}\}}$ subject to \eqref{eq:ineqWeightsofNodes} for a given $\delta$:
\begin{equation}\label{eq:ineqWeightsofNodesApp}
\hspace{-3mm}
 {\sigma^{(k)}_{j}}^{\star}\hspace{-2mm} =\hspace{-1mm} \frac{D^2\mu  ({\mu}-\delta\eta)}{\eta^4{\Phi}N_{j-1}|\mathcal{L}|} \Vert\widetilde{\nabla F (\mathbf{w}^{(k-1)})}\Vert^2, \hspace{.1mm}1\hspace{-.1mm}\leq\hspace{-.1mm} j\leq \hspace{-.1mm}|\mathcal{L}|,\forall k.\hspace{-3mm}
\end{equation}
The parameter $\delta$ can be tuned by the main server at the beginning of training to guarantee a desired linear convergence, or it can be tuned by \eqref{eq:deltaLinear} to satisfy a desired accuracy at a certain global iteration. In both cases, the server broadcasts this parameter among the nodes at the beginning of training. With $\delta$ and $\Vert\nabla F (\mathbf{w}^{(k-1)})\Vert$ in hand, along with the ML model characteristics ($D,\mu,\eta$) and networked related parameters ($\Phi,|\mathcal{L}|,N_{j-1}$) that can be once broadcast by the server, all the parent nodes of the clusters can calculate \eqref{eq:ineqWeightsofNodesApp} at each global iteration, which then can be used in~\eqref{eq:IterFogLCons2} to tune the number of D2D rounds for the children nodes in real-time.

 \begin{algorithm}[t]
 	\caption{\small Adaptive D2D round tuning at each cluster for non-convex loss functions}\label{alg:tuneConsNN}
 	 {\footnotesize
 	\SetKwFunction{Union}{Union}\SetKwFunction{FindCompress}{FindCompress}
 	\SetKwInOut{Input}{input}\SetKwInOut{Output}{output}
    \Input{Tolerable error of aggregations $\psi$, global aggregation count $k$, cluster index ${C} = {L}_{j,i}$.}
    \Output{Number of D2D rounds $\theta^{(k)}_{{L}_{j,i}}$ for the cluster.}
    Nodes inside the cluster iteratively compute (a) and (b) of~\eqref{eq:appDiv}. \label{alg:2innerfirst}\\
    Parent node of cluster samples one child and computes~\eqref{eq:appDiv}.\\
    Parent node of the cluster computes the required rounds of D2D as follows with $\sigma_j  =  {{\psi D^2}/({ {{\Phi}}N_{j-1}|\mathcal{L}| }})$:
    \vspace{-1.5mm}
    \begin{equation}\label{eq:IterFogLConsAlg}
  \hspace{-17mm}
  \begin{cases}
    \hspace{-.5mm} \theta^{(k)}_{{L}_{j,i}} \hspace{-1mm}\geq\hspace{-1mm} \frac{\log\left(\sigma_j \right)-2\log\left({\big\vert\mathcal{L}^{(k)}_{j,i}\big\vert^{\frac{3}{2}} \Upsilon^{(k)}_{{L}_{j,i}}} \right)}{2\log\left(\lambda^{(k)}_{{L}_{j,i}} \right)}, & \hspace{-2mm}\textrm{if}~ \sigma_j\leq {\big\vert\mathcal{L}^{(k)}_{j,i}\big\vert^3 \left(\Upsilon^{(k)}_{{L}_{j,i}}\right)^2} \\
        \hspace{-.5mm} \theta^{(k)}_{{L}_{j,i}}\hspace{-1mm}\geq 0, &\hspace{-2.5mm} \textrm{otherwise}.
     \end{cases}
     \hspace{-10mm}
     \vspace{-5mm}
  \end{equation}
       \vspace{-2mm}
    }
 \end{algorithm}
A summary of this procedure for tuning the D2D rounds at a cluster is given in Algorithm~\ref{alg:tuneCons}. In the full {\tt MH-FL} method described in Algorithm~\ref{alg:GenCons}, this is (optionally) called for each cluster in the  lines marked via $**$. 

%The pseudo-code of the FogL using adaptive consensus round tuning is given by Algorithm~\ref{alg:GenCons}, where lines marked via $**$ are considered during the training, which refer to Algorithm~\ref{alg:tuneCons}. The above procedure would result in distributed adaptive tuning of the number of consensus rounds through time and space, which will be later revealed through numerical simulations.

%Also, the value of $F(\mathbf{w}^*)$ in the above equation can be estimated using the loss function used.

\subsubsection{Adaptive D2D tuning for non-convex loss functions}
Some contemporary ML models, such as neural networks, possess non-convex loss functions for which Assumption~\ref{assum:genConv} does not apply. In these cases, we can develop a heuristic approach to tune the D2D rounds of {\tt MH-FL} if we specify a maximum tolerable error of aggregations $\psi$ at each global iteration. The resulting procedure is given in Algorithm~\ref{alg:tuneConsNN}, which is called once for each cluster in Algorithm~\ref{alg:GenCons} in place of Algorithm~\ref{alg:tuneCons} (see the lines marked \#\#). To execute this, prior to the start of training, each parent node should receive $\psi$ and the number of nodes in its layer. Using Algorithm~\ref{alg:tuneConsNN}, we can show that the 2-norm of aggregation errors is always bounded by parameter $\psi$ (see Appendix~\ref{app:aggrErr} for the proof).
%~\cite{ourTechRep}). %In the next section, we will experimentally validate the performance of both Algorithms~\ref{alg:tuneCons} and~\ref{alg:tuneConsNN}.

%Neural networks, in general, possess complex loss functions that cannot be categorized as (strong) convex functions. So far, no tractable convergence bound for distributed  ML using NNs has been proposed. We were also not able to obtain such bounds for FogL. However, given the importance of NNs and their wide applications, we propose a distributed (heuristic) approach to design the consensus rounds of FogL upon using loss functions of NNs.  Our approach cannot guarantee a certain accuracy at a given global iteration, instead, it can guarantee a certain error of multi-stage aggregations conducted through the hierarchy of FogL denoted by parameter $\psi$. To execute our algorithm, prior to the start of training, each parent node should receive the value of $\psi$ and the number of nodes belonging to its layer.  We evaluate the performance of our algorithm through extensive simulations and demonstrate its effectiveness. It is easy to verify that using Algorithm~\ref{alg:tuneConsNN}, the  norm-2 of the error of the aggregations  is always bounded by parameter $\psi$ (see Appendix~\ref{app:aggrErr}). The pseudo-code of the FogL using adaptive consensus round tuning upon using NNs is given by Algorithm~\ref{alg:GenCons}, where lines marked via \#\# are considered during the training, which refer to Algorithm~\ref{alg:tuneConsNN}.

\vspace{-3.2mm}
\section{Experimental Evaluation}
\label{sec:num-res}
\noindent We conducted extensive numerical experiments to evaluate {\tt MH-FL}.
In this section, we present the setup 
% (Sec.~\ref{ssec:setup}) 
and results 
% (Sec.~\ref{ssec:results}) 
for a popular dataset and a sample fog topology. Additional results on more datasets and topologies can be found in Appendix~\ref{app:extraSim}.
% of our online technical report~\cite{ourTechRep}.
\begin{figure*}[t]
\centering
\begin{minipage}{.47\textwidth}
     \centering
     \includegraphics[width=\linewidth]{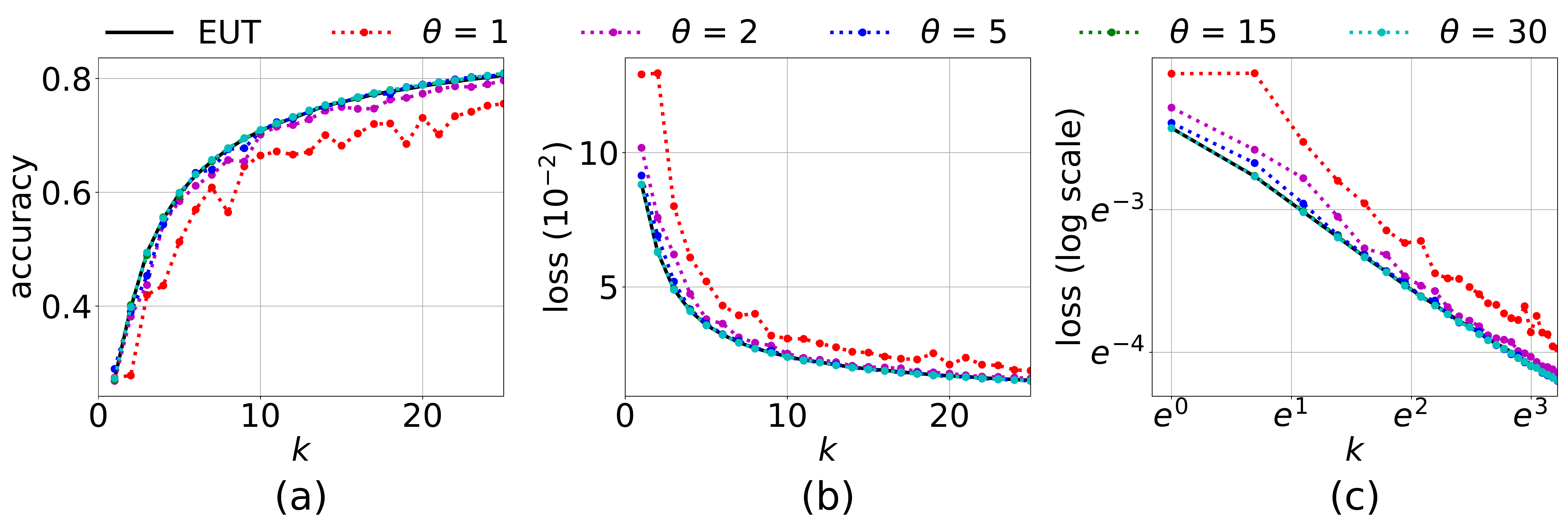}
     \caption{Performance comparison between baseline EUT and {\tt MH-FL} when a fixed number of D2D rounds $\theta$ is used at every cluster in the network, for non-i.i.d data. As the number of D2D rounds increases, {\tt MH-FL} performs more similar to the EUT baseline and the learning becomes more stable.}
     \label{fig:GenFigGoodIntuition_MNIST_125}
\end{minipage}%
\hspace{4mm}
\begin{minipage}{.47\textwidth}
  \centering
         \includegraphics[width=\linewidth]{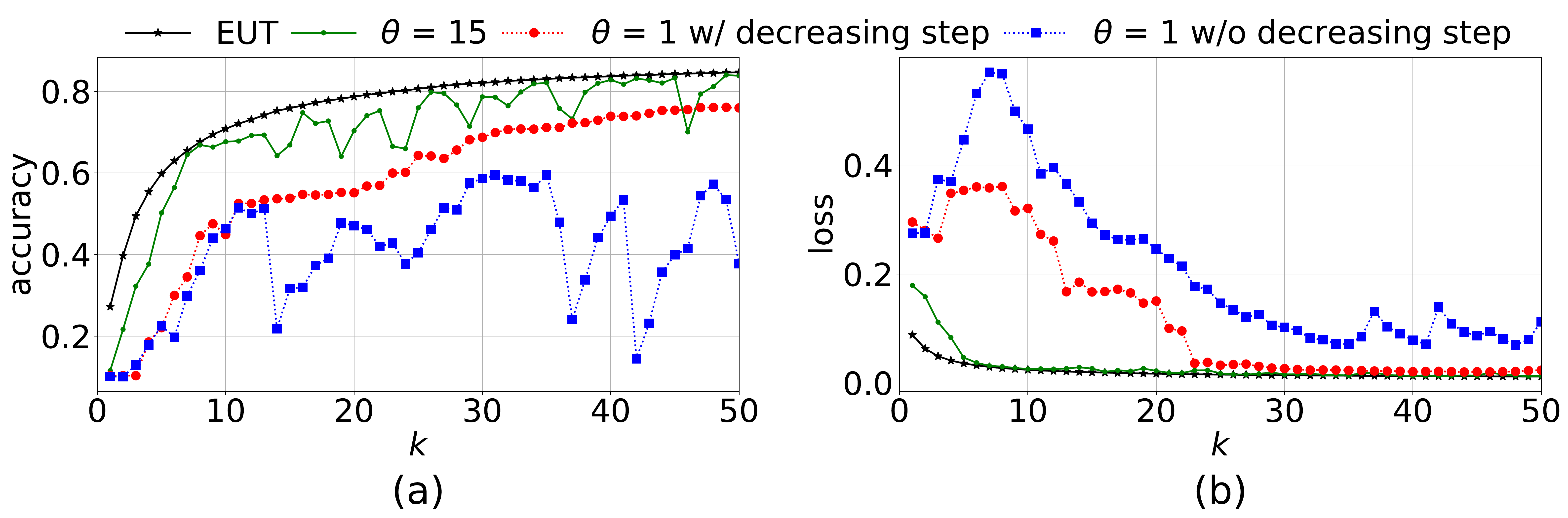}
     \caption{Performance comparison between baseline EUT, and {\tt MH-FL} with and without (w/o) decreasing the gradient descent step size. Decreasing the step size can provide convergence to the optimal solution in cases where a fixed step size is not capable, but also has a slower convergence speed.}
     \label{fig:decaying_learning_rate_MNIST_125}
\end{minipage}
 \vspace{-4.5mm}
\end{figure*}

\begin{figure*}[t]
\centering
\begin{minipage}{.47\textwidth}
        \centering
     \includegraphics[width=\linewidth]{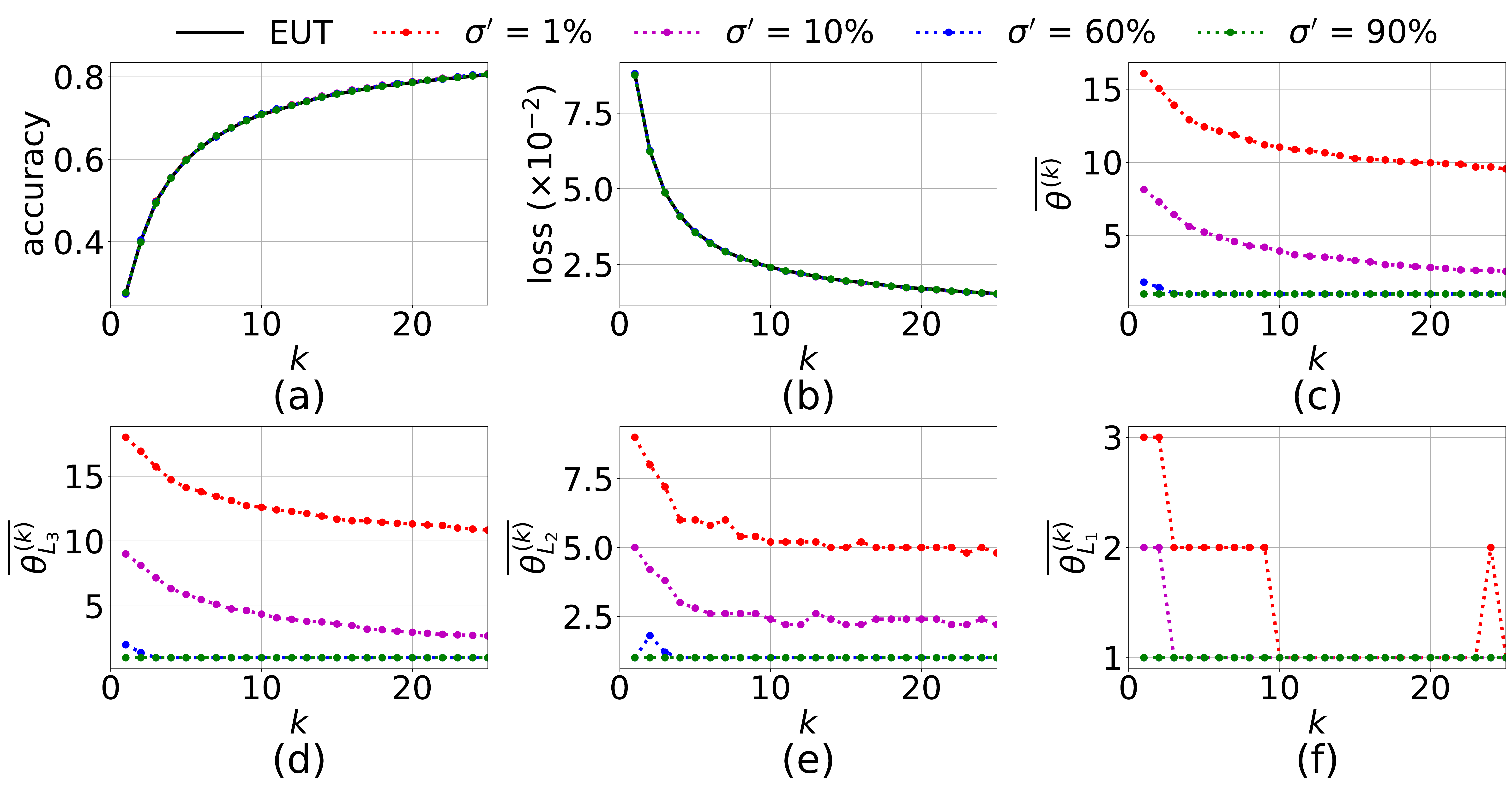}
     \caption{Performance comparison between baseline EUT and {\tt MH-FL} for i.i.d data when a finite optimality gap is tolerable. $\sigma_j$ at ${L}_j$ is fixed as $\sigma_j=\sigma' \max_{i}{\Upsilon_{{L}_{j,i}}^{(1)}}$. The tapering of D2D rounds  through time (c) and space (layers) (d)-(f) can be observed.}
     \label{fig:iidIncConSigma_MNIST_125}
\end{minipage}%
\hspace{4mm}
\begin{minipage}{.47\textwidth}
     \centering
     \includegraphics[width=\linewidth]{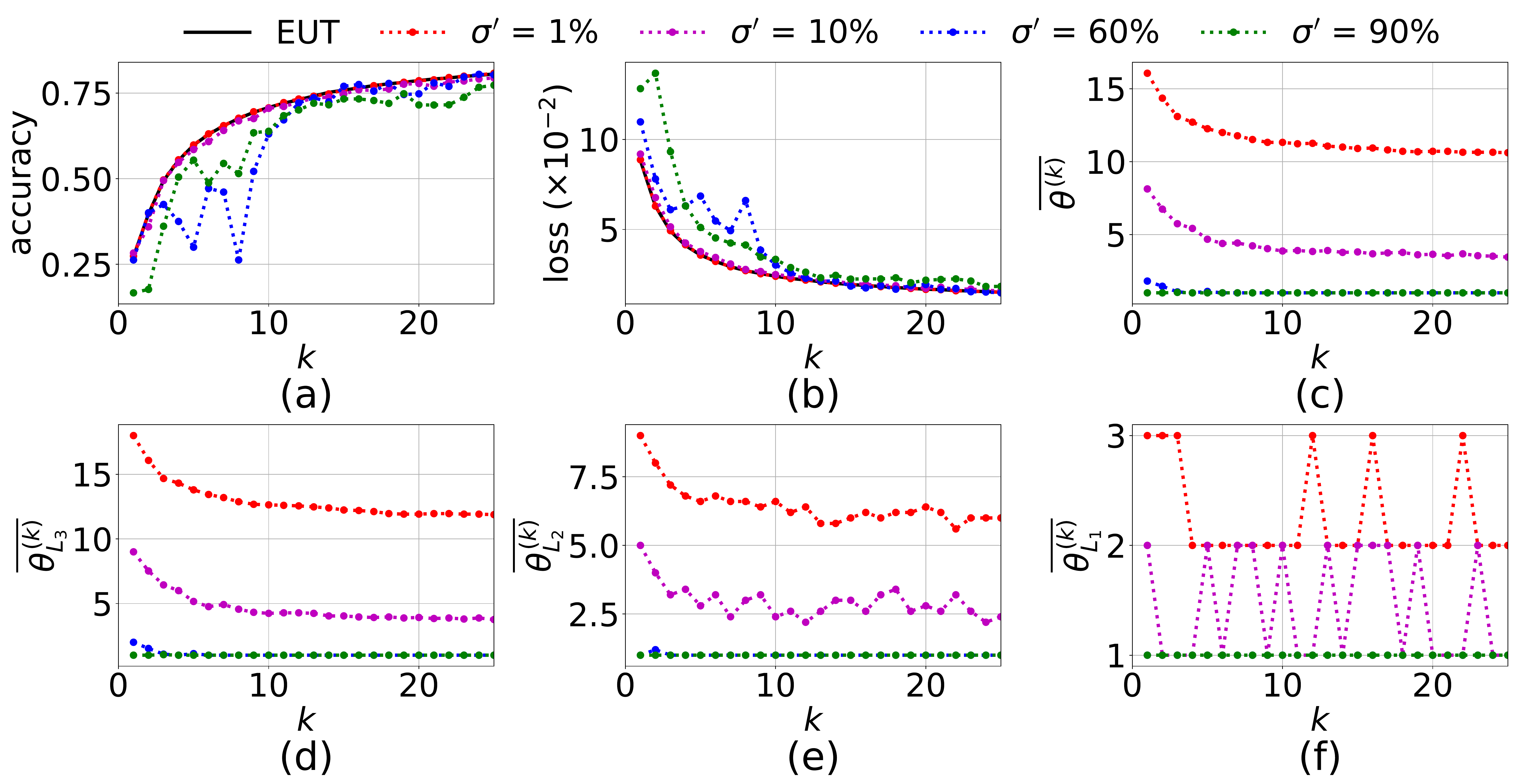}
     \caption{Performance comparison between baseline EUT and {\tt MH-FL} for non-i.i.d data when a finite optimality gap is tolerable. $\sigma_i$ is set as in Fig.~\ref{fig:iidIncConSigma_MNIST_125}. Smaller loss and higher accuracy are achieved with smaller $\sigma'$, implying more rounds of D2D are required.}
     \label{fig:non_iidIncConSigma_MNIST_125}
\end{minipage}
\vspace{-4.5mm}
\end{figure*}

\vspace{-0.1in}
\begin{figure*}[t]
\centering
\begin{minipage}{.47\textwidth}
       \centering
     \includegraphics[width=\linewidth]{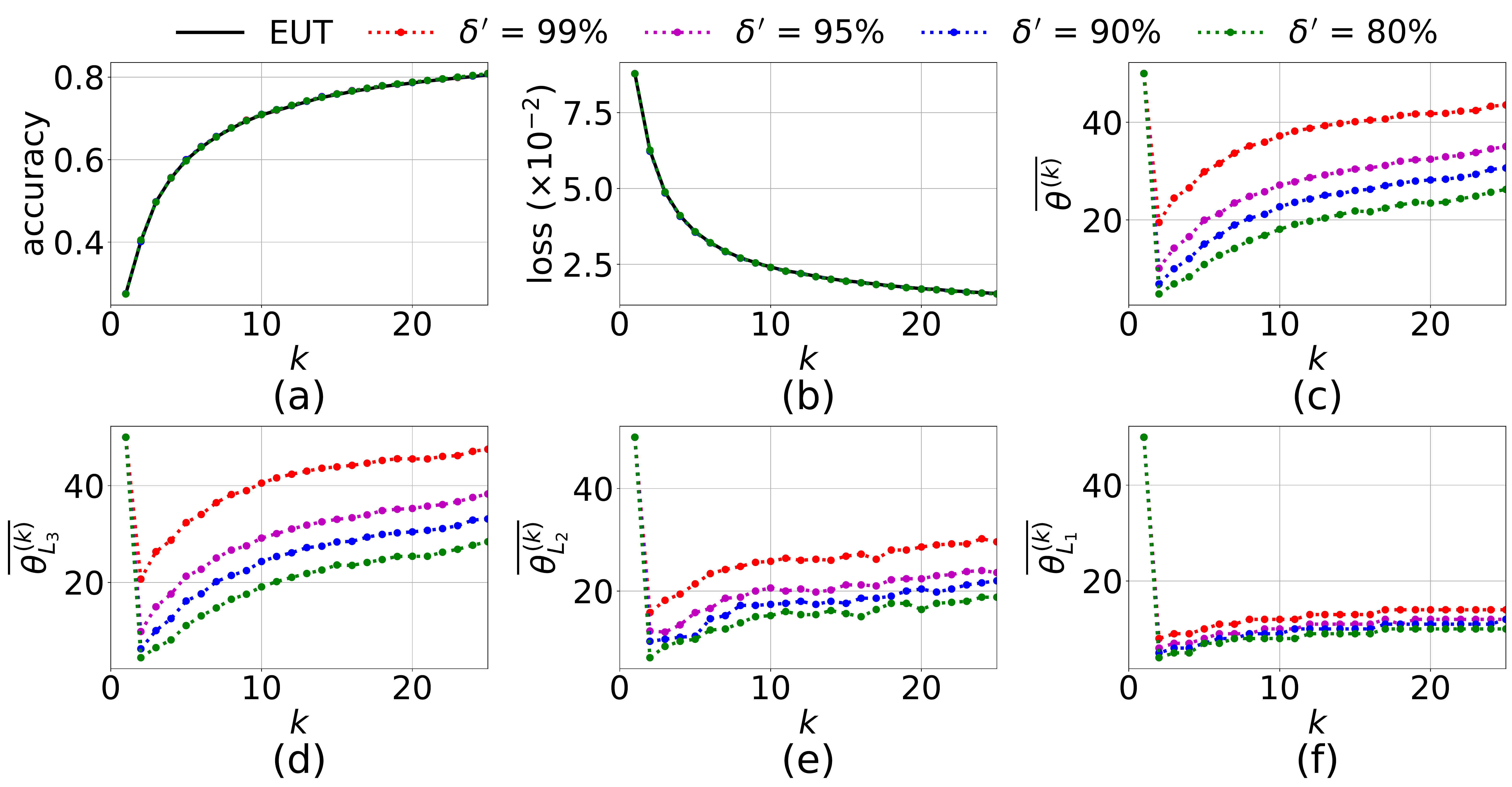}
     \caption{Performance comparison between baseline EUT and {\tt MH-FL} for i.i.d data when linear convergence to the optimal is desired. The value of $\delta$ is set at $\delta=\delta' \frac{\mu}{\eta}$. Boosting of the D2D rounds through time can be observed in (c)-(f) as $k$ is increased. Also, tapering through space can be observed by comparing the D2D rounds across the bottom subplots.}
     \label{fig:iidConsConDelta_MNIST_125}
\end{minipage}%
\hspace{4mm}
\begin{minipage}{.47\textwidth}
      \centering
     \includegraphics[width=\linewidth]{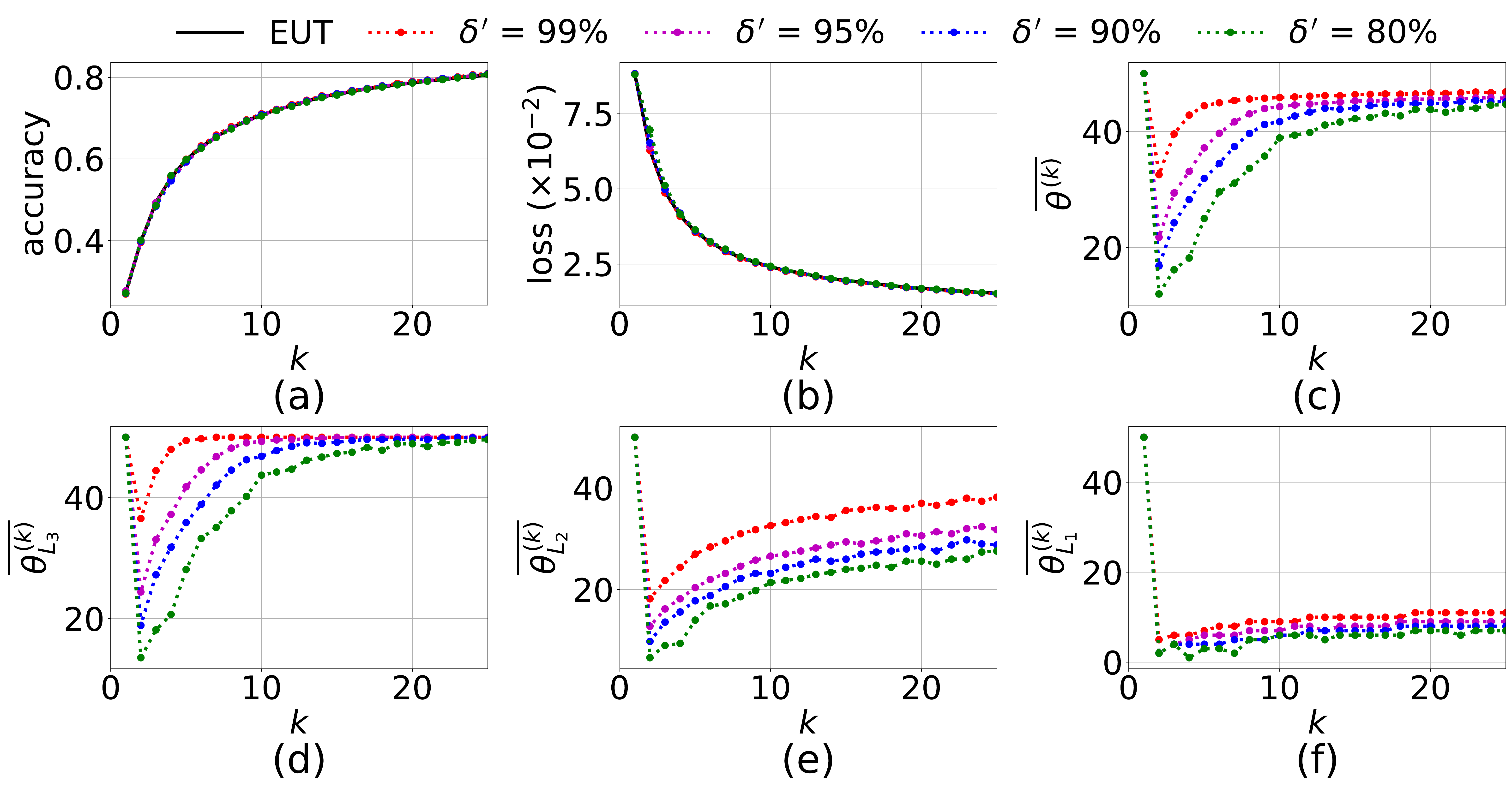}
     \caption{Performance comparison between baseline EUT and {\tt MH-FL} for non-i.i.d data when linear convergence to the optimal is desired. The value of $\delta$ is set as in Fig.~\ref{fig:iidConsConDelta_MNIST_125}.
     Smaller values of loss and higher accuracy are both associated with larger value of $\delta$, which results in lower error tolerance and more rounds of D2D.}
     \label{fig:non_iidConsConDelta_MNIST_125}
\end{minipage}
\vspace{-7.5mm}
\end{figure*}

\vspace{-0.1in}
\begin{figure*}[t]
\vspace{-2mm}
\centering
\begin{minipage}{.47\textwidth}
         \centering
     \includegraphics[width=\linewidth]{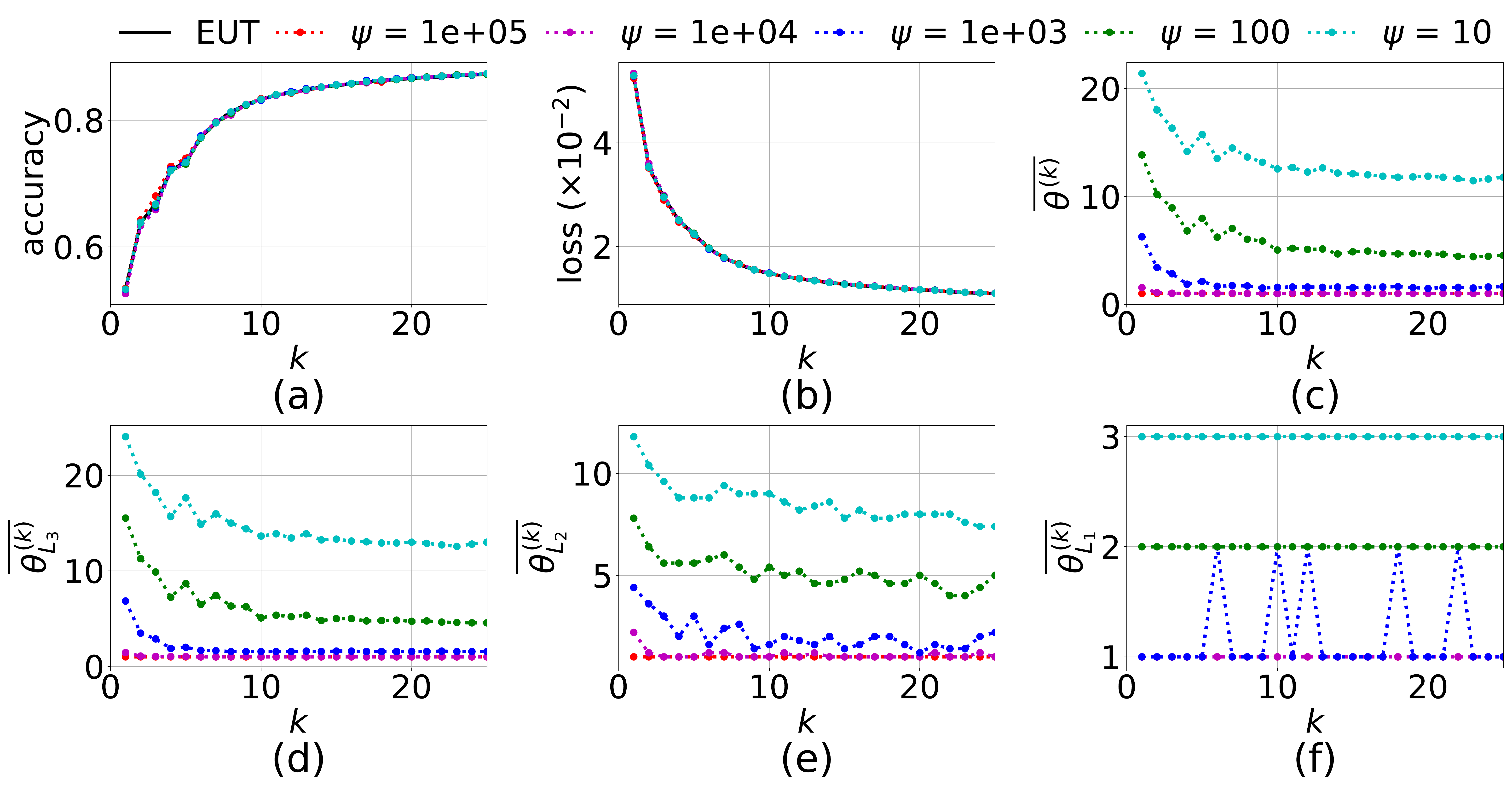}
     \caption{Performance comparison between baseline EUT and {\tt MH-FL}  under i.i.d data using NNs with different values of $\psi$. Tapering the D2D rounds through time can be observed. Also, tapering through space can be observed by comparing the D2D rounds across the bottom subplots.}
     \label{fig:iidNNIncConec_MNIST_125}
\end{minipage}%
\hspace{2mm}
\begin{minipage}{.47\textwidth}
         \centering
     \includegraphics[width=\linewidth]{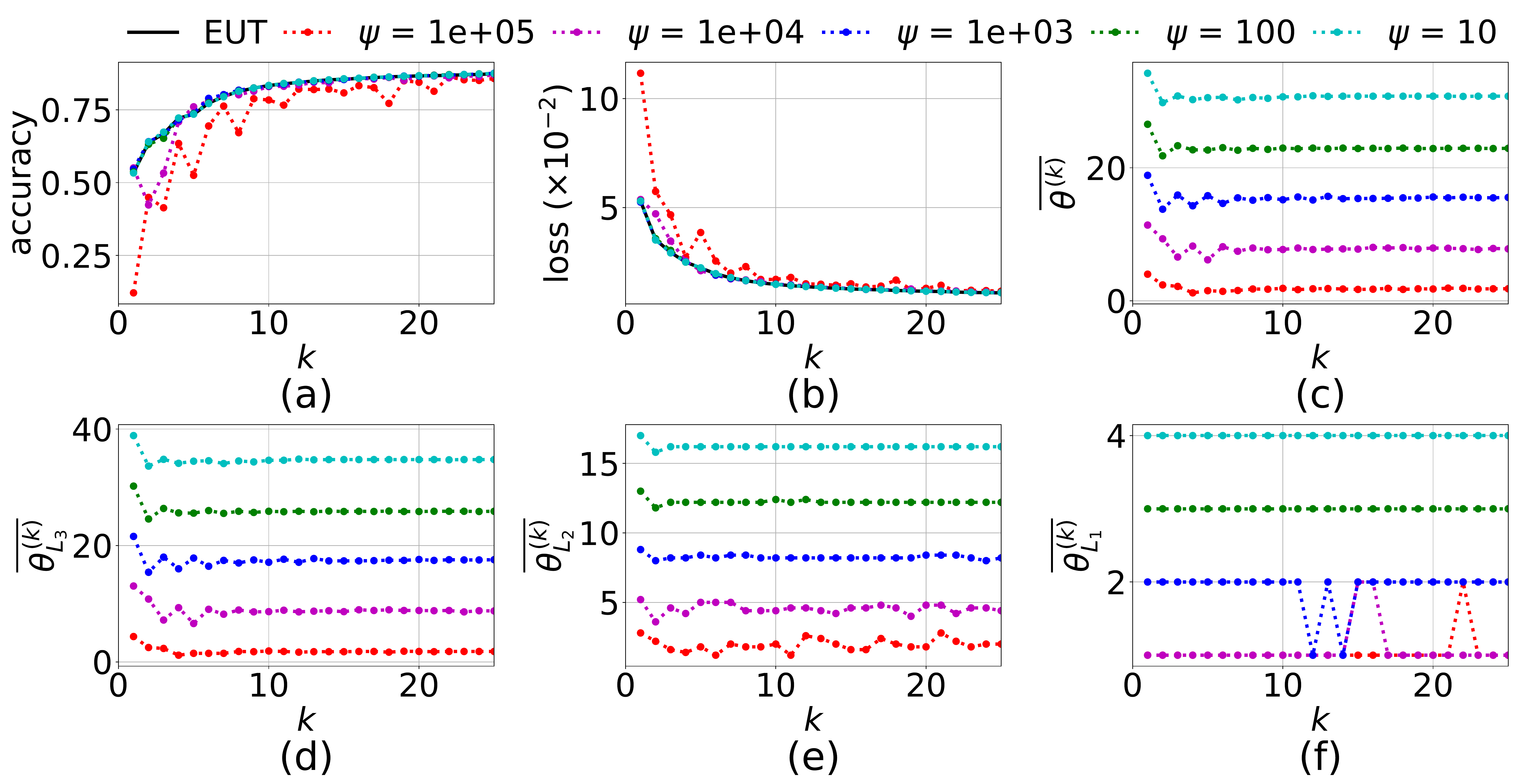}
     \caption{Performance comparison between baseline EUT and {\tt MH-FL} under non-i.i.d data using NNs with different values of $\psi$. Lower loss and higher accuracy are associated with smaller values of $\psi$, which result in lower error tolerance and larger numbers of D2D rounds.}
     \label{fig:non_iidConsConDeltaNN_MNIST_125}
\end{minipage}
\vspace{-5mm}
\end{figure*}

\begin{figure*}[t]
\centering
\begin{minipage}{.285\textwidth}
     \centering
         \includegraphics[width=0.875\linewidth]{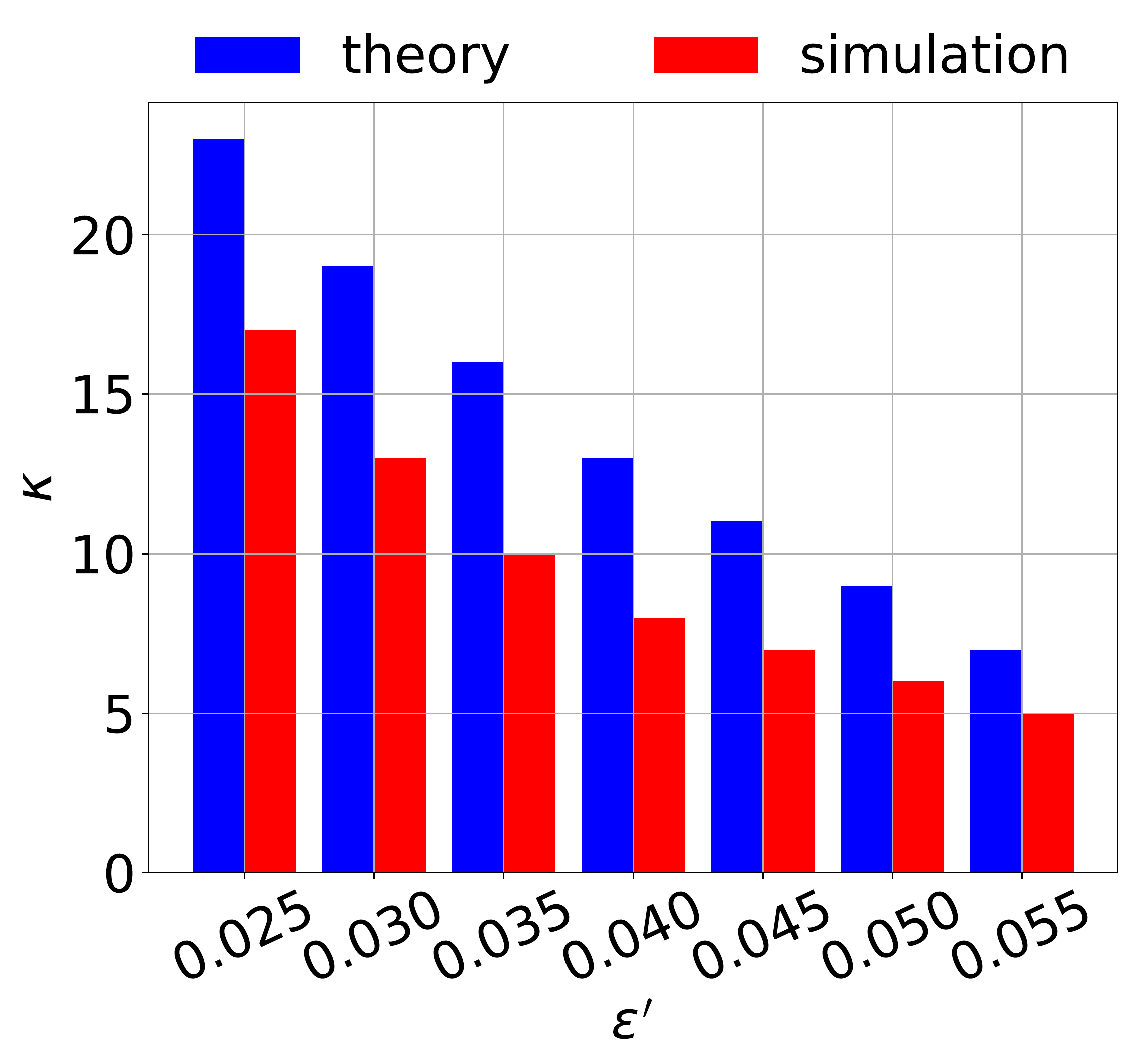}
     \caption{Theoretical vs. simulation results regarding the number of global iterations to achieve an accuracy of $\epsilon' (F(\mathbf{w}^{(0)})-F(\mathbf{w}^*))$ for different $\epsilon'$. Convergence in practice is faster than the derived upper bound.}
     \label{fig:non_iidTheoPracSimDiff_MNIST_125}
\end{minipage}%
\hspace{6mm}
\begin{minipage}{.285\textwidth}
     \centering
     \includegraphics[width=\linewidth]{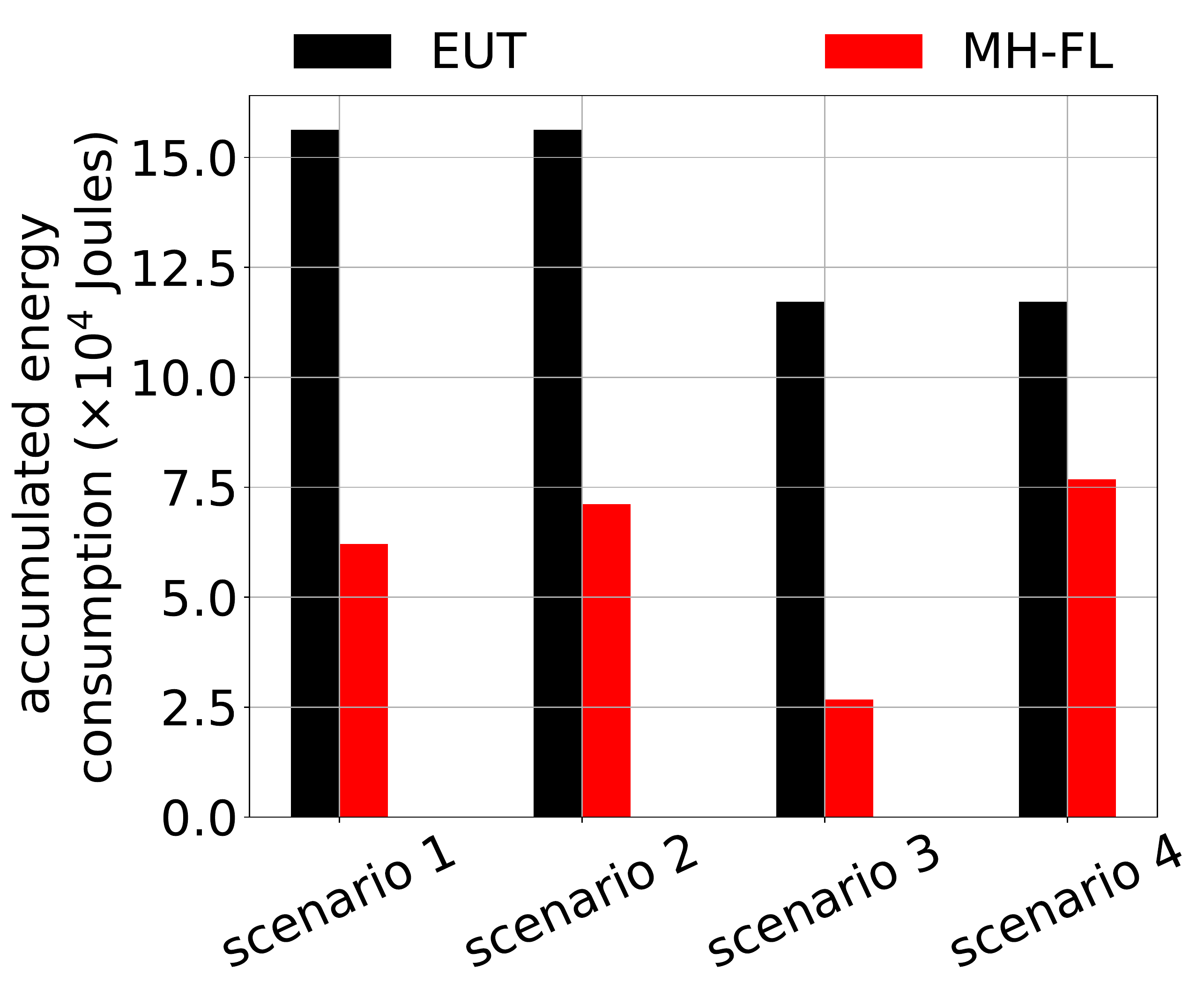}
     \caption{Accumulated energy consumption of EUT and {\tt MH-FL} over scenario 1: $\sigma' = 0.1$ from Fig.~\ref{fig:iidIncConSigma_MNIST_125}, scenario 2: $\sigma' = 0.1$ from Fig.~\ref{fig:non_iidIncConSigma_MNIST_125}, scenario 3: $\psi = 10^4$ from Fig.~\ref{fig:iidNNIncConec_MNIST_125}, and scenario 4: $\psi = 10^4$ from Fig.~\ref{fig:non_iidConsConDeltaNN_MNIST_125}.}
     \label{fig:accum_energy_125_MNIST}
\end{minipage}%
\hspace{6mm}
\begin{minipage}{.285\textwidth}
     \centering
     \includegraphics[width=\linewidth]{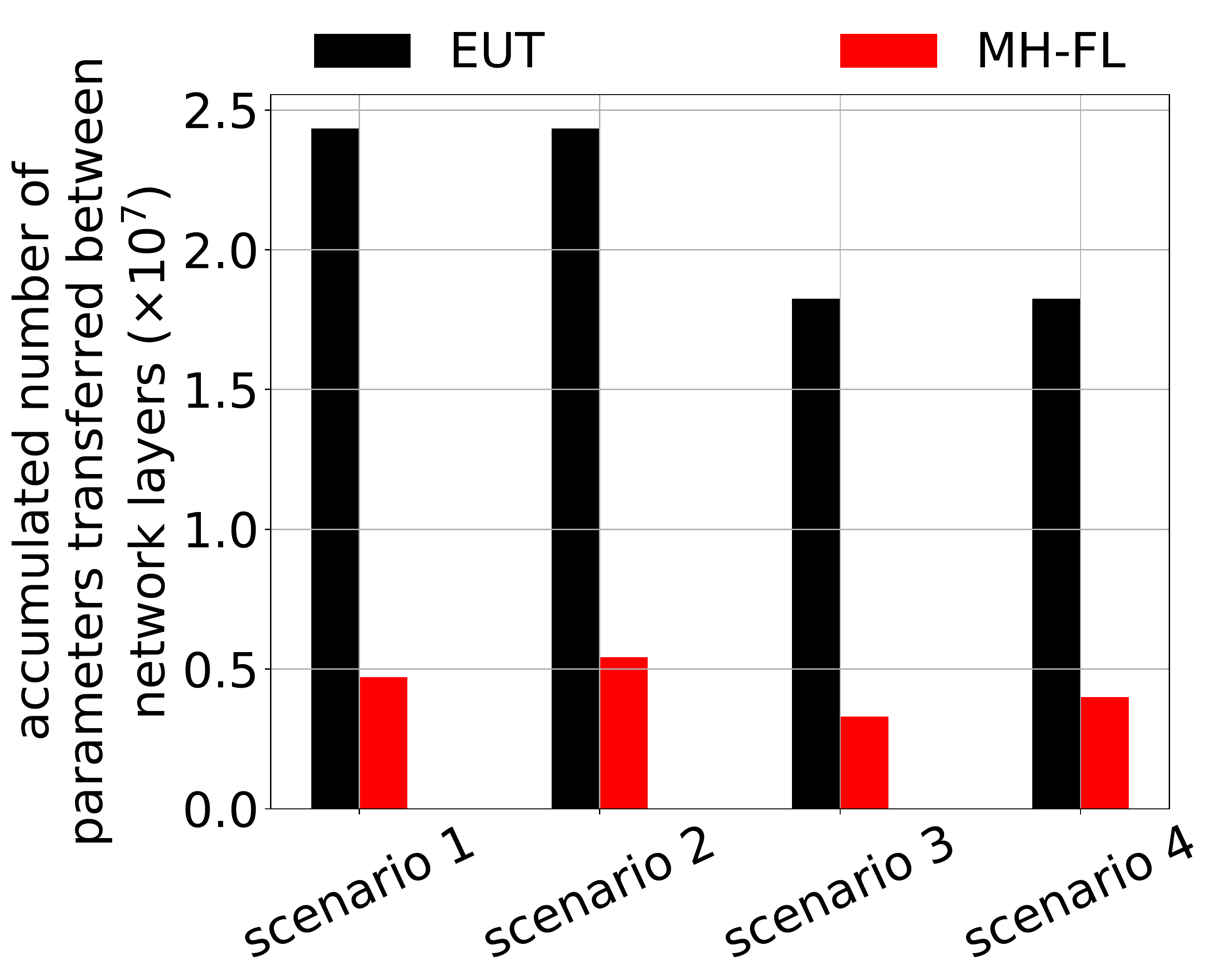}
     \caption{Comparison of parameters transferred among layers in EUT vs {\tt MH-FL} over scenario 1: $\sigma' = 0.1$ from Fig.~\ref{fig:iidIncConSigma_MNIST_125}, scenario 2: $\sigma' = 0.1$ from Fig.~\ref{fig:non_iidIncConSigma_MNIST_125}, scenario 3: $\psi = 10^4$ from Fig.~\ref{fig:iidNNIncConec_MNIST_125}, and scenario 4: $\psi = 10^4$ from Fig.~\ref{fig:non_iidConsConDeltaNN_MNIST_125}.}
     \label{fig:accumData_MNIST_125}
\end{minipage}
\vspace{-7mm}
\end{figure*}

\subsection{Experimental Setup}
\label{ssec:setup}
We consider a fog network consisting of a main server and three subsequent layers. There are $125$ edge devices in the bottom layer (${L}_3$), clustered into groups of $5$ nodes. Each of these clusters communicates with one of $25$ parent nodes in layer ${L}_2$. The nodes at this layer are in turn clustered into groups of $5$, with $5$ parent nodes at layer ${L}_1$ that communicate with the main server. We consider the cases where (i) all clusters are configured to operate in LUT mode and (ii) all are EUT, which allows us to evaluate the performance differences in terms of model convergence, energy consumption, and parameters transferred. In the LUT case, network topology within clusters follows a random geometric graph~\cite{1300540}. See Appendix~\ref{app:extraSim} for the detailed discussion of our implementation. %of~\cite{ourTechRep}).

%All clusters are configured to operate in LUT mode, which we compare  and the performances are compared with a similar network where all the clusters operate in EUT mode. This setting helps us evaluate the worst case scenario because a configuration with all clusters operating in LUT mode suffers from consensus errors at every cluster.

We consider a 10-class image classification task on the standard MNIST dataset of 70K handwritten digits (\url{http://yann.lecun.com/exdb/mnist/}). We evaluate with both regularized SVM and fully-connected neural network (NN) classifiers as loss functions; SVM satisfies Assumption~\ref{assum:genConv} while NN does not. Unless stated otherwise, the results are presented using SVM. Samples are distributed across devices in either an i.i.d. or non-i.i.d. manner; for i.i.d., each device has samples from each class, while for non-i.i.d., each device has samples from only of the 10 classes. More details on the dataset, classifiers, and hyperparameter tuning procedure are available in Appendix~\ref{app:extraSim}: there, we also provide additional results on the Fashion MNIST (F-MNIST) dataset and for a network of $625$ edge devices.
% , finding that  the results are qualitatively similar.

%We present the results obtained from using the MNIST dataset and use regularized SVM and fully-connected neural Network (NN) classifiers. The dataset distribution among the edge devices is considered to be either i.i.d or non-i.i.d. For i.i.d distribution, each edge device has samples from each label of the dataset, while under non-i.i.d distribution, each edge device has samples from only one of the labels.
%These are the extreme ends of possible split of data among nodes in terms of label distribution, helping us evaluate the overall robustness as well as differences in characteristics of our technique under different settings.
\vspace{-4mm}
\subsection{{\tt MH-FL} with Fixed Consensus Rounds}
\label{ssec:results}
\vspace{-.2mm}
\subsubsection{{\tt MH-FL} with fixed step size}
We consider a scenario in which the number of D2D rounds is set to be a constant value $\theta$ over all clusters, which provides useful insights for the rest of the results. In Fig.~\ref{fig:GenFigGoodIntuition_MNIST_125}, we compare the performance of {\tt MH-FL} when all the clusters work in LUT mode and have fixed rounds of D2D with the case where the network consists of all EUT clusters (referred to as ``EUT baseline''). The EUT case is identical (in terms of convergence) to carrying out centralized gradient descent over the entire dataset. We see that increasing the number of consensus at different layers increases the accuracy and stability of the model training for {\tt MH-FL}. Although convergence is not achieved in all cases (in particular, when $\theta$ is $1$ and $2$), if the number of D2D rounds chosen is larger than $15$, comparable performance to EUT is achieved. This performance is characterized by linear convergence, as can be seen in Fig.~\ref{fig:GenFigGoodIntuition_MNIST_125}(c) with logarithmic axis scales. 

\subsubsection{{\tt MH-FL} with decaying step size}
The effect of decreasing the gradient descent step size (Proposition~\ref{prop:conv_0_dem_step}) is depicted in Fig.~\ref{fig:decaying_learning_rate_MNIST_125}. This verifies that the decay can suppress the finite optimality bound and provide convergence to the optimal model. Also, the convergence occurs at a slower pace compared with the baseline, which is in line with our theoretical results (convergence rate of $O(1/k)$). Fig.~\ref{fig:decaying_learning_rate_MNIST_125} further reveals the inherent trade-off between conducting a higher number of D2D rounds with a constant learning rate (higher power consumption from more rounds, but   with a fast convergence speed) and performing a fewer number of D2D rounds with a decaying learning rate (lower power consumption with a slower convergence).

\vspace{-3mm}
\subsection{{\tt MH-FL} with Adaptive D2D Round Tuning}
We next study the case when our distributed D2D tuning scheme is utilized. The results are depicted for both convergence cases, i.e., where a finite optimality gap is tolerable (Figs.~\ref{fig:iidIncConSigma_MNIST_125},~\ref{fig:non_iidIncConSigma_MNIST_125}) and when the linear convergence to the optimal solution is desired (Figs.~\ref{fig:iidConsConDelta_MNIST_125},~\ref{fig:non_iidConsConDelta_MNIST_125}). Recall that Propositions~\ref{prop:boundedConsConv}\&\ref{prop:FogConv} obtain the sufficient number of D2D rounds based on an upper bound; for this experiment, we observed that $\log({\sigma_{j}})$ and $\log({\sigma^{(k)}_{j}})$ in \eqref{eq:IterFogLCons} and~\eqref{eq:IterFogLCons2} can be scaled and used as $\log(\chi {\sigma^{(k)}_{j}})$ and $\log(\chi {\sigma^{(k)}_{j}})$,  $\forall j$, where $\chi\in[1,15]$ to obtain fewer rounds of D2D while satisfying the desired convergence behavior.
\subsubsection{{\tt MH-FL} with finite optimality gap}
Fig.~\ref{fig:iidIncConSigma_MNIST_125} depicts the result for the case where local datasets are i.i.d., with the values of $\{\sigma_j\}_{j=1}^{|\mathcal{L}|}$ depicted. In the figures, $\overline{\theta^{(k)}}$ denotes the average number of D2D rounds employed by clusters over all the network layers at global iteration $k$, and $\overline{\theta^{(k)}_{{L}_j}}$ denotes the average number of D2D rounds at iteration $k$ in layer ${{L}_j}$.  We observe that (i) the D2D rounds performed in the network is tapered through time (subplot c), and (ii) the D2D rounds performed at different network layers is tapered through space, where higher layers of the network perform fewer rounds (subplots d-f). We perform the same experiment with non-i.i.d datasets across the nodes in Fig.~\ref{fig:non_iidIncConSigma_MNIST_125}. Comparing Fig.~\ref{fig:non_iidIncConSigma_MNIST_125} to~\ref{fig:iidIncConSigma_MNIST_125}, it can be observed that non-i.i.d. introduces oscillations on the number of D2D performed at different network layers, and the smaller values of D2D control parameters result in larger numbers of D2D rounds which leads to more stable training.

\subsubsection{{\tt MH-FL} with linear convergence} The same experiment is repeated in Figs.~\ref{fig:iidConsConDelta_MNIST_125},~\ref{fig:non_iidConsConDelta_MNIST_125} for the linear convergence case. We see that the number of D2D rounds is tapered through space (subplots d-f) and is boosted over time index $k$ (subplots c-f). This is due to the decrease in the norm of gradient in the right hand side of~\eqref{eq:ineqWeightsofNodesApp} over time, which calls for an increment in the number of D2D rounds. Comparing Figs.~\ref{fig:iidConsConDelta_MNIST_125} and~\ref{fig:non_iidConsConDelta_MNIST_125} with Figs.~\ref{fig:iidIncConSigma_MNIST_125} and \ref{fig:non_iidIncConSigma_MNIST_125}, we see that guaranteeing the linear convergence comes with the tradeoff of performing a larger number of D2D rounds over time. In Figs.~\ref{fig:iidIncConSigma_MNIST_125},~\ref{fig:non_iidIncConSigma_MNIST_125}, a small optimality gap is achieved, while the number of D2D rounds is tapered over time.

\subsubsection{{\tt MH-FL} with adaptive rounds of D2D for non-convex ML models}
Recall that we developed Algorithm 3 for non-convex ML loss functions. Figs.~\ref{fig:iidNNIncConec_MNIST_125} and~\ref{fig:non_iidConsConDeltaNN_MNIST_125} give results with NNs for different values of tolerable error of aggregations $\psi$, under i.i.d and non-i.i.d data distributions, respectively. These figures show the effect of tolerable error of aggregations on the performance of NNs; by decreasing the tolerable error, the number of D2D rounds is increased, and the performance is enhanced. These figures also reveal a tapering of the number of consensus through time and space in the i.i.d case.

\subsubsection{Analytical vs. experimental bound comparison}
 We investigate the number of aggregations required to obtain a certain accuracy under linear convergence (Corollary~\ref{cor:numIterCertainAccur}). In Fig.~\ref{fig:non_iidTheoPracSimDiff_MNIST_125}, we compare the result obtained from Policy B using \eqref{eq:ineqWeightsofNodesApp} to that observed in our experiments. The results indicate that the theoretical bounds are reasonably tight.

%\ali{
%In Fig.~\ref{fig:non_iidTheoPracSimDiff_MNIST_125}, we investigate the number of global iterations required to obtain a certain accuracy under the linear convergence in Corollary~\ref{cor:numIterCertainAccur}. We compare the result obtained from Policy B with \eqref{eq:coeffgapApp} in Sec.~\ref{ssec:control} to that observed in the experiment. The figure indicates that the theoretical results are indeed an upper bound on the value obtained in practice.

%consider the case that linear convergence to the optimal solution is desired and set the value of~$\delta$ according to  Corollary~\ref{cor:numIterCertainAccur} to achieve certain accuracy at different target global iterations. The number of consensus rounds are obtained using our distributed consensus tuning. The figure reveals that the theoretical results are indeed an upper bound on the value obtained in practice.}
% In Fig.~\ref{fig:non_iidTheoPracSimDiff_MNIST_125}, we compare our analytical upper bound in Corollary~\ref{cor:numGlobIterCertainAccur} on the number of global iterations required to achieve a certain accuracy with the number observed in practice. Several different target accuracies are considered, in the case that linear convergence to the optimal solution is desired. We see that the sufficient condition derived in Corollary~\ref{cor:numGlobIterCertainAccur} is indeed an upper bound on the value obtained in practice.

\vspace{-3.3mm}
\subsection{Network Resource Utilization}
\vspace{-.2mm}
% We saw that {\tt MH-FL} can achieve comparable classification performance to the EUT baseline.
We now study the network resource utilization of {\tt MH-FL}. In particular, we consider two metrics: (i) the amount of data transferred between the network layers, and (ii) the accumulated energy consumption of the edge devices. In both cases, EUT and {\tt MH-FL} are trained to reach $98\%$ of the final accuracy achieved after $50$ iterations of centralized gradient descent. We consider four scenarios, corresponding to those used in Figs.~\ref{fig:iidIncConSigma_MNIST_125},~\ref{fig:non_iidIncConSigma_MNIST_125},~\ref{fig:iidNNIncConec_MNIST_125},~\ref{fig:non_iidConsConDeltaNN_MNIST_125}.
To obtain the accumulated energy, we consider the transmission power of end devices to be $10$dbm in D2D and $24$dbm in uplink mode~\cite{hmila2019energy,dominic2020joint}, and assume that transmission of parameters at each round occurs with data rate of $1 \mbox{Mbits/s}$ with $32$-bit quantization per model parameter element. The accumulated energy consumption of the edge devices through the training phase is depicted in Fig.~\ref{fig:accum_energy_125_MNIST}, which reveals around $50\%$ energy saving on average as compared to the EUT baseline. The accumulated number of parameters transferred over the network layers are shown in Fig.~\ref{fig:accumData_MNIST_125}, revealing $80\%$ reduction in the number of parameters transferred over the layers as compared to the baseline.
We conduct further numerical studies in Appendix~\ref{appendix:psiSigma} to reveal the impact of our D2D control parameters $\{\sigma_j\}_{j=1}^{|\mathcal{L}|}$ and our tolerable aggregation error $\psi$ on the energy and transmit parameters savings.

\vspace{-3mm}
\section{Conclusion and Future Work}

% \vspace{-.5mm}
\noindent We developed multi-stage hybrid federated learning~({\tt MH-FL}), a novel methodology which migrates the star topology of federated learning to a multi-layer cluster-based distributed architecture incorporating cooperative D2D communications, which constitutes a semi-decentralized learning architecture. We theoretically obtained the convergence bound of {\tt MH-MT} explicitly considering the time varying network topology, time varying number of D2D rounds at different network clusters, and inherent ML model characteristics. We proposed a set of policies for the number of D2D rounds conducted at different network clusters under which convergence either to a finite optimality gap or the global optimum can be achieved. We further used these policies to develop a set of adaptive distributed control algorithms that tune the number of D2D rounds at different clusters of the network in  real-time. 

% Our experimental results validated the improvement in network resource utilization {\tt MH-MT} achieves.

% \ali{This paper motivates several directions for future work. Investigating more specific system characteristics that have been considered for federated learning -- such as communication imperfections, interference management, mitigation of stragglers, and device scheduling -- for the multi-stage hybrid structure of fog networks is promising. Also, the proposed network dimension of federated learning motivates works on device synchronization, network (re-)formation, congestion aware data transfer and load balancing, and topology design. Also, studying the asynchronous federated learning paradigm~\cite{xie2019asynchronous} over the semi-decentralized architecture proposed in this paper is an interesting direction. 
This paper motivates several directions for future work. Investigating more specific system characteristics that have been considered for federated learning -- such as communication imperfections, interference management, mitigation of stragglers, and device scheduling -- for the multi-stage hybrid structure of fog networks is promising. Also, the proposed network dimension of federated learning motivates works on network (re-)formation, congestion-aware data transfer and load balancing, and topology design for performance optimization. Furthermore, integrating the recently proposed asynchronous federated learning paradigm~\cite{xie2019asynchronous} with the semi-decentralized architecture proposed in this paper is an interesting direction. Finally, in this work, we have assumed that the operation of device clusters as LUT vs. EUT is provided as an input to our methodology; namely, by the physical/link-layer protocols in place, where D2D communication links are established. A holistic, cross-layer optimization approach that jointly optimizes model training and resource utilization metrics over the partitioning of devices into LUT vs. EUT and the subsequent operation of LUT clusters is another promising future direction.
%utilization of stochastic gradient decent (SGD),
%In addition, studying asynchronous transmissions among the D2D devices and its effect on convergence is an interesting direction.
\vspace{-7mm}
\bibliographystyle{IEEEtran}
\bibliography{FogLRef}
\vspace{-16mm}

\begin{IEEEbiographynophoto}{Seyyedali Hosseinalipour (M'20)}
received his Ph.D. in EE from NCSU in 2020. He is currently a postdoctoral researcher at Purdue University. 
\end{IEEEbiographynophoto}
\vspace{-15mm}
\begin{IEEEbiographynophoto}{Sheikh Shams Azam} is a Ph.D. student at Purdue University. He received his B.Tech. in ECE from NITK, India in 2015.
 
\end{IEEEbiographynophoto}
\vspace{-15mm}
\begin{IEEEbiographynophoto}{Christopher G. Brinton (SM'20)}
is an Assistant Professor of ECE at Purdue University. He received his Ph.D. in EE from Princeton University in 2016.
\end{IEEEbiographynophoto}
\vspace{-15mm}
\begin{IEEEbiographynophoto}{Nicol\`{o} Michelusi (SM'19)} 
received his Ph.D. in EE from University of Padova, Italy, in 2013. He is an Assistant Professor at Arizona State University.
\end{IEEEbiographynophoto}
\vspace{-15mm}
% \vspace{-0.15in}
\begin{IEEEbiographynophoto}{Vaneet Aggarwal (SM'15)} received his Ph.D. in EE from Princeton University in 2010. He is currently an Associate Professor at Purdue University.
\end{IEEEbiographynophoto}
% \vspace{-0.15in}
\vspace{-15mm}
\begin{IEEEbiographynophoto}{David Love (F'15)} is the Nick Trbovich Professor of ECE at Purdue University. He received the Ph.D. degree in EE from University of Texas at Austin in 2004.
\end{IEEEbiographynophoto}
\vspace{-15mm}
% \vspace{-0.55in}
\begin{IEEEbiographynophoto}{Huaiyu Dai (F’17)}
received the Ph.D. degree in EE from Princeton University in 2002. He is currently a Professor of ECE at NCSU, holding the title of University Faculty Scholar.
\end{IEEEbiographynophoto}
% \vspace{-0.65in}
\newpage
\vspace{30mm}
  \newpage 
\vspace{10mm}
\begingroup
\let\clearpage\relax 
\onecolumn %%% For single column
\appendices
\section{Proof of Theorem~\ref{theo:consbase}}\label{app:ConsGen}
 \begin{proof}
 We carry out the proof in three parts: In part I, we obtain the convergence behavior of {\tt MH-FL} given arbitrary aggregation errors at the sampled devices. In part II, we obtain the aggregation error caused by the D2D consensus process. Finally, in part III, we derive the final convergence bound.
 
 \subsection{PART I: Convergence Bound for General Local Aggregation Error at the Sampled Nodes}
 
  We first aim to bound the per-iteration decrease of the gap between the function $F(\mathbf{w}^{(k)})$ and $F(\mathbf{w}^*)$. Using the Taylor expansion and the $\eta$-smoothness of function $F$, the following quadratic upper-bound can be obtained:
 \begin{equation}\label{eq:quadBound}
 \begin{aligned}
     F(\mathbf{w}^{(k)})\leq& F(\mathbf{w}^{(k-1)})+ (\mathbf{w}^{(k)}-\mathbf{w}^{(k-1)})^\top \nabla F(\mathbf{w}^{(k-1)})+\frac{\eta}{2}\norm{\mathbf{w}^{(k)}-\mathbf{w}^{(k-1)}}^2,~~ \forall k.
     \end{aligned}
 \end{equation}
 To find the relationship between $\textbf{w}^{(k-1)}$ and $\textbf{w}^{(k)}$, we follow the procedure described in the main text. For parent node $a_{p}$, let $a'_{p+1}$ denote the corresponding sampled node, $\forall p$, e.g., in the following nested sums $a'_{|\mathcal{L}|}$ denotes the sampled node in the last layer by parent node $a_{|\mathcal{L}|-1}$ in its above layer. This corresponds to an arbitrary realization of the children sampling at different parent nodes. Let $a'_{1}$ denote the  sampled node by the main server in $L_1$. The model parameter of this node is given by
\begin{equation}\label{eq:root}
  \hspace{-22mm}
  \begin{aligned}
    & \widehat{\mathbf{w}}^{(k)}_{{a'_1}}=\frac{\displaystyle \sum_{{a_{1}}
    \in \mathcal{L}^{(k)}_{{1},{1}}} \sum_{a_{2}
    \in \mathcal{Q}^{(k)}(a_{1})} \sum_{a_{3}
    \in \mathcal{Q}^{(k)}(a_{2})}\cdots  \sum_{a_{|\mathcal{L}|}
    \in \mathcal{Q}^{(k)}({a_{|\mathcal{L}|-1}})} |\mathcal{D}_{a_{|\mathcal{L}|}}|\mathbf{w}^{(k-1)}_{{a_{|\mathcal{L}|}}}} {|\mathcal{L}^{(k)}_{{1},1}|}\\& 
    -\frac{\displaystyle \sum_{{a_{1}}
    \in \mathcal{L}^{(k)}_{{1},{1}}} \sum_{a_{2}
    \in \mathcal{Q}^{(k)}(a_{1})} \sum_{a_{3}
    \in \mathcal{Q}^{(k)}(a_{2})}\cdots  \sum_{a_{|\mathcal{L}|}
    \in \mathcal{Q}^{(k)}({a_{|\mathcal{L}|-1}})}\beta |\mathcal{D}_{a_{|\mathcal{L}|}}|\nabla f_{{a_{|\mathcal{L}|}}}(\mathbf{w}^{(k-1)}_{{a_{|\mathcal{L}|}}})} {|\mathcal{L}^{(k)}_{{1},1}|}\\&
    +\sum_{{a_{1}}
    \in \mathcal{L}^{(k)}_{{1},{1}}} \sum_{a_{2}
    \in \mathcal{Q}^{(k)}(a_{1})} \sum_{a_{3}
    \in \mathcal{Q}^{(k)}(a_{2})}\cdots  \sum_{a_{|\mathcal{L}|-1}
    \in \mathcal{Q}^{(k)}({a_{|\mathcal{L}|-2}})} \frac{\mathbbm{1}^{(k)}_{\left\{{{Q}(a_{|\mathcal{L}|-1})}\right\}}|\mathcal{Q}^{(k)}(a_{|\mathcal{L}|-1})|  \mathbf{c}^{(k)}_{a'_{|\mathcal{L}|}}}{|\mathcal{L}^{(k)}_{{1},1}|}\\&
    +\sum_{{a_{1}}
    \in \mathcal{L}^{(k)}_{{1},{1}}} \sum_{a_{2}
    \in \mathcal{Q}^{(k)}(a_{1})} \sum_{a_{3}
    \in \mathcal{Q}^{(k)}(a_{2})}\cdots  \sum_{a_{|\mathcal{L}|-2}
    \in \mathcal{Q}^{(k)}({a_{|\mathcal{L}|-3}})} \frac{\mathbbm{1}^{(k)}_{\left\{{{Q}(a_{|\mathcal{L}|-2})}\right\}}|\mathcal{Q}^{(k)}(a_{|\mathcal{L}|-2})|  \mathbf{c}^{(k)}_{a'_{|\mathcal{L}|-1}}}{|\mathcal{L}^{(k)}_{{1},1}|}+
    \\&\vdots
    \\&+\sum_{a_1
    \in \mathcal{L}^{(k)}_{{1},1}} \frac{\mathbbm{1}^{(k)}_{\left\{{{Q}(a_1)}\right\}}|\mathcal{Q}^{(k)}(a_1)|  \mathbf{c}^{(k)}_{a'_2}}{|\mathcal{L}^{(k)}_{{1},1}|}+\mathbbm{1}^{(k)}_{\left\{{{L}_{{1},1}}\right\}} \mathbf{c}^{(k)}_{a'_1},
       \end{aligned}
       \hspace{-20mm}
 \end{equation}
which the main server uses to obtain the next global parameter as follows (due to the existence of the indicator function in the last term of the above expression, the following expression holds regardless of the operating mode of the cluster at layer ${L}_1$):
  \begin{equation}\label{up:proofGlob}
    \textbf{w}^{(k)}= \frac{|\mathcal{L}^{(k)}_{{1},1}| {\widehat{\mathbf{w}}}^{(k)}_{{a'_1}}}{D}.
 \end{equation}
Based on \eqref{eq:localupdateOverrride}, it can be verified that
\begin{equation}
  \begin{aligned}
    &\displaystyle \sum_{{a_{1}}
    \in \mathcal{L}^{(k)}_{{1},{1}}} \sum_{a_{2}
    \in \mathcal{Q}^{(k)}(a_{1})} \sum_{a_{3}
    \in \mathcal{Q}^{(k)}(a_{2})}\cdots  \sum_{a_{|\mathcal{L}|}
    \in \mathcal{Q}^{(k)}({a_{|\mathcal{L}|-1}})} |\mathcal{D}_{a_{|\mathcal{L}|}}|\mathbf{w}^{(k-1)}_{{a_{|\mathcal{L}|}}} =D\textbf{w}^{(k-1)}.
       \end{aligned}
       \hspace{-15mm}
 \end{equation}
Also, using~\eqref{eq:globlossinit}, we have
\begin{equation}
  \hspace{-27mm}
  \begin{aligned}
  {\displaystyle \sum_{{a_{1}}
    \in \mathcal{L}^{(k)}_{{1},{1}}} \sum_{a_{2}
    \in \mathcal{Q}^{(k)}(a_{1})} \sum_{a_{3}
    \in \mathcal{Q}^{(k)}(a_{2})}\cdots  \sum_{a_{|\mathcal{L}|}
    \in \mathcal{Q}^{(k)}({a_{|\mathcal{L}|-1}})}\beta |\mathcal{D}_{a_{|\mathcal{L}|}}|\nabla f_{{a_{|\mathcal{L}|}}}(\mathbf{w}^{(k-1)}_{{a_{|\mathcal{L}|}}})}
    = \beta D\nabla F (\textbf{w}^{(k-1)}).
       \end{aligned}
       \hspace{-20mm}
 \end{equation}
 Replacing the above two equations in~\eqref{eq:root} and performing the update given by~\eqref{up:proofGlob}, we get
 \begin{equation}\label{eq:roott}
  \hspace{-24mm}
  \begin{aligned}
    & {\textbf{w}}^{(k)}=\textbf{w}^{(k-1)}-\beta \nabla F(\textbf{w}^{(k-1)}) +\frac{1}{D}\Bigg(\sum_{{a_{1}}
    \in \mathcal{L}^{(k)}_{{1},{1}}} \sum_{a_{2}
    \in \mathcal{Q}^{(k)}(a_{1})} \sum_{a_{3}
    \in \mathcal{Q}^{(k)}(a_{2})}\cdots  \sum_{a_{|\mathcal{L}|-1}
    \in \mathcal{Q}^{(k)}({a_{|\mathcal{L}|-2}})} {\mathbbm{1}^{(k)}_{\left\{{{Q}(a_{|\mathcal{L}|-1})}\right\}}|\mathcal{Q}^{(k)}(a_{|\mathcal{L}|-1})|  \mathbf{c}^{(k)}_{a'_{|\mathcal{L}|}}}\\&
    +\sum_{{a_{1}}
    \in \mathcal{L}^{(k)}_{{1},{1}}} \sum_{a_{2}
    \in \mathcal{Q}^{(k)}(a_{1})} \sum_{a_{3}
    \in \mathcal{Q}^{(k)}(a_{2})}\cdots  \sum_{a_{|\mathcal{L}|-2}
    \in \mathcal{Q}^{(k)}({a_{|\mathcal{L}|-3}})} {\mathbbm{1}^{(k)}_{\left\{{{Q}(a_{|\mathcal{L}|-2})}\right\}}|\mathcal{Q}^{(k)}(a_{|\mathcal{L}|-2})|  \mathbf{c}^{(k)}_{a'_{|\mathcal{L}|-1}}}+
    \\&\vdots
    \\&+\sum_{a_1
    \in \mathcal{L}^{(k)}_{{1},1}} {\mathbbm{1}^{(k)}_{\left\{{{Q}(a_1)}\right\}}|\mathcal{Q}^{(k)}(a_1)|  \mathbf{c}^{(k)}_{a'_2}}+\mathbbm{1}^{(k)}_{\left\{{{L}_{{1},1}}\right\}} |\mathcal{L}^{(k)}_{{1},1}|\mathbf{c}^{(k)}_{a'_1}\Bigg),
       \end{aligned}
       \hspace{-20mm}
 \end{equation}
 Calculating $\textbf{w}^{(k)}-\textbf{w}^{(k-1)}$ using the above equation and replacing the result in~\eqref{eq:quadBound} yields
  \begin{equation}\label{eq:quadBound2}
 \begin{aligned}
     &F(\mathbf{w}^{(k)})-F(\mathbf{w}^{(k-1)})\leq  \left(\frac{\eta\beta^2}{2}-\beta\right)\norm{\nabla F(\mathbf{w}^{(k-1)})}^2\\&+\left(\frac{1-\beta \eta}{D}\right)\Bigg(\sum_{{a_{1}}
    \in \mathcal{L}^{(k)}_{{1},{1}}} \sum_{a_{2}
    \in \mathcal{Q}^{(k)}(a_{1})} \sum_{a_{3}
    \in \mathcal{Q}^{(k)}(a_{2})}\cdots  \sum_{a_{|\mathcal{L}|-1}
    \in \mathcal{Q}^{(k)}({a_{|\mathcal{L}|-2}})} {\mathbbm{1}^{(k)}_{\left\{{{Q}(a_{|\mathcal{L}|-1})}\right\}}|\mathcal{Q}^{(k)}(a_{|\mathcal{L}|-1})|  \mathbf{c}^{(k)}_{a'_{|\mathcal{L}|}}}\\&
    +\sum_{{a_{1}}
    \in \mathcal{L}^{(k)}_{{1},{1}}} \sum_{a_{2}
    \in \mathcal{Q}^{(k)}(a_{1})} \sum_{a_{3}
    \in \mathcal{Q}^{(k)}(a_{2})}\cdots  \sum_{a_{|\mathcal{L}|-2}
    \in \mathcal{Q}^{(k)}({a_{|\mathcal{L}|-3}})} {\mathbbm{1}^{(k)}_{\left\{{{Q}(a_{|\mathcal{L}|-2})}\right\}}|\mathcal{Q}^{(k)}(a_{|\mathcal{L}|-2})|  \mathbf{c}^{(k)}_{a'_{|\mathcal{L}|-1}}}+\cdots
    \\&+\sum_{a_1
    \in \mathcal{L}^{(k)}_{{1},1}} {\mathbbm{1}^{(k)}_{\left\{{{Q}(a_1)}\right\}}|\mathcal{Q}^{(k)}(a_1)|  \mathbf{c}^{(k)}_{a'_2}}+\mathbbm{1}^{(k)}_{\left\{{{L}_{{1},1}}\right\}} |\mathcal{L}^{(k)}_{{1},1}|\mathbf{c}^{(k)}_{a'_1}\Bigg)^\top\nabla F(\mathbf{w}^{(k-1)})\\
 &+\frac{\eta}{2D^2}\Big\Vert\sum_{{a_{1}}
    \in \mathcal{L}^{(k)}_{{1},{1}}} \sum_{a_{2}
    \in \mathcal{Q}^{(k)}(a_{1})} \sum_{a_{3}
    \in \mathcal{Q}^{(k)}(a_{2})}\cdots  \sum_{a_{|\mathcal{L}|-1}
    \in \mathcal{Q}^{(k)}({a_{|\mathcal{L}|-2}})} {\mathbbm{1}^{(k)}_{\left\{{{Q}(a_{|\mathcal{L}|-1})}\right\}}|\mathcal{Q}^{(k)}(a_{|\mathcal{L}|-1})|  \mathbf{c}^{(k)}_{a'_{|\mathcal{L}|}}}\\&
    +\sum_{{a_{1}}
    \in \mathcal{L}^{(k)}_{{1},{1}}} \sum_{a_{2}
    \in \mathcal{Q}^{(k)}(a_{1})} \sum_{a_{3}
    \in \mathcal{Q}^{(k)}(a_{2})}\cdots  \sum_{a_{|\mathcal{L}|-2}
    \in \mathcal{Q}^{(k)}({a_{|\mathcal{L}|-3}})} {\mathbbm{1}^{(k)}_{\left\{{{Q}(a_{|\mathcal{L}|-2})}\right\}}|\mathcal{Q}^{(k)}(a_{|\mathcal{L}|-2})|  \mathbf{c}^{(k)}_{a'_{|\mathcal{L}|-1}}}+\cdots
    \\&+\sum_{a_1
    \in \mathcal{L}^{(k)}_{{1},1}} {\mathbbm{1}^{(k)}_{\left\{{{Q}(a_1)}\right\}}|\mathcal{Q}^{(k)}(a_1)|  \mathbf{c}^{(k)}_{a'_2}}+\mathbbm{1}^{(k)}_{\left\{{{L}_{{1},1}}\right\}} |\mathcal{L}^{(k)}_{{1},1}|\mathbf{c}^{(k)}_{a'_1}\Big\Vert^2.
     \end{aligned}
     \hspace{-6mm}
 \end{equation}
 Tuning the learning rate as $\beta=\frac{1}{\eta}$, we obtain
 \begin{equation}\label{eq:quadBound33}
 \hspace{-5mm}
 \begin{aligned}
     &F(\mathbf{w}^{(k)})-F(\mathbf{w}^{(k-1)})\leq \frac{-1}{2\eta}\norm{\nabla F(\mathbf{w}^{(k-1)})}^2+\\
 &\frac{\eta}{2D^2}\Big\Vert\sum_{{a_{1}}
    \in \mathcal{L}^{(k)}_{{1},{1}}} \sum_{a_{2}
    \in \mathcal{Q}^{(k)}(a_{1})} \sum_{a_{3}
    \in \mathcal{Q}^{(k)}(a_{2})}\cdots  \sum_{a_{|\mathcal{L}|-1}
    \in \mathcal{Q}^{(k)}({a_{|\mathcal{L}|-2}})} {\mathbbm{1}^{(k)}_{\left\{{{Q}(a_{|\mathcal{L}|-1})}\right\}}|\mathcal{Q}^{(k)}(a_{|\mathcal{L}|-1})|  \mathbf{c}^{(k)}_{a'_{|\mathcal{L}|}}}\\&
    +\sum_{{a_{1}}
    \in \mathcal{L}^{(k)}_{{1},{1}}} \sum_{a_{2}
    \in \mathcal{Q}^{(k)}(a_{1})} \sum_{a_{3}
    \in \mathcal{Q}^{(k)}(a_{2})}\cdots  \sum_{a_{|\mathcal{L}|-2}
    \in \mathcal{Q}^{(k)}({a_{|\mathcal{L}|-3}})} {\mathbbm{1}^{(k)}_{\left\{{{Q}(a_{|\mathcal{L}|-2})}\right\}}|\mathcal{Q}^{(k)}(a_{|\mathcal{L}|-2})|  \mathbf{c}^{(k)}_{a'_{|\mathcal{L}|-1}}}+
    \\&\vdots
    \\&+\sum_{a_1
    \in \mathcal{L}^{(k)}_{{1},1}} {\mathbbm{1}^{(k)}_{\left\{{{Q}(a_1)}\right\}}|\mathcal{Q}^{(k)}(a_1)|  \mathbf{c}^{(k)}_{a'_2}}+\mathbbm{1}^{(k)}_{\left\{{{L}_{{1},1}}\right\}} |\mathcal{L}^{(k)}_{{1},1}|\mathbf{c}^{(k)}_{a'_1}\Big\Vert^2.
     \end{aligned}
     \hspace{-5mm}
 \end{equation}
 Considering the right hand side of the above inequality, to bound $\norm{\nabla F(\mathbf{w}^{(k-1)})}^2$,  we use the strong convexity property of $F$. Considering the strong convexity criterion in Assumption~\ref{assum:genConv} with $x$ replaced by $\mathbf{w}^{(k-1)}$ and minimizing the both hand sides, the minimum of the left hand side occurs at $y=\mathbf{w}^*$ and the minimum of the right hand side occurs at $y=\mathbf{w}^{(k-1)}-\frac{1}{\mu} \nabla F(\mathbf{w}^{(k-1)})$. Replacing these values in the strong convexity criterion in Assumption~\ref{assum:genConv} results in Polyak-Lojasiewicz inequality~\cite{polyak1963gradient} in the following form:
 \begin{equation}\label{proofboundGradStrongConv}
     \norm{\nabla F(\mathbf{w}^{(k-1)})}^2 \geq (F(\mathbf{w}^{(k-1)})- F(\mathbf{w}^*)){2\mu},
 \end{equation}
 which yields
 \begin{equation}\label{eq:quadBound3}
 \hspace{-5mm}
 \begin{aligned}
     &F(\mathbf{w}^{(k)})-F(\mathbf{w}^{(k-1)})\leq \frac{-\mu}{\eta}(F(\mathbf{w}^{(k-1)})- F(\mathbf{w}^*))+\\
 &\frac{\eta}{2D^2}\Bigg[\Big\Vert\sum_{{a_{1}}
    \in \mathcal{L}^{(k)}_{{1},{1}}} \sum_{a_{2}
    \in \mathcal{Q}^{(k)}(a_{1})} \sum_{a_{3}
    \in \mathcal{Q}^{(k)}(a_{2})}\cdots  \sum_{a_{|\mathcal{L}|-1}
    \in \mathcal{Q}^{(k)}({a_{|\mathcal{L}|-2}})} {\mathbbm{1}^{(k)}_{\left\{{{Q}(a_{|\mathcal{L}|-1})}\right\}}|\mathcal{Q}^{(k)}(a_{|\mathcal{L}|-1})|  \mathbf{c}^{(k)}_{a'_{|\mathcal{L}|}}}\\&
    +\sum_{{a_{1}}
    \in \mathcal{L}^{(k)}_{{1},{1}}} \sum_{a_{2}
    \in \mathcal{Q}^{(k)}(a_{1})} \sum_{a_{3}
    \in \mathcal{Q}^{(k)}(a_{2})}\cdots  \sum_{a_{|\mathcal{L}|-2}
    \in \mathcal{Q}^{(k)}({a_{|\mathcal{L}|-3}})} {\mathbbm{1}^{(k)}_{\left\{{{Q}(a_{|\mathcal{L}|-2})}\right\}}|\mathcal{Q}^{(k)}(a_{|\mathcal{L}|-2})|  \mathbf{c}^{(k)}_{a'_{|\mathcal{L}|-1}}}+
    \\&\vdots
    \\&+\sum_{a_1
    \in \mathcal{L}^{(k)}_{{1},1}} {\mathbbm{1}^{(k)}_{\left\{{{Q}(a_1)}\right\}}|\mathcal{Q}^{(k)}(a_1)|  \mathbf{c}^{(k)}_{a'_2}}+\mathbbm{1}^{(k)}_{\left\{{{L}_{{1},1}}\right\}} |\mathcal{L}^{(k)}_{{1},1}|\mathbf{c}^{(k)}_{a'_1}\Big\Vert^2\Bigg].
     \end{aligned}
     \hspace{-5mm}
 \end{equation}
 Then, we perform the following algebraic steps to bound the second term on the right hand side of the above inequality:
\begin{equation}\label{eq:proofGenBound1}
\hspace{-19mm}
     \begin{aligned}
 &\Big\Vert\sum_{{a_{1}}
    \in \mathcal{L}^{(k)}_{{1},{1}}} \sum_{a_{2}
    \in \mathcal{Q}^{(k)}(a_{1})} \sum_{a_{3}
    \in \mathcal{Q}^{(k)}(a_{2})}\cdots  \sum_{a_{|\mathcal{L}|-1}
    \in \mathcal{Q}^{(k)}({a_{|\mathcal{L}|-2}})} {\mathbbm{1}^{(k)}_{\left\{{{Q}(a_{|\mathcal{L}|-1})}\right\}}|\mathcal{Q}^{(k)}(a_{|\mathcal{L}|-1})|  \mathbf{c}^{(k)}_{a'_{|\mathcal{L}|}}}\\&
    +\sum_{{a_{1}}
    \in \mathcal{L}^{(k)}_{{1},{1}}} \sum_{a_{2}
    \in \mathcal{Q}^{(k)}(a_{1})} \sum_{a_{3}
    \in \mathcal{Q}^{(k)}(a_{2})}\cdots  \sum_{a_{|\mathcal{L}|-2}
    \in \mathcal{Q}^{(k)}({a_{|\mathcal{L}|-3}})} {\mathbbm{1}^{(k)}_{\left\{{{Q}(a_{|\mathcal{L}|-2})}\right\}}|\mathcal{Q}^{(k)}(a_{|\mathcal{L}|-2})|  \mathbf{c}^{(k)}_{a'_{|\mathcal{L}|-1}}}+
    \cdots
    \\&+\sum_{a_1
    \in \mathcal{L}^{(k)}_{{1},1}} {\mathbbm{1}^{(k)}_{\left\{{{Q}(a_1)}\right\}}|\mathcal{Q}^{(k)}(a_1)|  \mathbf{c}^{(k)}_{a'_2}}+\mathbbm{1}^{(k)}_{\left\{{{L}_{{1},1}}\right\}} |\mathcal{L}^{(k)}_{{1},1}|\mathbf{c}^{(k)}_{a'_1}\Big\Vert^2\\&%%%%%%%%%%%%%% 
    \leq 
    \Bigg( \Big\Vert\sum_{{a_{1}}
    \in \mathcal{L}^{(k)}_{{1},{1}}} \sum_{a_{2}
    \in \mathcal{Q}^{(k)}(a_{1})} \sum_{a_{3}
    \in \mathcal{Q}^{(k)}(a_{2})}\cdots  \sum_{a_{|\mathcal{L}|-1}
    \in \mathcal{Q}^{(k)}({a_{|\mathcal{L}|-2}})} {\mathbbm{1}^{(k)}_{\left\{{{Q}(a_{|\mathcal{L}|-1})}\right\}}|\mathcal{Q}^{(k)}(a_{|\mathcal{L}|-1})|  \mathbf{c}^{(k)}_{a'_{|\mathcal{L}|}}}\Big\Vert\\&+\Big\Vert\sum_{{a_{1}}
    \in \mathcal{L}^{(k)}_{{1},{1}}} \sum_{a_{2}
    \in \mathcal{Q}^{(k)}(a_{1})} \sum_{a_{3}
    \in \mathcal{Q}^{(k)}(a_{2})}\cdots  \sum_{a_{|\mathcal{L}|-2}
    \in \mathcal{Q}^{(k)}({a_{|\mathcal{L}|-3}})} {\mathbbm{1}^{(k)}_{\left\{{{Q}(a_{|\mathcal{L}|-2})}\right\}}|\mathcal{Q}^{(k)}(a_{|\mathcal{L}|-2})|  \mathbf{c}^{(k)}_{a'_{|\mathcal{L}|-1}}}
   \Big\Vert \\&+\cdots
    \\&+\Big\Vert\sum_{a_1
    \in \mathcal{L}^{(k)}_{{1},1}} {\mathbbm{1}^{(k)}_{\left\{{{Q}(a_1)}\right\}}|\mathcal{Q}^{(k)}(a_1)|  \mathbf{c}^{(k)}_{a'_2}}\Big\Vert+\Big\Vert\mathbbm{1}^{(k)}_{\left\{{{L}_{{1},1}}\right\}} |\mathcal{L}^{(k)}_{{1},1}|\mathbf{c}^{(k)}_{a'_1}\Big\Vert\Bigg)^2
   \\&%%%%%%%%%%%%%%%%%%%%%%%%%%%
     \leq \Bigg( \sum_{{a_{1}}
    \in \mathcal{L}^{(k)}_{{1},{1}}} \sum_{a_{2}
    \in \mathcal{Q}^{(k)}(a_{1})} \sum_{a_{3}
    \in \mathcal{Q}^{(k)}(a_{2})}\cdots  \sum_{a_{|\mathcal{L}|-1}
    \in \mathcal{Q}^{(k)}({a_{|\mathcal{L}|-2}})} {\mathbbm{1}^{(k)}_{\left\{{{Q}(a_{|\mathcal{L}|-1})}\right\}}|\mathcal{Q}^{(k)}(a_{|\mathcal{L}|-1})| \Big\Vert \mathbf{c}^{(k)}_{a'_{|\mathcal{L}|}}}\Big\Vert\\&+\sum_{{a_{1}}
    \in \mathcal{L}^{(k)}_{{1},{1}}} \sum_{a_{2}
    \in \mathcal{Q}^{(k)}(a_{1})} \sum_{a_{3}
    \in \mathcal{Q}^{(k)}(a_{2})}\cdots  \sum_{a_{|\mathcal{L}|-2}
    \in \mathcal{Q}^{(k)}({a_{|\mathcal{L}|-3}})} {\mathbbm{1}^{(k)}_{\left\{{{Q}(a_{|\mathcal{L}|-2})}\right\}}|\mathcal{Q}^{(k)}(a_{|\mathcal{L}|-2})|  \Big\Vert\mathbf{c}^{(k)}_{a'_{|\mathcal{L}|-1}}}
   \Big\Vert \\&+\cdots
    \\&+\sum_{a_1
    \in \mathcal{L}^{(k)}_{{1},1}} {\mathbbm{1}^{(k)}_{\left\{{{Q}(a_1)}\right\}}|\mathcal{Q}^{(k)}(a_1)|^2 \Big\Vert\mathbf{c}^{(k)}_{a'_2}}\Big\Vert+\mathbbm{1}^{(k)}_{\left\{{{L}_{{1},1}}\right\}}  |\mathcal{L}^{(k)}_{{1},1}|\Big\Vert \mathbf{c}^{(k)}_{a'_1}\Big\Vert\Bigg)^2%%%%%%%%%%%%%%%%%
    \\&\overset{(a)}{\leq} \\&{\Phi} \Bigg[  \sum_{{a_{1}}
    \in \mathcal{L}^{(k)}_{{1},{1}}} \sum_{a_{2}
    \in \mathcal{Q}^{(k)}(a_{1})} \sum_{a_{3}
    \in \mathcal{Q}^{(k)}(a_{2})}\cdots  \sum_{a_{|\mathcal{L}|-1}
    \in \mathcal{Q}^{(k)}({a_{|\mathcal{L}|-2}})} {\mathbbm{1}^{(k)}_{\left\{{{Q}(a_{|\mathcal{L}|-1})}\right\}}|\mathcal{Q}^{(k)}(a_{|\mathcal{L}|-1})|^2 \Big\Vert \mathbf{c}^{(k)}_{a'_{|\mathcal{L}|}}}\Big\Vert^2\\&+\sum_{{a_{1}}
    \in \mathcal{L}^{(k)}_{{1},{1}}} \sum_{a_{2}
    \in \mathcal{Q}^{(k)}(a_{1})} \sum_{a_{3}
    \in \mathcal{Q}^{(k)}(a_{2})}\cdots  \sum_{a_{|\mathcal{L}|-2}
    \in \mathcal{Q}^{(k)}({a_{|\mathcal{L}|-3}})} {\mathbbm{1}^{(k)}_{\left\{{{Q}(a_{|\mathcal{L}|-2})}\right\}}|\mathcal{Q}^{(k)}(a_{|\mathcal{L}|-2})|^2  \Big\Vert\mathbf{c}^{(k)}_{a'_{|\mathcal{L}|-1}}}
   \Big\Vert^2 \\&+\cdots
    \\&+\sum_{a_1
    \in \mathcal{L}^{(k)}_{{1},1}} {\mathbbm{1}^{(k)}_{\left\{{{Q}(a_1)}\right\}}|\mathcal{Q}^{(k)}(a_1)|^2  \Big\Vert\mathbf{c}^{(k)}_{a'_2}}\Big\Vert^2+\mathbbm{1}^{(k)}_{\left\{{{L}_{{1},1}}\right\}}  |\mathcal{L}^{(k)}_{{1},1}|^2\Big\Vert \mathbf{c}^{(k)}_{a'_1}\Big\Vert^2\Bigg],
         \end{aligned}
         \hspace{-20mm}
\end{equation}
where the triangle inequality, e.g, for vectors $\textbf{a}_i,~1\leq i\leq n$: $\norm{\sum_{i=1}^{n}\textbf{a}_i}\leq \sum_{i=1}^{n} \norm{\textbf{a}_i}$,  is applied sequentially and
\begin{equation}
    \Phi= N_{{|\mathcal{L}|-1}}+N_{{|\mathcal{L}|-2}}+\cdots+N_{{1}}+1.
\end{equation}
Also, inequality (a) in~\eqref{eq:proofGenBound1} is obtained by exploiting the Cauchy–Schwarz inequality,$<\textbf{a},\textbf{a}'>~\leq \norm{\textbf{a}}.\norm{\textbf{a}'}$, which results in the following bound, where $\mathbf{q}=[q_1,\cdots,q_b]$:
\begin{equation}
\begin{aligned}
    &\left(\sum_{a=1}^{b} q_i \right)^2= \left(<\mathbf{1},\mathbf{q}> \right)^2
    \leq b \sum_{a=1}^{b} q^2_i.
    \end{aligned}
\end{equation}
%%%%%%%%%%%%%%%%%%%%%%%%%%%%%%%%%%%%%%%%%%%
%%%%%%%%%%%%%%%%%%%%%%%%%%%%%%%%%%%%%%%%%%%
%%%%%%%%%%%%%%%%%%%%%%%%%%%%%%%%%%%%%%%%
%%%%%%%%%%%%%%%%%%%%%%%%%%%%%%%%%%%%%%%%%%

\subsection{PART II: Finding the Consensus (local aggregation) Error in Each LUT Cluster}
To further find each of the error terms in the right hand side of~\eqref{eq:proofGenBound1}, we need to bound $\Big\Vert \mathbf{c}^{(k)}_{{}{a'_p}}\Big\Vert^2$, $1\leq p \leq |\mathcal{L}|$. For notations simplicity we consider bounding $\Big\Vert \mathbf{c}^{(k)}_{{a'_{|\mathcal{L}|}}}\Big\Vert^2$ for the case where sampling is conducted from node $a'_{|\mathcal{L}|}$, $a'_{|\mathcal{L}|} \in \mathcal{C}^{(k)}$, where LUT cluster ${C}$ is located in the bottom-most layer .
% Let $c_1,\cdots,c_{|\mathcal{C}^{(k)}|}$ denote the nodes belonging to cluster ${C}$, where $a'_{|\mathcal{L}|}$ is among them.

The evolution of nodes' parameters during D2D communications in this cluster can be described by~\eqref{eq:consensus} as
\begin{equation}\label{eq:con2s}
      \widehat{\mathbf{W}}_{{C}}^{(k)}= \left(\mathbf{V}^{(k)}_{{C}}\right)^{\theta^{(k)}_{{C}}} \widetilde{\mathbf{W}}_{{C}}^{(k)}.
  \end{equation}
  Let matrix $\overline{\mathbf{W}}_{{C}}^{(k)}$ denote
 the average of the vector of parameters in cluster ${C}$. This matrix can be described as
 \begin{equation}\label{eq:aveDefProof}
     \overline{{\mathbf{W}}}_{{C}}^{(k)}= \frac{\textbf{1}_{|\mathcal{C}^{(k)}|} \textbf{1}_{|\mathcal{C}^{(k)}|}^\top {\widetilde{\mathbf{W}}}_{{C}}^{(k)}}{|\mathcal{C}^{(k)}|}.
 \end{equation}
  We next define the local aggregation error matrix $\bm{E}_C^{(k)}$ for cluster $C$ at the instance of global aggregation $k$, which satisfies the following equality:
  \begin{equation}
      \mathbf{E}_C^{(k)}=\widehat{\mathbf{W}}_{{C}}^{(k)}-\overline{\mathbf{W}}_{{C}}^{(k)}.
  \end{equation}
  Note that $[\mathbf{E}_C^{(k)}]_{a'_{|\mathcal{L}|},:}=\mathbf{c}^{(k)}_{a'_{|\mathcal{L}|}}$, where $[\mathbf{E}_C^{(k)}]_{a'_{|\mathcal{L}|},:}$ denotes the row describing the parameter of node $a'_{|\mathcal{L}|}$.
  Note that $\mathbf{1}^\top \mathbf{E}_C^{(k)}=\bm{0}$, and thus $(\mathbf{1} \mathbf{1}^\top) \mathbf{E}_C^{(k)}=\mathbf{0}$ and accordingly
  \begin{equation}
  \begin{aligned}
      \mathbf{E}_C^{(k)}&=\left(\mathbf{I}-\frac{\mathbf{1} \mathbf{1}^\top}{|\mathcal{C}^{(k)}|}\right)\mathbf{E}_C^{(k)}=\left(\mathbf{I}-\frac{\mathbf{1} \mathbf{1}^\top}{|\mathcal{C}^{(k)}|}\right)\left(\widehat{\mathbf{W}}_{{C}}^{(k)}-\overline{\mathbf{W}}_{{C}}^{(k)}\right)
      \\&
      =\left(\mathbf{I}-\frac{\mathbf{1} \mathbf{1}^\top}{|\mathcal{C}^{(k)}|}\right)\left(\left(\mathbf{V}^{(k)}_{{C}}\right)^{\theta^{(k)}_{{C}}} \widetilde{\mathbf{W}}_{{C}}^{(k)}-\overline{\mathbf{W}}_{{C}}^{(k)}\right)=\left(\mathbf{I}-\frac{\mathbf{1} \mathbf{1}^\top}{|\mathcal{C}^{(k)}|}\right)\left(\left(\mathbf{V}^{(k)}_{{C}}\right)^{\theta^{(k)}_{{C}}} \widetilde{\mathbf{W}}_{{C}}^{(k)}-\left(\mathbf{V}^{(k)}_{{C}}\right)^{\theta^{(k)}_{{C}}}\overline{\mathbf{W}}_{{C}}^{(k)}\right)
      \\&
      =\left(\left(\mathbf{V}^{(k)}_{{C}}\right)^{\theta^{(k)}_{{C}}}-\frac{\mathbf{1} \mathbf{1}^\top}{|\mathcal{C}^{(k)}|}\right)\left(\widetilde{\mathbf{W}}_{{C}}^{(k)}-\overline{\mathbf{W}}_{{C}}^{(k)}\right),
      \end{aligned}
  \end{equation}
  where $\mathbf{I}$ denotes the identity matrix. In the above equalities we have used the facts that (i) $\left(\mathbf{V}^{(k)}_{{C}}\right)^{\theta^{(k)}_{{C}}} \overline{\mathbf{W}}_{{C}}^{(k)}=\overline{\mathbf{W}}_{{C}}^{(k)}$ since performing consensus on averaged matrix does not change the resulting parameters, and (ii) $\frac{\mathbf{1} \mathbf{1}^\top}{|\mathcal{C}^{(k)}|} \left(\mathbf{V}^{(k)}_{{C}}\right)^{\theta^{(k)}_{{C}}}=\frac{\mathbf{1} \mathbf{1}^\top}{|\mathcal{C}^{(k)}|}$ according to Assumption~\ref{assump:cons} since $\left(\mathbf{V}^{(k)}_{{C}}\right)^{\theta^{(k)}_{{C}}}$ is also double stochastic.

Using the above properties, we finally bound $\Vert\mathbf{c}^{(k)}_{a'_{|\mathcal{L}|}}\Vert$ as follows:
\begin{equation}
\begin{aligned}
    \big\Vert\mathbf{c}^{(k)}_{a'_{|\mathcal{L}|}}\big\Vert ^2 &\leq \textrm{trace} \left((\mathbf{E}_C^{(k)})^\top \mathbf{E}_C^{(k)} \right)=\textrm{trace}\left(\left(\widetilde{\mathbf{W}}_{{C}}^{(k)}-\overline{\mathbf{W}}_{{C}}^{(k)}\right)^\top\left(\left(\mathbf{V}^{(k)}_{{C}}\right)^{\theta^{(k)}_{{C}}}-\frac{\mathbf{1} \mathbf{1}^\top}{|\mathcal{C}^{(k)}|}\right)^2\left(\widetilde{\mathbf{W}}_{{C}}^{(k)}-\overline{\mathbf{W}}_{{C}}^{(k)}\right)\right)\\&
    \leq (\lambda^{(k)}_C)^{2\theta_C^{(k)}} \sum_{q\in\mathcal{C}^{(k)}} \Vert \widetilde{\mathbf{w}}^{(k)}_{{q}} -\overline{\mathbf{w}}^{(k)}_{C} \big\Vert^2\leq (\lambda^{(k)}_C)^{2\theta_C^{(k)}} \frac{1}{|\mathcal{C}^{(k)}|}\sum_{q,q'\in\mathcal{C}^{(k)}} \Vert \widetilde{\mathbf{w}}^{(k)}_{{q}} -\overline{\mathbf{w}}^{(k)}_{q'} \big\Vert^2\\&\leq (\lambda^{(k)}_C)^{2\theta_C^{(k)}} {|\mathcal{C}^{(k)}|}\max_{q,q'\in\mathcal{C}^{(k)}} \Vert \widetilde{\mathbf{w}}^{(k)}_{{q}} -\overline{\mathbf{w}}^{(k)}_{q'} \big\Vert^2\leq (\lambda^{(k)}_C)^{2\theta_C^{(k)}} {|\mathcal{C}^{(k)}|}\left(\Upsilon_C^{(k)}\right)^2,
    \end{aligned}
\end{equation}  
  where $\overline{\mathbf{w}}^{(k)}_{C}$  denotes the vector of average of parameters inside the cluster and we used the fact that $\left(\mathbf{V}^{(k)}_{{C}}\right)^{\theta^{(k)}_{{C}}}-\frac{\mathbf{1} \mathbf{1}^\top}{|\mathcal{C}^{(k)}|}=\left(\mathbf{V}^{(k)}_{{C}}\right)^{\theta^{(k)}_{{C}}}\left(\mathbf I-\frac{\mathbf{1} \mathbf{1}^\top}{|\mathcal{C}^{(k)}|}\right)=\left(\mathbf{V}^{(k)}_{{C}}\right)^{\theta^{(k)}_{{C}}}\left(\mathbf I-\frac{\mathbf{1} \mathbf{1}^\top}{|\mathcal{C}^{(k)}|}\right)^{\theta^{(k)}_{{C}}}=\left(\mathbf{V}^{(k)}_{{C}}-\frac{\mathbf{1} \mathbf{1}^\top}{|\mathcal{C}^{(k)}|}\right)^{\theta^{(k)}_{{C}}}$ (note that $\left(\mathbf I-\frac{\mathbf{1} \mathbf{1}^\top}{|\mathcal{C}^{(k)}|}\right)$ is a projection matrix) is a real symmetric matrix.

  The above mentioned proof can be generalized to every cluster with slight modifications, which results in
    \begin{equation}
      \Big\Vert \mathbf{c}^{(k)}_{a'_p}\Big\Vert^2\leq \left(\lambda_{{C}}\right)^{2\theta^{(k)}_{{C}}} |\mathcal{C}^{(k)}|\left({\Upsilon^{(k)}_{{C}}}\right)^2,~a'_{p}\in \mathcal{C}.
  \end{equation}
  
  \subsection{PART III: Obtaining the Final Convergence Bound}
 
  Replacing the above inequality in~\eqref{eq:proofGenBound1} combined with~\eqref{eq:quadBound3} gives us
  \begin{equation}\label{eq:quadBound4}
 \hspace{-20mm}
 \begin{aligned}
     &F(\mathbf{w}^{(k)})-F(\mathbf{w}^{(k-1)})\leq \frac{-\mu}{\eta} \left(F(\mathbf{w}^{(k-1)})- F(\mathbf{w}^*)\right)+\\
 &\frac{\eta{\Phi}}{2D^2} \Bigg[\sum_{{a_{1}}
    \in \mathcal{L}^{(k)}_{{1},{1}}} \sum_{a_{2}
    \in \mathcal{Q}^{(k)}(a_{1})} \cdots \hspace{-3mm} \sum_{a_{|\mathcal{L}|-1}
    \in \mathcal{Q}^{(k)}({a_{|\mathcal{L}|-2}})} \mathbbm{1}^{(k)}_{ \left\{{Q}(a_{|\mathcal{L}|-1})\right\}}|\mathcal{Q}^{(k)}(a_{|\mathcal{L}|-1})|^3
    \left(\lambda^{(k)}_{{Q}(a_{|\mathcal{L}|-1})}\right)^{2\theta^{(k)}_{{Q}(a_{|\mathcal{L}|-1})}} \left(\Upsilon^{(k)}_{{Q}(a_{|\mathcal{L}|-1})}\right)^2
    \\&+\sum_{{a_{1}}
    \in \mathcal{L}^{(k)}_{{1},{1}}} \sum_{a_{2}
    \in \mathcal{Q}^{(k)}(a_{1})} \cdots  \hspace{-3mm} \sum_{a_{|\mathcal{L}|-2}
    \in \mathcal{Q}^{(k)}({a_{|\mathcal{L}|-3}})} \mathbbm{1}^{(k)}_{ \left\{{Q}(a_{|\mathcal{L}|-2})\right\}}|\mathcal{Q}^{(k)}(a_{|\mathcal{L}|-2})|^3  \left(\lambda^{(k)}_{{Q}(a_{|\mathcal{L}|-2})}\right)^{2\theta^{(k)}_{{Q}(a_{|\mathcal{L}|-2})}} \left(\Upsilon^{(k)}_{{Q}(a_{|\mathcal{L}|-2})}\right)^2 \\&+\cdots+\sum_{a_1
    \in \mathcal{L}^{(k)}_{{1},1}} \mathbbm{1}^{(k)}_{ \left\{{Q}(a_1)\right\}}|\mathcal{Q}^{(k)}(a_1)|^3 \left(\lambda^{(k)}_{{Q}(a_1)}\right)^{2\theta^{(k)}_{{Q}(a_1)}} \left(\Upsilon^{(k)}_{{Q}(a_1)}\right)^2+\mathbbm{1}^{(k)}_{\left\{{{L}_{{1},1}}\right\}}|\mathcal{L}^{(k)}_{{1},1}|^3 \left(\lambda^{(k)}_{{L}_{{1},1}}\right)^{2\theta^{(k)}_{{L}_{{1},1}}} \left(\Upsilon^{(k)}_{{L}_{{1},1}}\right)^2\Bigg].
     \end{aligned}
     \hspace{-15mm}
 \end{equation}
 Adding $F(\textbf{w}^{(k-1)})$ to both hands sides and subtracting $F(\textbf{w}^*)$ from both hand sides, we get
 \begin{equation}\label{eq:quadBound5}
 \hspace{-12mm}
 \begin{aligned}
     &F(\mathbf{w}^{(k)})-F(\mathbf{w}^{*})\leq (1-\frac{\mu}{\eta}) \underbrace{\left(F(\mathbf{w}^{(k-1)})- F(\mathbf{w}^*)\right)}_{(a)}+\\
 &\frac{\eta{\Phi}}{2D^2} \Bigg[\sum_{{a_{1}}
    \in \mathcal{L}^{(k)}_{{1},{1}}} \sum_{a_{2}
    \in \mathcal{Q}^{(k)}(a_{1})} \cdots \hspace{-3mm} \sum_{a_{|\mathcal{L}|-1}
    \in \mathcal{Q}^{(k)}({a_{|\mathcal{L}|-2}})} \mathbbm{1}^{(k)}_{ \left\{{Q}(a_{|\mathcal{L}|-1})\right\}}|\mathcal{Q}^{(k)}(a_{|\mathcal{L}|-1})|^3
    \left(\lambda^{(k)}_{{Q}(a_{|\mathcal{L}|-1})}\right)^{2\theta^{(k)}_{{Q}(a_{|\mathcal{L}|-1})}} \left(\Upsilon^{(k)}_{{Q}(a_{|\mathcal{L}|-1})}\right)^2
    \\&+\sum_{{a_{1}}
    \in \mathcal{L}^{(k)}_{{1},{1}}} \sum_{a_{2}
    \in \mathcal{Q}^{(k)}(a_{1})} \cdots  \hspace{-3mm} \sum_{a_{|\mathcal{L}|-2}
    \in \mathcal{Q}^{(k)}({a_{|\mathcal{L}|-3}})} \mathbbm{1}^{(k)}_{ \left\{{Q}(a_{|\mathcal{L}|-2})\right\}}|\mathcal{Q}^{(k)}(a_{|\mathcal{L}|-2})|^3  \left(\lambda^{(k)}_{{Q}(a_{|\mathcal{L}|-2})}\right)^{2\theta^{(k)}_{{Q}(a_{|\mathcal{L}|-2})}} \left(\Upsilon^{(k)}_{{Q}(a_{|\mathcal{L}|-2})}\right)^2 \\&+\cdots+\sum_{a_1
    \in \mathcal{L}^{(k)}_{{1},1}} \mathbbm{1}^{(k)}_{ \left\{{Q}(a_1)\right\}}|\mathcal{Q}^{(k)}(a_1)|^3 \left(\lambda^{(k)}_{{Q}(a_1)}\right)^{2\theta^{(k)}_{{Q}(a_1)}} \left(\Upsilon^{(k)}_{{Q}(a_1)}\right)^2+\mathbbm{1}^{(k)}_{\left\{{{L}_{{1},1}}\right\}}|\mathcal{L}^{(k)}_{{1},1}|^3 \left(\lambda^{(k)}_{{L}_{{1},1}}\right)^{2\theta^{(k)}_{{L}_{{1},1}}} \left(\Upsilon^{(k)}_{{L}_{{1},1}}\right)^2\Bigg].
     \end{aligned}
     \hspace{-15mm}
 \end{equation}
 Expanding term (a) on the right hand side of the inequality in a recursive manner leads to the theorem result.
  \end{proof}
% \section{Proof of Proposition~\ref{prop:numPhiAccur}}\label{app:7}
% \begin{proof}
%  The proof is based on geometric series analysis. To be done!
% \end{proof}
% \section{Proof of Proposition~\ref{prop:deviceSampling}}\label{app:8}
% \begin{proof}
%  To continue, we use the bound given in \eqref{eq:proof1}. In this case, we have~\eqref{eq:longConvSmapl1}. Following a similar procedure used to derive~\eqref{proofMidWay1}, we get:
 
%   \begin{equation}\label{proofMidWay1}
%  \begin{aligned}
%      &\mathbb{E}F(\mathbf{w}^{(k)})-F(\mathbf{w}^{*})\leq \\
%     & \frac{L}{2}\mathbb{E}\norm{\mathbf{e}^{(k)}}^2 + \frac{L-\mu}{L}(F(\mathbf{w}^{(k-1)})- F(\mathbf{w}^*)).
%      \end{aligned}
%  \end{equation}
 
%  The first term is bounded in~\eqref{eq:longConvSmapl2} {\color{red} This is the place that is still open, I think we may be able to do some analysis using the cluster Sampling with Equal/Non-equal probabilities of selection.} Now if the sequence in which the clusters are chosen are known prior to the start of the algorithm, we get~\eqref{eq:longConvSmapl3}.
 
%  , in which the last inequality is the result of the L-smoothness of the objective function, which results in: 
 
% \begin{equation}\label{eq:midSmooth}
%     F(y)\leq F(x)+(y-x)^\top \nabla F(x)+\frac{L}{2} \verty-x\Vert^2
% \end{equation}. It can be verified that the minimum of the R.H.S is achieved when $y=x-1/L\nabla F(x)$. Minimizing the both hand sides of \eqref{eq:midSmooth} leads to: $\Vert\nabla F(x)\Vert^2\leq 2 L [F(x)-F(x^*)]$.
%  \end{proof}

\section{Proof of Proposition~\ref{prop:boundedConsConv}}\label{app:boundedConsConv}
\noindent Consider the bound on the number of D2D that is given in the proposition statement. For ${L}_{j,i} $, $\textrm{if}~ \sigma_{j}\leq {|\mathcal{L}^{(k)}_{j,i}|^3 \left(\Upsilon_{{L}_{j,i}}^{(k)}\right)^2}$, $\forall i$, the proposed number of D2D guarantees $\theta^{(k)}_{{L}_{j,i}} \geq \frac{\log\left({\sigma_{j}}\right)-2\log\left({\big\vert\mathcal{L}^{(k)}_{j,i}\big\vert^{\frac{3}{2}} \Upsilon_{{L}_{j,i}}^{(k)}} \right)}{2\log\left(\lambda^{(k)}_{{L}_{j,i}} \right)} $, which results in
\begin{equation}
\begin{aligned}
&\theta^{(k)}_{{L}_{j,i}} \geq \frac{\log\left({\sigma_{j}}\right)-2\log\left({\big\vert\mathcal{L}^{(k)}_{j,i}\big\vert^{\frac{3}{2}} \Upsilon_{{L}_{j,i}}^{(k)}} \right)}{2\log\left(\lambda^{(k)}_{{L}_{j,i}} \right)}\\
&\Rightarrow \theta^{(k)}_{{L}_{j,i}} \geq \frac{1}{2} \frac{\log\left(\frac{\sigma_{j}}{{|\mathcal{L}^{(k)}_{j,i}|^3 \left(\Upsilon_{{L}_{j,i}}^{(k)}\right)^2}}\right)}{\log\left(\lambda^{(k)}_{{L}_{j,i}} \right)}\\
&
\Rightarrow \left(\lambda^{(k)}_{{L}_{j,i}}\right)^{2\theta^{(k)}_{{L}_{j,i}}}\leq \frac{\sigma_{j}}{{|\mathcal{L}^{(k)}_{j,i}|^3 \left(\Upsilon_{{L}_{j,i}}^{(k)}\right)^2}},~\forall k,
\end{aligned}
\end{equation}
where the last inequality is due to the facts that $\frac{\log a}{\log b} =\log^b_a$, $a^{\log^b_a}=b$, and $\lambda^{(k)}_{{L}_{j,i}}<1$. Also, for cluster ${L}_{j,i} $,  $\textrm{if}~ \sigma_{j}\geq {|\mathcal{L}^{(k)}_{j,i}|^3 \left(\Upsilon_{{L}_{j,i}}^{(k)}\right)^2}$, any $\theta^{(k)}_{{L}_{j,i}}\geq 0$ ensures $\sigma_{j}\geq {|\mathcal{L}^{(k)}_{j,i}|^3 \left(\Upsilon_{{L}_{j,i}}^{(k)}\right)^2}\left(\lambda^{(k)}_{{L}_{j,i}}\right)^{2\theta^{(k)}_{{L}_{j,i}}}$, $\forall k$.  Replacing the above result in~\eqref{eq:ConsTh1}, we get
\begin{equation}\label{eq:simpConvCons}
     \hspace{-18mm}
     \begin{aligned}
     &F(\mathbf{w}^{(k-1)})-F(\mathbf{w}^{*})\leq \frac{\eta{\Phi}}{2D^2} \sum_{t=0}^{k-1}
     \left(\frac{\eta-\mu}{\eta}\right)^{t}   \Bigg[ \\& \sum_{{a_{1}}
    \in \mathcal{L}^{(k)}_{{1},{1}}} \sum_{a_{2}
    \in \mathcal{Q}^{(k)}(a_{1})} \cdots \hspace{-3mm} \sum_{a_{|\mathcal{L}|-1}
    \in \mathcal{Q}^{(k)}({a_{|\mathcal{L}|-2}})}\mathbbm{1}^{(k-t)}_{ \left\{\mathcal{Q}({a_{|\mathcal{L}|-1}})\right\}}\sigma_{|\mathcal{L}|}+
    \\&\sum_{{a_{1}}
    \in \mathcal{L}^{(k)}_{{1},{1}}} \sum_{a_{2}
    \in \mathcal{Q}^{(k)}(a_{1})} \cdots  \hspace{-3mm} \sum_{a_{|\mathcal{L}|-2}
    \in \mathcal{Q}^{(k)}({a_{|\mathcal{L}|-3}})} \mathbbm{1}^{(k-t)}_{ \left\{\mathcal{Q}(a_{|\mathcal{L}|-2})\right\}}\sigma_{|\mathcal{L}|-1}+\cdots \\&+\sum_{a_1
    \in \mathcal{L}^{(k)}_{{1},1}} \mathbbm{1}^{(k-t)}_{ \left\{\mathcal{Q}(a_1)\right\}}\sigma_{2}+\mathbbm{1}^{(k-t)}_{\left\{{{L}_{{1},1}}\right\}}\sigma_{1}\Bigg]
    +\left(\frac{\eta-\mu}{\eta}\right)^{k}\left(F(\mathbf{w}^{(0)})- F(\mathbf{w}^*)\right) \\
   & \leq \frac{\eta{\Phi}}{2D^2} \sum_{t=0}^{k-1}
     \left(\frac{\eta-\mu}{\eta}\right)^{t}   \Bigg[ N_{|\mathcal{L}|-1}\sigma_{|\mathcal{L}|}+
    N_{|\mathcal{L}|-2}\sigma_{|\mathcal{L}|-1} +\cdots\\&+N_{1}\sigma_{2}+{N_{0}\sigma_{1}\Bigg]
    + \left(\frac{\eta-\mu}{\eta}\right)^{k}\left(F(\mathbf{w}^{(0)})- F(\mathbf{w}^*)\right) }.
     \end{aligned}
     \hspace{-10mm}
 \end{equation}
 Taking the limit with respect to $k$, we get
 \begin{equation}
     \hspace{-18mm}
     \begin{aligned}
     &\lim_{k\rightarrow \infty}F(\mathbf{w}^{(k-1)})-F(\mathbf{w}^{*})\leq \frac{\eta{\Phi}}{2D^2}\left(\sum_{j=0}^{|\mathcal{L}|-1}\sigma_{j+1} N_j\right) \frac{1}{1-(\frac{\eta-\mu}{\eta})},
     \end{aligned}
     \hspace{-10mm}
 \end{equation}
 which concludes the proof.

\section{Proof of Proposition~\ref{prop:FogConv}}\label{app:ConLinCons}
 \noindent Consider the per-iteration decrease of the objective function given by~\eqref{eq:quadBound5}. Following a similar procedure as Appendix~\ref{app:boundedConsConv}, given the proposed number of D2D rounds in the proposition statement, we get
\begin{equation}\label{eq:quadBound6}
 \hspace{-5mm}
 \begin{aligned}
     &F(\mathbf{w}^{(k)})-F(\mathbf{w}^{*})\leq (1-\frac{\mu}{\eta}) \left(F(\mathbf{w}^{(k-1)})- F(\mathbf{w}^*)\right)+\frac{\eta{\Phi}}{2D^2} \Bigg[ \sum_{j=0}^{|\mathcal{L}|-1} \sigma^{(k)}_{j+1} N_j \Bigg].
     \end{aligned}
     \hspace{-5mm}
 \end{equation}
 
 Using the fact that $\nabla F(\textbf{w}^*)=0$ combined with $\eta$-smoothness of $F$, we get 
 \begin{equation}
 \hspace{-3mm}
     \norm{\nabla F(\textbf{w}^{(k-1)})}=\norm{\nabla F(\textbf{w}^{(k-1)})-\nabla F(\textbf{w}^*)}\leq \eta\Vert \textbf{w}^{(k-1)}- \textbf{w}^*\Vert.
      \hspace{-2mm}
 \end{equation} 
Also, it is straightforward to verify that strong convexity of $F$, expressed in Assumption~\ref{assum:genConv}, implies the following inequality:
\begin{equation}
    \mu/2 \norm{\textbf{w}^{(k-1)}-\textbf{w}^*}^2\leq F(\textbf{w}^{(k-1)})-F(\textbf{w}^*).
\end{equation}
Combining the above results with the condition given in the proposition statement, i.e., \eqref{eq:ineqWeightsofNodes}, we get
\begin{equation}
 \begin{aligned}
       &\sum_{j=0}^{|\mathcal{L}|-1}\sigma^{(k)}_{j+1} N_j\leq  \frac{D^2\mu  ({\mu}-\delta\eta)}{\eta^4{\Phi}} \norm{\nabla F(\textbf{w}^{(k-1)})}^2\\&
       \leq \frac{D^2\mu  ({\mu}-\delta\eta)}{\eta^2{\Phi}} \norm{ \textbf{w}^{(k-1)}-\textbf{w}^*}^2\\&\leq \frac{2D^2 ({\mu}-\delta\eta)}{\eta^2{\Phi}} \left(F(\textbf{w}^{(k-1)})-F(\textbf{w}^*) \right) 
       \end{aligned}
 \end{equation}
 By replacing the above inequality in~\eqref{eq:quadBound6} we obtain
 \begin{equation}\label{eq:quadBound8}
 \hspace{-5mm}
 \begin{aligned}
     &F(\mathbf{w}^{(k)})-F(\mathbf{w}^{*})\leq (1-{\mu}/{\eta})\left(F(\mathbf{w}^{(k-1)})- F(\mathbf{w}^*)\right)
 +({\mu/\eta -\delta})\left(F(\mathbf{w}^{(k-1)})- F(\mathbf{w}^*)\right),
     \end{aligned}
     \hspace{-5mm}
 \end{equation}
 which readily leads to the proposition result.
\section{Proof of Corollary~\ref{cor:numIterCertainAccur}}\label{app:numIterCertainAccur}
 \noindent 
 
\noindent  Regarding the first condition, at global iteration $\kappa$, using the number of consensus given in the corollary statement, according to~\eqref{eq:simpConvCons}, we have
 \begin{equation}
     \hspace{-18mm}
     \begin{aligned}
    &F(\mathbf{w}^{(\kappa)})-F(\mathbf{w}^{*})
   \leq \frac{\eta{\Phi}}{2D^2} \sum_{j=0}^{|\mathcal{L}|-1} \sigma_{j+1} N_j
   \frac{1-\left(1-\frac{\mu}{\eta}\right)^\kappa}{\mu/\eta}
    +\left(\frac{\eta-\mu}{\eta}\right)^{\kappa}\left(F(\mathbf{w}^{(0)})- F(\mathbf{w}^*)\right) \\&
    =\left(1-\frac{\mu}{\eta}\right)^\kappa \left(F(\mathbf{w}^{(0)})- F(\mathbf{w}^*)-\frac{\eta^2{\Phi}}{2\mu D^2} \sum_{j=0}^{|\mathcal{L}|-1} \sigma_{j+1} N_j \right)
    +\frac{\eta^2{\Phi}}{2\mu D^2} \sum_{j=0}^{|\mathcal{L}|-1} \sigma_{j+1} N_j.
     \end{aligned}
     \hspace{-10mm}
 \end{equation}
 Thus to satisfy the accuracy requirement, it is sufficient to have
 \begin{equation}\label{eq:suffGapk}
     \hspace{-18mm}
     \begin{aligned}
   &\left(1-\frac{\mu}{\eta}\right)^\kappa \left(F(\mathbf{w}^{(0)})- F(\mathbf{w}^*)-\frac{\eta^2{\Phi}}{2\mu D^2} \sum_{j=0}^{|\mathcal{L}|-1} \sigma_{j+1} N_j \right)
    +\frac{\eta^2{\Phi}}{2\mu D^2} \sum_{j=0}^{|\mathcal{L}|-1} \sigma_{j+1} N_j \leq \epsilon.
     \end{aligned}
     \hspace{-10mm}
 \end{equation}
 Performing some algebraic steps leads to~\eqref{eq:coeffgap}.
 
 Regarding the second condition, given the number of D2D rounds stated in the proposition statement, we first recursively expand the right hand side of~\eqref{eq:quadBound8} to get
 \begin{equation}\label{eq:linearPrec}
 \hspace{-5mm}
 \begin{aligned}
     &F(\mathbf{w}^{(\kappa)})-F(\mathbf{w}^{*})\leq (1-\delta)^{\kappa}\left(F(\mathbf{w}^{(0)})- F(\mathbf{w}^*)\right).
     \end{aligned}
     \hspace{-5mm}
 \end{equation}
 Thus, to satisfy the desired accuracy, it is sufficient to have
 \begin{equation}\label{eq:linearK}
 \hspace{-5mm}
 \begin{aligned}
    & (1-\delta)^{\kappa}[F(\mathbf{w}^{(0)})- F(\mathbf{w}^*)]\leq \epsilon,
     \end{aligned}
     \hspace{-5mm}
 \end{equation}
 which readily leads to~\eqref{eq:deltaLinear}. Note that the criterion given in the corollary statement for $\epsilon$ guarantees that: $0< \delta \leq \mu/\eta$.
\section{Proof of Corollary~\ref{cor:numGlobIterCertainAccur}}\label{app:numGlobIterCertainAccur}
\noindent Regarding the first condition, upon using the number of D2D rounds described in the corollary statement, we get~\eqref{eq:suffGapk}, which can be written as
 \begin{equation}
     \hspace{-18mm}
     \begin{aligned}
   &\left(1-\frac{\mu}{\eta}\right)^\kappa
     \leq \frac{\epsilon-\frac{\eta^2{\Phi}}{2\mu D^2} \sum_{j=0}^{|\mathcal{L}|-1} \sigma_{j+1} N_j}{F(\mathbf{w}^{(0)})- F(\mathbf{w}^*)-\frac{\eta^2{\Phi}}{2\mu D^2} \sum_{j=0}^{|\mathcal{L}|-1} \sigma_{j+1} N_j}.
     \end{aligned}
     \hspace{-10mm}
 \end{equation}
To obtain $\kappa$, we need to take the logarithm with base $1-\mu/\eta$, where $0<1-\mu/\eta<1$. Using the characteristic of the logarithm upon having a positive base less than one, we get
    \begin{equation}
     \begin{aligned}
    &\kappa\geq \log_{1-\mu/\eta}^{\left(\epsilon-\frac{\eta^2{\Phi}}{2\mu D^2} \sum_{j=0}^{|\mathcal{L}|-1} \sigma_{j+1} N_j \right){\left(F(\mathbf{w}^{(0)})- F(\mathbf{w}^*)-\frac{\eta^2{\Phi}}{2\mu D^2} \sum_{j=0}^{|\mathcal{L}|-1} \sigma_{j+1} N_j \right)^{-1}}},
     \end{aligned}
     \hspace{-5mm}
\end{equation}
which can be written as~\eqref{eq:kGapCons}.

Regarding the second condition, upon using the number of D2D rounds described in the corollary statement, we get~\eqref{eq:linearK}. To obtain $\kappa$, we take the logarithm with base $1-\delta$ from both hand sides of the equation, using the fact that $0<1-\delta<1$ and the characteristic of the logarithm upon having a positive base less than one, we get
 \begin{equation}
 \hspace{-5mm}
 \begin{aligned}
    & \kappa\geq \log_{1-\delta}^{\left(\epsilon/\left(F(\mathbf{w}^{(0)})- F(\mathbf{w}^*)\right)\right)},
     \end{aligned}
     \hspace{-5mm}
 \end{equation}
which can be written as~\eqref{cor:klieaner}.
\section{Proof of Proposition~\ref{prop:conv_0_dem_step}}\label{app:conv_0_dem_step}
 \noindent Upon sharing the gradients, the nodes in the bottom layer share their scaled gradients (multiplying their gradients by their number of data points), while the rest of the procedure, i.e., traversing of the gradients over the hierarchy, is the same as sharing the parameters. For parent node $a_{p}$, let $a'_{p+1}$ denote the corresponding sampled node, $\forall p$, e.g., in the following nested sums $a'_{|\mathcal{L}|}$ denotes the sampled node in the last layer by parent node $a_{|\mathcal{L}|-1}$ in its above layer. Let $\hat{\textbf{g}}^{(k)}_{a'_1}$ denote the sampled value by the main server at global iteration $k$. It can be verified that we have
\begin{equation}\label{eq:rootGrad}
  \hspace{-22mm}
  \begin{aligned}
    &  
    \hat{\textbf{g}}^{(k)}_{a'_1}=\frac{\displaystyle \sum_{{a_{1}}
    \in \mathcal{L}^{(k)}_{{1},{1}}} \sum_{a_{2}
    \in \mathcal{Q}^{(k)}(a_{1})} \sum_{a_{3}
    \in \mathcal{Q}^{(k)}(a_{2})}\cdots  \sum_{a_{|\mathcal{L}|}
    \in \mathcal{Q}^{(k)}({a_{|\mathcal{L}|-1}})} |\mathcal{D}_{a_{|\mathcal{L}|}}|\nabla f_{{a_{|\mathcal{L}|}}}(\mathbf{w}^{(k-1)}_{{a_{|\mathcal{L}|}}})} {|\mathcal{L}^{(k)}_{{1},1}|}\\&
    +\sum_{{a_{1}}
    \in \mathcal{L}^{(k)}_{{1},{1}}} \sum_{a_{2}
    \in \mathcal{Q}^{(k)}(a_{1})} \sum_{a_{3}
    \in \mathcal{Q}^{(k)}(a_{2})}\cdots  \sum_{a_{|\mathcal{L}|-1}
    \in \mathcal{Q}^{(k)}({a_{|\mathcal{L}|-2}})} \frac{\mathbbm{1}^{(k)}_{\left\{{{Q}(a_{|\mathcal{L}|-1})}\right\}}|\mathcal{Q}^{(k)}(a_{|\mathcal{L}|-1})|  \mathbf{c}^{(k)}_{a'_{|\mathcal{L}|}}}{|\mathcal{L}^{(k)}_{{1},1}|}\\&
    +\sum_{{a_{1}}
    \in \mathcal{L}^{(k)}_{{1},{1}}} \sum_{a_{2}
    \in \mathcal{Q}^{(k)}(a_{1})} \sum_{a_{3}
    \in \mathcal{Q}^{(k)}(a_{2})}\cdots  \sum_{a_{|\mathcal{L}|-2}
    \in \mathcal{Q}^{(k)}({a_{|\mathcal{L}|-3}})} \frac{\mathbbm{1}^{(k)}_{\left\{{{Q}(a_{|\mathcal{L}|-2})}\right\}}|\mathcal{Q}^{(k)}(a_{|\mathcal{L}|-2})|  \mathbf{c}^{(k)}_{a'_{|\mathcal{L}|-1}}}{|\mathcal{L}^{(k)}_{{1},1}|}+
    \\&\vdots
    \\&+\sum_{a_1
    \in \mathcal{L}^{(k)}_{{1},1}} \frac{\mathbbm{1}^{(k)}_{\left\{{{Q}(a_1)}\right\}}|\mathcal{Q}^{(k)}(a_1)|  \mathbf{c}^{(k)}_{a'_2}}{|\mathcal{L}^{(k)}_{{1},1}|}+\mathbbm{1}^{(k)}_{\left\{{{L}_{{1},1}}\right\}} \mathbf{c}^{(k)}_{a'_1}.
       \end{aligned}
       \hspace{-20mm}
 \end{equation}
 The main server then uses this vector as the estimation of global gradient and builds the parameter vector for the next iteration as follows (note that although the root only receives the gradients, it has the knowledge of the previous parameters that it broadcast, i.e., $\textbf{w}^{(k-1)}$):
 \begin{equation}\label{eq:rootGrad}
  \hspace{-22mm}
  \begin{aligned}
    & \widehat{\textbf{w}}^{(k)}_{a'_1}=D\frac{\textbf{w}^{(k-1)}}{|\mathcal{L}^{(k)}_{{1},1}|}-\\& 
    \beta_{k-1}\Bigg[\frac{\displaystyle \sum_{{a_{1}}
    \in \mathcal{L}^{(k)}_{{1},{1}}} \sum_{a_{2}
    \in \mathcal{Q}^{(k)}(a_{1})} \sum_{a_{3}
    \in \mathcal{Q}^{(k)}(a_{2})}\cdots  \sum_{a_{|\mathcal{L}|}
    \in \mathcal{Q}^{(k)}({a_{|\mathcal{L}|-1}})} |\mathcal{D}_{a_{|\mathcal{L}|}}|\nabla f_{{a_{|\mathcal{L}|}}}(\mathbf{w}^{(k-1)}_{{a_{|\mathcal{L}|}}})} {|\mathcal{L}^{(k)}_{{1},1}|}\\&
    +\sum_{{a_{1}}
    \in \mathcal{L}^{(k)}_{{1},{1}}} \sum_{a_{2}
    \in \mathcal{Q}^{(k)}(a_{1})} \sum_{a_{3}
    \in \mathcal{Q}^{(k)}(a_{2})}\cdots  \sum_{a_{|\mathcal{L}|-1}
    \in \mathcal{Q}^{(k)}({a_{|\mathcal{L}|-2}})} \frac{\mathbbm{1}^{(k)}_{\left\{{{Q}(a_{|\mathcal{L}|-1})}\right\}}|\mathcal{Q}^{(k)}(a_{|\mathcal{L}|-1})|  \mathbf{c}^{(k)}_{a'_{|\mathcal{L}|}}}{|\mathcal{L}^{(k)}_{{1},1}|}\\&
    +\sum_{{a_{1}}
    \in \mathcal{L}^{(k)}_{{1},{1}}} \sum_{a_{2}
    \in \mathcal{Q}^{(k)}(a_{1})} \sum_{a_{3}
    \in \mathcal{Q}^{(k)}(a_{2})}\cdots  \sum_{a_{|\mathcal{L}|-2}
    \in \mathcal{Q}^{(k)}({a_{|\mathcal{L}|-3}})} \frac{\mathbbm{1}^{(k)}_{\left\{{{Q}(a_{|\mathcal{L}|-2})}\right\}}|\mathcal{Q}^{(k)}(a_{|\mathcal{L}|-2})|  \mathbf{c}^{(k)}_{a'_{|\mathcal{L}|-1}}}{|\mathcal{L}^{(k)}_{{1},1}|}+
    \\&\vdots
    \\&+\sum_{a_1
    \in \mathcal{L}^{(k)}_{{1},1}} \frac{\mathbbm{1}^{(k)}_{\left\{{{Q}(a_1)}\right\}}|\mathcal{Q}^{(k)}(a_1)|  \mathbf{c}^{(k)}_{a'_2}}{|\mathcal{L}^{(k)}_{{1},1}|}+\mathbbm{1}^{(k)}_{\left\{{{L}_{{1},1}}\right\}} \mathbf{c}^{(k)}_{a'_1}\Bigg],
       \end{aligned}
       \hspace{-20mm}
 \end{equation}
which is used to obtain the next global parameter (due to the existence of the indicator function in the last term of the above expression, the following expression holds regardless of the operating mode of the cluster at layer ${L}_1$):
  \begin{equation}\label{up:proofGlobGrad}
    \textbf{w}^{(k)}= \frac{|\mathcal{L}^{(k)}_{{1},1}| \widehat{\textbf{w}}^{(k)}_{a'_1}}{D}.
 \end{equation}
According to~\eqref{eq:globlossinit}, it can be verified that
\begin{equation}
  \hspace{-27mm}
  \begin{aligned}
  {\displaystyle \displaystyle \sum_{{a_{1}}
    \in \mathcal{L}^{(k)}_{{1},{1}}} \sum_{a_{2}
    \in \mathcal{Q}^{(k)}(a_{1})} \sum_{a_{3}
    \in \mathcal{Q}^{(k)}(a_{2})}\cdots  \sum_{a_{|\mathcal{L}|}
    \in \mathcal{Q}^{(k)}({a_{|\mathcal{L}|-1}})} |\mathcal{D}_{a_{|\mathcal{L}|}}|\nabla f_{a_{|\mathcal{L}|}}(\mathbf{w}^{(k-1)}_{a_{|\mathcal{L}|}})}
    = D\nabla F (\textbf{w}^{(k-1)}).
       \end{aligned}
       \hspace{-20mm}
 \end{equation}
 Replacing the above equation in~\eqref{eq:rootGrad} and performing the update given by~\eqref{up:proofGlobGrad}, we get
 \begin{equation}
  \hspace{-24mm}
  \begin{aligned}
    & {\textbf{w}}^{(k)}=\textbf{w}^{(k-1)}-\beta_{k-1}\Bigg[ \nabla F(\textbf{w}^{(k-1)}) 
    +\frac{1}{D}\Bigg(\sum_{{a_{1}}
    \in \mathcal{L}^{(k)}_{{1},{1}}} \sum_{a_{2}
    \in \mathcal{Q}^{(k)}(a_{1})} \cdots  \sum_{a_{|\mathcal{L}|-1}
    \in \mathcal{Q}^{(k)}({a_{|\mathcal{L}|-2}})} {\mathbbm{1}^{(k)}_{\left\{{{Q}(a_{|\mathcal{L}|-1})}\right\}}|\mathcal{Q}^{(k)}(a_{|\mathcal{L}|-1})|  \mathbf{c}^{(k)}_{a'_{|\mathcal{L}|}}}\\&
    +\sum_{{a_{1}}
    \in \mathcal{L}^{(k)}_{{1},{1}}} \sum_{a_{2}
    \in \mathcal{Q}^{(k)}(a_{1})} \cdots  \sum_{a_{|\mathcal{L}|-2}
    \in \mathcal{Q}^{(k)}({a_{|\mathcal{L}|-3}})} {\mathbbm{1}^{(k)}_{\left\{{{Q}(a_{|\mathcal{L}|-2})}\right\}}|\mathcal{Q}^{(k)}(a_{|\mathcal{L}|-2})|  \mathbf{c}^{(k)}_{a'_{|\mathcal{L}|-1}}}+
    \\&\vdots
    \\&+\sum_{a_1
    \in \mathcal{L}^{(k)}_{{1},1}} {\mathbbm{1}^{(k)}_{\left\{{{Q}(a_1)}\right\}}|\mathcal{Q}^{(k)}(a_1)|  \mathbf{c}^{(k)}_{a'_2}}+\mathbbm{1}^{(k)}_{\left\{{{L}_{{1},1}}\right\}} |\mathcal{L}^{(k)}_{{1},1}|\mathbf{c}^{(k)}_{a'_1}\Bigg)\Bigg].
       \end{aligned}
       \hspace{-20mm}
 \end{equation}
 Using the above equality in \eqref{eq:quadBound}, we have
 \begin{equation}\label{eq:stConSam}
 \hspace{-20mm}
 \begin{aligned}
 &F(\mathbf{w}^{(k)})\leq F(\mathbf{w}^{(k-1)}) -\beta_{k-1}\left( \nabla F (\mathbf{w}^{(k-1)})+\mathbf{c}^{(k)}\right)^\top\nabla F(\mathbf{w}^{(k-1)})
 + \frac{\eta}{2} \beta_{k-1}^2\Vert \nabla F (\mathbf{w}^{(k-1)})+\mathbf{c}^{(k)}\Vert^2\\
     & =F(\mathbf{w}^{(k-1)}) -\beta_{k-1} \norm{\nabla F(\mathbf{w}^{(k-1)})}^2 -\beta_{k-1}\left(\mathbf{c}^{(k)}\right)^\top\nabla F(\mathbf{w}^{(k-1)})
+ \frac{\eta}{2}\beta_{k-1}^2 \Vert \nabla F (\mathbf{w}^{(k-1)})+\mathbf{c}^{(k)}\Vert^2\\
 &= F(\mathbf{w}^{(k-1)}) -\beta_{k-1} \norm{\nabla F(\mathbf{w}^{(k-1)})}^2 -\beta_{k-1}\left(\mathbf{c}^{(k)}\right)^\top\nabla F(\mathbf{w}^{(k-1)})
 + \frac{\eta\beta_{k-1}^2}{2} \Vert\nabla F (\mathbf{w}^{(k-1)})\Vert^2 \\
 &+\beta_{k-1}^2 \eta\left(\nabla F (\mathbf{w}^{(k-1)})^\top \mathbf{c}^{(k)}\right) +\frac{\eta\beta_{k-1}^2}{2}\norm{\mathbf{c}^{(k)}}^2, 
 \end{aligned}
 \hspace{-20mm}
 \end{equation}

 where
 \begin{equation}
 \hspace{-20mm}
     \begin{aligned}
  &\mathbf{c}^{(k)}\triangleq \frac{1}{D}\Bigg(\sum_{{a_{1}}
    \in \mathcal{L}^{(k)}_{{1},{1}}} \sum_{a_{2}
    \in \mathcal{Q}^{(k)}(a_{1})} \sum_{a_{3}
    \in \mathcal{Q}^{(k)}(a_{2})}\cdots  \sum_{a_{|\mathcal{L}|-1}
    \in \mathcal{Q}^{(k)}({a_{|\mathcal{L}|-2}})} {\mathbbm{1}^{(k)}_{\left\{{{Q}(a_{|\mathcal{L}|-1})}\right\}}|\mathcal{Q}^{(k)}(a_{|\mathcal{L}|-1})|  \mathbf{c}^{(k)}_{a'_{|\mathcal{L}|}}}\\&
    +\sum_{{a_{1}}
    \in \mathcal{L}^{(k)}_{{1},{1}}} \sum_{a_{2}
    \in \mathcal{Q}^{(k)}(a_{1})} \sum_{a_{3}
    \in \mathcal{Q}^{(k)}(a_{2})}\cdots  \sum_{a_{|\mathcal{L}|-2}
    \in \mathcal{Q}^{(k)}({a_{|\mathcal{L}|-3}})} {\mathbbm{1}^{(k)}_{\left\{{{Q}(a_{|\mathcal{L}|-2})}\right\}}|\mathcal{Q}^{(k)}(a_{|\mathcal{L}|-2})|  \mathbf{c}^{(k)}_{a'_{|\mathcal{L}|-1}}}+
    \\&\vdots
    \\&+\sum_{a_1
    \in \mathcal{L}^{(k)}_{{1},1}} {\mathbbm{1}^{(k)}_{\left\{{{Q}(a_1)}\right\}}|\mathcal{Q}^{(k)}(a_1)|  \mathbf{c}^{(k)}_{a'_2}}+\mathbbm{1}^{(k)}_{\left\{{{L}_{{1},1}}\right\}} |\mathcal{L}^{(k)}_{{1},1}|\mathbf{c}^{(k)}_{a'_1}\Bigg).
       \end{aligned}
       \hspace{-20mm}
 \end{equation}
 
 Taking the expectation from both hand sides (with respect to the consensus errors) and using the fact that upon using the consensus method, when one node is sampled uniformly at random we have:\footnote{Assume a set of $n$ numbers denoted by $x_1,\cdots,x_n$ with mean $\bar{x}$. Assume that $X$ denotes a random variable with probability mass function $p(X=x_i)=\frac{1}{n}$, $1\leq i\leq n$. It is straightforward to verify that $E(X-\bar{x})=0$.} $\mathbb{E}[\mathbf{c}^{(k)}_{{}{a'_p}}]=\textbf{0}$, $\forall p$. This implies $\mathbb{E}[\mathbf{c}^{(k)}]=\textbf{0}$, $\forall k$, replacing which in~\eqref{eq:stConSam} gives us
 \begin{equation}
     \begin{aligned}
         &\mathbb{E} [F(\mathbf{w}^{(k)})]\leq F(\mathbf{w}^{(k-1)})-(1-\frac{\eta\beta_{k-1}}{2})\beta_{k-1}\norm{\nabla F(\mathbf{w}^{(k-1)})}^2+\frac{\eta \beta_{k-1}^2}{2} E[\Vert\mathbf{c}^{(k)}\Vert^2].
     \end{aligned}
 \end{equation}
 
 Using the fact that  $\beta_0 \leq 1/\eta$, we get $\beta_{k}\leq 1/\eta$, and thus $1-\eta\beta_k/2\geq1/2$, $\forall k$. Using this in the above inequality gives us
 \begin{equation}
     \begin{aligned}
         &\mathbb{E} [F(\mathbf{w}^{(k)})]\leq F(\mathbf{w}^{(k-1)})-\frac{\beta_{k-1}}{2}\norm{\nabla F(\mathbf{w}^{(k-1)})}^2+\frac{\eta \beta_{k-1}^2}{2} E[\Vert\mathbf{c}^{(k)}\Vert^2].
     \end{aligned}
 \end{equation}
Using the strong convexity, we get Polyak-Lojasiewicz inequality~\cite{polyak1963gradient} in the following form: $\Vert\nabla F(\mathbf{w}^{(k-1)})\Vert^2\geq 2\mu [F(\mathbf{w}^{(k-1)})-F(\mathbf{w}^*)]$, using which in the above inequality yields
 \begin{equation}
     \begin{aligned}
         &\mathbb{E} [F(\mathbf{w}^{(k)})]\leq F(\mathbf{w}^{(k-1)})-\beta_{k-1}\mu[F(\mathbf{w}^{(k-1)})-F(\mathbf{w}^*)]+\frac{\eta \beta_{k-1}^2}{2} E[\Vert\mathbf{c}^{(k)}\Vert^2],
     \end{aligned}
 \end{equation}
 or, equivalently
  \begin{equation}
     \begin{aligned}
         &\mathbb{E} [F(\mathbf{w}^{(k)})]-F(\mathbf{w}^*)\leq (1-\beta_{k-1}\mu)[F(\mathbf{w}^{(k-1)})-F(\mathbf{w}^*)]+\frac{\eta \beta_{k-1}^2}{2} E[\Vert\mathbf{c}^{(k)}\Vert^2].
     \end{aligned}
 \end{equation}
 Taking total expectation, with respect to all the consensus errors until iteration $k$, from both hand sides results in
   \begin{equation}\label{eq:expBoundConvDim}
   \hspace{-10mm}
     \begin{aligned}
         &\mathbb{E} [F(\mathbf{w}^{(k)})-F(\mathbf{w}^*)]\leq (1-\beta_{k-1}\mu)\mathbb{E}[F(\mathbf{w}^{(k-1)})-F(\mathbf{w}^*)]+\frac{\eta \beta_{k-1}^2}{2} E[\Vert\mathbf{c}^{(k)}\Vert^2].
     \end{aligned}
      \hspace{-10mm}
 \end{equation}
We continue the proof by carrying out  an induction.  The proposition result trivially holds for iteration $0$. Assume that the result holds for iteration $k$, i.e.,   $\mathbb{E} [F(\mathbf{w}^{(k)})-F(\mathbf{w}^*)]\leq \frac{\Gamma}{k+\lambda}$. We aim to show that the result also holds for iteration $k+1$. Using~\eqref{eq:expBoundConvDim}, we get
  \begin{equation}
   \hspace{-10mm}
     \begin{aligned}
         &\mathbb{E} [F(\mathbf{w}^{(k+1)})-F(\mathbf{w}^*)]\leq (1-\beta_{k}\mu)\mathbb{E}[F(\mathbf{w}^{(k)})-F(\mathbf{w}^*)]+\frac{\eta \beta_{k}^2}{2} E[\Vert\mathbf{c}^{(k+1)}\Vert^2],
     \end{aligned}
      \hspace{-10mm}
 \end{equation}
 which results in
    \begin{equation}\label{eq:expBoundConvDim2}
   \hspace{-10mm}
     \begin{aligned}
         &\mathbb{E} [F(\mathbf{w}^{(k+1)})-F(\mathbf{w}^*)]\leq (1-\frac{\alpha}{k+\lambda}\mu)\frac{\Gamma}{k+\lambda}+ \frac{\eta \alpha^2}{2(k+\lambda)^2} \mathbb{E}[\Vert\mathbf{c}^{(k+1)}\Vert^2]\\&=\left(\frac{k+\lambda-\alpha\mu}{(k+\lambda)^2}\right)\Gamma+ \frac{\eta \alpha^2}{2(k+\lambda)^2} \mathbb{E}[\Vert\mathbf{c}^{(k+1)}\Vert^2]\\&=\left(\frac{k+\lambda-1}{(k+\lambda)^2}\right)\Gamma-\frac{\alpha\mu-1}{(k+\lambda)^2}\Gamma+ \frac{\eta \alpha^2}{2(k+\lambda)^2} \mathbb{E}[\Vert\mathbf{c}^{(k+1)}\Vert^2]    . 
     \end{aligned}
      \hspace{-10mm}
 \end{equation}
  Note that using a similar method as Appendix~\ref{app:ConsGen}, we can get\footnote{Note that if for every realization of random variable $X$, inequality $\norm{X}^2<y$ holds, then we get: $\mathbb{E}[ \norm{X}^2]<y$.}
 \begin{equation}
 \hspace{-15mm}
     \begin{aligned}
  &\mathbb{E}\left[\norm{\mathbf{c}^{(k+1)}}^2\right]\leq \frac{{\Phi}}{D^2} \Bigg[\\&\sum_{{a_{1}}
    \in \mathcal{L}^{(k+1)}_{{1},{1}}} \sum_{a_{2}
    \in \mathcal{Q}^{(k+1)}(a_{1})} \cdots \hspace{-3mm} \sum_{a_{|\mathcal{L}|-1}
    \in \mathcal{Q}^{(k+1)}({a_{|\mathcal{L}|-2}})} \mathbbm{1}^{(k+1)}_{ \left\{{Q}(a_{|\mathcal{L}|-1})\right\}}|\mathcal{Q}^{(k+1)}(a_{|\mathcal{L}|-1})|^3
    \left(\lambda^{(k+1)}_{{Q}(a_{|\mathcal{L}|-1})}\right)^{2\theta^{(k+1)}_{{Q}(a_{|\mathcal{L}|-1})}} \left(\Upsilon^{(k+1)}_{{Q}(a_{|\mathcal{L}|-1})}\right)^2
    \\&+\sum_{{a_{1}}
    \in \mathcal{L}^{(k+1)}_{1,1}} \sum_{a_{2}
    \in \mathcal{Q}^{(k+1)}(a_{1})} \cdots  \hspace{-3mm} \sum_{a_{|\mathcal{L}|-2}
    \in \mathcal{Q}^{(k+1)}({a_{|\mathcal{L}|-3}})} \mathbbm{1}^{(k+1)}_{ \left\{{Q}(a_{|\mathcal{L}|-2})\right\}}|\mathcal{Q}^{(k+1)}(a_{|\mathcal{L}|-2})|^3  \left(\lambda^{(k+1)}_{{Q}(a_{|\mathcal{L}|-2})}\right)^{2\theta^{(k+1)}_{{Q}(a_{|\mathcal{L}|-2})}} \left(\Upsilon^{(k+1)}_{{Q}(a_{|\mathcal{L}|-2})}\right)^2 \\&+\cdots+\sum_{a_1
    \in \mathcal{L}^{(k+1)}_{1,1}} \mathbbm{1}^{(k+1)}_{ \left\{{Q}^{(k+1)}(a_1)\right\}}|\mathcal{Q}^{(k+1)}(a_1)|^3 \left(\lambda^{(k+1)}_{{Q}(a_1)}\right)^{2\theta^{(k+1)}_{{Q}(a_1)}} \left(\Upsilon^{(k+1)}_{{Q}(a_1)}\right)^2+\mathbbm{1}^{(k+1)}_{\left\{{{L}_{1,1}}\right\}}|\mathcal{L}^{(k+1)}_{1,1}|^3 \left(\lambda^{(k+1)}_{{L}_{1,1}}\right)^{2\theta^{(k+1)}_{{L}_{1,1}}} \left(\Upsilon^{(k+1)}_{{L}_{1,1}}\right)^2\Bigg].
       \end{aligned}
       \hspace{-20mm}
 \end{equation}
 Using the number of D2D rounds given in the proposition, similar to the approach taken in Appendix~\ref{app:boundedConsConv} it can be verified that $E[\Vert\mathbf{c}^{(k+1)}\Vert^2] \leq C= \frac{{\Phi}}{D^2}\sum_{j=0}^{|\mathcal{L}|-1} N_j \sigma_{j+1}$, $\forall k$. Using this and the definition of $\Gamma$ in~\eqref{eq:gammadef}, we get: $\Gamma\geq \frac{\eta\alpha^2 C}{2(\alpha \mu -1)}$, $\forall k$. Using this result in the last line of~\eqref{eq:expBoundConvDim2}, we get
   \begin{equation}\label{eq:expBoundConvDim3}
   \hspace{-10mm}
     \begin{aligned}
         &\mathbb{E} [F(\mathbf{w}^{(k+1)})-F(\mathbf{w}^*)]     \leq \left(\frac{k+\lambda-1}{(k+\lambda)^2}\right)\Gamma.
     \end{aligned}
      \hspace{-10mm}
 \end{equation}
 Note that since $k+\lambda >1$, we have $(k+\lambda)^2\geq (k+\lambda-1)(k+\lambda+1)$. Using this fact in~\eqref{eq:expBoundConvDim3}, we obtain
 \begin{equation}
     \mathbb{E} [F(\mathbf{w}^{(k+1)})-F(\mathbf{w}^*)]\leq \left(\frac{1}{k+\lambda+1}\right)\Gamma,
 \end{equation}
 which completes the induction and thus the proof.

 \begin{table*}[b]
\begin{minipage}{0.99\textwidth}
\vspace{-.2mm}
\hrulefill
{\footnotesize
\begin{equation}\label{eq:ClustSamp}
\vspace{-2mm}
 \hspace{-14mm}
 \begin{aligned}
     &F(\mathbf{w}^{(k)})-F(\mathbf{w}^{*})\leq
      \left[ \prod_{l=1}^{k} \left(1-\frac{\mu}{\eta}+8\frac{c_2}{ D^2}\left({D-D^{(l)}_s}\right)^2  \right) \right]\left(F(\mathbf{w}^{(0)})- F(\mathbf{w}^*)\right)+\Vast(\sum_{t=1}^{k} \left[ \prod_{l=t+1}^{k} \left(1-\frac{\mu}{\eta}+8\frac{c_2}{ D^2}\left({D-D^{(l)}_s}\right)^2  \right) \right]\\
 &\vast(\frac{\eta{\Phi}}{\left(D^{(t)}_s\right)^2} \Big[ \sum_{{a_{1}}
    \in \mathcal{L}^{(k)}_{{1},{1}}} \sum_{a_{2}
    \in \mathcal{Q}^{(k)}(a_{1})} \cdots  \sum_{a_{|\mathcal{L}|-1}
    \in \mathcal{Q}^{(k)}({a_{|\mathcal{L}|-2}})} \mathbbm{1}^{(t)}_{ \left\{\mathcal{Q}(a_{|\mathcal{L}|-1})\right\}}|\mathcal{Q}^{(k)}(a_{|\mathcal{L}|-1})|^3
    \left(\lambda^{(t)}_{\mathcal{Q}(a_{|\mathcal{L}|-1})}\right)^{2\theta^{(t)}_{\mathcal{Q}(a_{|\mathcal{L}|-1})}} \left(\Upsilon^{(t)}_{\mathcal{Q}(a_{|\mathcal{L}|-1})}\right)^2
    \\&+\sum_{{a_{1}}
    \in \mathcal{L}^{(k)}_{{1},{1}}} \sum_{a_{2}
    \in \mathcal{Q}^{(k)}(a_{1})} \cdots  \sum_{a_{|\mathcal{L}|-2}
    \in \mathcal{Q}^{(k)}({a_{|\mathcal{L}|-3}})} \mathbbm{1}^{(t)}_{ \left\{\mathcal{Q}(a_{|\mathcal{L}|-2})\right\}}|\mathcal{Q}^{(k)}(a_{|\mathcal{L}|-2})|^3  \left(\lambda^{(t)}_{\mathcal{Q}(a_{|\mathcal{L}|-2})}\right)^{2\theta^{(t)}_{\mathcal{Q}(a_{|\mathcal{L}|-2})}} \left(\Upsilon^{(t)}_{\mathcal{Q}(a_{|\mathcal{L}|-2})}\right)^2 +\cdots\\&+\sum_{a_1
    \in \mathcal{L}^{(k)}_{{1},1}} \mathbbm{1}^{(t)}_{ \left\{\mathcal{Q}(a_1)\right\}}|\mathcal{Q}^{(k)}(a_1)|^3 \left(\lambda^{(t)}_{\mathcal{Q}(a_1)}\right)^{2\theta^{(t)}_{\mathcal{Q}(a_1)}} \left(\Upsilon^{(t)}_{\mathcal{Q}(a_1)}\right)^2\\&+\mathbbm{1}^{(t)}_{\left\{{{L}_{{1},1}}\right\}}|\mathcal{L}^{(k)}_{{1},1}|^3 \left(\lambda^{(t)}_{\mathcal{L}_{{1},1}}\right)^{2\theta^{(t)}_{\mathcal{L}_{{1},1}}} \left(\Upsilon^{(t)}_{\mathcal{L}_{{1},1}}\right)^2\Big]+\frac{4}{\eta}\left(\frac{D-D^{(t)}_s}{D}\right)^2 c_1\vast)\Vast)
     \end{aligned}
     \hspace{-15mm}
 \end{equation}
 }
 \end{minipage}
 \vspace{-3mm}
 \end{table*}
   \pagebreak
\section{Cluster Sampling}
\label{app:cluster}
In a system of a million/billion users, one technique that a main server can use to  reduce the network load is to engage a portion of the devices in each global iteration. We realize this in FogL via \textit{cluster sampling} using which at each global iteration, a portion of the clusters of the bottom-most layer are engaged in model training, which we call them as \textit{active clusters}. We assume that at each global iteration $k$, the main server engages a set of $|\mathcal{S}^{(k)}|$ clusters in the learning, where each element of the set $\mathcal{S}^{(k)}$ corresponds to one cluster in the bottom-most layer. Consequently, we partition the nodes in different layers into \textit{active nodes} (those that are through the path between an active cluster and the main server) and \textit{passive nodes}. Similarly, for the clusters of the middle layers, if the cluster contains at least one active node, it is called an \textit{active cluster}. To capture these dynamics, with some abuse of notation, let $\mathbbm{1}^{(k)}_{ \left\{{C}\right\}}$ take the value of $1$ if cluster ${C}$ is both in active mode and operates in LUT mode in global aggregation $k$, and $0$ otherwise.  To conduct analysis, in addition to our assumptions made in Assumptions~\ref{assum:genConv} and~\ref{assump:cons}, we also consider the following assumption that is common in stochastic optimization literature~\cite{bertsekas1996neuro}:
    \begin{equation}\label{eq:berts}
    \hspace{-7mm}
        \exists c_1\geq 0,c_2\geq 1: \Vert\nabla f_i(x)\Vert^2\leq c_1+c_2\Vert\nabla F(x)\Vert^2,~~\forall i,x.
         \hspace{-7mm}
    \end{equation}
  \begin{proposition}\label{prop:deviceSamplingCons}
 For global iteration $k$ of {\tt MH-FL} with cluster sampling, the upper bound of convergence of  the objective function is given by~\eqref{eq:ClustSamp}, where $D^{(k)}_S$ denotes the total number of data points of the sampled devices at iteration $k$, i.e., $D^{(k)}_s=\sum_{n
    \in \mathcal{N}_{|\mathcal{L}|}} \mathbbm{1}^{(k)}_{ \left\{\mathcal{B}(n)\right\}}|\mathcal{D}_n|$, with $\mathcal{B}(n)$ referring to the cluster that node $n$ belongs to.\footnote{It is assumed that $\prod_{j=k+1}^{k} c_j =1$, $\forall c_j$.}
\end{proposition}
 \begin{proof}
To find the relationship between $\textbf{w}^{(k)}$ and $\textbf{w}^{(k-1)}$, we follow the procedure described in the main text. Let $\mathbbm{1}^{(k)}_{ \left\{{S}({C})\right\}}$ take the value of $1$ when cluster ${C}$ is in active mode in global aggregation $k$, and $0$ otherwise. Also, with some abuse of notation, let $\mathbbm{1}^{(k)}_{ \left\{{C}\right\}}$ take the value of $1$ if cluster ${C}$ is both in active mode and operates in LUT mode in global aggregation $k$, and $0$ otherwise.
   \begin{table*}[b]
\begin{minipage}{0.99\textwidth}
 \hrulefill
\begin{equation}\label{eq:rootSample}
  \hspace{-22mm}
  \begin{aligned}
    & \widehat{\mathbf{w}}^{(k)}_{{a'_1}}=\frac{\displaystyle \sum_{{a_{1}}
    \in \mathcal{L}^{(k)}_{{1},{1}}} \sum_{a_{2}
    \in \mathcal{Q}^{(k)}(a_{1})} \sum_{a_{3}
    \in \mathcal{Q}^{(k)}(a_{2})}\cdots  \sum_{a_{|\mathcal{L}|}
    \in \mathcal{Q}^{(k)}({a_{|\mathcal{L}|-1}})} \mathbbm{1}^{(k)}_{ \left\{{S}({Q}({a_{|\mathcal{L}|-1}}))\right\}}|\mathcal{D}_{a_{|\mathcal{L}|}}|\mathbf{w}^{(k-1)}_{{a_{|\mathcal{L}|}}}} {|\mathcal{L}^{(k)}_{{1},1}|}\\& 
    -\frac{\displaystyle \sum_{{a_{1}}
    \in \mathcal{L}^{(k)}_{{1},{1}}} \sum_{a_{2}
    \in \mathcal{Q}^{(k)}(a_{1})} \sum_{a_{3}
    \in \mathcal{Q}^{(k)}(a_{2})}\cdots  \sum_{a_{|\mathcal{L}|}
    \in \mathcal{Q}^{(k)}({a_{|\mathcal{L}|-1}})}\mathbbm{1}^{(k)}_{ \left\{{S}({Q}({a_{|\mathcal{L}|-1}}))\right\}}\beta |\mathcal{D}_{a_{|\mathcal{L}|}}|\nabla f_{{a_{|\mathcal{L}|}}}(\mathbf{w}^{(k-1)}_{{a_{|\mathcal{L}|}}})} {|\mathcal{L}^{(k)}_{{1},1}|}\\&
    +\sum_{{a_{1}}
    \in \mathcal{L}^{(k)}_{{1},{1}}} \sum_{a_{2}
    \in \mathcal{Q}^{(k)}(a_{1})} \sum_{a_{3}
    \in \mathcal{Q}^{(k)}(a_{2})}\cdots  \sum_{a_{|\mathcal{L}|-1}
    \in \mathcal{Q}^{(k)}({a_{|\mathcal{L}|-2}})} \frac{\mathbbm{1}^{(k)}_{\left\{{{Q}(a_{|\mathcal{L}|-1})}\right\}}|\mathcal{Q}^{(k)}(a_{|\mathcal{L}|-1})|  \mathbf{c}^{(k)}_{a'_{|\mathcal{L}|}}}{|\mathcal{L}^{(k)}_{{1},1}|}\\&
    +\sum_{{a_{1}}
    \in \mathcal{L}^{(k)}_{{1},{1}}} \sum_{a_{2}
    \in \mathcal{Q}^{(k)}(a_{1})} \sum_{a_{3}
    \in \mathcal{Q}^{(k)}(a_{2})}\cdots  \sum_{a_{|\mathcal{L}|-2}
    \in \mathcal{Q}^{(k)}({a_{|\mathcal{L}|-3}})} \frac{\mathbbm{1}^{(k)}_{\left\{{{Q}(a_{|\mathcal{L}|-2})}\right\}}|\mathcal{Q}^{(k)}(a_{|\mathcal{L}|-2})|  \mathbf{c}^{(k)}_{a'_{|\mathcal{L}|-1}}}{|\mathcal{L}^{(k)}_{{1},1}|}+
    \\&\vdots
    \\&+\sum_{a_1
    \in \mathcal{L}^{(k)}_{{1},1}} \frac{\mathbbm{1}^{(k)}_{\left\{{{Q}(a_1)}\right\}}|\mathcal{Q}^{(k)}(a_1)|  \mathbf{c}^{(k)}_{a'_2}}{|\mathcal{L}^{(k)}_{{1},1}|}+\mathbbm{1}^{(k)}_{\left\{{{L}_{{1},1}}\right\}} \mathbf{c}^{(k)}_{a'_1}
       \end{aligned}
       \hspace{-20mm}
 \end{equation}
\end{minipage}
\end{table*}
  It can be verified that, at global iteration $k$, the parameter of the  node located in the ${L}_1$ sampled by the main server, referred to as ${{a'_1}}$, is given by~\eqref{eq:rootSample}, which
  is used by the server to obtain the next global parameter as follows:
  \begin{equation}\label{up:proofGlobSamp}
   \textbf{w}^{(k)}= \frac{|\mathcal{L}^{(k)}_{{1},1}| {\mathbf{w}}^{(k)}_{{{a'_1}}}}{D^{(k)}_s},
 \end{equation}
where $D^{(k)}_s=\sum_{n
    \in \mathcal{N}_{|\mathcal{L}|}} \mathbbm{1}^{(k)}_{ \left\{\mathcal{B}(n)\right\}}|\mathcal{D}_n|$, with $\mathcal{B}(n)$ referring to the cluster that the node $n$ belongs to, is the total number of data points available at the sampled devices at global aggregation $k$, which is assumed to be known to the server (in this case the server needs the knowledge of the number of data points available at the active clusters).
Following a similar procedure described in Appendix~\ref{app:ConsGen}, we obtain~\eqref{eq:sample1ww}.
\begin{table*}[t]
\begin{minipage}{0.99\textwidth}
\begin{equation}\label{eq:sample1ww}
  \hspace{-22mm}
  \begin{aligned}
    & {\textbf{w}}^{(k)}=\mathbf{w}^{(k-1)} - 
  \displaystyle  \sum_{{a_{1}}
    \in \mathcal{L}^{(k)}_{{1},{1}}} \sum_{a_{2}
    \in \mathcal{Q}^{(k)}(a_{1})} \sum_{a_{3}
    \in \mathcal{Q}^{(k)}(a_{2})}\cdots  \sum_{a_{|\mathcal{L}|}
    \in \mathcal{Q}^{(k)}({a_{|\mathcal{L}|-1}})}\mathbbm{1}^{(k)}_{ \left\{{S}({Q}({a_{|\mathcal{L}|-1}}))\right\}} \beta \frac{|\mathcal{D}_{{a_{|\mathcal{L}|}}}|}{D^{(k)}_s}\nabla f_{{a_{|\mathcal{L}|}}}(\mathbf{w}^{(k-1)}_{{a_{|\mathcal{L}|}}})
    \\&
    +\frac{1}{D^{(k)}_s}\Bigg[\sum_{{a_{1}}
    \in \mathcal{L}^{(k)}_{{1},{1}}} \sum_{a_{2}
    \in \mathcal{Q}^{(k)}(a_{1})} \sum_{a_{3}
    \in \mathcal{Q}^{(k)}(a_{2})}\cdots  \sum_{a_{|\mathcal{L}|-1}
    \in \mathcal{Q}^{(k)}({a_{|\mathcal{L}|-2}})} {\mathbbm{1}^{(k)}_{\left\{{{Q}(a_{|\mathcal{L}|-1})}\right\}}|\mathcal{Q}^{(k)}(a_{|\mathcal{L}|-1})|  \mathbf{c}^{(k)}_{a'_{|\mathcal{L}|}}}\\&
    +\sum_{{a_{1}}
    \in \mathcal{L}^{(k)}_{{1},{1}}} \sum_{a_{2}
    \in \mathcal{Q}^{(k)}(a_{1})} \sum_{a_{3}
    \in \mathcal{Q}^{(k)}(a_{2})}\cdots  \sum_{a_{|\mathcal{L}|-2}
    \in \mathcal{Q}^{(k)}({a_{|\mathcal{L}|-3}})} {\mathbbm{1}^{(k)}_{\left\{{{Q}(a_{|\mathcal{L}|-2})}\right\}}|\mathcal{Q}^{(k)}(a_{|\mathcal{L}|-2})|  \mathbf{c}^{(k)}_{a'_{|\mathcal{L}|-1}}}+
    \\&\vdots
    \\&+\sum_{a_1
    \in \mathcal{L}^{(k)}_{{1},1}} {\mathbbm{1}^{(k)}_{\left\{{{Q}(a_1)}\right\}}|\mathcal{Q}^{(k)}(a_1)|  \mathbf{c}^{(k)}_{a'_2}}+\mathbbm{1}^{(k)}_{\left\{{{L}_{{1},1}}\right\}} |\mathcal{L}^{(k)}_{{1},1}|\mathbf{c}^{(k)}_{a'_1}\Bigg]
       \end{aligned}
       \hspace{-20mm}
 \end{equation}
 \hrulefill
 \end{minipage}
 \end{table*}
 Let us define $\mathbf{\varpi}^{(k)}$ as follows:
 \begin{equation}
 \hspace{-10mm}
 \begin{aligned}
     &\mathbf{\varpi}^{(k)}
\triangleq\frac{1}{D^{(k)}_s}\Bigg[\sum_{{a_{1}}
    \in \mathcal{L}^{(k)}_{{1},{1}}} \sum_{a_{2}
    \in \mathcal{Q}^{(k)}(a_{1})} \sum_{a_{3}
    \in \mathcal{Q}^{(k)}(a_{2})}\cdots  \sum_{a_{|\mathcal{L}|-1}
    \in \mathcal{Q}^{(k)}({a_{|\mathcal{L}|-2}})} {\mathbbm{1}^{(k)}_{\left\{{{Q}(a_{|\mathcal{L}|-1})}\right\}}|\mathcal{Q}^{(k)}(a_{|\mathcal{L}|-1})|  \mathbf{c}^{(k)}_{a'_{|\mathcal{L}|}}}\\&
    +\sum_{{a_{1}}
    \in \mathcal{L}^{(k)}_{{1},{1}}} \sum_{a_{2}
    \in \mathcal{Q}^{(k)}(a_{1})} \sum_{a_{3}
    \in \mathcal{Q}^{(k)}(a_{2})}\cdots  \sum_{a_{|\mathcal{L}|-2}
    \in \mathcal{Q}^{(k)}({a_{|\mathcal{L}|-3}})} {\mathbbm{1}^{(k)}_{\left\{{{Q}(a_{|\mathcal{L}|-2})}\right\}}|\mathcal{Q}^{(k)}(a_{|\mathcal{L}|-2})|  \mathbf{c}^{(k)}_{a'_{|\mathcal{L}|-1}}}+
    \\&\vdots
    \\&+\sum_{a_1
    \in \mathcal{L}^{(k)}_{{1},1}} {\mathbbm{1}^{(k)}_{\left\{{{Q}(a_1)}\right\}}|\mathcal{Q}^{(k)}(a_1)|  \mathbf{c}^{(k)}_{a'_2}}+\mathbbm{1}^{(k)}_{\left\{{{L}_{{1},1}}\right\}} |\mathcal{L}^{(k)}_{{1},1}|\mathbf{c}^{(k)}_{a'_1}\Bigg]. \end{aligned}\hspace{-10mm}\end{equation}
By adding and subtracting a term, we rewrite \eqref{eq:sample1ww} as follows:
\begin{equation}\label{eq:sample1w}
  \hspace{-22mm}
  \begin{aligned}
    & {\textbf{w}}^{(k)}=\mathbf{w}^{(k-1)} -
  \displaystyle \sum_{{a_{1}}
    \in \mathcal{L}^{(k)}_{{1},{1}}} \sum_{a_{2}
    \in \mathcal{Q}^{(k)}(a_{1})} \sum_{a_{3}
    \in \mathcal{Q}^{(k)}(a_{2})}\cdots  \sum_{a_{|\mathcal{L}|}
    \in \mathcal{Q}^{(k)}({a_{|\mathcal{L}|-1}})}\mathbbm{1}^{(k)}_{ \left\{{S}({Q}({a_{|\mathcal{L}|-1}}))\right\}} \beta \frac{|\mathcal{D}_{{a_{|\mathcal{L}|}}}|}{D^{(k)}_s}\nabla f_{{a_{|\mathcal{L}|}}}(\mathbf{w}^{(k-1)}_{{a_{|\mathcal{L}|}}})\\&
    +\displaystyle \sum_{{a_{1}}
    \in \mathcal{L}^{(k)}_{{1},{1}}} \sum_{a_{2}
    \in \mathcal{Q}^{(k)}(a_{1})} \sum_{a_{3}
    \in \mathcal{Q}^{(k)}(a_{2})}\cdots  \sum_{a_{|\mathcal{L}|}
    \in \mathcal{Q}^{(k)}({a_{|\mathcal{L}|-1}})}\beta \frac{|\mathcal{D}_{{a_{|\mathcal{L}|}}}|}{D}\nabla f_{{a_{|\mathcal{L}|}}}(\mathbf{w}^{(k-1)}_{{a_{|\mathcal{L}|}}})
    \\&-\displaystyle \sum_{{a_{1}}
    \in \mathcal{L}^{(k)}_{{1},{1}}} \sum_{a_{2}
    \in \mathcal{Q}^{(k)}(a_{1})} \sum_{a_{3}
    \in \mathcal{Q}^{(k)}(a_{2})}\cdots  \sum_{a_{|\mathcal{L}|}
    \in \mathcal{Q}^{(k)}({a_{|\mathcal{L}|-1}})} \beta \frac{|\mathcal{D}_{{a_{|\mathcal{L}|}}}|}{D}\nabla f_{{a_{|\mathcal{L}|}}}(\mathbf{w}^{(k-1)}_{{a_{|\mathcal{L}|}}})
    +\mathbf{\varpi}^{(k)},
       \end{aligned}
       \hspace{-20mm}
 \end{equation}
 or equivalently
 \begin{equation}\label{eq:sample1w}
  \hspace{-25mm}
  \begin{aligned}
    & {\textbf{w}}^{(k)}=\mathbf{w}^{(k-1)} -\beta \displaystyle \nabla F(\mathbf{w}^{(k-1)})\\& 
  -\displaystyle \sum_{{a_{1}}
    \in \mathcal{L}^{(k)}_{{1},{1}}} \sum_{a_{2}
    \in \mathcal{Q}^{(k)}(a_{1})} \sum_{a_{3}
    \in \mathcal{Q}^{(k)}(a_{2})}\cdots  \sum_{a_{|\mathcal{L}|}
    \in \mathcal{Q}^{(k)}({a_{|\mathcal{L}|-1}})}\mathbbm{1}^{(k)}_{ \left\{{S}({Q}({a_{|\mathcal{L}|-1}}))\right\}} \beta \frac{|\mathcal{D}_{{a_{|\mathcal{L}|}}}|}{D^{(k)}_s}\nabla f_{{a_{|\mathcal{L}|}}}(\mathbf{w}^{(k-1)}_{{a_{|\mathcal{L}|}}})
    \\&+\displaystyle \sum_{{a_{1}}
    \in \mathcal{L}^{(k)}_{{1},{1}}} \sum_{a_{2}
    \in \mathcal{Q}^{(k)}(a_{1})} \sum_{a_{3}
    \in \mathcal{Q}^{(k)}(a_{2})}\cdots  \sum_{a_{|\mathcal{L}|}
    \in \mathcal{Q}^{(k)}({a_{|\mathcal{L}|-1}})} \beta \frac{|\mathcal{D}_{{a_{|\mathcal{L}|}}}|}{D}\nabla f_{{a_{|\mathcal{L}|}}}(\mathbf{w}^{(k-1)}_{{a_{|\mathcal{L}|}}})
    +\mathbf{\varpi}^{(k)}.
       \end{aligned}
       \hspace{-20mm}
 \end{equation}
 Let us define $\varrho^{({k})}$ as follows:
\begin{equation}\label{eq:sample1w}
  \hspace{-22mm}
  \begin{aligned}
  &\varrho^{({k})}\triangleq
  {\beta}\Bigg[-\displaystyle \sum_{{a_{1}}
    \in \mathcal{L}^{(k)}_{{1},{1}}} \sum_{a_{2}
    \in \mathcal{Q}^{(k)}(a_{1})} \sum_{a_{3}
    \in \mathcal{Q}^{(k)}(a_{2})}\cdots  \sum_{a_{|\mathcal{L}|}
    \in \mathcal{Q}^{(k)}({a_{|\mathcal{L}|-1}})}\mathbbm{1}^{(k)}_{ \left\{{S}({Q}({a_{|\mathcal{L}|-1}}))\right\}}  \frac{|\mathcal{D}_{{a_{|\mathcal{L}|}}}|}{D^{(k)}_s}\nabla f_{{a_{|\mathcal{L}|}}}(\mathbf{w}^{(k-1)}_{{a_{|\mathcal{L}|}}})
    \\&+\displaystyle \sum_{{a_{1}}
    \in \mathcal{L}^{(k)}_{{1},{1}}} \sum_{a_{2}
    \in \mathcal{Q}^{(k)}(a_{1})} \sum_{a_{3}
    \in \mathcal{Q}^{(k)}(a_{2})}\cdots  \sum_{a_{|\mathcal{L}|}
    \in \mathcal{Q}^{(k)}({a_{|\mathcal{L}|-1}})} \frac{|\mathcal{D}_{{a_{|\mathcal{L}|}}}|}{D}\nabla f_{{a_{|\mathcal{L}|}}}(\mathbf{w}^{(k-1)}_{{a_{|\mathcal{L}|}}})
    + \frac{1}{\beta}\mathbf{\varpi}^{(k)}\Bigg].
       \end{aligned}
       \hspace{-20mm}
 \end{equation}
 For global iteration $k$, let $\bar{\mathcal{S}}^{(k)}$ denotes the set of passive clusters, which is the complementary set of ${\mathcal{S}}^{(k)}$, i.e., $\bar{\mathcal{S}}^{(k)}\cup {\mathcal{S}}^{(k)} = \mathcal{L}_{|\mathcal{L}|}$, $\bar{\mathcal{S}}^{(k)}\cap {\mathcal{S}}^{(k)}= \emptyset$, where $\mathcal{L}_{|\mathcal{L}|}$ denotes the set of all clusters located in the bottom-most layer.
 Let $\mathbbm{1}^{(k)}_{ \left\{\bar{S}({C})\right\}}$ take the value of $1$ when cluster ${C}$ is in passive mode in global aggregation $k$, and $0$ otherwise.
Following the procedure described in the proof of Appendix~\ref{app:ConsGen}, we first aim to bound $\mathbb{E}[\Vert\varrho^{({k})}\Vert^2]$. The procedure is described in~\eqref{eq:longConvConsSmapl3}.  In that series of simplifications in~\eqref{eq:longConvConsSmapl3}, the triangle inequality is applied repeatedly. In inequality (a), we have used the fact that $(\Vert\mathbf{a}\Vert + \Vert\mathbf{b}\Vert)^2\leq 2( \Vert \mathbf{a}\Vert^2 + \Vert \mathbf{b}\Vert^2)  $,
in inequality (b) we have used the fact that $\frac{1}{D^{(k)}_s}=\frac{1}{D}-\frac{D^{(k)}_s-D}{(D) (D^{(k)}_s)}$, in (c) we have used~\eqref{eq:berts}, and in inequality (d) we have used the smoothness definition in Assumption~\ref{assum:genConv} that can also be written as
\begin{equation}
    \hspace{-4mm}   F(\mathbf{y})\leq F(\mathbf{x})+(\mathbf{y}-\mathbf{x})^\top \nabla F(\mathbf{x}) +\frac{\eta}{2}\norm{\mathbf{y}-\mathbf{x}}^2,~\forall \mathbf{x},\mathbf{y},
\end{equation}
 minimizing the both hand sides of which results in: $\norm{\nabla F(\mathbf{w})}^2 \leq  2\eta( F(\mathbf{w})- F(\mathbf{w}^*))$, $\forall \mathbf{w}$. Note that $\Vert \mathbf{\varpi}^{(k)}\Vert^2$ can be obtained similar to Appendix~\ref{app:ConsGen} as~\eqref{eq:samp11}.
\begin{table*}[t]
\begin{minipage}{0.99\textwidth}
\begin{equation}\label{eq:samp11}
\hspace{-9mm}
\begin{aligned}
   &\Vert \mathbf{\varpi}^{(k)}\Vert^2 \leq \frac{{\Phi}}{\left(D^{(k)}_s\right)^2} \Bigg[\\&\sum_{{a_{1}}
    \in \mathcal{L}^{(k)}_{{1},{1}}} \sum_{a_{2}
    \in \mathcal{Q}^{(k)}(a_{1})} \cdots \hspace{-3mm} \sum_{a_{|\mathcal{L}|-1}
    \in \mathcal{Q}^{(k)}({a_{|\mathcal{L}|-2}})} \mathbbm{1}^{(k)}_{ \left\{{Q}(a_{|\mathcal{L}|-1})\right\}}|\mathcal{Q}^{(k)}(a_{|\mathcal{L}|-1})|^3
    \left(\lambda^{(k)}_{{Q}(a_{|\mathcal{L}|-1})}\right)^{2\theta^{(k)}_{{Q}(a_{|\mathcal{L}|-1})}} \left(\Upsilon^{(k)}_{{Q}(a_{|\mathcal{L}|-1})}\right)^2
    \\&+\sum_{{a_{1}}
    \in \mathcal{L}^{(k)}_{{1},{1}}} \sum_{a_{2}
    \in \mathcal{Q}^{(k)}(a_{1})} \cdots  \hspace{-3mm} \sum_{a_{|\mathcal{L}|-2}
    \in \mathcal{Q}^{(k)}({a_{|\mathcal{L}|-3}})} \mathbbm{1}^{(k)}_{ \left\{{Q}(a_{|\mathcal{L}|-2})\right\}}|\mathcal{Q}^{(k)}(a_{|\mathcal{L}|-2})|^3  \left(\lambda^{(k)}_{{Q}(a_{|\mathcal{L}|-2})}\right)^{2\theta^{(k)}_{{Q}(a_{|\mathcal{L}|-2})}} \left(\Upsilon^{(k)}_{{Q}(a_{|\mathcal{L}|-2})}\right)^2 \\&+\cdots+\sum_{a_1
    \in \mathcal{L}^{(k)}_{{1},1}} \mathbbm{1}^{(k)}_{ \left\{{Q}(a_1)\right\}}|\mathcal{Q}^{(k)}(a_1)|^3 \left(\lambda^{(k)}_{{Q}(a_1)}\right)^{2\theta^{(k)}_{{Q}(a_1)}} \left(\Upsilon^{(k)}_{{Q}(a_1)}\right)^2+\mathbbm{1}^{(k)}_{\left\{{{L}_{{1},1}}\right\}}|\mathcal{L}^{(k)}_{{1},1}|^3 \left(\lambda^{(k)}_{{L}_{{1},1}}\right)^{2\theta^{(k)}_{{L}_{{1},1}}} \left(\Upsilon^{(k)}_{{L}_{{1},1}}\right)^2\Bigg]
    \end{aligned}
    \hspace{-15mm}
\end{equation}
\hrulefill
\end{minipage}
\end{table*}
Replacing this with $\beta=\frac{1}{\eta}$ in the bound in~\eqref{eq:longConvConsSmapl3}, and following the procedure of proof in Appendix~\ref{app:ConsGen}, we get~\eqref{eq:samp22}, which can be recursively expanded to get the bound in the proposition statement.
 \end{proof}
 \begin{remark} The methodology used to derive all the previous results regarding the convergence and the number of D2D can be studied for this scenario with cluster sampling, which we leave as future work. 
 One key observation from~\eqref{eq:ClustSamp} is that upon increasing the number of active clusters, often resulting in increasing $D_s^{(k)}$, $\forall k$, the right hand side of~\eqref{eq:ClustSamp} starts to decrease, which implies a higher training accuracy, and the similarity between the bounds \eqref{eq:ConsTh1} and~\eqref{eq:ClustSamp} increases. In the limiting case $D_s^{(k)}=D$, $\forall k$, bound~\eqref{eq:ClustSamp} can be written similarly to ~\eqref{eq:ClustSamp}, where $\frac{\eta{\Phi}}{2D^2}$ in~\eqref{eq:ConsTh1} would be replaced by a larger value $\frac{\eta{\Phi}}{D^2}$. 
 \end{remark}
\begin{table*}[t]
\begin{minipage}{0.99\textwidth}
\begin{equation}\label{eq:samp22}
\hspace{-7mm}
 \begin{aligned}
     &F(\mathbf{w}^{(k)})-F(\mathbf{w}^{*})\leq (1-\frac{\mu}{\eta})\left(F(\mathbf{w}^{(k-1)})- F(\mathbf{w}^*)\right)+ \frac{\eta}{2}\Bigg[\\&\frac{2{\Phi}}{\left(D^{(k)}_s\right)^2} \Big[\sum_{{a_{1}}
    \in \mathcal{L}^{(k)}_{{1},{1}}} \sum_{a_{2}
    \in \mathcal{Q}^{(k)}(a_{1})} \cdots \hspace{-3mm} \sum_{a_{|\mathcal{L}|-1}
    \in \mathcal{Q}^{(k)}({a_{|\mathcal{L}|-2}})} \mathbbm{1}^{(k)}_{ \left\{{Q}(a_{|\mathcal{L}|-1})\right\}}|\mathcal{Q}^{(k)}(a_{|\mathcal{L}|-1})|^3
    \left(\lambda^{(k)}_{{Q}(a_{|\mathcal{L}|-1})}\right)^{2\theta^{(k)}_{{Q}(a_{|\mathcal{L}|-1})}} \left(\Upsilon^{(k)}_{{Q}(a_{|\mathcal{L}|-1})}\right)^2
    \\&+\sum_{{a_{1}}
    \in \mathcal{L}^{(k)}_{{1},{1}}} \sum_{a_{2}
    \in \mathcal{Q}^{(k)}(a_{1})} \cdots  \hspace{-3mm} \sum_{a_{|\mathcal{L}|-2}
    \in \mathcal{Q}^{(k)}({a_{|\mathcal{L}|-3}})} \mathbbm{1}^{(k)}_{ \left\{{Q}(a_{|\mathcal{L}|-2})\right\}}|\mathcal{Q}^{(k)}(a_{|\mathcal{L}|-2})|^3  \left(\lambda^{(k)}_{{Q}(a_{|\mathcal{L}|-2})}\right)^{2\theta^{(k)}_{{Q}(a_{|\mathcal{L}|-2})}} \left(\Upsilon^{(k)}_{{Q}(a_{|\mathcal{L}|-2})}\right)^2 \\&+\cdots+\sum_{a_1
    \in \mathcal{L}^{(k)}_{{1},1}} \mathbbm{1}^{(k)}_{ \left\{{Q}(a_1)\right\}}|\mathcal{Q}^{(k)}(a_1)|^3 \left(\lambda^{(k)}_{{Q}(a_1)}\right)^{2\theta^{(k)}_{{Q}(a_1)}} \left(\Upsilon^{(k)}_{{Q}(a_1)}\right)^2+\mathbbm{1}^{(k)}_{\left\{{{L}_{{1},1}}\right\}}|\mathcal{L}^{(k)}_{{1},1}|^3 \left(\lambda^{(k)}_{{L}_{{1},1}}\right)^{2\theta^{(k)}_{{L}_{{1},1}}} \left(\Upsilon^{(k)}_{{L}_{{1},1}}\right)^2\Big]\\&+  \frac{8}{\eta^2}\left(\frac{D-D^{(k)}_s}{D}\right)^2 (c_1+2c_2 \eta(F(\textbf{w}^{(k-1)})-F(\textbf{w}^*))\Bigg]
     \end{aligned}
     \hspace{-18mm}
 \end{equation}
\hrulefill
\vspace{-1mm}
\end{minipage}
\end{table*}
\pagebreak
  \begin{table*}
 {\footnotesize
 \begin{equation}\label{eq:longConvConsSmapl3}
 \hspace{-25mm}
 \begin{aligned}
     &\frac{1}{\beta^2}\Vert\varrho^{({k})}\Vert^2=\Bigg\Vert-\displaystyle \sum_{{a_{1}}
    \in \mathcal{L}^{(k)}_{{1},{1}}} \sum_{a_{2}
    \in \mathcal{Q}^{(k)}(a_{1})} \sum_{a_{3}
    \in \mathcal{Q}^{(k)}(a_{2})}\cdots  \sum_{a_{|\mathcal{L}|}
    \in \mathcal{Q}^{(k)}({a_{|\mathcal{L}|-1}})}\mathbbm{1}^{(k)}_{ \left\{{S}({Q}({a_{|\mathcal{L}|-1}}))\right\}}  \frac{|\mathcal{D}_{{a_{|\mathcal{L}|}}}|}{D^{(k)}_s}\nabla f_{{a_{|\mathcal{L}|}}}(\mathbf{w}^{(k-1)}_{{a_{|\mathcal{L}|}}})
    \\&+\displaystyle \sum_{{a_{1}}
    \in \mathcal{L}^{(k)}_{{1},{1}}} \sum_{a_{2}
    \in \mathcal{Q}^{(k)}(a_{1})} \sum_{a_{3}
    \in \mathcal{Q}^{(k)}(a_{2})}\cdots  \sum_{a_{|\mathcal{L}|}
    \in \mathcal{Q}^{(k)}({a_{|\mathcal{L}|-1}})} \frac{|\mathcal{D}_{{a_{|\mathcal{L}|}}}|}{D}\nabla f_{{a_{|\mathcal{L}|}}}(\mathbf{w}^{(k-1)}_{{a_{|\mathcal{L}|}}})
    + \frac{1}{\beta}\mathbf{\varpi}^{(k)}\Bigg\Vert^2\\&\leq%%%%%%%%
     %%%%%%%%%%%%%%%%%%%%%%%%%%%%%%%%
   \Bigg(\norm{\frac{1}{\beta}\mathbf{\varpi}^{(k)}  }+\Bigg\Vert 
    \displaystyle \sum_{{a_{1}}
    \in \mathcal{L}^{(k)}_{{1},{1}}} \sum_{a_{2}
    \in \mathcal{Q}^{(k)}(a_{1})} \sum_{a_{3}
    \in \mathcal{Q}^{(k)}(a_{2})}\cdots  \sum_{a_{|\mathcal{L}|}
    \in \mathcal{Q}^{(k)}({a_{|\mathcal{L}|-1}})} \frac{|\mathcal{D}_{{a_{|\mathcal{L}|}}}|}{D}\nabla f_{{a_{|\mathcal{L}|}}}(\mathbf{w}^{(k-1)}_{{a_{|\mathcal{L}|}}})\\&-\displaystyle \sum_{{a_{1}}
    \in \mathcal{L}^{(k)}_{{1},{1}}} \sum_{a_{2}
    \in \mathcal{Q}^{(k)}(a_{1})} \sum_{a_{3}
    \in \mathcal{Q}^{(k)}(a_{2})}\cdots  \sum_{a_{|\mathcal{L}|}
    \in \mathcal{Q}^{(k)}({a_{|\mathcal{L}|-1}})}\mathbbm{1}^{(k)}_{ \left\{{S}({Q}({a_{|\mathcal{L}|-1}}))\right\}}  \frac{|\mathcal{D}_{{a_{|\mathcal{L}|}}}|}{D^{(k)}_s}\nabla f_{{a_{|\mathcal{L}|}}}(\mathbf{w}^{(k-1)}_{{a_{|\mathcal{L}|}}})\Bigg\Vert \Bigg)^2\\&\overset{(a)}{\leq}%%%%%%%%
     %%%%%%%%%%%%%%%%%%%%%%%%%%%%%%%%
        2 \norm{\frac{1}{\beta}\mathbf{\varpi}^{(k)}  }^2+2\Bigg\Vert\displaystyle \displaystyle \sum_{{a_{1}}
    \in \mathcal{L}^{(k)}_{{1},{1}}} \sum_{a_{2}
    \in \mathcal{Q}^{(k)}(a_{1})} \sum_{a_{3}
    \in \mathcal{Q}^{(k)}(a_{2})}\cdots  \sum_{a_{|\mathcal{L}|}
    \in \mathcal{Q}^{(k)}({a_{|\mathcal{L}|-1}})} \frac{|\mathcal{D}_{{a_{|\mathcal{L}|}}}|}{D}\nabla f_{{a_{|\mathcal{L}|}}}(\mathbf{w}^{(k-1)}_{{a_{|\mathcal{L}|}}})\\&-\displaystyle \sum_{{a_{1}}
    \in \mathcal{L}^{(k)}_{{1},{1}}} \sum_{a_{2}
    \in \mathcal{Q}^{(k)}(a_{1})} \sum_{a_{3}
    \in \mathcal{Q}^{(k)}(a_{2})}\cdots  \sum_{a_{|\mathcal{L}|}
    \in \mathcal{Q}^{(k)}({a_{|\mathcal{L}|-1}})}\mathbbm{1}^{(k)}_{ \left\{{S}({Q}({a_{|\mathcal{L}|-1}}))\right\}}  \frac{|\mathcal{D}_{{a_{|\mathcal{L}|}}}|}{D^{(k)}_s}\nabla f_{{a_{|\mathcal{L}|}}}(\mathbf{w}^{(k-1)}_{{a_{|\mathcal{L}|}}})\Bigg\Vert^2\\&\overset{(b)}{\leq}%%%%%%%%
     %%%%%%%%%%%%%%%%%%%%%%%%%%%%%%%%
     2\frac{1}{\beta^2}\norm{\mathbf{\varpi}^{(k)}  }^2+
       2\Bigg\Vert \displaystyle \displaystyle \sum_{{a_{1}}
    \in \mathcal{L}^{(k)}_{{1},{1}}} \sum_{a_{2}
    \in \mathcal{Q}^{(k)}(a_{1})} \sum_{a_{3}
    \in \mathcal{Q}^{(k)}(a_{2})}\cdots  \sum_{a_{|\mathcal{L}|}
    \in \mathcal{Q}^{(k)}({a_{|\mathcal{L}|-1}})} \mathbbm{1}^{(k)}_{ \left\{\bar{S}({Q}({a_{|\mathcal{L}|-1}}))\right\}} \frac{|\mathcal{D}_{{a_{|\mathcal{L}|}}}|}{D}\nabla f_{{a_{|\mathcal{L}|}}}(\mathbf{w}^{(k-1)}_{{a_{|\mathcal{L}|}}})\\&
    -\displaystyle \sum_{{a_{1}}
    \in \mathcal{L}^{(k)}_{{1},{1}}} \sum_{a_{2}
    \in \mathcal{Q}^{(k)}(a_{1})} \sum_{a_{3}
    \in \mathcal{Q}^{(k)}(a_{2})}\cdots  \sum_{a_{|\mathcal{L}|}
    \in \mathcal{Q}^{(k)}({a_{|\mathcal{L}|-1}})}\mathbbm{1}^{(k)}_{ \left\{{S}({Q}({a_{|\mathcal{L}|-1}}))\right\}}  \frac{(D-D^{(k)}_s)|\mathcal{D}_{a_{|\mathcal{L}|}}|}{(D)(D^{(k)}_s)}\nabla f_{{a_{|\mathcal{L}|}}}(\mathbf{w}^{(k-1)}_{{a_{|\mathcal{L}|}}})\Bigg\Vert^2
\\&{\leq} %%%%%%%%
     %%%%%%%%%%%%%%%%%%%%%%%%%%%%%%%%
     2\frac{1}{\beta^2} \norm{\mathbf{\varpi}^{(k)}  }^2+
       2\Bigg( \displaystyle \displaystyle \sum_{{a_{1}}
    \in \mathcal{L}^{(k)}_{{1},{1}}} \sum_{a_{2}
    \in \mathcal{Q}^{(k)}(a_{1})} \sum_{a_{3}
    \in \mathcal{Q}^{(k)}(a_{2})}\cdots  \sum_{a_{|\mathcal{L}|}
    \in \mathcal{Q}^{(k)}({a_{|\mathcal{L}|-1}})} \mathbbm{1}^{(k)}_{ \left\{\bar{S}({Q}({a_{|\mathcal{L}|-1}}))\right\}} \frac{|\mathcal{D}_{{a_{|\mathcal{L}|}}}|}{D}\Big\Vert\nabla f_{{a_{|\mathcal{L}|}}}(\mathbf{w}^{(k-1)}_{{a_{|\mathcal{L}|}}})\Big\Vert\\&
    -\displaystyle \sum_{{a_{1}}
    \in \mathcal{L}^{(k)}_{{1},{1}}} \sum_{a_{2}
    \in \mathcal{Q}^{(k)}(a_{1})} \sum_{a_{3}
    \in \mathcal{Q}^{(k)}(a_{2})}\cdots  \sum_{a_{|\mathcal{L}|}
    \in \mathcal{Q}^{(k)}({a_{|\mathcal{L}|-1}})}\mathbbm{1}^{(k)}_{ \left\{{S}({Q}({a_{|\mathcal{L}|-1}}))\right\}}  \frac{(D-D^{(k)}_s)|\mathcal{D}_{a_{|\mathcal{L}|}}|}{(D)(D^{(k)}_s)}\Big\Vert\nabla f_{{a_{|\mathcal{L}|}}}(\mathbf{w}^{(k-1)}_{{a_{|\mathcal{L}|}}})\Big\Vert\Bigg)^2\\&{\leq}
         %%%%%%%%%%%%%%%%%%%%%%%%%%%%
        2\frac{1}{\beta^2} \norm{\mathbf{\varpi}^{(k)}  }^2+
       2\Bigg( \displaystyle \displaystyle \sum_{{a_{1}}
    \in \mathcal{L}^{(k)}_{{1},{1}}} \sum_{a_{2}
    \in \mathcal{Q}^{(k)}(a_{1})} \sum_{a_{3}
    \in \mathcal{Q}^{(k)}(a_{2})}\cdots  \sum_{a_{|\mathcal{L}|}
    \in \mathcal{Q}^{(k)}({a_{|\mathcal{L}|-1}})} \mathbbm{1}^{(k)}_{ \left\{\bar{S}({Q}({a_{|\mathcal{L}|-1}}))\right\}} \frac{|\mathcal{D}_{{a_{|\mathcal{L}|}}}|}{D}\max_{a_{|\mathcal{L}|}
    \in \mathcal{Q}^{(k)}({a_{|\mathcal{L}|-1}})}\Big(\big\Vert\nabla f_{{a_{|\mathcal{L}|}}}(\mathbf{w}^{(k-1)}_{{a_{|\mathcal{L}|}}})\big\Vert\Big)\\&+\displaystyle \sum_{{a_{1}}
    \in \mathcal{L}^{(k)}_{{1},{1}}} \sum_{a_{2}
    \in \mathcal{Q}^{(k)}(a_{1})} \sum_{a_{3}
    \in \mathcal{Q}^{(k)}(a_{2})}\cdots  \sum_{a_{|\mathcal{L}|}
    \in \mathcal{Q}^{(k)}({a_{|\mathcal{L}|-1}})}\mathbbm{1}^{(k)}_{ \left\{{S}({Q}({a_{|\mathcal{L}|-1}}))\right\}}  \frac{(D-D^{(k)}_s)|\mathcal{D}_{a_{|\mathcal{L}|}}|}{(D)(D^{(k)}_s)}\max_{a_{|\mathcal{L}|}
    \in \mathcal{Q}^{(k)}({a_{|\mathcal{L}|-1}})}\Big(\big\Vert\nabla f_{{a_{|\mathcal{L}|}}}(\mathbf{w}^{(k-1)}_{{a_{|\mathcal{L}|}}})\big\Vert\Big)\Bigg)^2\\&{=}
          %%%%%%%%%%%%%%%%%%%%%%%%%%%%
         2\frac{1}{\beta^2} \norm{\mathbf{\varpi}^{(k)}  }^2+
       2\Bigg( \frac{D-D^{(k)}_s}{D}\max_{a \in \mathcal{N}_{|\mathcal{L}|}}\Big(\big\Vert\nabla f_{a}(\mathbf{w}^{(k-1)}_{a})\Vert\Big)+ \frac{(D-D^{(k)}_s)D^{(k)}_s}{(D)(D^{(k)}_s)}\max_{a\in \mathcal{N}_{|\mathcal{L}|}}\Big(\big\Vert\nabla f_{a}(\mathbf{w}^{(k-1)}_{a})\big\Vert\Big)\Bigg)^2\\&\leq
          %%%%%%%%%%%%%%%%%%%%%%%%%%%%
          \frac{2}{\beta^2} \norm{\mathbf{\varpi}^{(k)}  }^2+
       2\Bigg[\Bigg( 2\frac{D-D^{(k)}_s}{D}\max_{a\in \mathcal{N}_{|\mathcal{L}|}}\Big(\big\Vert\nabla f_{a}(\mathbf{w}^{(k-1)}_{a})\big\Vert\Big)\Bigg)^2\Bigg]\leq
          %%%%%%%%%%%%%%%%%%%%%%%%%%%%
          \frac{2}{\beta^2} \norm{\mathbf{\varpi}^{(k)}  }^2+
      8\Bigg(\frac{D-D^{(k)}_s}{D}\max_{a\in \mathcal{N}_{|\mathcal{L}|}}\Big(\big\Vert\nabla f_{a}(\mathbf{w}^{(k-1)}_{a})\big\Vert\Big)\Bigg)^2\\&{\leq}
      %%%%%%%%%%%%%%%%%%%%%%%%%%%%%%%%
       \frac{2}{\beta^2}\norm{\mathbf{\varpi}^{(k)}  }^2+
      8\left(\frac{D-D^{(k)}_s}{D}\right)^2\Bigg(\max_{a\in \mathcal{N}_{|\mathcal{L}|}}\Big(\big\Vert\nabla f_{a}(\mathbf{w}^{(k-1)}_{a})\big\Vert\Big)\Bigg)^2\\&\overset{(c)}{\leq} \frac{2}{\beta^2}\norm{\mathbf{\varpi}^{(k)}  }^2+
      8\left(\frac{D-D^{(k)}_s}{D}\right)^2\left(c_1+c_2\Vert F(\mathbf{w}^{(k-1)})\Vert^2\right)\\&\overset{(d)}{\leq}
      2\frac{1}{\beta^2} \norm{\mathbf{\varpi}^{(k)}  }^2+8\left(\frac{D-D^{(k)}_s}{D}\right)^2\left(c_1+2c_2 \eta\left(F(\textbf{w}^{(k-1)})-F(\textbf{w}^*) \right)\right)
          %%%%%%%%%%%%%%%%%%%%%%%%%%%%
     \end{aligned}
     \hspace{-21mm}
 \end{equation}
 }
 \hrulefill
 \end{table*}

 \pagebreak
The following appendix is the last appendix of the paper concerned with theoretical analysis, which is followed by another appendix containing extensive numerical simulations.
\pagebreak
   \section{Aggregation Error upon using Algorithm~\ref{alg:tuneConsNN}}\label{app:aggrErr}
   \noindent According to~\eqref{eq:root}, the aggregation error at the $k$-th global aggregation is given by
   \begin{equation}
  \hspace{-24mm}
  \begin{aligned}
    & \mathbf{e}^{(k)}=\frac{1}{D}\Bigg(\sum_{{a_{1}}
    \in \mathcal{L}^{(k)}_{{1},{1}}} \sum_{a_{2}
    \in \mathcal{Q}^{(k)}(a_{1})} \sum_{a_{3}
    \in \mathcal{Q}^{(k)}(a_{2})}\cdots  \sum_{a_{|\mathcal{L}|-1}
    \in \mathcal{Q}^{(k)}({a_{|\mathcal{L}|-2}})} {\mathbbm{1}^{(k)}_{\left\{{{Q}(a_{|\mathcal{L}|-1})}\right\}}|\mathcal{Q}^{(k)}(a_{|\mathcal{L}|-1})|  \mathbf{c}^{(k)}_{a'_{|\mathcal{L}|}}}\\&
    +\sum_{{a_{1}}
    \in \mathcal{L}^{(k)}_{{1},{1}}} \sum_{a_{2}
    \in \mathcal{Q}^{(k)}(a_{1})} \sum_{a_{3}
    \in \mathcal{Q}^{(k)}(a_{2})}\cdots  \sum_{a_{|\mathcal{L}|-2}
    \in \mathcal{Q}^{(k)}({a_{|\mathcal{L}|-3}})} {\mathbbm{1}^{(k)}_{\left\{{{Q}(a_{|\mathcal{L}|-2})}\right\}}|\mathcal{Q}^{(k)}(a_{|\mathcal{L}|-2})|  \mathbf{c}^{(k)}_{a'_{|\mathcal{L}|-1}}}+
    \\&\vdots
    \\&+\sum_{a_1
    \in \mathcal{L}^{(k)}_{{1},1}} {\mathbbm{1}^{(k)}_{\left\{{{Q}(a_1)}\right\}}|\mathcal{Q}^{(k)}(a_1)|  \mathbf{c}^{(k)}_{a'_2}}+\mathbbm{1}^{(k)}_{\left\{{{L}_{{1},1}}\right\}} |\mathcal{L}^{(k)}_{{1},1}|\mathbf{c}^{(k)}_{a'_1}\Bigg).
       \end{aligned}
       \hspace{-20mm}
 \end{equation}
 Following a similar procedure described in Appendix~\ref{app:ConsGen}, we get
  \begin{equation} 
  \hspace{-24mm}
  \begin{aligned}
    & \Vert\mathbf{e}^{(k)}\Vert^2\leq\frac{\Phi}{D^2} \Bigg[\\& \sum_{{a_{1}}
    \in \mathcal{L}^{(k)}_{{1},{1}}} \sum_{a_{2}
    \in \mathcal{Q}^{(k)}(a_{1})} \cdots \hspace{-3mm} \sum_{a_{|\mathcal{L}|-1}
    \in \mathcal{Q}^{(k)}({a_{|\mathcal{L}|-2}})} \mathbbm{1}^{(k)}_{ \left\{{Q}(a_{|\mathcal{L}|-1})\right\}}|\mathcal{Q}^{(k)}(a_{|\mathcal{L}|-1})|^3
    \left(\lambda^{(k)}_{{Q}(a_{|\mathcal{L}|-1})}\right)^{2\theta^{(k)}_{{Q}(a_{|\mathcal{L}|-1})}} \left(\Upsilon^{(k)}_{{Q}(a_{|\mathcal{L}|-1})}\right)^2
    \\&+\sum_{{a_{1}}
    \in \mathcal{L}^{(k)}_{{1},{1}}} \sum_{a_{2}
    \in \mathcal{Q}^{(k)}(a_{1})} \cdots  \hspace{-3mm} \sum_{a_{|\mathcal{L}|-2}
    \in \mathcal{Q}^{(k)}({a_{|\mathcal{L}|-3}})} \mathbbm{1}^{(k)}_{ \left\{{Q}(a_{|\mathcal{L}|-2})\right\}}|\mathcal{Q}^{(k)}(a_{|\mathcal{L}|-2})|^3  \left(\lambda^{(k)}_{{Q}(a_{|\mathcal{L}|-2})}\right)^{2\theta^{(k)}_{{Q}(a_{|\mathcal{L}|-2})}} \left(\Upsilon^{(k)}_{{Q}(a_{|\mathcal{L}|-2})}\right)^2 \\&+\cdots+\sum_{a_1
    \in \mathcal{L}^{(k)}_{{1},1}} \mathbbm{1}^{(k)}_{ \left\{{Q}(a_1)\right\}}|\mathcal{Q}^{(k)}(a_1)|^3 \left(\lambda^{(k)}_{{Q}(a_1)}\right)^{2\theta^{(k)}_{{Q}(a_1)}} \left(\Upsilon^{(k)}_{{Q}(a_1)}\right)^2+\mathbbm{1}^{(k)}_{\left\{{{L}_{{1},1}}\right\}}|\mathcal{L}^{(k)}_{{1},1}|^3 \left(\lambda^{(k)}_{{L}_{{1},1}}\right)^{2\theta^{(k)}_{{L}_{{1},1}}} \left(\Upsilon^{(k)}_{{L}_{{1},1}}\right)^2\Bigg],
       \end{aligned}
       \hspace{-20mm}
 \end{equation}
 where $\Phi= N_{{|\mathcal{L}|-1}}+N_{{|\mathcal{L}|-2}}+\cdots+N_{{1}}+1$.
By tuning the number of D2D according  to~\eqref{eq:IterFogLConsAlg}, following a similar procedure as Appendix~\ref{app:boundedConsConv}, we get
  \begin{equation} 
  \hspace{-24mm}
  \begin{aligned}
    & \Vert\mathbf{e}^{(k)}\Vert^2\leq\frac{\Phi}{D^2} \Bigg[  \sum_{{a_{1}}
    \in \mathcal{L}^{(k)}_{{1},{1}}} \sum_{a_{2}
    \in \mathcal{Q}^{(k)}(a_{1})} \cdots \hspace{-3mm} \sum_{a_{|\mathcal{L}|-1}
    \in \mathcal{Q}^{(k)}({a_{|\mathcal{L}|-2}})} \mathbbm{1}^{(k)}_{ \left\{{Q}(a_{|\mathcal{L}|-1})\right\}}{\frac{\psi }{ \frac{{\Phi}}{ D^2}N_{|\mathcal{L}|-1}|\mathcal{L}|}}
    \\&+\sum_{{a_{1}}
    \in \mathcal{L}^{(k)}_{{1},{1}}} \sum_{a_{2}
    \in \mathcal{Q}^{(k)}(a_{1})} \cdots  \hspace{-3mm} \sum_{a_{|\mathcal{L}|-2}
    \in \mathcal{Q}^{(k)}({a_{|\mathcal{L}|-3}})} \mathbbm{1}^{(k)}_{ \left\{{Q}(a_{|\mathcal{L}|-2})\right\}}{\frac{\psi }{ \frac{{\Phi}}{ D^2}N_{|\mathcal{L}|-2}|\mathcal{L}|}} +\cdots
    \\&+\sum_{a_1
    \in \mathcal{L}^{(k)}_{{1},1}} \mathbbm{1}^{(k)}_{ \left\{{Q}(a_1)\right\}}{\frac{\psi }{ \frac{{\Phi}}{ D^2}N_{1}|\mathcal{L}| }}+\mathbbm{1}^{(k)}_{\left\{{{L}_{{1},1}}\right\}}{\frac{\psi }{ \frac{{\Phi}}{ D^2}N_{0}|\mathcal{L}| }}\Bigg]
    \\&\leq \frac{\Phi}{D^2} \Bigg[  \sum_{{a_{1}}
    \in \mathcal{L}^{(k)}_{{1},{1}}} \sum_{a_{2}
    \in \mathcal{Q}^{(k)}(a_{1})} \cdots \hspace{-3mm} \sum_{a_{|\mathcal{L}|-1}
    \in \mathcal{Q}^{(k)}({a_{|\mathcal{L}|-2}})} {\frac{\psi }{ \frac{{\Phi}}{ D^2}N_{|\mathcal{L}|-1}|\mathcal{L}|}}
    \\&+\sum_{{a_{1}}
    \in \mathcal{L}^{(k)}_{{1},{1}}} \sum_{a_{2}
    \in \mathcal{Q}^{(k)}(a_{1})} \cdots  \hspace{-3mm} \sum_{a_{|\mathcal{L}|-2}
    \in \mathcal{Q}^{(k)}({a_{|\mathcal{L}|-3}})} \mathbbm{1}^{(k)}_{ \left\{{Q}(a_{|\mathcal{L}|-2})\right\}} {\frac{\psi }{ \frac{{\Phi}}{ D^2}N_{|\mathcal{L}|-2}|\mathcal{L}|}} +\cdots
    \\&+\sum_{a_1
    \in \mathcal{L}^{(k)}_{{1},1}} {\frac{\psi }{ \frac{{\Phi}}{ D^2}N_{1}|\mathcal{L}| }}+{\frac{\psi }{ \frac{{\Phi}}{ D^2}N_{0} |\mathcal{L}|}}\Bigg]
    \\&= \frac{\Phi}{D^2} \Bigg[ \underbrace{ {\frac{\psi }{ \frac{{\Phi}}{ D^2}|\mathcal{L}|}}
   + {\frac{\psi }{ \frac{{\Phi}}{ D^2}|\mathcal{L}|}} \cdots
  + {\frac{\psi }{ \frac{{\Phi}}{ D^2}|\mathcal{L}| }}+{\frac{\psi }{ \frac{{\Phi}}{ D^2} |\mathcal{L}|}}}_{|\mathcal{L}| ~\textrm{terms}}\Bigg].
       \end{aligned}
       \hspace{-20mm}
 \end{equation}
 Thus, we have
 \begin{equation}
     \Vert\mathbf{e}^{(k)}\Vert^2\leq \psi.
 \end{equation}
\pagebreak
\section{Details of the Simulations Setting and Further Simulations}\label{app:extraSim}
In this section, we first present some details regarding  simulations settings and parameter tuning and then present a series of simulation results regarding the choice of different datasets and larger network size as compared to the main text. 
Our entire Python implementation, including the set of hyperparameters used in each experiment, can be found at the following Github repository: https://github.com/shams-sam/Federated2Fog".
% The exact set of hyperparameters used on each network can be found at our implementation repository: \url{https://github.com/shams-sam/Federated2Fog}. 
\subsection{Simulation Setting}
\subsubsection{Setup}
 All simulations are performed on a single machine with 64GB RAM and 8GB GPU memory, which emulates the learning through a distributed learning framework \textit{PySyft} that helps spin off virtually disjoint nodes with mutually exclusive model parameters and datasets, working on top of \textit{PyTorch} machine learning library.
 \subsubsection{Classifiers} We consider two different classifiers - regularized Support Vector Machine (SVM) and fully-connected Neural Network (NN), initialized with a copy of global model before the learning process begins on each node participating in the learning process. 
 
 The regularized SVM is tuned to satisfy the strong convexity with $\mu=0.1$. We also use the estimated value of $\eta=10$ (similar values are observed in~\cite{8664630}). The NN classifier is a simple fully connected network with a single hidden layer and no convolutional units. \textit{Softmax} activation at the output layer gives the class \textit{logits} and the overall training optimizes negative log-likelihood loss function with L2 regularization. 
 
 Input size for both the models, SVM and NN is $28\times 28=784$, with output size $10$. The number of parameters optimized by the networks $M$ is given by $M=(784+1)\times 10 = 7850$.
\subsubsection{Datasets and Data Distribution among the Nodes} We consider two datasets MNIST and F-MNIST (Fashion MNIST)\footnote{https://github.com/zalandoresearch/fashion-mnist}, each of which contain $60000$ training samples and $10000$ testing samples. MNIST consists of handwritten digits $0-9$, while F-MNIST consists of images associated with $10$ classes in clothing domain. Each dataset consists of $28\times 28$ grayscale images.

The datasets are distributed over nodes such that all nodes have approximately equal number of training samples. However the training samples, maybe either be i.i.d or non-i.i.d distributed. For i.i.d distribution, each node participating in the learning process has samples from each class of the dataset, while under non-i.i.d distribution, each node has access to only one of the classes. These are the extreme ends of possible split of data among nodes in terms of class distribution, helping us evaluate the overall robustness as well as differences in characteristics of our technique under different settings.

\subsubsection{Network Formation} 
We consider two network configurations: (i) the network consists of $125$ edge devices; (ii) the network consists of $625$ edge devices. For the former case, we consider a fog network consisting of a main server and three sub-layers, to build our fog network we start with the $125$ worker nodes in the bottom-most layer (${L}_3$) and dedicated local datasets sampled as explained above. The worker nodes update the local models with a copy of parameters from latest global model at the start of each iteration. The worker nodes are then clustered in groups of $5$ to communicate with one of the $25$ aggregators in their upper layer (i.e., ${L}_2$), such that there is a $1$-to-$1$ mapping between the clusters and the aggregators. Similarly the nodes in layer (${L}_2$) are clustered and communicate with the $5$ aggregators in the layer ${L}_1$, followed by clustering and communicating the $5$ nodes with the main server. 

For the latter case, we consider a fog network consisting of a main server and four sub-layers, to build our fog network we start with the $625$ worker nodes in the bottom-most layer (i.e., ${L}_4$) and dedicated local datasets sampled as explained above. The worker nodes update the local models with a copy of parameters from latest global model at the start of each iteration. The worker nodes are then clustered in groups of $5$ to communicate with one of the $125$ aggregators in their upper layer (i.e., ${L}_3$), such that there is a 1-to-1 mapping between the clusters and aggregators. Similarly nodes in ${L}_3$ are again clustered and communicate to the $25$ aggregators in the upper layer (i.e., ${L}_2$). This is followed by clustering of these nodes in groups of $5$ and communicating with $5$ aggregators in layer ${L}_1$, which then communicate with the main server. 

The connectivity among the nodes within a cluster is simulated using random geometric graphs with increasing connectivity as we traverse from the bottom-most layer to the main server. In our random geometric graph construction, nodes are placed in a circle disc with radius $100$m uniformly at random, where the existence of edge (i.e., D2D link) between two nodes is assumed if the distance between the nodes is less than a threshold ($\varphi$).
For the case with $125$ edge device, layer ${L}_3$ has $\varphi=40$m, followed by layer ${L}_2$ with  $\varphi=50$m and layer ${L}_1$ with $\varphi=60$m. 
For the case with $625$ edge device, in layer ${L}_3$ and ${L}_4$ we have $\varphi=40$, followed by layer ${L}_2$ with $\varphi=50$ and layer ${L}_1$ with $\varphi=60$. We use \textit{NetworkX}\footnote{https://networkx.github.io} library of \textit{Python} for generating the graph. We adjust the radius parameter of the graph generator such that the average degree of the graph is within tolerance region of 0.2 from the desired degree of the graph.

For the D2D communications, we consider the common choice of the weights~\cite{xiao2004fast} that gives $\textbf{z}_{n}^{(t+1)} = \textbf{z}_{n}^{(t)}+d^{(k)}_{{C}}\sum_{m\in \mathcal{\zeta}^{(k)}(n)} (\textbf{z}_{m}^{(t)}-\textbf{z}_{n}^{(t)})$, $0 < d^{(k)}_{{C}} < 1 / D^{(k)}_{{C}}$, for any node $n$ in arbitrary cluster ${C}$, where $D^{(k)}_{{C}}$ is the maximum degree of the nodes in $G^{(k)}_{{C}}$. Using this implementation, the nodes inside LUT cluster ${C}$ only need to have the knowledge of the parameter $d^{(k)}_{C}$, which is broadcast by the respective parent node.

We summarize the simulation parameters in Table~\ref{table:11}.

\begin{table}[h]
\caption{Summary of parameter values employed in our simulations.\label{table:11}}
\centering
\begin{tabular}{|c| c|} 
 \hline
 Parameter & Value \\ 
 \hline
 Number of Edge Devices & $125$, $625$\\
  \hline
 Number of Layers of the Network & $4$, $5$\\
 \hline
 Number of Devices Per Cluster & $5$\\
 \hline
 Random Geometric Graph Threshold $\varphi$ & $[40,60]m$  \\ 
 \hline
 Smoothness $\eta$ & 10  \\
 \hline
 Strong Convexity $\mu$ & 0.1  \\
 \hline
Number of Data points $D$ & $60,000$  \\
 \hline
 Uplink Transmit Power of Devices & $24$dBm \\
 \hline
 D2D Transmit Power of Devices & $10$dBm \\
 \hline
 D2D/Uplink Delay of Transmission of Parameters& $0.25$ Sec\\
 \hline
\end{tabular}
\end{table}

\subsection{Further Simulation Results}
 This section presents the plots from complimentary experiments from Section~\ref{sec:num-res}. In the following, we explain the relationship between the figures presented in this appendix and the simulation results presented in the main text.

%%%%%%%%%%%%%%%%%%%%%%%%%%%%%% Start MNIST 125
%%%%%%%%%%%%%%%%%%%%%%%%%%%%%%
%%%%%%%%%%%%%%%%%%%%%%%%%%%%%%
%%%%%%%%%%%%%%%%%%%%%%%%%%%%%

% \input{MNIST_125_Results}

Fig.~\ref{fig:GenFigGoodIntuition_MNIST_125} from main text is repeated in Fig.~\ref{fig:GenFigGoodIntuition_mnist_625} for MNIST dataset distributed over 625 edge devices, Fig.~\ref{fig:GenFigGoodIntuition_FMNIST_125} for FMNIST dataset distributed over 125 edge devices and Fig.~\ref{fig:GenFigGoodIntuition_FMNIST_625} for FMNIST dataset distributed over 625 edge devices.

Fig.~\ref{fig:decaying_learning_rate_MNIST_125} from main text is repeated in Fig.~\ref{fig:decaying_learning_rate_mnist_625} for MNIST dataset distributed over 625 edge devices, Fig.~\ref{fig:decaying_learning_rate_FMNIST_125} for FMNIST dataset distributed over 125 edge devices and Fig.~\ref{fig:decaying_learning_rate_FMNIST_625} for FMNIST dataset distributed over 625 edge devices.

Fig.~\ref{fig:iidIncConSigma_MNIST_125} from main text is repeated in Fig.~\ref{fig:iidIncConSigma_mnist_625} for MNIST dataset distributed over 625 edge devices, Fig.~\ref{fig:iidIncConSigma_FMNIST_125} for FMNIST dataset distributed over 125 edge devices and Fig.~\ref{fig:iidIncConSigma_FMNIST_625} for FMNIST dataset distributed over 625 edge devices.

Fig.~\ref{fig:non_iidIncConSigma_MNIST_125} from main text is repeated in Fig.~\ref{fig:non_iidIncConSigma_mnist_625} for MNIST dataset distributed over 625 edge devices, Fig.~\ref{fig:non_iidIncConSigma_FMNIST_125} for FMNIST dataset distributed over 125 edge devices and Fig.~\ref{fig:non_iidIncConSigma_FMNIST_625} for FMNIST dataset distributed over 625 edge devices.

Fig.~\ref{fig:iidConsConDelta_MNIST_125} from main text is repeated in Fig.~\ref{fig:iidConsConDelta_mnist_625} for MNIST dataset distributed over 625 edge devices, Fig.~\ref{fig:iidConsConDelta_FMNIST_125} for FMNIST dataset distributed over 125 edge devices and Fig.~\ref{fig:iidConsConDelta_FMNIST_625} for FMNIST dataset distributed over 625 edge devices.

Fig.~\ref{fig:non_iidConsConDelta_MNIST_125} from main text is repeated in Fig.~\ref{fig:non_iidConsConDelta_mnist_625} for MNIST dataset distributed over 625 edge devices, Fig.~\ref{fig:non_iidConsConDelta_FMNIST_125} for FMNIST dataset distributed over 125 edge devices and Fig.~\ref{fig:non_iidConsConDelta_FMNIST_625} for FMNIST dataset distributed over 625 edge devices.

Fig.~\ref{fig:iidNNIncConec_MNIST_125} from main text is repeated in Fig.~\ref{fig:iidNNIncConec_mnist_625} for MNIST dataset distributed over 625 edge devices, Fig.~\ref{fig:iidNNIncConec_FMNIST_125} for FMNIST dataset distributed over 125 edge devices and Fig.~\ref{fig:iidNNIncConec_FMNIST_625} for FMNIST dataset distributed over 625 edge devices.

Fig.~\ref{fig:non_iidConsConDeltaNN_MNIST_125} from main text is repeated in Fig.~\ref{fig:non_iidConsConDeltaNN_mnist_625} for MNIST dataset distributed over 625 edge devices, Fig.~\ref{fig:non_iidConsConDeltaNN_FMNIST_125} for FMNIST dataset distributed over 125 edge devices and Fig.~\ref{fig:non_iidConsConDeltaNN_FMNIST_625} for FMNIST dataset distributed over 625 edge devices.

Fig.~\ref{fig:non_iidTheoPracSimDiff_MNIST_125} from main text is repeated in Fig.~\ref{fig:non_iidTheoPracSimDiff_mnist_625} for MNIST dataset distributed over 625 edge devices, Fig.~\ref{fig:non_iidTheoPracSimDiff_FMNIST_125} for FMNIST dataset distributed over 125 edge devices and Fig.~\ref{fig:non_iidTheoPracSimDiff_FMNIST_625} for FMNIST dataset distributed over 625 edge devices.

Fig.~\ref{fig:accum_energy_125_MNIST} from main text is repeated in Fig.~\ref{fig:accum_energy_625_MNIST} for MNIST dataset distributed over 625 edge devices, Fig.~\ref{fig:accum_energy_125_FMNIST} for FMNIST dataset distributed over 125 edge devices and Fig.~\ref{fig:accum_energy_625_FMNIST} for FMNIST dataset distributed over 625 edge devices.

Fig.~\ref{fig:accumData_MNIST_125} from main text is repeated in Fig.~\ref{fig:accumData_625_MNIST} for MNIST dataset distributed over 625 edge devices, Fig.~\ref{fig:accumData_FMNIST_125} for FMNIST dataset distributed over 125 edge devices and Fig.~\ref{fig:accumData_FMNIST_625} for FMNIST dataset distributed over 625 edge devices.

\subsection{Energy and Parameter Transmission Savings under Various D2D Control Parameters and Tolerable Aggregation Errors}\label{appendix:psiSigma}

\textbf{Varying $\sigma$:} To demonstrate the effect of $\sigma$ on the energy and data traffic savings, we set  $\sigma_j$ at layer ${L}_j$ as $\sigma_j=\sigma' \max_{i}{\Upsilon_{{L}_{j,i}}^{(1)}}$, where $\Upsilon_{{L}_{j,i}}^{(1)}$ is the divergence of parameters at the beginning of model training at $i$-th cluster of layer $j$, and change the value of $\sigma'$ in our experiments. Note that higher values of $\left\{\sigma_j\right\}_{j=1}^{|\mathcal{L}|}$ are associated with a looser condition on the D2D consensus formation error (see Proposition 1 and 2). This implies that increasing $\left\{\sigma_j\right\}_{j=1}^{|\mathcal{L}|}$ often results in performing fewer D2D rounds across the fog layers.
The results for varying values of $\sigma'$ are depicted in
Fig.~\ref{fig:pc_mnist_125_sig}
and~\ref{fig:pt_mnist_125_sig} (for MNIST dataset and $125$ nodes);  Fig.~\ref{fig:pc_mnist_625_sig} and~\ref{fig:pt_mnist_625_sig} (for MNIST dataset and $625$ nodes); Fig.~\ref{fig:pc_fmnist_125_sig} and~\ref{fig:pt_fmnist_125_sig} (for FMNIST dataset and $125$ nodes);
Fig.~\ref{fig:pc_fmnist_625_sig}
and~\ref{fig:pt_fmnist_625_sig} (for FMNIST dataset and $625$ nodes). 

\begin{itemize}
    \item Considering the energy consumption (i.e.,~Figs.~\ref{fig:pc_mnist_125_sig},\ref{fig:pc_mnist_625_sig},\ref{fig:pc_fmnist_125_sig},\ref{fig:pc_fmnist_625_sig}), increasing $\sigma'$ often results in more energy savings since the nodes conduct less D2D rounds. However, after a certain threshold, increasing $\sigma'$ may lead to slight increase in energy consumption for {\tt MH-FL} (e.g., increasing $\sigma'$ from $0.6$ to $0.9$ in Fig.~\ref{fig:pc_mnist_125_sig}). That is because decreasing the D2D communication rounds below a threshold may have a significantly negative impact on the convergence speed, where the model may need considerably higher number of global aggregation iterations to reach the desired accuracy. 
    \item Considering the parameter transmission savings (i.e.,~Figs.~\ref{fig:pt_mnist_125_sig},\ref{fig:pt_mnist_625_sig},\ref{fig:pt_fmnist_125_sig},\ref{fig:pt_fmnist_625_sig}), it can be noted that increasing $\sigma'$ often results in a slight increase in parameter transmission for {\tt MH-FL}, since the model may need a few more global aggregations to reach the desired accuracy when the D2D rounds conducted are decreased.     Note that in all the cases, {\tt MH-FL} outperforms the EUT baseline method in terms of both energy consumption and parameter transmissions.
\end{itemize}

\textbf{Varying $\psi$:} Note that parameter $\psi$ controls the 2-norm of aggregation errors when {\tt MH-FL} is used with non-convex loss functions. In particular, smaller values of $\psi$ impose a smaller tolerable error of aggregation at the server, which call for higher number of D2D rounds across the nodes to decrease the local aggregation errors. The results are depicted in Fig.~\ref{fig:pc_mnist_125}
and~\ref{fig:pt_mnist_125} (for MNIST dataset and $125$ nodes);  Fig.~\ref{fig:pc_mnist_625} and~\ref{fig:pt_mnist_625} (for MNIST dataset and $625$ nodes); Fig.~\ref{fig:pc_fmnist_125} and~\ref{fig:pt_fmnist_125} (for FMNIST dataset and $125$ nodes);
Fig.~\ref{fig:pc_fmnist_625}
and~\ref{fig:pt_fmnist_625} (for FMNIST dataset and $625$ nodes). 

\begin{itemize}
    \item Considering the energy consumption (i.e.,~Figs.~\ref{fig:pc_mnist_125},\ref{fig:pc_mnist_625},\ref{fig:pc_fmnist_125},\ref{fig:pc_fmnist_625}), increasing $\psi$  results in more energy savings since the nodes conduct fewer D2D rounds. Also, it can be seen that for small values of $\psi$ (e.g. $\psi=10$ in these figures), {\tt MH-FL} has a higher energy consumption as compared to EUT baseline since the number of D2D communication rounds become unreasonably high for such choices of $\psi$. However for moderate to high value of $\psi$ (e.g., $\psi\geq 10^3$ in these figures), {\tt MH-FL} always outperforms the EUT baseline in terms of energy consumption.
\item Considering the parameter transmission savings (i.e.,~Figs.~\ref{fig:pt_mnist_125},\ref{fig:pt_mnist_625},\ref{fig:pt_fmnist_125},\ref{fig:pt_fmnist_625}), increasing $\psi$ often results in increasing the number of parameter transmissions for {\tt MH-FL} since the model may need more time (i.e., higher number of global aggregations) to reach the desired accuracy. Nevertheless, {\tt MH-FL} outperforms the EUT baseline in all the scenarios due the sampling of a single node from each LUT cluster.
\end{itemize}

\newpage
%%%%%%%%%%%%%%%%%%%%%%%%%%%%%% End MNIST 125
%%%%%%%%%%%%%%%%%%%%%%%%%%%%%%
%%%%%%%%%%%%%%%%%%%%%%%%%%%%%%
%%%%%%%%%%%%%%%%%%%%%%%%%%%%%

%%%%%%%%%%%%%%%%%%%%%%%%%%%%%% Start MNIST 625
%%%%%%%%%%%%%%%%%%%%%%%%%%%%%%
%%%%%%%%%%%%%%%%%%%%%%%%%%%%%%
%%%%%%%%%%%%%%%%%%%%%%%%%%%%%

\begin{figure}
\centering
\begin{minipage}{.48\textwidth}
     \centering
     \includegraphics[width=\linewidth]{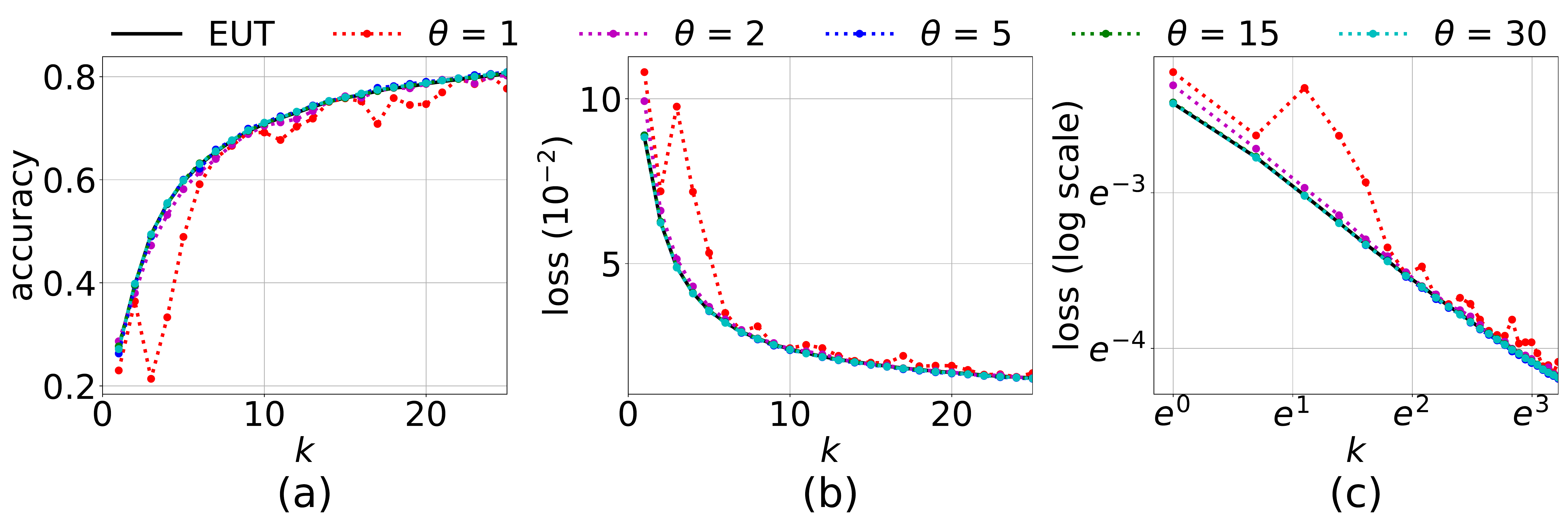}
     \caption{Performance comparison between baseline EUT and {\tt MH-MT} when a fixed number of D2D rounds $\theta$ is used at every cluster of the network, for non-i.i.d. As the number of D2D rounds increases, {\tt MH-MT} performs more similar to the EUT baseline and the learning is more stable. (MNIST, $625$ Edge Devices)}
     \label{fig:GenFigGoodIntuition_mnist_625}
\end{minipage}%
\hspace{2mm}
\begin{minipage}{.48\textwidth}
  \centering
     \includegraphics[width=\linewidth]{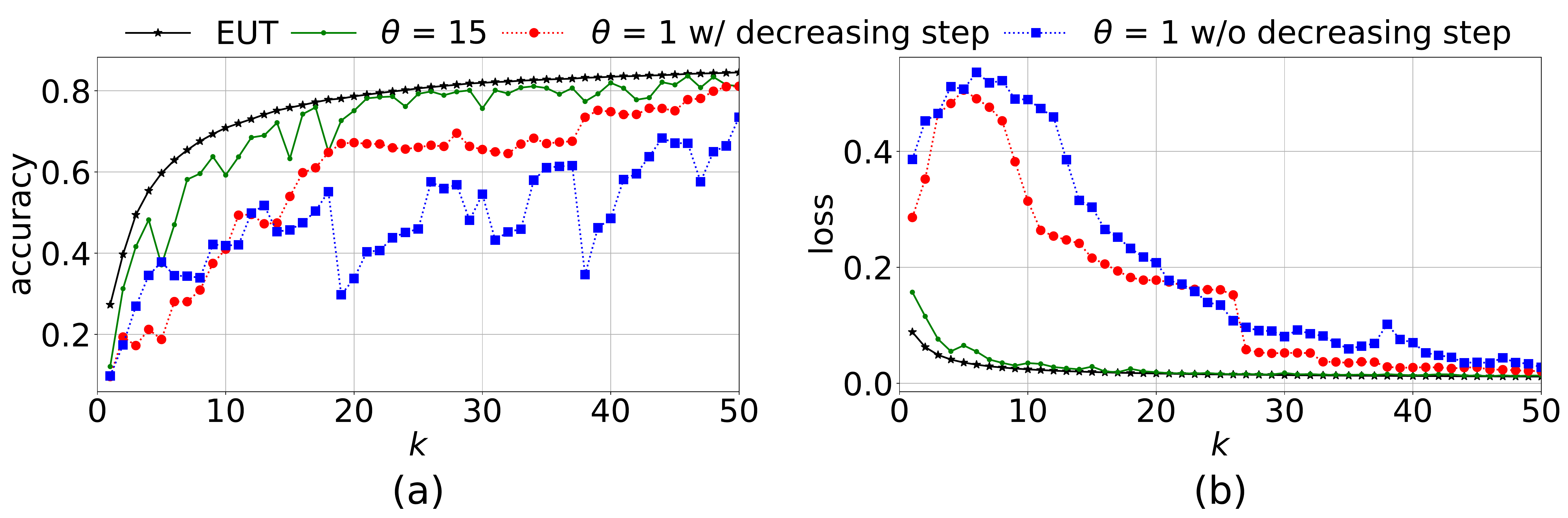}
     \caption{Performance comparison between baseline EUT, and {\tt MH-MT} with and without (w/o) decreasing the gradient descent step size. Decreasing the step size can provide convergence to the optimal solution in cases where a fixed step size is not capable, but also has a slower convergence speed. (MNIST, $625$ Edge Devices)}
     \label{fig:decaying_learning_rate_mnist_625}
\end{minipage}
\end{figure}

\begin{figure}
\centering
\begin{minipage}{.48\textwidth}
        \centering
     \includegraphics[width=\linewidth]{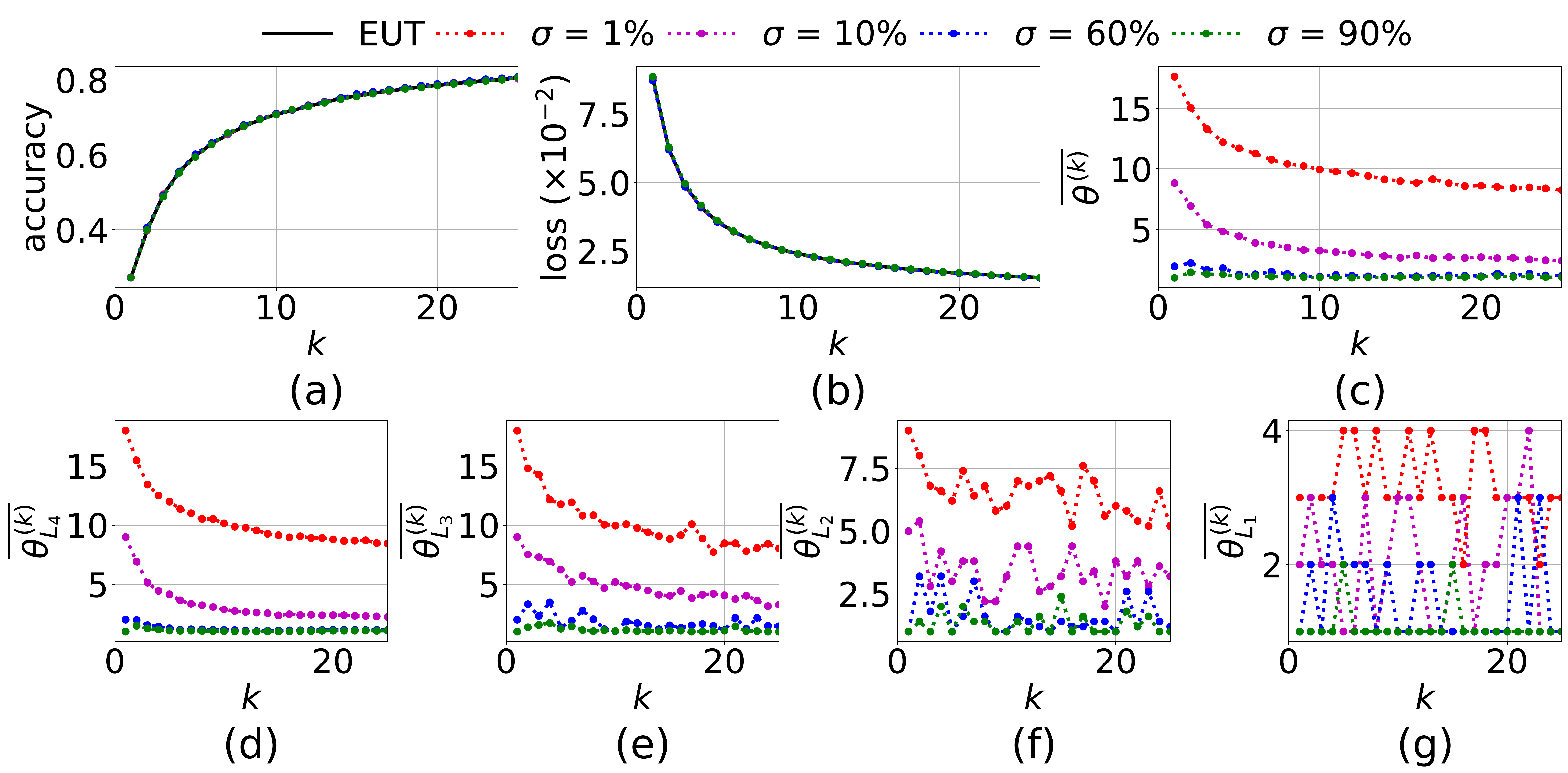}
     \caption{Performance comparison between baseline EUT and {\tt MH-MT} for i.i.d when a finite optimality gap is tolerable. $\sigma_j$ at ${L}_j$ is fixed as $\sigma_j=\sigma' \max_{i}{\Upsilon_{{L}_{j,i}}^{(1)}}$. Tapering of D2D rounds through time and space (layers) can be observed. (MNIST, $625$ Edge Devices)}
     \label{fig:iidIncConSigma_mnist_625}
\end{minipage}%
\hspace{2mm}
\begin{minipage}{.48\textwidth}
     \centering
     \includegraphics[width=\linewidth]{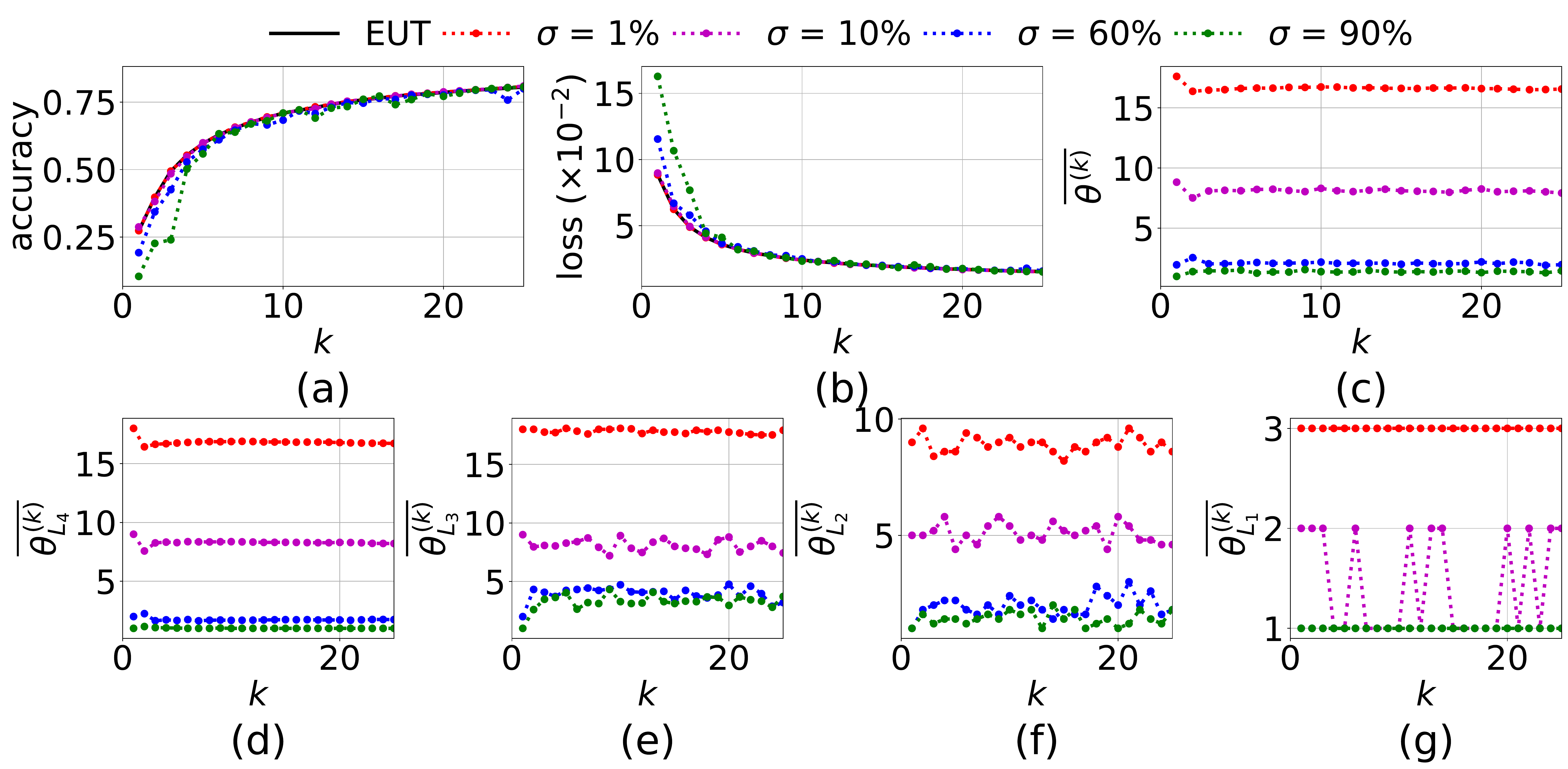}
     \caption{Performance comparison between baseline EUT and {\tt MH-MT} for non-i.i.d when a finite optimality gap is tolerable. $\sigma_i$ is set as in Fig.~\ref{fig:iidIncConSigma_mnist_625}. Smaller loss and higher accuracy are achieved with smaller $\sigma'$, implying more rounds of consensus. (MNIST, $625$ Edge Devices)}
     \label{fig:non_iidIncConSigma_mnist_625}
\end{minipage}
\end{figure}

\begin{figure}
\centering
\begin{minipage}{.48\textwidth}
       \centering
     \includegraphics[width=\linewidth]{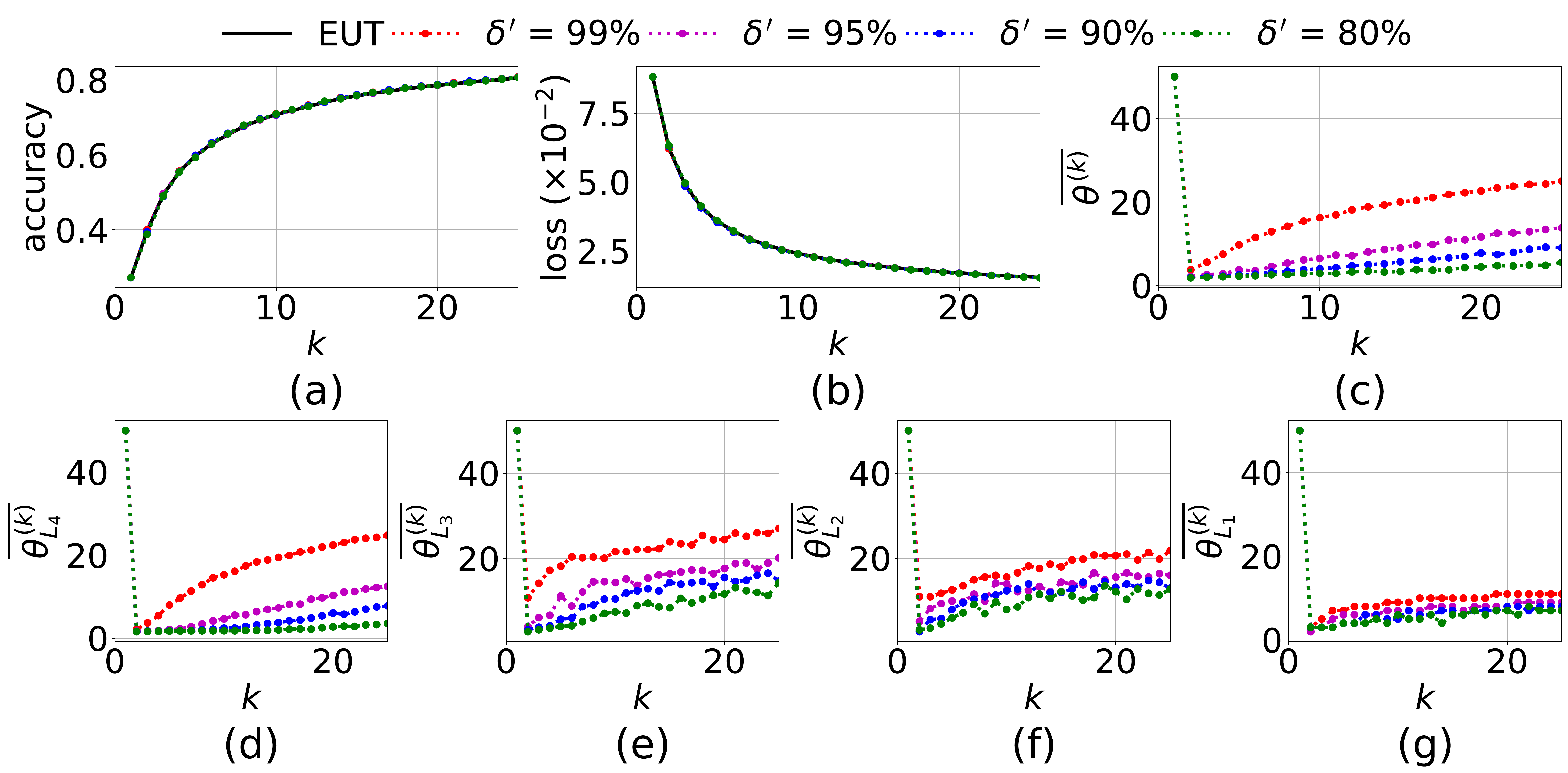}
     \caption{Performance comparison between baseline EUT and {\tt MH-MT} for i.i.d. when linear convergence to the optimal is desired. The value of $\delta$ is set at $\delta=\delta' \frac{\mu}{\eta}$. Boosting of the D2D rounds through time can be observed. Also, tapering through space can be observed by comparing the D2D rounds in the bottom subplots. (MNIST, $625$ Edge Devices)}
     \label{fig:iidConsConDelta_mnist_625}
\end{minipage}%
\hspace{2mm}
\begin{minipage}{.48\textwidth}
      \centering
     \includegraphics[width=\linewidth]{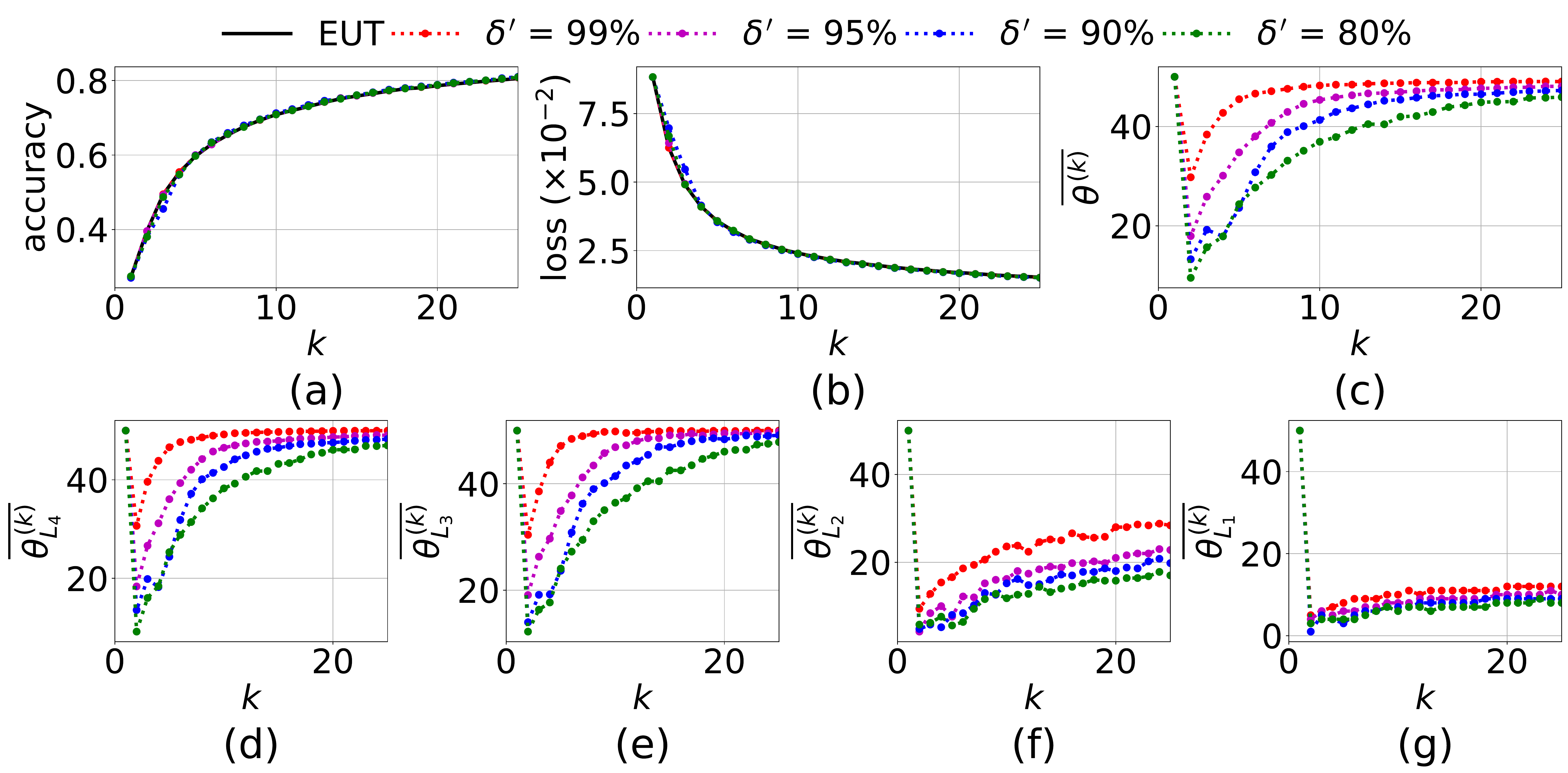}
     \caption{Performance comparison between baseline EUT and {\tt MH-MT} for non-i.i.d when linear convergence to the optimal is desired. The value of $\delta$ is set as in Fig.~\ref{fig:iidConsConDelta_mnist_625}.
     Smaller values of loss and higher accuracy are both associated with larger value of $\delta$, which results in lower error tolerance and more rounds of consensus. (MNIST, $625$ Edge Devices)}
     \label{fig:non_iidConsConDelta_mnist_625}
\end{minipage}
\end{figure}
 
 \begin{figure}
\centering
\begin{minipage}{.48\textwidth}
         \centering
     \includegraphics[width=\linewidth]{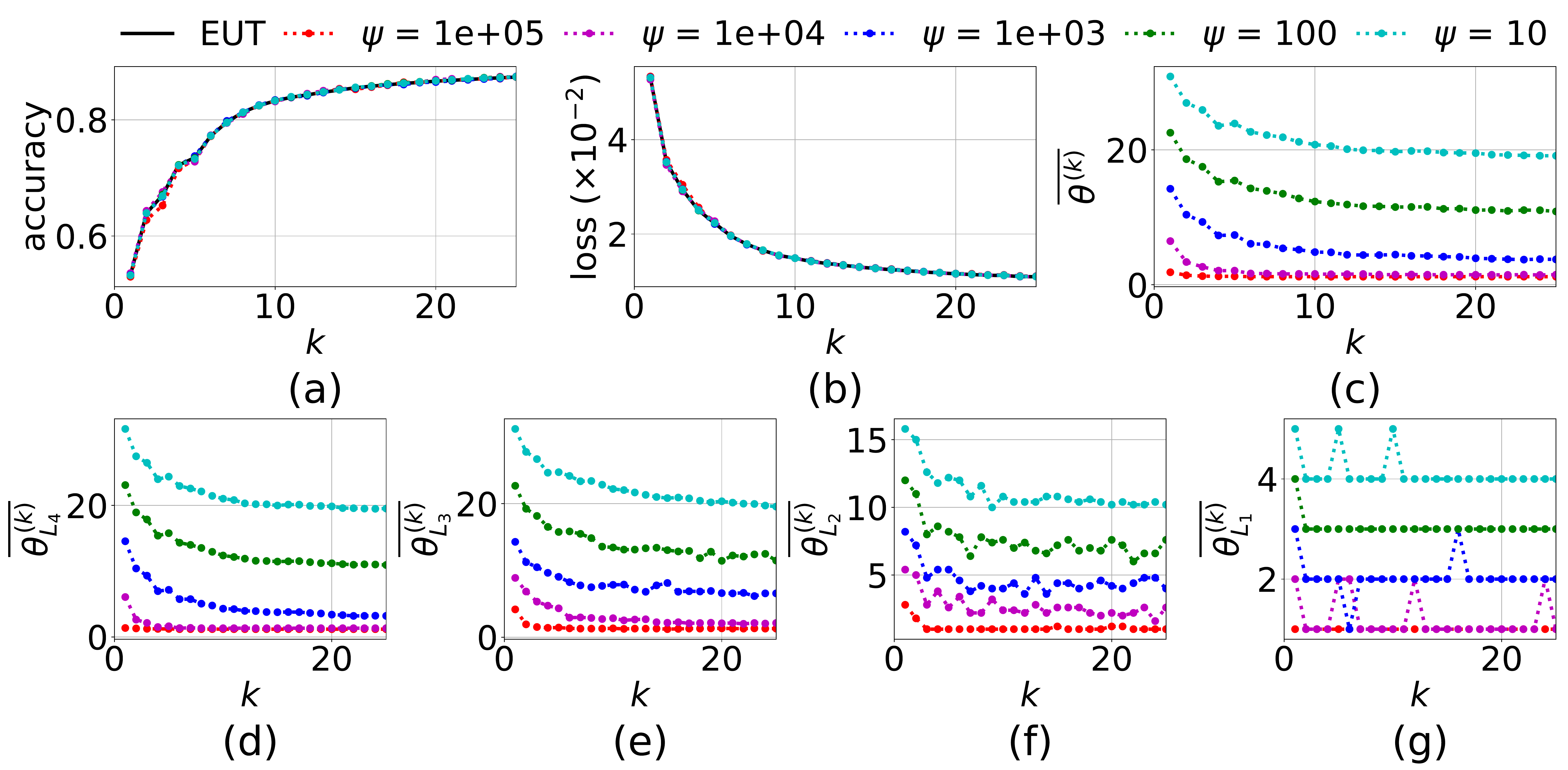}
     \caption{Performance comparison between baseline EUT and {\tt MH-MT}  under i.i.d using NNs with different values of $\psi$. Tapering the D2D rounds through time can be observed. Also, tapering through space can be observed by comparing the D2D rounds in the bottom subplots. (MNIST, $625$ Edge Devices)}
     \label{fig:iidNNIncConec_mnist_625}
\end{minipage}%
\hspace{2mm}
\begin{minipage}{.48\textwidth}
         \centering
     \includegraphics[width=\linewidth]{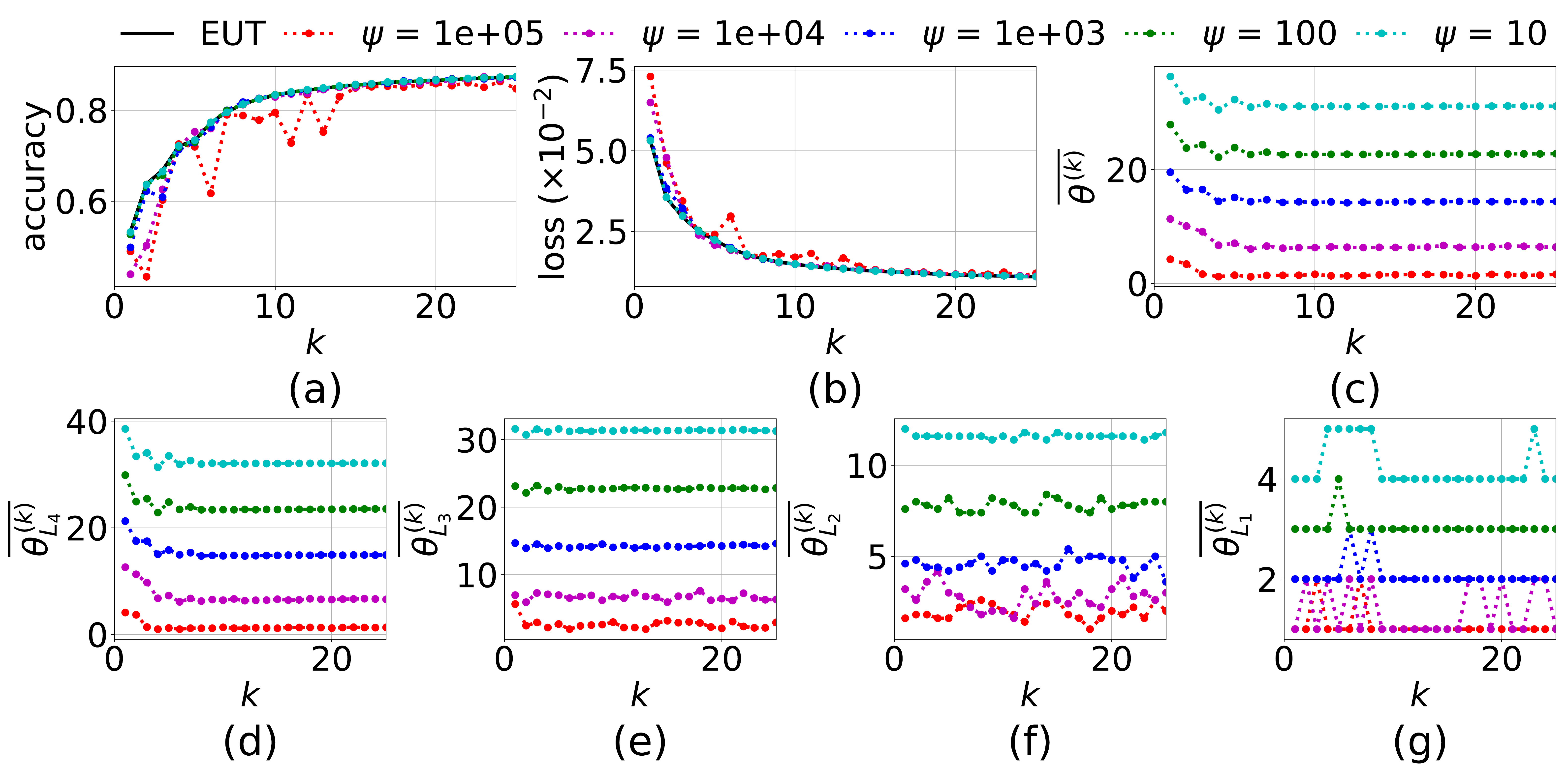}
     \caption{Performance comparison between baseline EUT and {\tt MH-MT} under non-i.i.d. using NNs with different values of $\psi$. Lower loss and higher accuracy are associated with smaller values of $\psi$, which result in lower error tolerance and larger values of D2D rounds over time. (MNIST, $625$ Edge Devices)}
     \label{fig:non_iidConsConDeltaNN_mnist_625}
\end{minipage}
\end{figure}

 \begin{figure}
\centering
\begin{minipage}{.32\textwidth}
     \centering
         \includegraphics[width=0.875\linewidth]{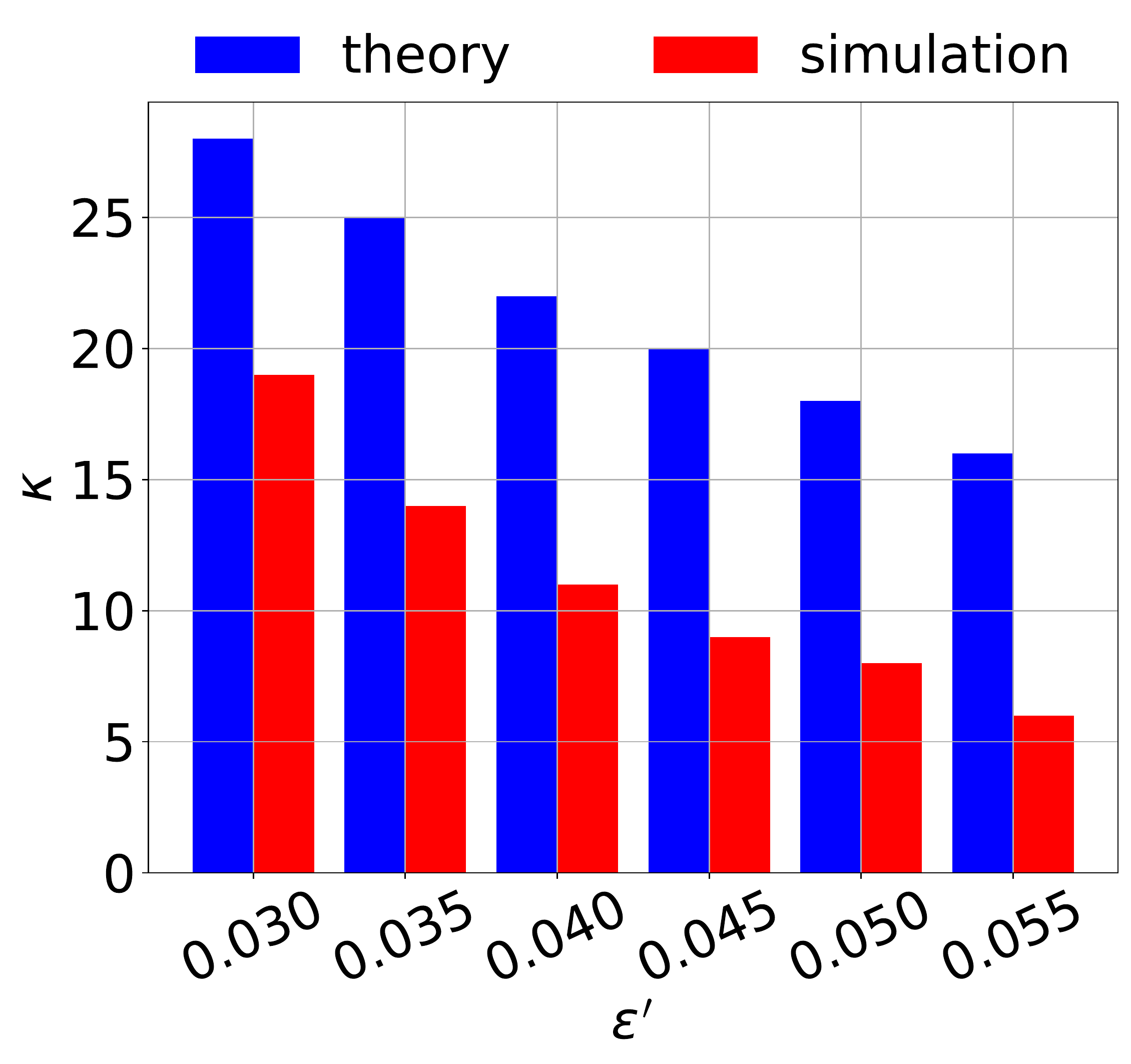}
     \caption{Comparison between the theoretical and simulation results regarding the number of global iterations to achieve an accuracy of $\epsilon' (F(\textbf{w}^{(0)})-F(\textbf{w}^*))$ for different $\epsilon'$. Convergence in practice is faster than the derived upper bound. (MNIST, $625$ Edge Devices)}
     \label{fig:non_iidTheoPracSimDiff_mnist_625}
\end{minipage}%
\hspace{2mm}
\begin{minipage}{.32\textwidth}
     \centering
     \includegraphics[width=\linewidth]{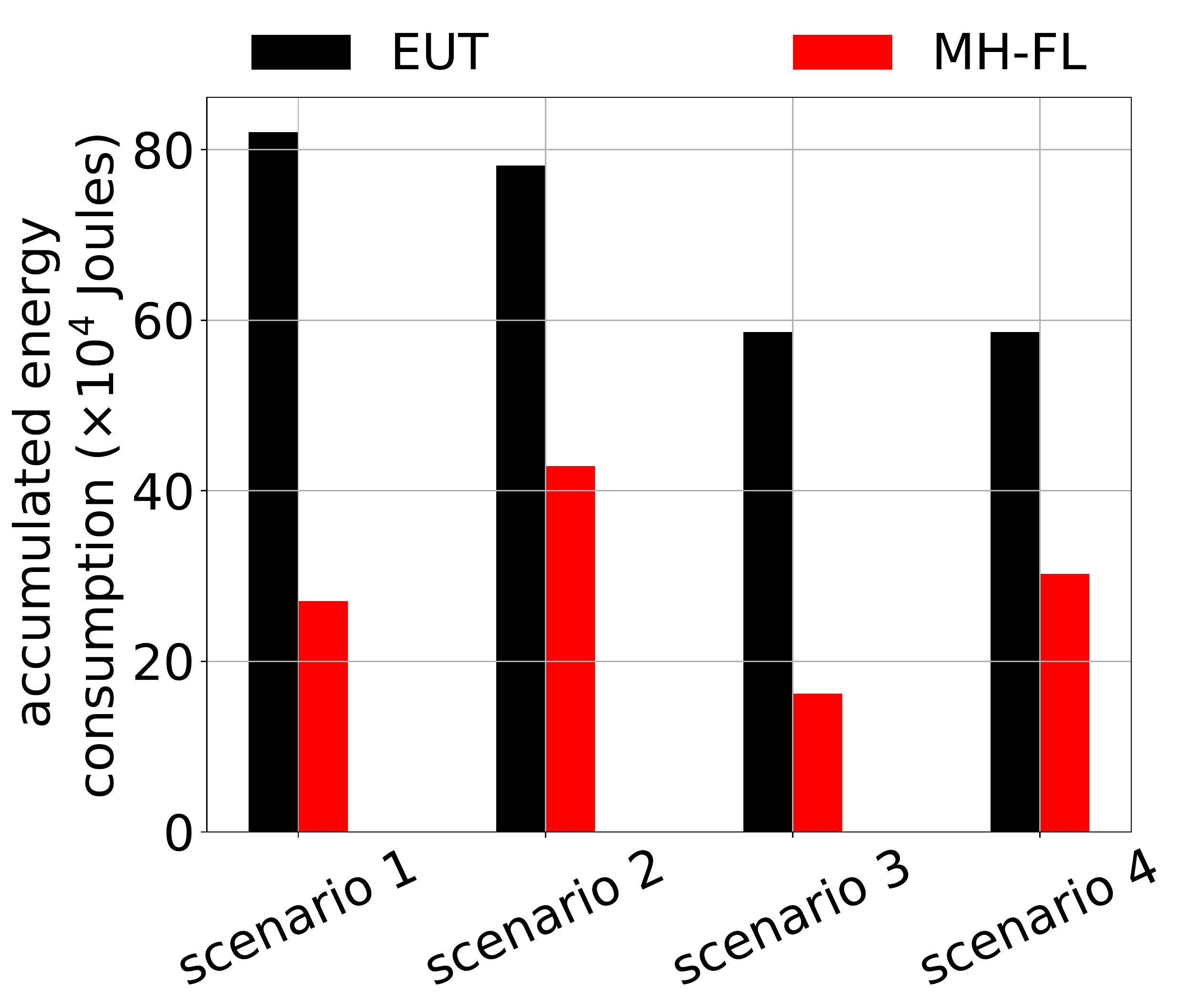}
     \caption{Comparison of accumulated energy consumption between EUT and {\tt MH-MT} over scenario 1: $\sigma' = 0.1$ from Fig.~\ref{fig:iidIncConSigma_mnist_625}, scenario 2: $\sigma' = 0.1$ from Fig.~\ref{fig:non_iidIncConSigma_mnist_625}, scenario 3: $\psi = 10^4$ from Fig.~\ref{fig:iidNNIncConec_mnist_625}, and scenario 4: $\psi = 10^4$ from Fig.~\ref{fig:non_iidConsConDeltaNN_mnist_625}. (MNIST, $625$ Edge Devices)}
     \label{fig:accum_energy_625_MNIST}
\end{minipage}%
\hspace{2mm}
\begin{minipage}{.32\textwidth}
     \centering
     \includegraphics[width=\linewidth]{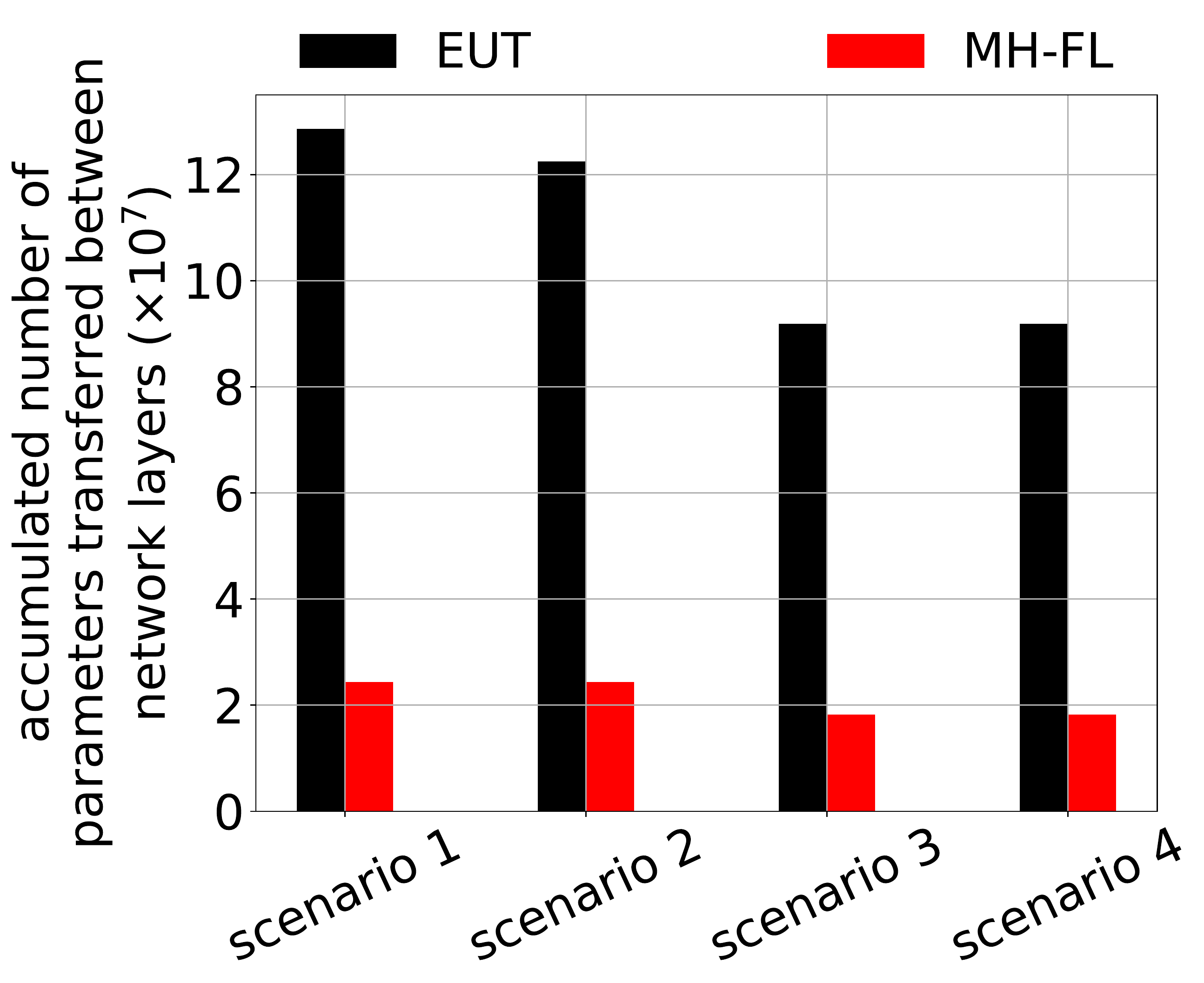}
     \caption{Comparison of parameters transferred among layers in EUT vs {\tt MH-MT} over scenario 1: $\sigma' = 0.1$ from Fig.~\ref{fig:iidIncConSigma_mnist_625}, scenario 2: $\sigma' = 0.1$ from Fig.~\ref{fig:non_iidIncConSigma_mnist_625}, scenario 3: $\psi = 10^4$ from Fig.~\ref{fig:iidNNIncConec_mnist_625}, and scenario 4: $\psi = 10^4$ from Fig.~\ref{fig:non_iidConsConDeltaNN_mnist_625}. (MNIST, $625$ Edge Devices)}
     \label{fig:accumData_625_MNIST}
\end{minipage}
\end{figure}

  \begin{figure}
     
 \end{figure}

%%%%%%%%%%%%%%%%%%%%%%%%%%%%%% End MNIST 625
%%%%%%%%%%%%%%%%%%%%%%%%%%%%%%%
%%%%%%%%%%%%%%%%%%%%%%%%%%%%
%%%%%%%%%%%%%%%%%%%%%%%%%%%%%

%%%%%%%%%%%%%%%%%%%%%%%%%%%%%% Start F-MNIST 125
%%%%%%%%%%%%%%%%%%%%%%%%%%%%%%
%%%%%%%%%%%%%%%%%%%%%%%%%%%%%%
%%%%%%%%%%%%%%%%%%%%%%%%%%%%%

% \begin{figure}
% \centering
% \begin{minipage}{.48\textwidth}
%      \centering
%      \includegraphics[width=\linewidth]{fmnist/sim_comparison_across_consensus_rounds-eps-converted-to.pdf}
%      \caption{Comparison of performance with different number of consensus rounds on MNIST dataset. As the number of rounds increase fog training matches FL. (FMNIST-125)}
%      \label{fig:GenFigGoodIntuitionFMNIST}
% \end{minipage}%
% \hspace{2mm}
% \begin{minipage}{.48\textwidth}
%   \centering
%          \includegraphics[width=\linewidth]{fmnist/sim_comparison_of_decaying_lr_on_convergence-eps-converted-to.pdf}
%      \caption{Convergence is improved if learning rate decays as training progresses for a previously diverging training with constant learning rate (FMNIST-125)}
%      \label{fig:decaying_learning_rate_FMNIST}
% \end{minipage}
% \end{figure}

\begin{figure}
\vspace{-.05mm}
\centering
\begin{minipage}{.48\textwidth}
     \centering
     \includegraphics[width=\linewidth]{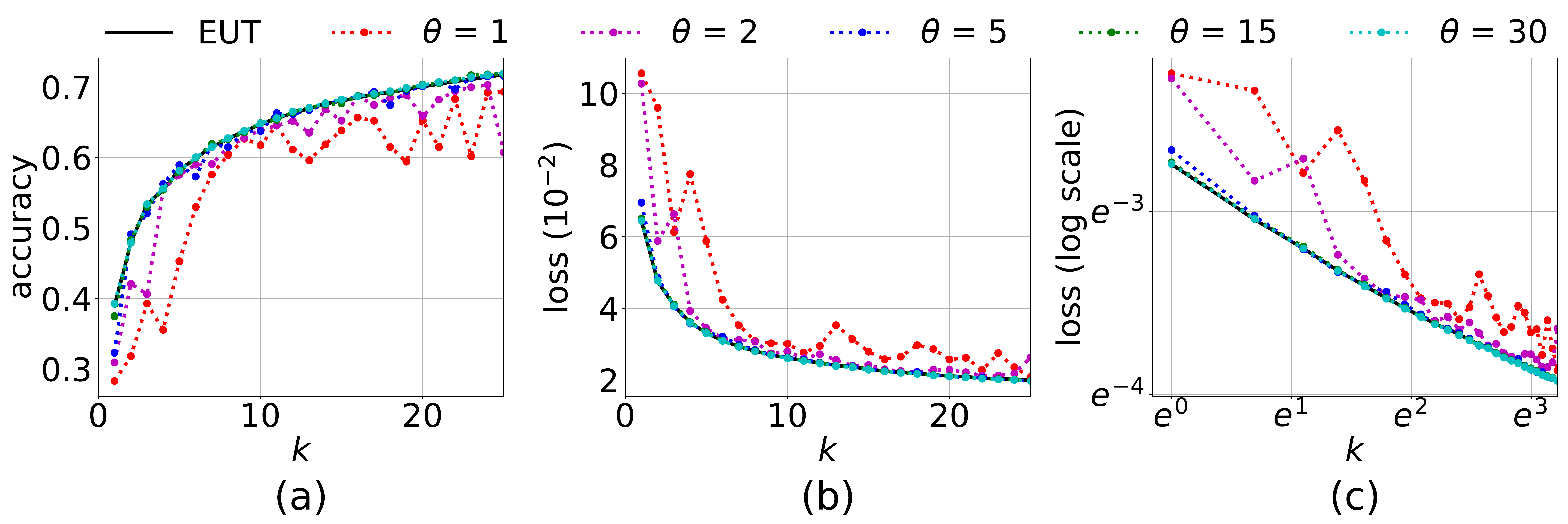}
     \caption{Performance comparison between baseline EUT and {\tt MH-MT} when a fixed number of D2D rounds $\theta$ is used at every cluster of the network, for non-i.i.d. As the number of D2D rounds increases, {\tt MH-MT} performs more similar to the EUT baseline and the learning is more stable. (FMNIST, $125$ Edge Devices)}
     \label{fig:GenFigGoodIntuition_FMNIST_125}
\end{minipage}%
\hspace{2mm}
\begin{minipage}{.48\textwidth}
  \centering
         \includegraphics[width=\linewidth]{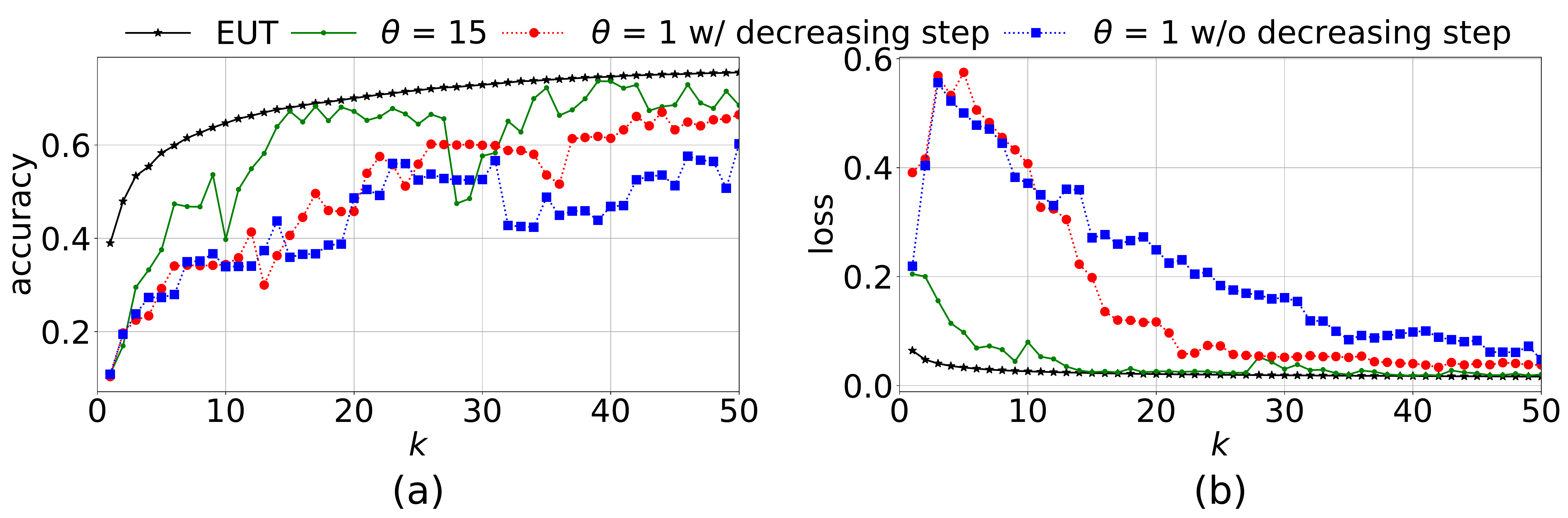}
     \caption{Performance comparison between baseline EUT, and {\tt MH-MT} with and without (w/o) decreasing the gradient descent step size. Decreasing the step size can provide convergence to the optimal solution in cases where a fixed step size is not capable, but also has a slower convergence speed. (FMNIST, $125$ Edge Devices)}
     \label{fig:decaying_learning_rate_FMNIST_125}
\end{minipage}
\end{figure}

\begin{figure}
\centering
\begin{minipage}{.48\textwidth}
        \centering
     \includegraphics[width=\linewidth]{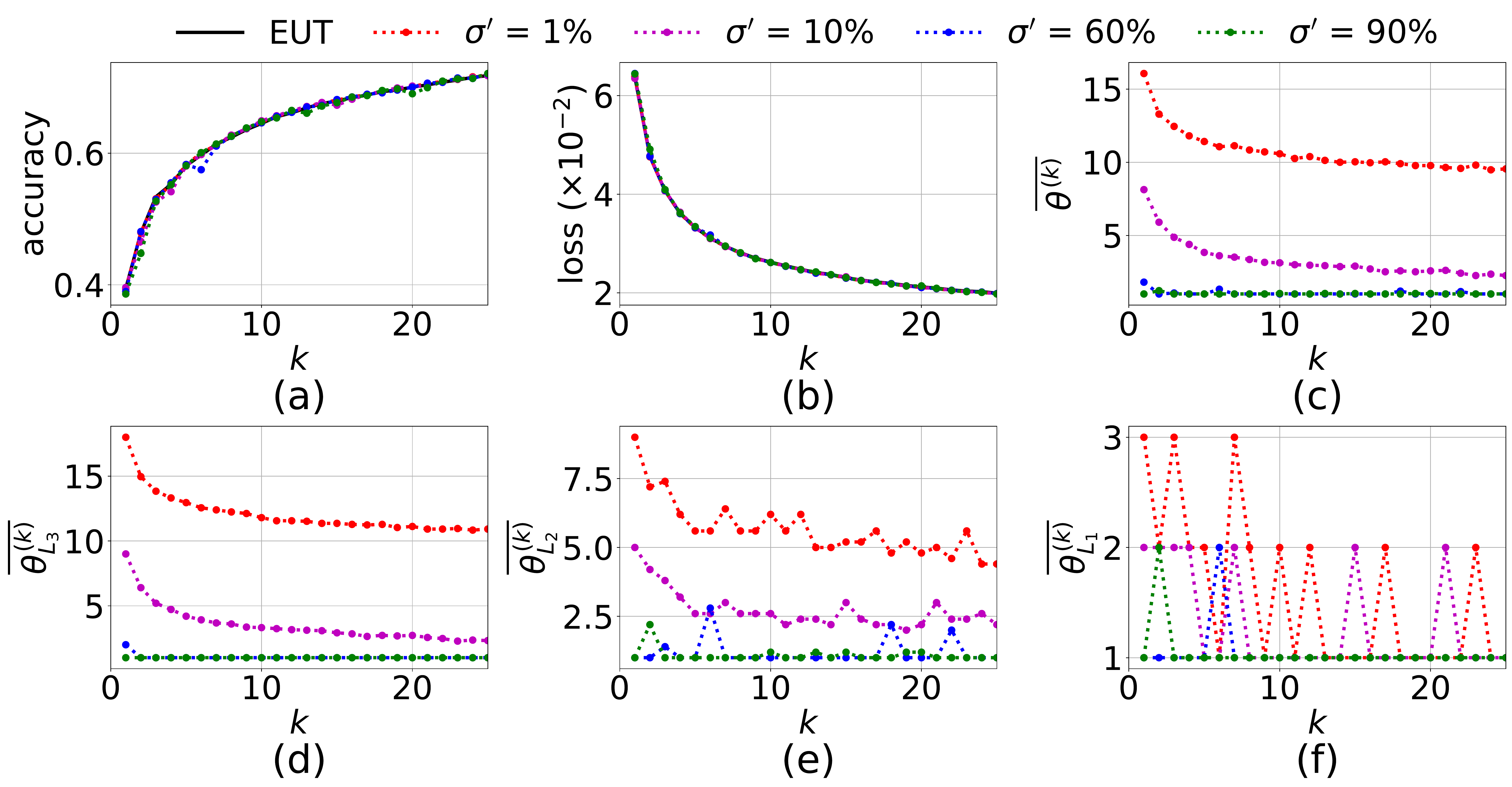}
     \caption{Performance comparison between baseline EUT and {\tt MH-MT} for i.i.d when a finite optimality gap is tolerable. $\sigma_j$ at ${L}_j$ is fixed as $\sigma_j=\sigma' \max_{i}{\Upsilon_{{L}_{j,i}}^{(1)}}$. Tapering of D2D rounds through time and space (layers) can be observed. (FMNIST, $125$ Edge Devices)}
     \label{fig:iidIncConSigma_FMNIST_125}
\end{minipage}%
\hspace{2mm}
\begin{minipage}{.48\textwidth}
     \centering
     \includegraphics[width=\linewidth]{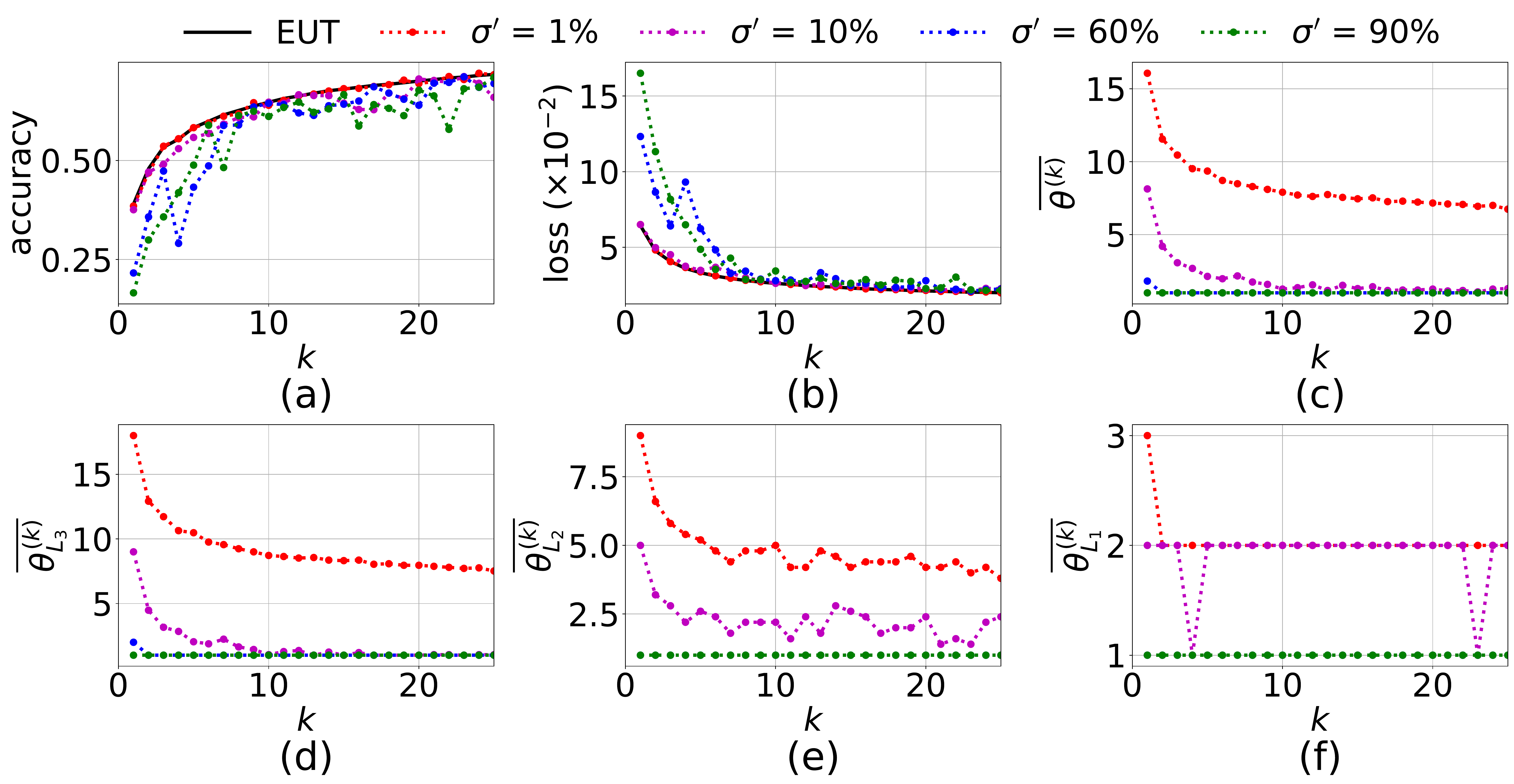}
     \caption{Performance comparison between baseline EUT and {\tt MH-MT} for non-i.i.d when a finite optimality gap is tolerable. $\sigma_i$ is set as in Fig.~\ref{fig:iidIncConSigma_FMNIST_125}. Smaller loss and higher accuracy are achieved with smaller $\sigma'$, implying more rounds of consensus. (FMNIST, $125$ Edge Devices)}
     \label{fig:non_iidIncConSigma_FMNIST_125}
\end{minipage}
\end{figure}

\begin{figure}
\centering
\begin{minipage}{.48\textwidth}
       \centering
     \includegraphics[width=\linewidth]{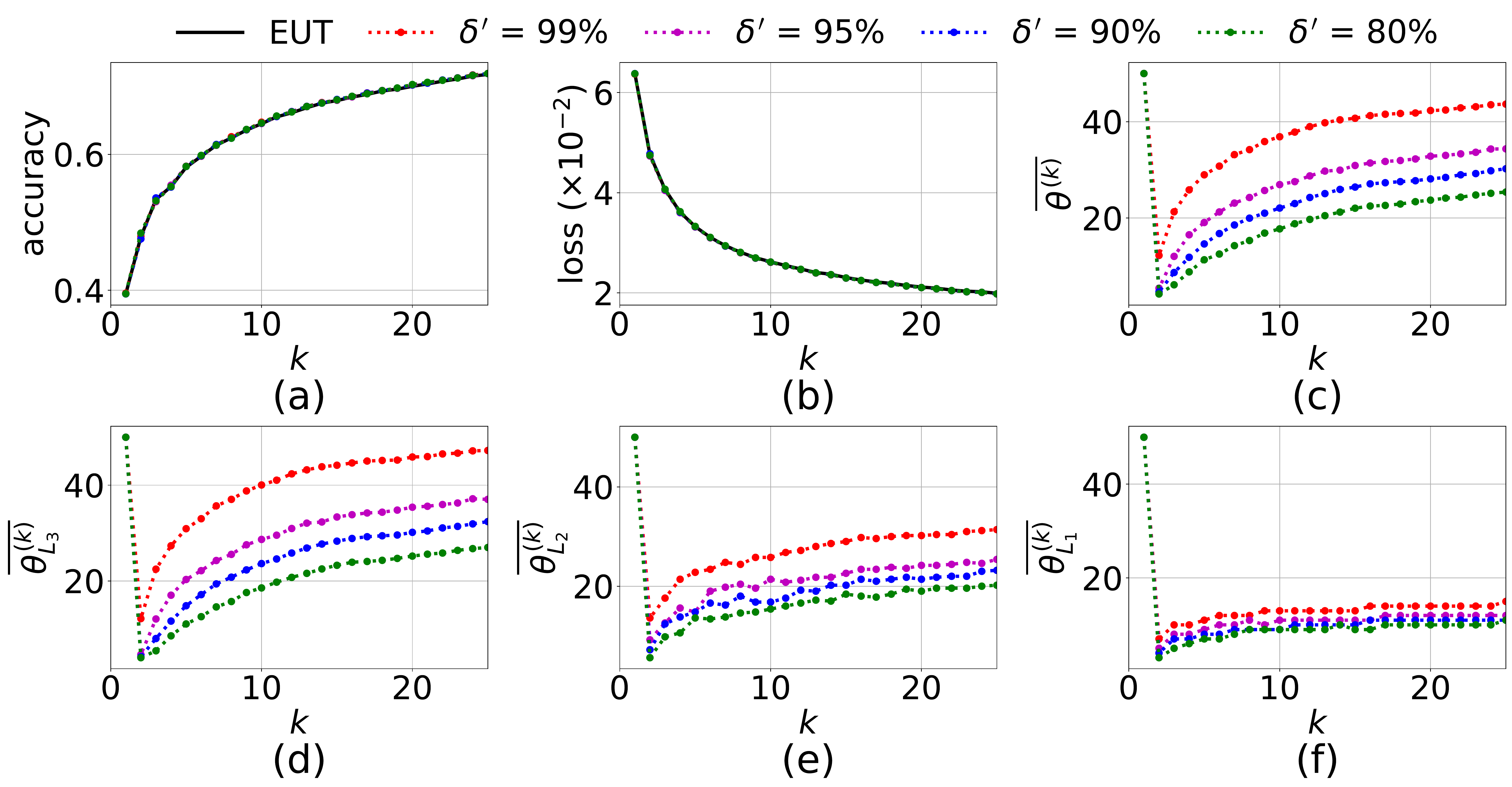}
     \caption{Performance comparison between baseline EUT and {\tt MH-MT} for i.i.d. when linear convergence to the optimal is desired. The value of $\delta$ is set at $\delta=\delta' \frac{\mu}{\eta}$. Boosting of the D2D rounds through time can be observed. Also, tapering through space can be observed by comparing the D2D rounds in the bottom subplots. (FMNIST, $125$ Edge Devices)}
     \label{fig:iidConsConDelta_FMNIST_125}
\end{minipage}%
\hspace{2mm}
\begin{minipage}{.48\textwidth}
      \centering
     \includegraphics[width=\linewidth]{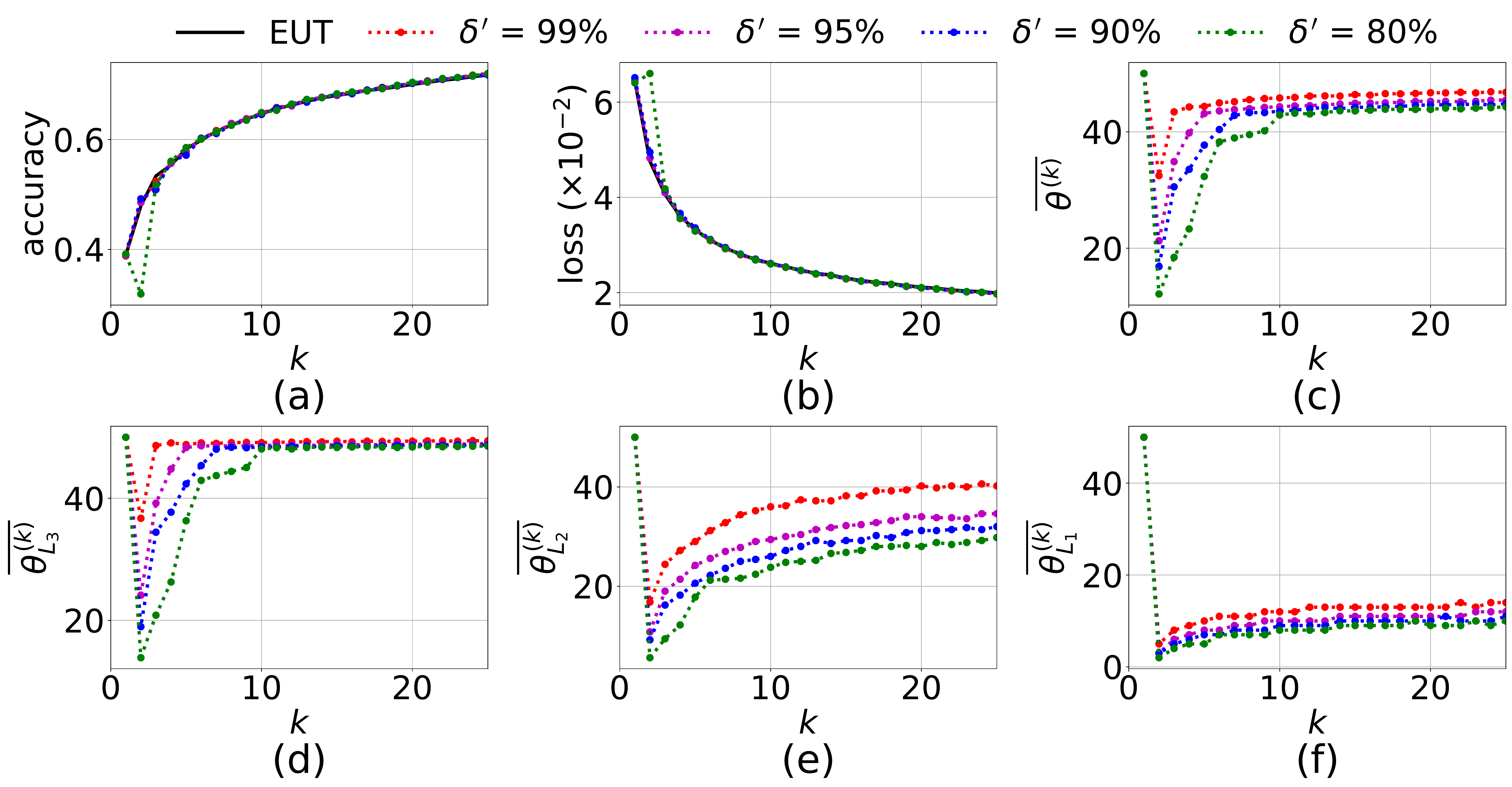}
     \caption{Performance comparison between baseline EUT and {\tt MH-MT} for non-i.i.d when linear convergence to the optimal is desired. The value of $\delta$ is set as in Fig.~\ref{fig:iidConsConDelta_FMNIST_125}.
     Smaller values of loss and higher accuracy are both associated with larger value of $\delta$, which results in lower error tolerance and more rounds of consensus. (FMNIST, $125$ Edge Devices)}
     \label{fig:non_iidConsConDelta_FMNIST_125}
\end{minipage}
\end{figure}
 
 \begin{figure}
\centering
\begin{minipage}{.48\textwidth}
         \centering
     \includegraphics[width=\linewidth]{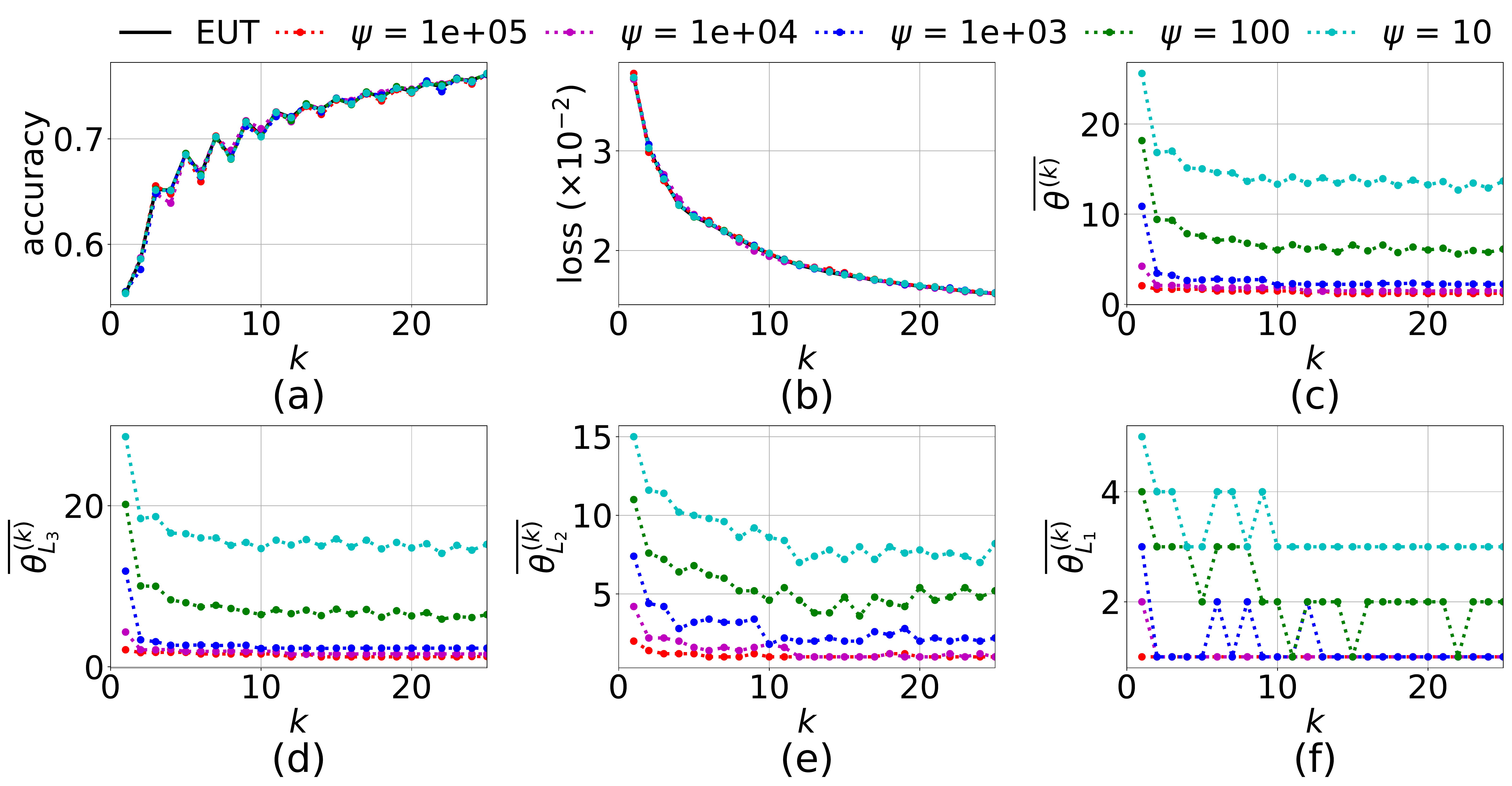}
     \caption{Performance comparison between baseline EUT and {\tt MH-MT}  under i.i.d using NNs with different values of $\psi$. Tapering the D2D rounds through time can be observed. Also, tapering through space can be observed by comparing the D2D rounds in the bottom subplots. (FMNIST, $125$ Edge Devices)}
     \label{fig:iidNNIncConec_FMNIST_125}
\end{minipage}%
\hspace{2mm}
\begin{minipage}{.48\textwidth}
         \centering
     \includegraphics[width=\linewidth]{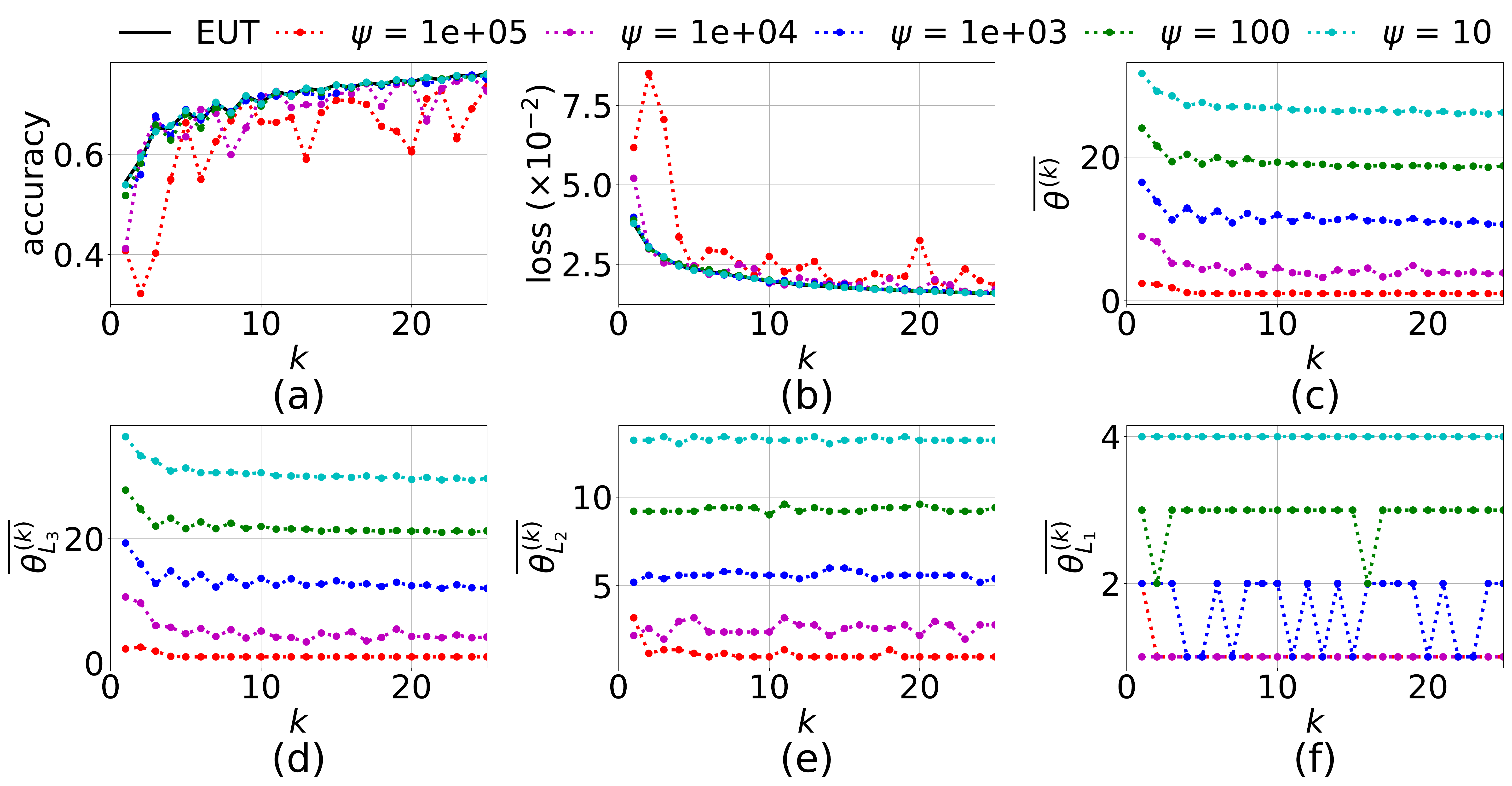}
     \caption{Performance comparison between baseline EUT and {\tt MH-MT} under non-i.i.d. using NNs with different values of $\psi$. Lower loss and higher accuracy are associated with smaller values of $\psi$, which result in lower error tolerance and larger values of D2D rounds over time. (FMNIST, $125$ Edge Devices)}
     \label{fig:non_iidConsConDeltaNN_FMNIST_125}
\end{minipage}
\end{figure}

%   \begin{figure}
%      \centering
%      \includegraphics[width=0.24\linewidth]{fmnist/sim_comparison_theory_practical_delta-eps-converted-to.pdf}
%      \caption{Comaparison of kappa theoretically derived vs practically observed on FMNIST dataset trained on 125 worker nodes. (FMNIST, $125$ Edge Devices)}
%      \label{fig:iidNNIncConec_FMNIST_125}
%  \end{figure}
 
 \begin{figure}
\centering
\begin{minipage}{.32\textwidth}
     \centering
         \includegraphics[width=0.875\linewidth]{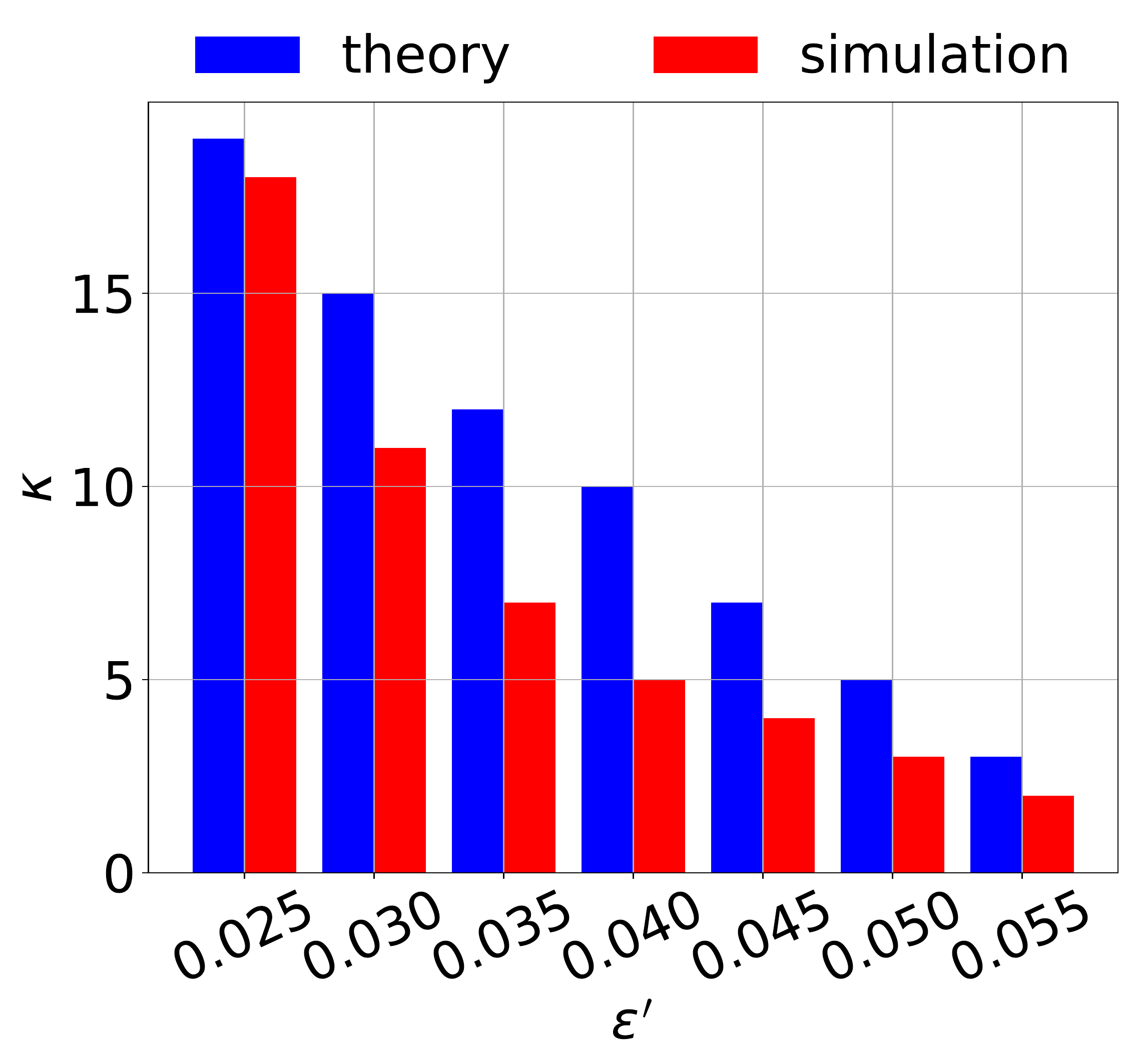}
     \caption{Comparison between the theoretical and simulation results regarding the number of global iterations to achieve an accuracy of $\epsilon' (F(\mathbf{w}^{(0)})-F(\mathbf{w}^*))$ for different $\epsilon'$. Convergence in practice is faster than the derived upper bound. (FMNIST, $125$ Edge Devices)}
     \label{fig:non_iidTheoPracSimDiff_FMNIST_125}
\end{minipage}%
\hspace{2mm}
\begin{minipage}{.32\textwidth}
     \centering
     \includegraphics[width=\linewidth]{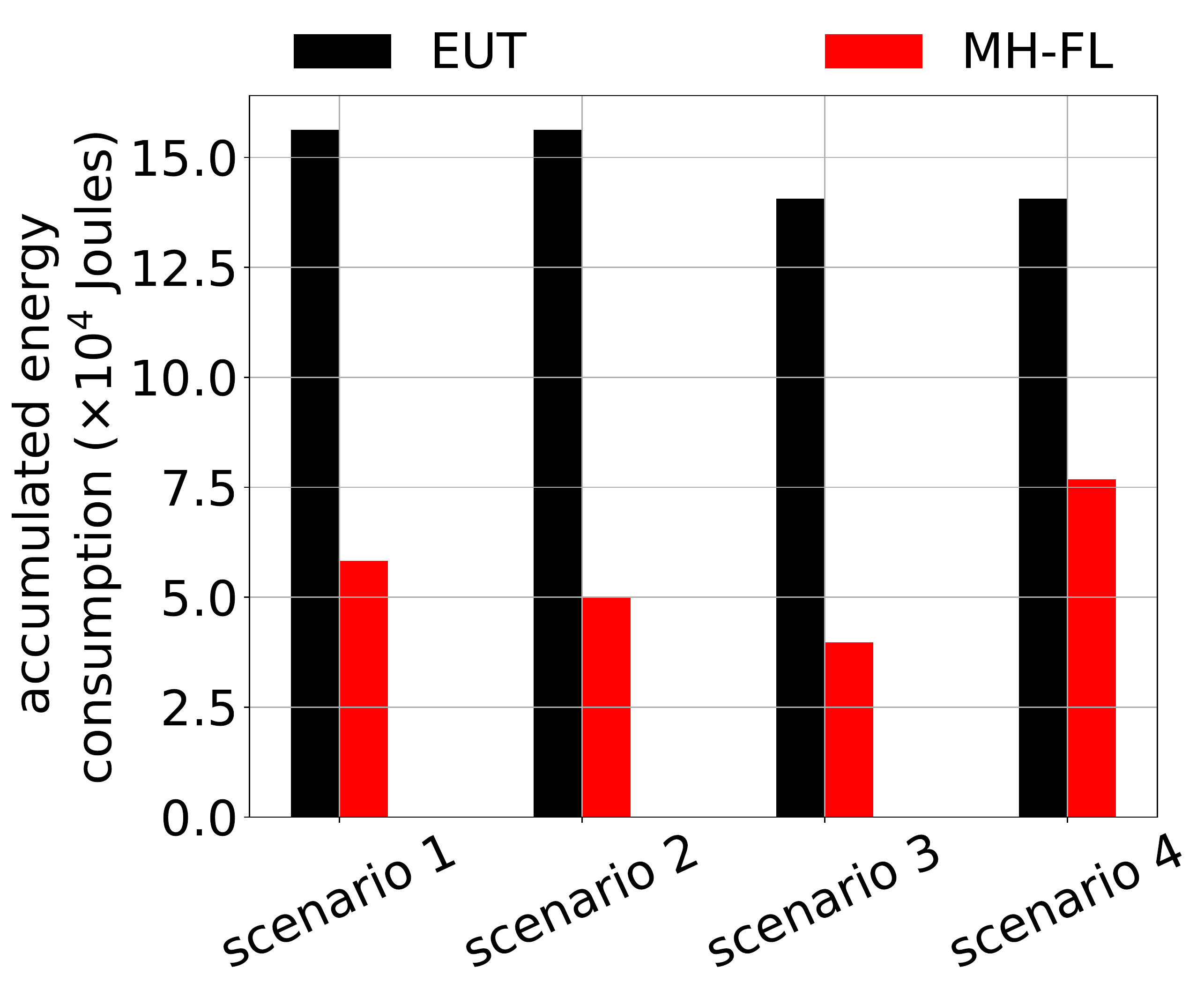}
     \caption{Comparison of accumulated energy consumption between EUT and {\tt MH-MT} over scenario 1: $\sigma' = 0.1$ from Fig.~\ref{fig:iidIncConSigma_FMNIST_125}, scenario 2: $\sigma' = 0.1$ from Fig.~\ref{fig:non_iidIncConSigma_FMNIST_125}, scenario 3: $\psi = 10^4$ from Fig.~\ref{fig:iidNNIncConec_FMNIST_125}, and scenario 4: $\psi = 10^4$ from Fig.~\ref{fig:non_iidConsConDeltaNN_FMNIST_125}. (FMNIST, $125$ Edge Devices)}
     \label{fig:accum_energy_125_FMNIST}
\end{minipage}%
\hspace{2mm}
\begin{minipage}{.32\textwidth}
     \centering
     \includegraphics[width=\linewidth]{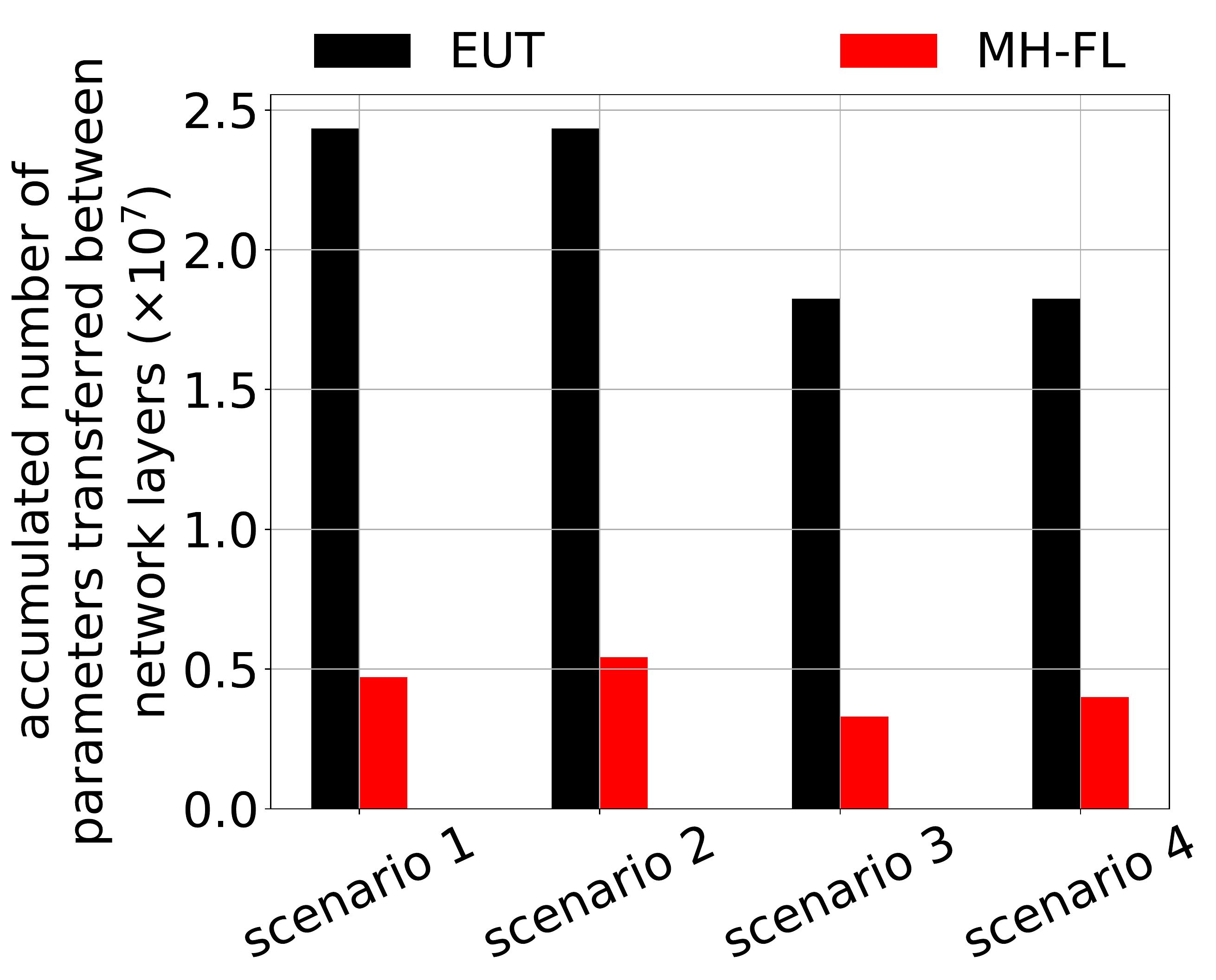}
     \caption{Comparison of parameters transferred among layers in EUT vs {\tt MH-MT} over scenario 1: $\sigma' = 0.1$ from Fig.~\ref{fig:iidIncConSigma_FMNIST_125}, scenario 2: $\sigma' = 0.1$ from Fig.~\ref{fig:non_iidIncConSigma_FMNIST_125}, scenario 3: $\psi = 10^4$ from Fig.~\ref{fig:iidNNIncConec_FMNIST_125}, and scenario 4: $\psi = 10^4$ from Fig.~\ref{fig:non_iidConsConDeltaNN_FMNIST_125}. (FMNIST, $125$ Edge Devices)}
     \label{fig:accumData_FMNIST_125}
\end{minipage}%\vspace{-5mm}
\end{figure}

%%%%%%%%%%%%%%%%%%%%%%%%%%%%%% End F-MNIST 125
%%%%%%%%%%%%%%%%%%%%%%%%%%%%%%
%%%%%%%%%%%%%%%%%%%%%%%%%%%%%%
%%%%%%%%%%%%%%%%%%%%%%%%%%%%%

%%%%%%%%%%%%%%%%%%%%%%%%%%%%%% Start F-MNIST 625
%%%%%%%%%%%%%%%%%%%%%%%%%%%%%%
%%%%%%%%%%%%%%%%%%%%%%%%%%%%%%
%%%%%%%%%%%%%%%%%%%%%%%%%%%%%

\begin{figure}
\centering
\begin{minipage}{.48\textwidth}
     \centering
     \includegraphics[width=\linewidth]{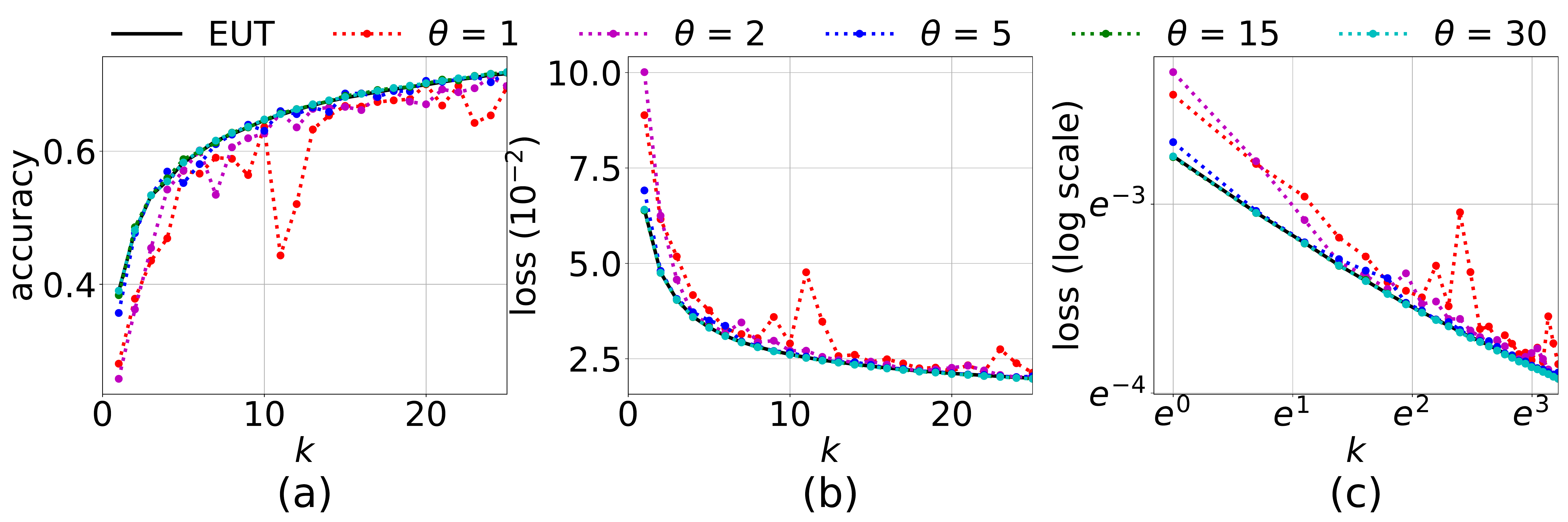}
     \caption{Performance comparison between baseline EUT and {\tt MH-MT} when a fixed number of D2D rounds $\theta$ is used at every cluster of the network, for non-i.i.d. As the number of D2D rounds increases, {\tt MH-MT} performs more similar to the EUT baseline and the learning is more stable. (FMNIST, $625$ Edge Devices)}
         \label{fig:GenFigGoodIntuition_FMNIST_625}
\end{minipage}%
\hspace{2mm}
\begin{minipage}{.48\textwidth}
  \centering
         \includegraphics[width=\linewidth]{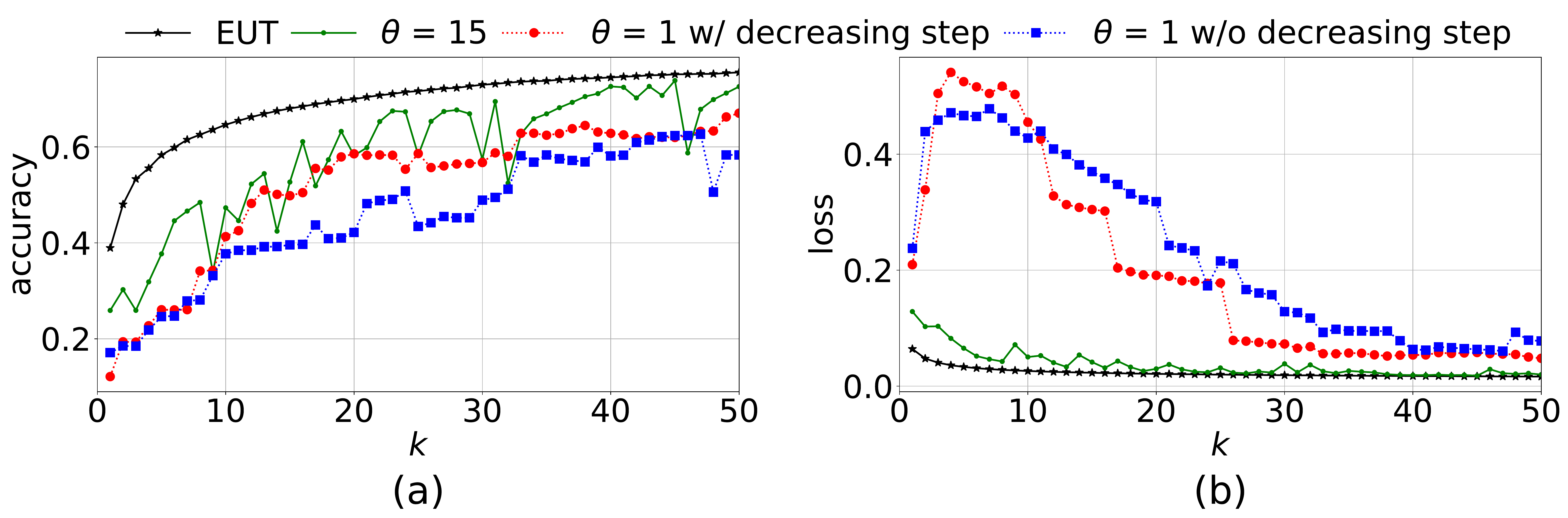}
     \caption{Performance comparison between baseline EUT, and {\tt MH-MT} with and without (w/o) decreasing the gradient descent step size. Decreasing the step size can provide convergence to the optimal solution in cases where a fixed step size is not capable, but also has a slower convergence speed. (FMNIST, $625$ Edge Devices)}
     \label{fig:decaying_learning_rate_FMNIST_625}
\end{minipage}
\end{figure}

\begin{figure}
\centering
\begin{minipage}{.48\textwidth}
        \centering
     \includegraphics[width=\linewidth]{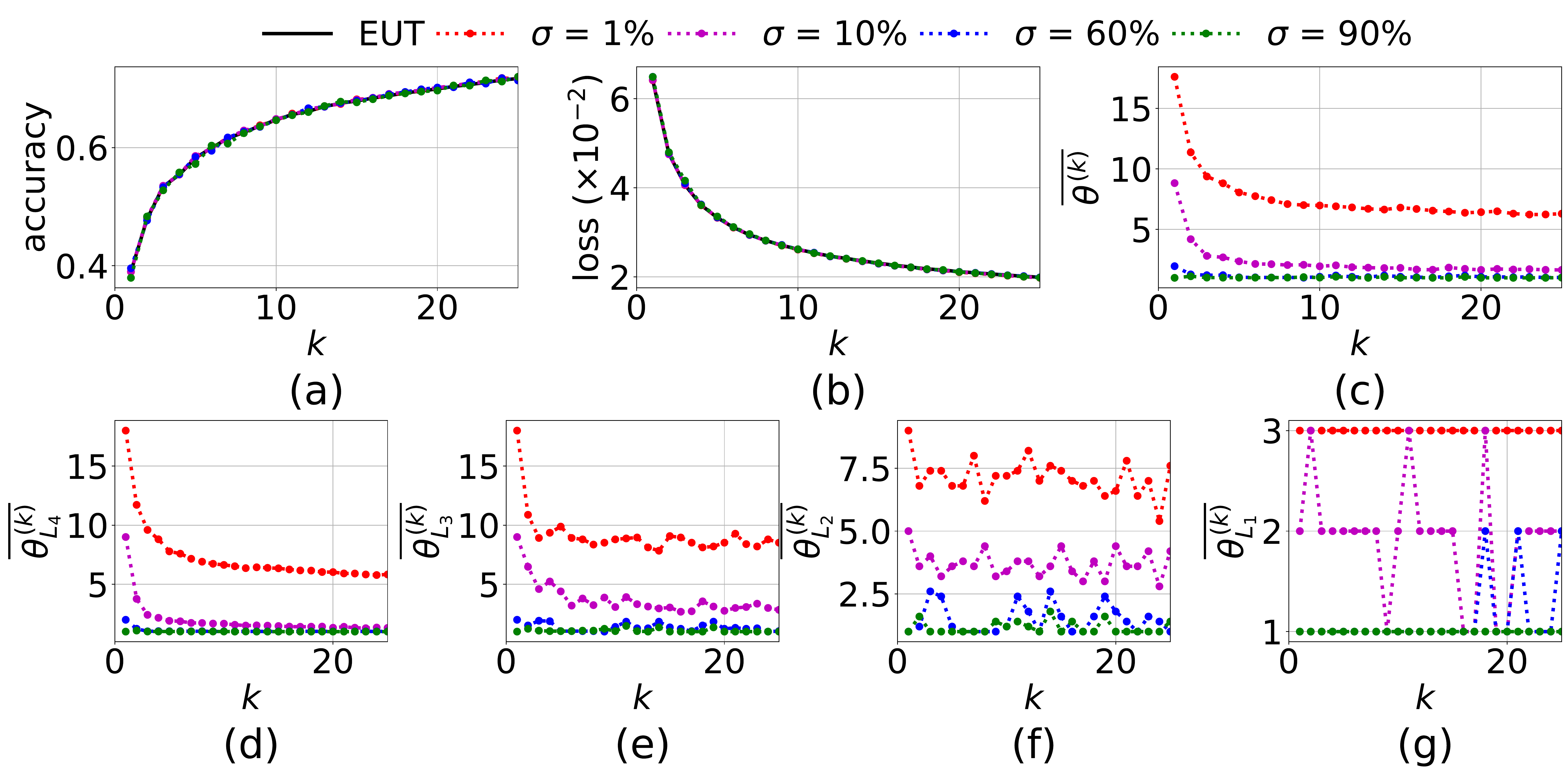}
     \caption{Performance comparison between baseline EUT and {\tt MH-MT} for i.i.d when a finite optimality gap is tolerable. $\sigma_j$ at ${L}_j$ is fixed as $\sigma_j=\sigma' \max_{i}{\Upsilon_{{L}_{j,i}}^{(1)}}$. Tapering of D2D rounds through time and space (layers) can be observed. (FMNIST, $625$ Edge Devices)}
     \label{fig:iidIncConSigma_FMNIST_625}
\end{minipage}%
\hspace{2mm}
\begin{minipage}{.48\textwidth}
     \centering
     \includegraphics[width=\linewidth]{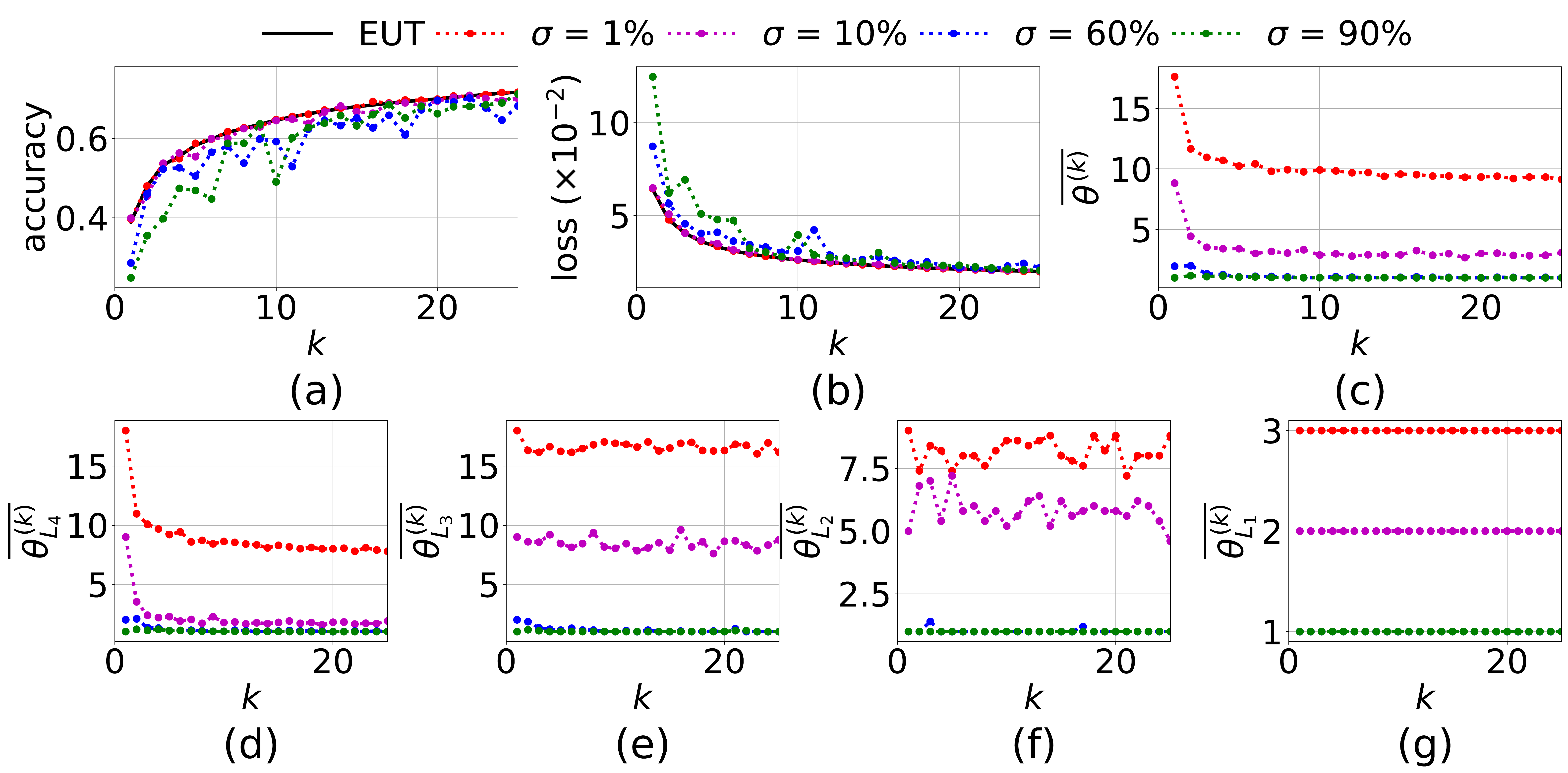}
     \caption{Performance comparison between baseline EUT and {\tt MH-MT} for non-i.i.d when a finite optimality gap is tolerable. $\sigma_i$ is set as in Fig.~\ref{fig:iidIncConSigma_FMNIST_625}. Smaller loss and higher accuracy are achieved with smaller $\sigma'$, implying more rounds of consensus. (FMNIST, $625$ Edge Devices)}
     \label{fig:non_iidIncConSigma_FMNIST_625}
\end{minipage}
\end{figure}

\begin{figure}
\centering
\begin{minipage}{.48\textwidth}
       \centering
     \includegraphics[width=\linewidth]{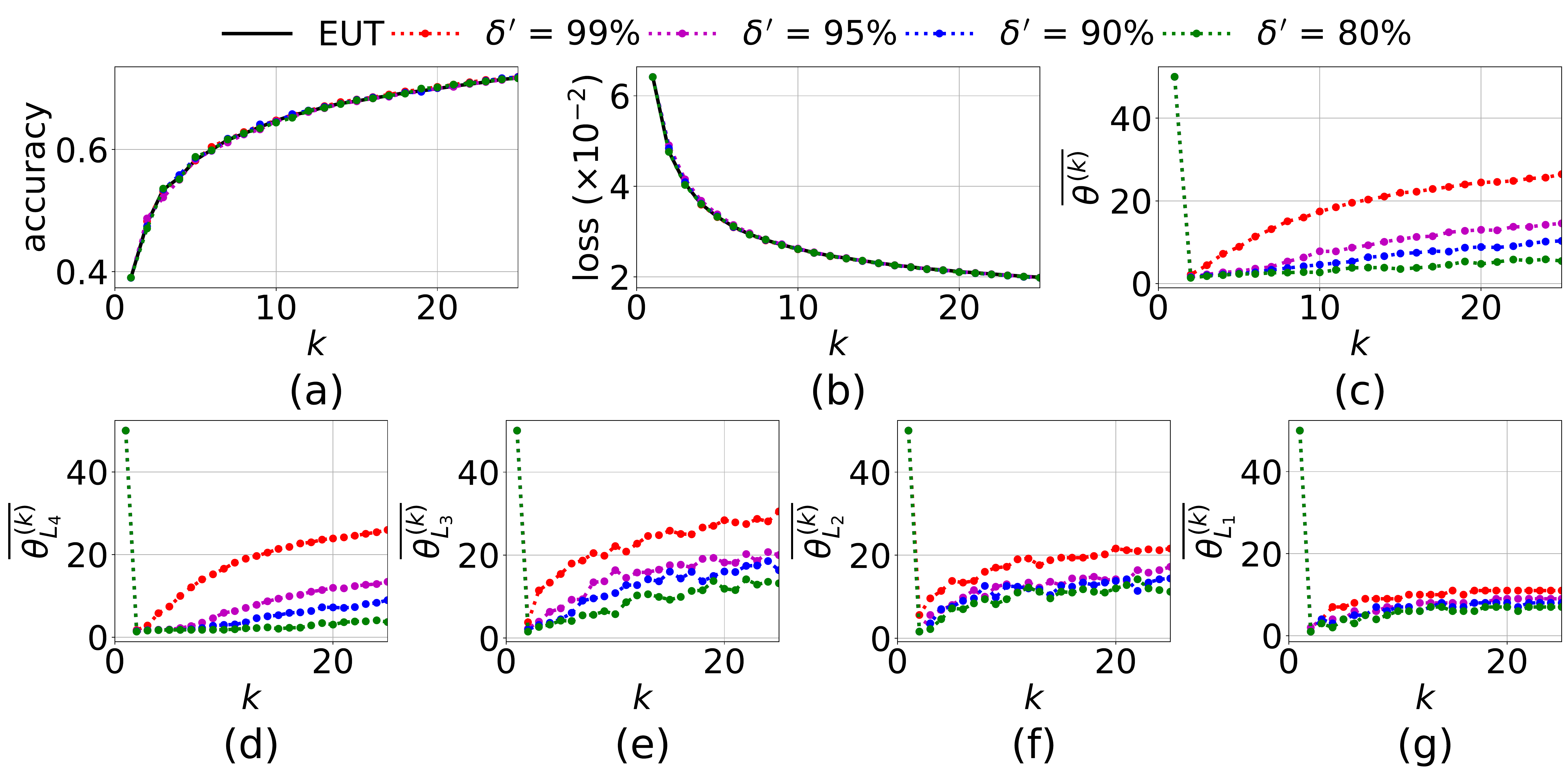}
     \caption{Performance comparison between baseline EUT and {\tt MH-MT} for i.i.d. when linear convergence to the optimal is desired. The value of $\delta$ is set at $\delta=\delta' \frac{\mu}{\eta}$. Boosting of the D2D rounds through time can be observed. Also, tapering through space can be observed by comparing the D2D rounds in the bottom subplots. (FMNIST, $625$ Edge Devices)}
     \label{fig:iidConsConDelta_FMNIST_625}
\end{minipage}%
\hspace{2mm}
\begin{minipage}{.48\textwidth}
      \centering
     \includegraphics[width=\linewidth]{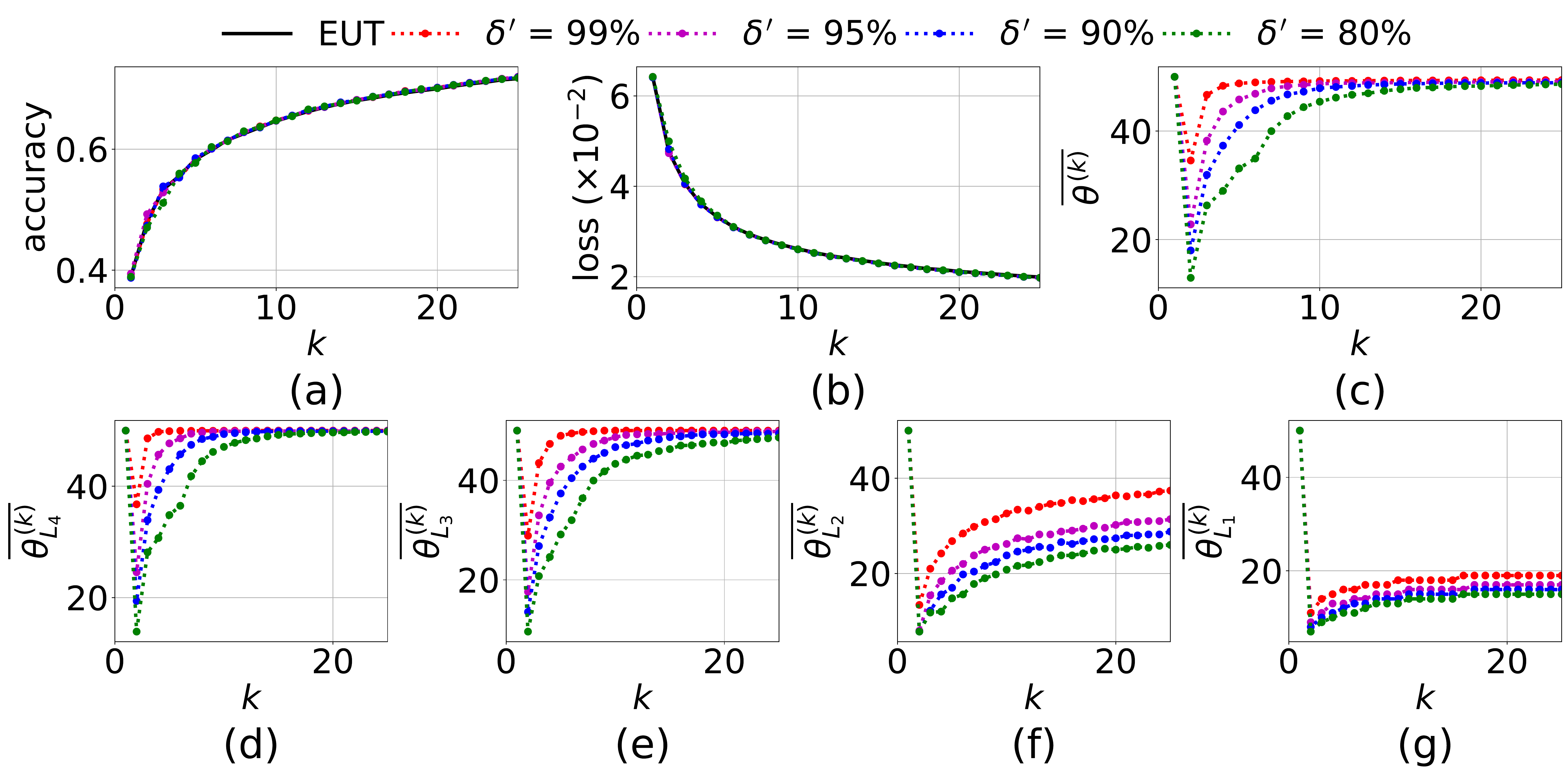}
     \caption{Performance comparison between baseline EUT and {\tt MH-MT} for non-i.i.d when linear convergence to the optimal is desired. The value of $\delta$ is set as in Fig.~\ref{fig:iidConsConDelta_FMNIST_625}.
     Smaller values of loss and higher accuracy are both associated with larger value of $\delta$, which results in lower error tolerance and more rounds of consensus. (FMNIST, $625$ Edge Devices)}
     \label{fig:non_iidConsConDelta_FMNIST_625}
\end{minipage}
\end{figure}
 
 \begin{figure}
\centering
\begin{minipage}{.48\textwidth}
         \centering
     \includegraphics[width=\linewidth]{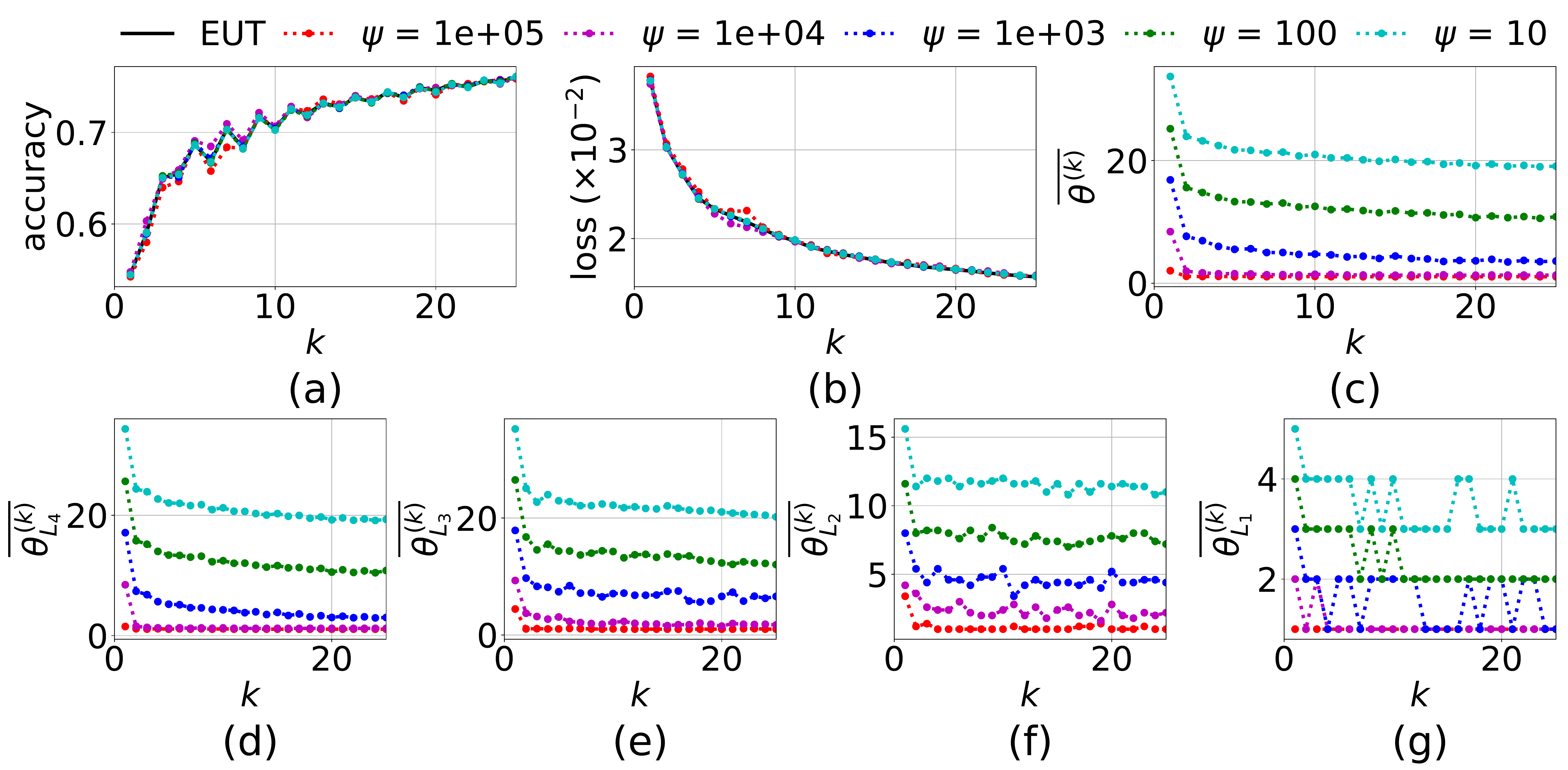}
     \caption{Performance comparison between baseline EUT and {\tt MH-MT}  under i.i.d using NNs with different values of $\psi$. Tapering the D2D rounds through time can be observed. Also, tapering through space can be observed by comparing the D2D rounds in the bottom subplots. (FMNIST, $625$ Edge Devices)}
     \label{fig:iidNNIncConec_FMNIST_625}
\end{minipage}%
\hspace{2mm}
\begin{minipage}{.48\textwidth}
         \centering
     \includegraphics[width=\linewidth]{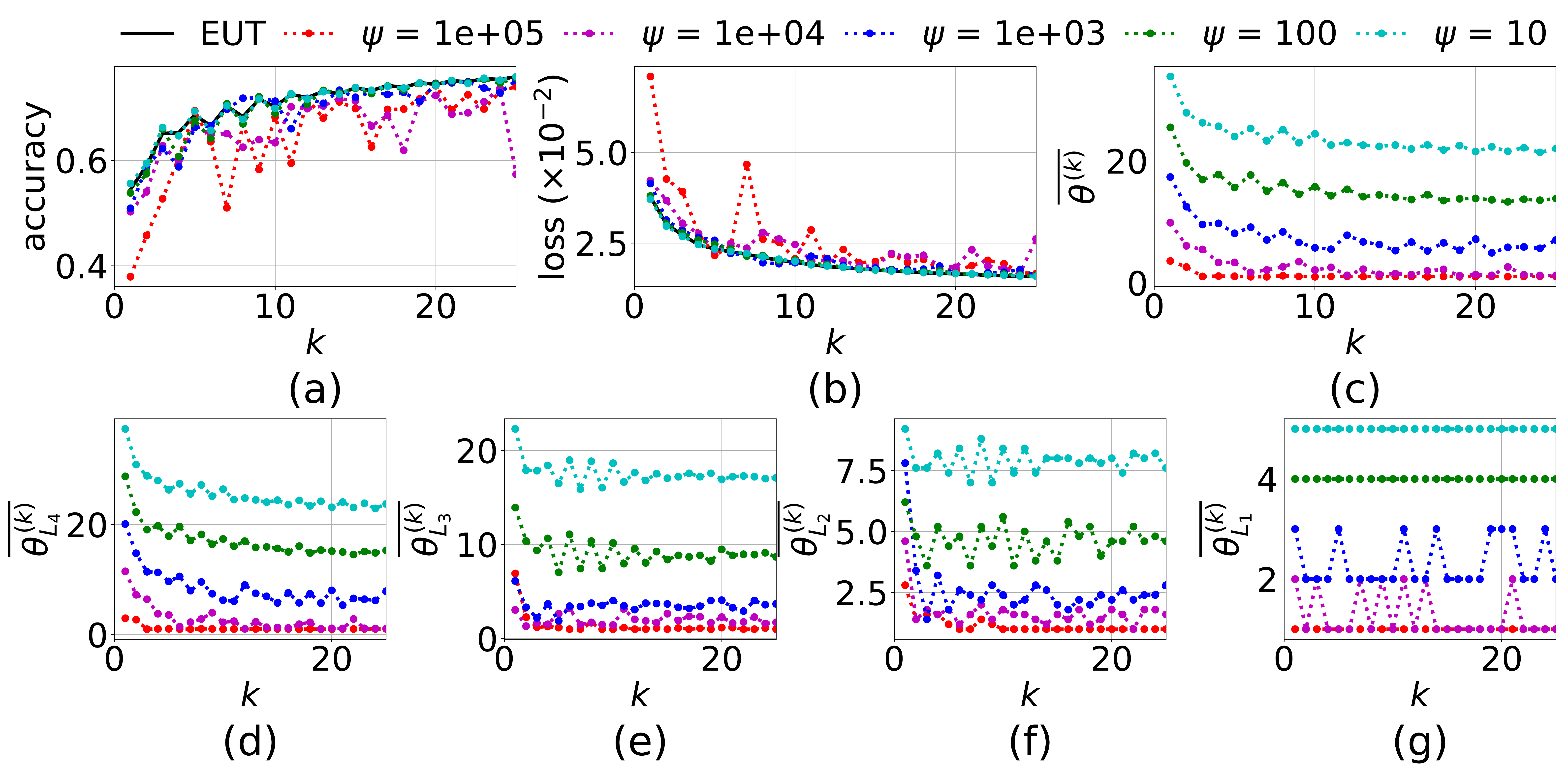}
     \caption{Performance comparison between baseline EUT and {\tt MH-MT} under non-i.i.d. using NNs with different values of $\psi$. Lower loss and higher accuracy are associated with smaller values of $\psi$, which result in lower error tolerance and larger values of D2D rounds over time. (FMNIST, $625$ Edge Devices)}
     \label{fig:non_iidConsConDeltaNN_FMNIST_625}
\end{minipage}
\end{figure}

%   \begin{figure}
%      \centering
%      \includegraphics[width=0.24\linewidth]{fmnist/sim_comparison_theory_practical_delta-eps-converted-to.pdf}
%      \caption{Comparison between the theoretical and simulation results regarding the number of global iterations to achieve an accuracy of $\epsilon' (F(w^{(0)})-F(w^*))$ for different $\epsilon'$. Convergence in practice is faster than the derived upper bound. (FMNIST, $625$ Edge Devices)}
%      \label{fig:iidNNIncConec_FMNIST_625}
%  \end{figure}
 
 \begin{figure}
\centering
\begin{minipage}{.32\textwidth}
     \centering
         \includegraphics[width=\linewidth]{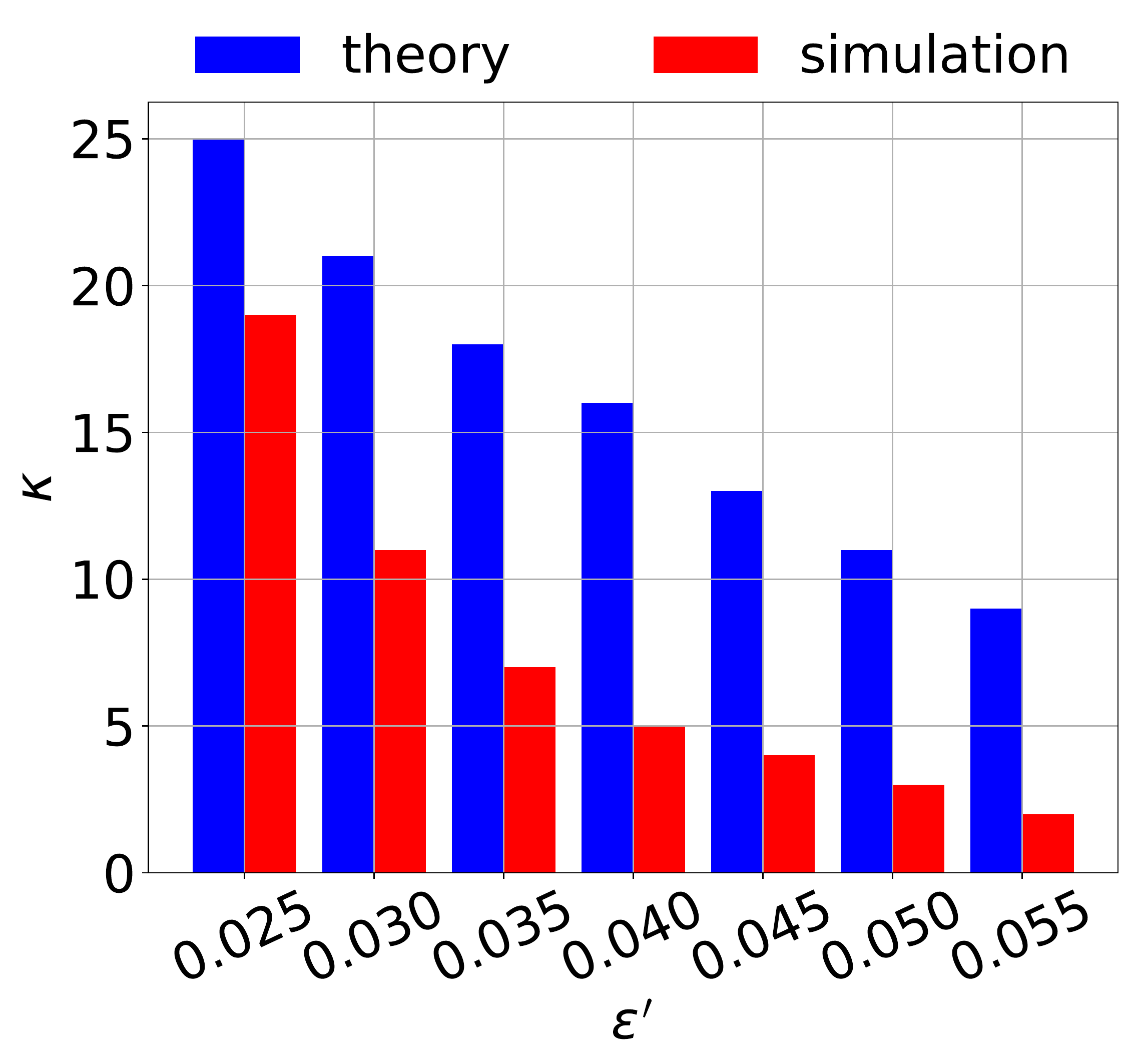}
     \caption{Comparison between the theoretical and simulation results regarding the number of global iterations to achieve an accuracy of $\epsilon' (F(\mathbf{w}^{(0)})-F(\mathbf{w}^*))$ for different $\epsilon'$. Convergence in practice is faster than the derived upper bound. (FMNIST, $625$ Edge Devices)}
     \label{fig:non_iidTheoPracSimDiff_FMNIST_625}
\end{minipage}%
\hspace{2mm}
\begin{minipage}{.32\textwidth}
     \centering
     \includegraphics[width=\linewidth]{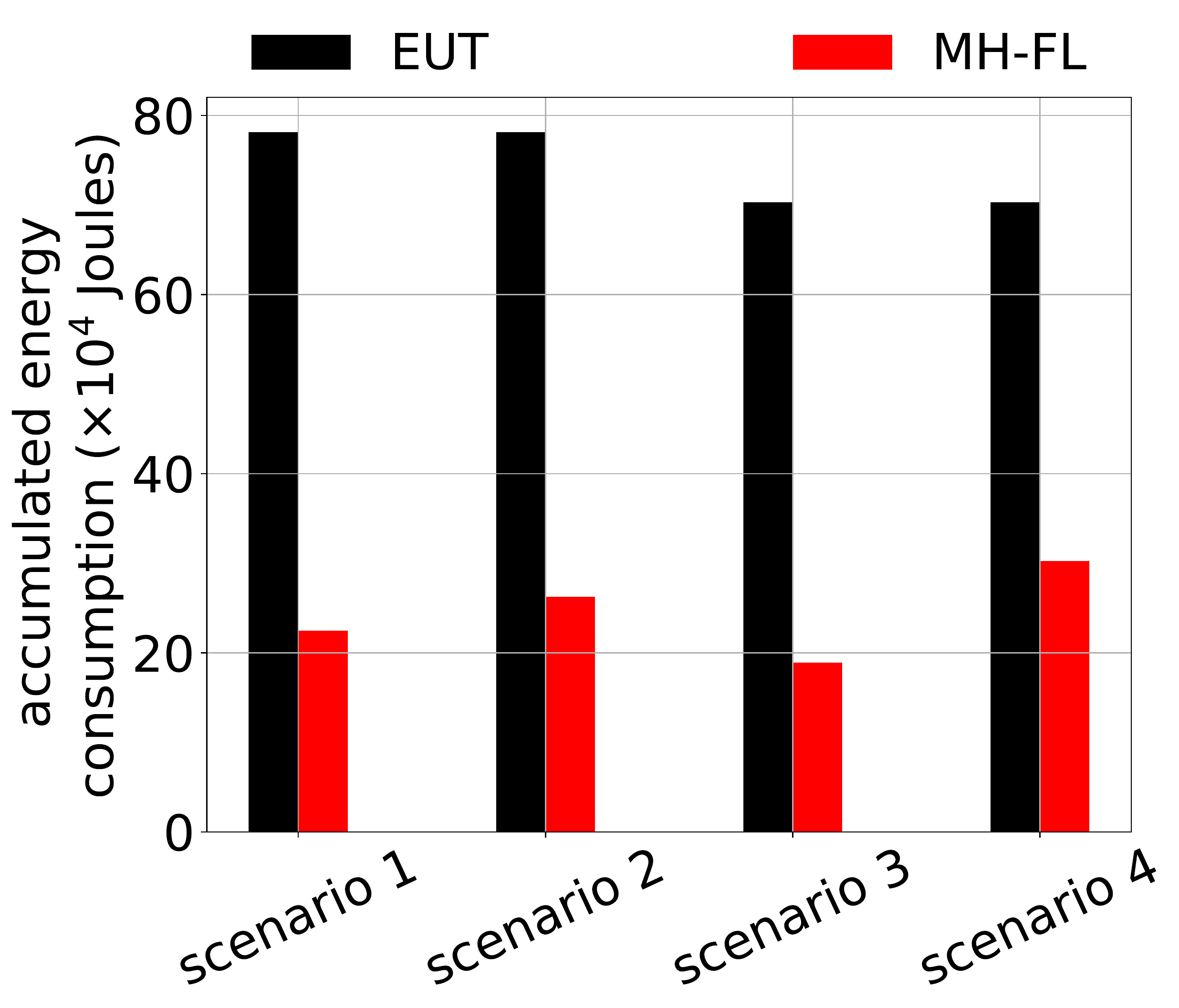}
     \caption{Comparison of accumulated energy consumption between EUT and {\tt MH-MT} over scenario 1: $\sigma' = 0.1$ from Fig.~\ref{fig:iidIncConSigma_FMNIST_625}, scenario 2: $\sigma' = 0.1$ from Fig.~\ref{fig:non_iidIncConSigma_FMNIST_625}, scenario 3: $\psi = 10^4$ from Fig.~\ref{fig:iidNNIncConec_FMNIST_625}, and scenario 4: $\psi = 10^4$ from Fig.~\ref{fig:non_iidConsConDeltaNN_FMNIST_625}. (FMNIST, $625$ Edge Devices)}
     \label{fig:accum_energy_625_FMNIST}
\end{minipage}%
\hspace{2mm}
\begin{minipage}{.32\textwidth}
     \centering
     \includegraphics[width=\linewidth]{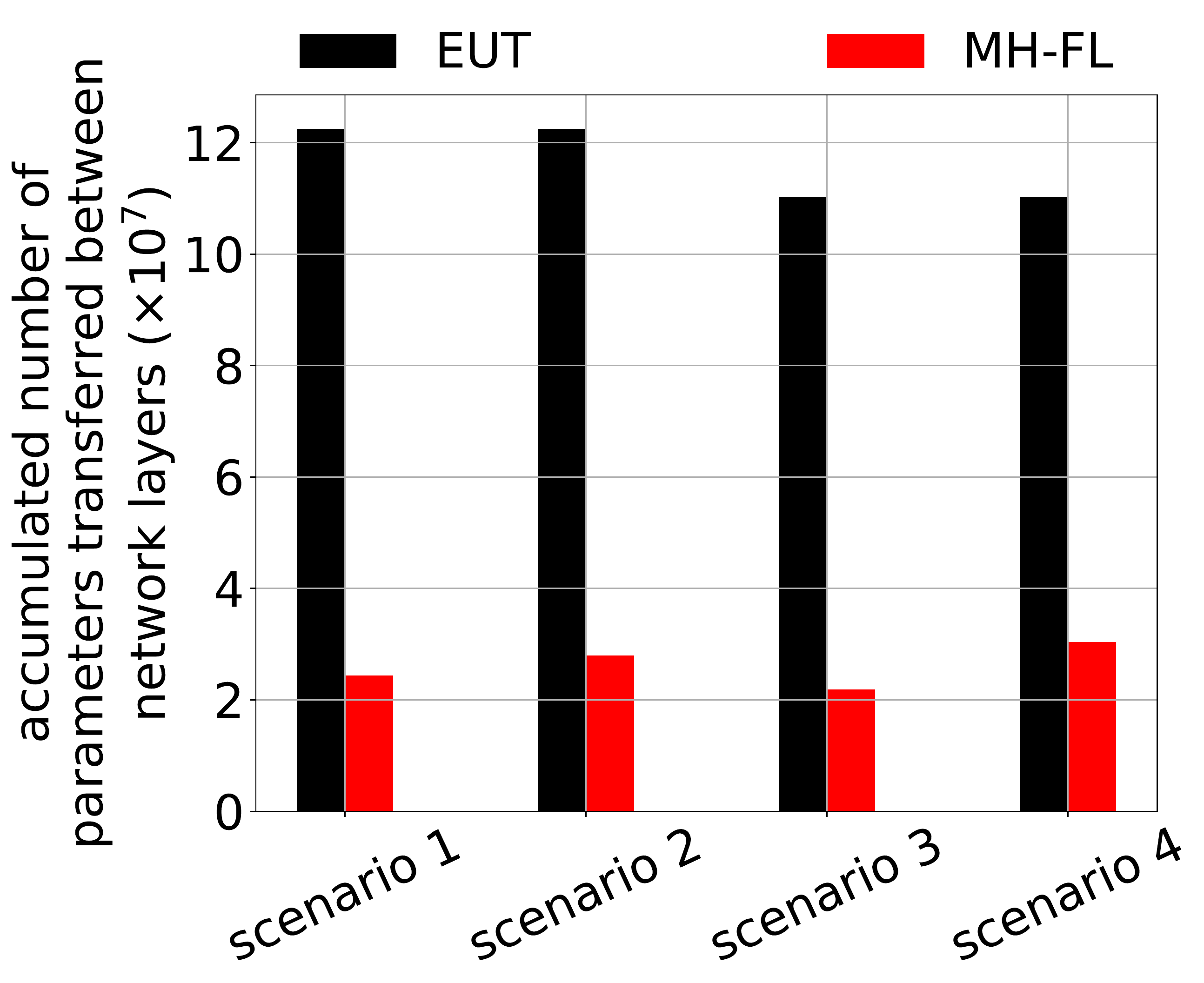}
     \caption{Comparison of parameters transferred among layers in EUT vs {\tt MH-MT} over scenario 1: $\sigma' = 0.1$ from Fig.~\ref{fig:iidIncConSigma_FMNIST_625}, scenario 2: $\sigma' = 0.1$ from Fig.~\ref{fig:non_iidIncConSigma_FMNIST_625}, scenario 3: $\psi = 10^4$ from Fig.~\ref{fig:iidNNIncConec_FMNIST_625}, and scenario 4: $\psi = 10^4$ from Fig.~\ref{fig:non_iidConsConDeltaNN_FMNIST_625}. (FMNIST, $625$ Edge Devices)}
     \label{fig:accumData_FMNIST_625}
\end{minipage}\vspace{-9mm}
\end{figure}

%%%%%%%%%%%%%%%%%%%%%%%%%%%%%% End F-MNIST 625
%%%%%%%%%%%%%%%%%%%%%%%%%%%%%%
%%%%%%%%%%%%%%%%%%%%%%%%%%%%%%
%%%%%%%%%%%%%%%%%%%%%%%%%%%%%

\begin{figure}
\centering
\begin{minipage}{.45\textwidth}
  \centering
    \includegraphics[width=\textwidth]{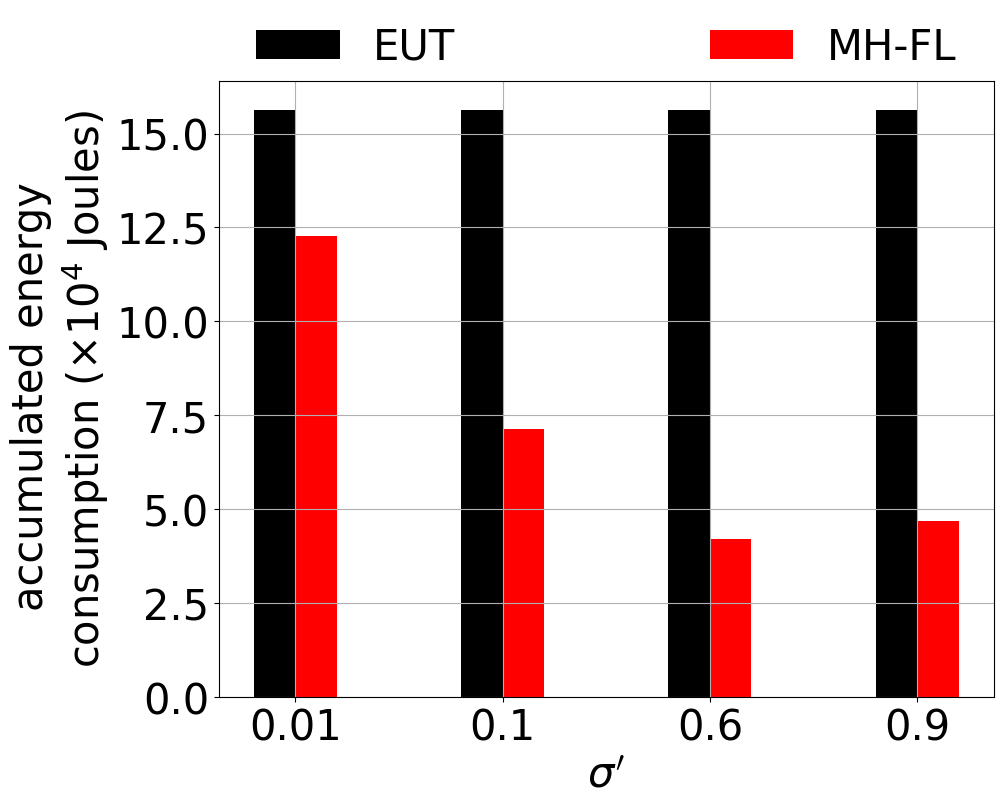}
    \caption{Energy consumption on MNIST w/ 125 nodes for varying $\sigma'$.}
    \label{fig:pc_mnist_125_sig}
\end{minipage}\quad\quad
\begin{minipage}{.45\textwidth}
  \centering
    \includegraphics[width=1.04\textwidth]{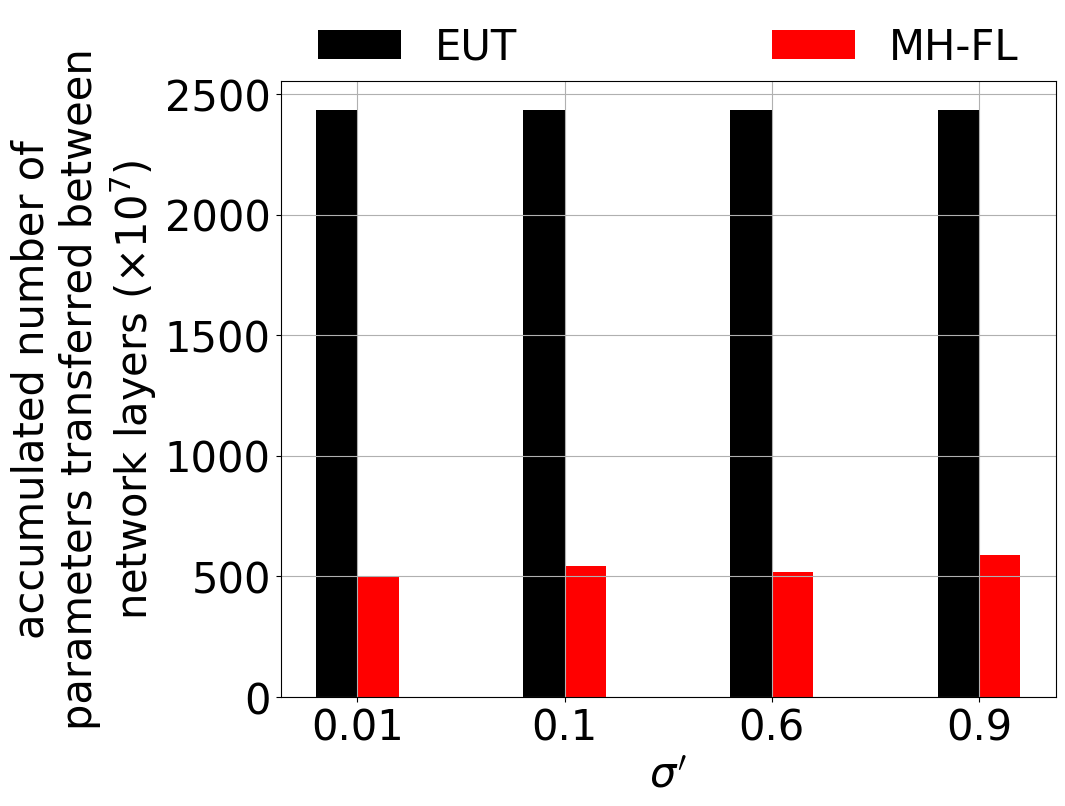}
    \caption{Parameters transferred on MNIST w/ 125 nodes for varying $\sigma'$.}
    \label{fig:pt_mnist_125_sig}
\end{minipage}
\end{figure}

\begin{figure}
\centering
\begin{minipage}{.45\textwidth}
    \centering
    \includegraphics[width=\textwidth]{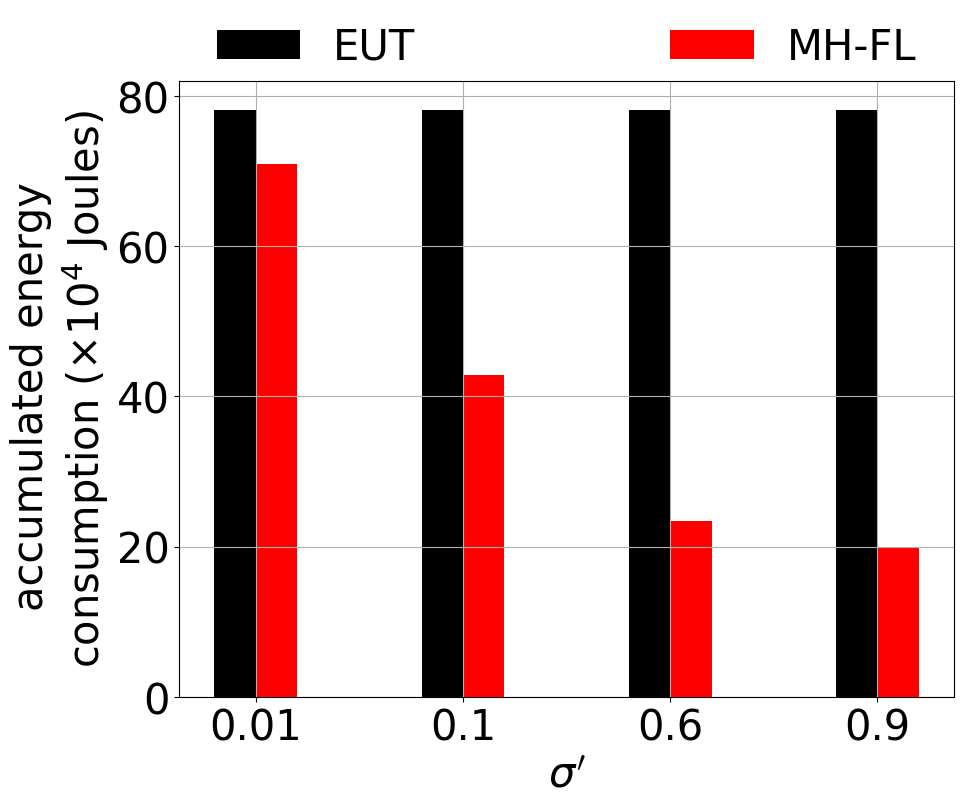}
    \caption{Energy consumption on MNIST w/ 625 nodes for varying $\sigma'$.}
    \label{fig:pc_mnist_625_sig}
\end{minipage}\quad\quad
\begin{minipage}{.45\textwidth}
    \centering
    \includegraphics[width=1.12\textwidth]{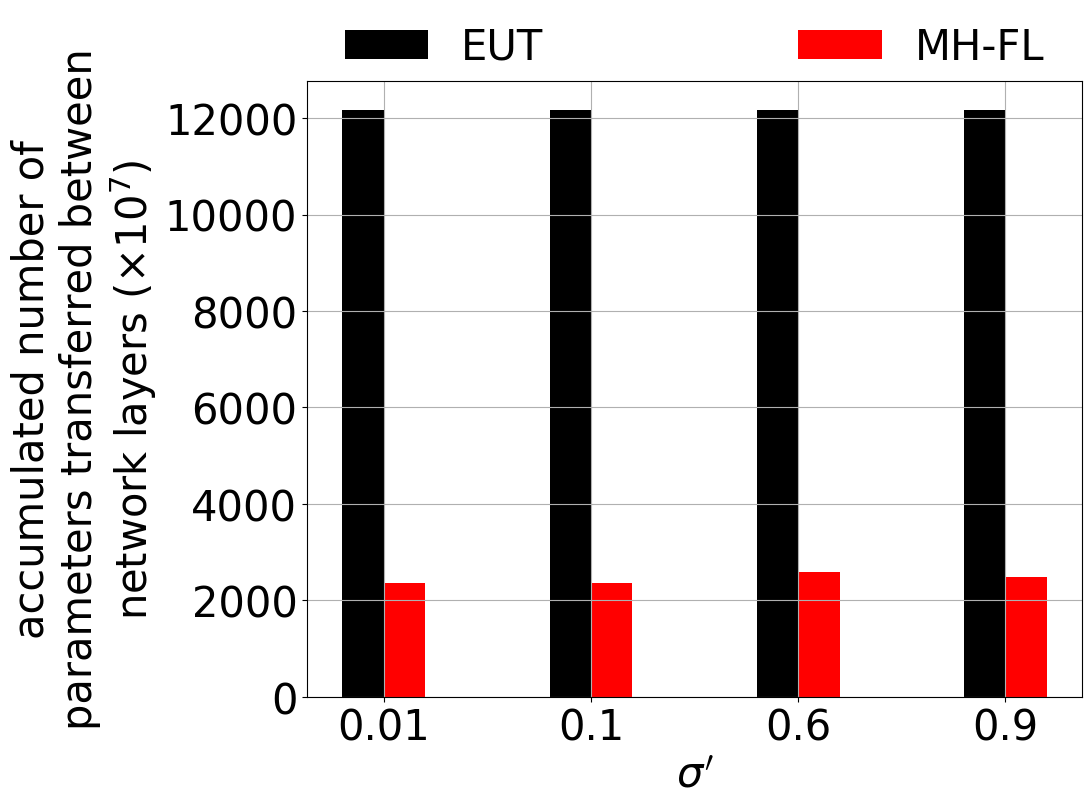}
    \caption{Parameters transferred on MNIST w/ 625 nodes for varying $\sigma'$.}
    \label{fig:pt_mnist_625_sig}
\end{minipage}
\end{figure}

\begin{figure}
\centering
\begin{minipage}{.45\textwidth}
    \centering
    \includegraphics[width=\textwidth]{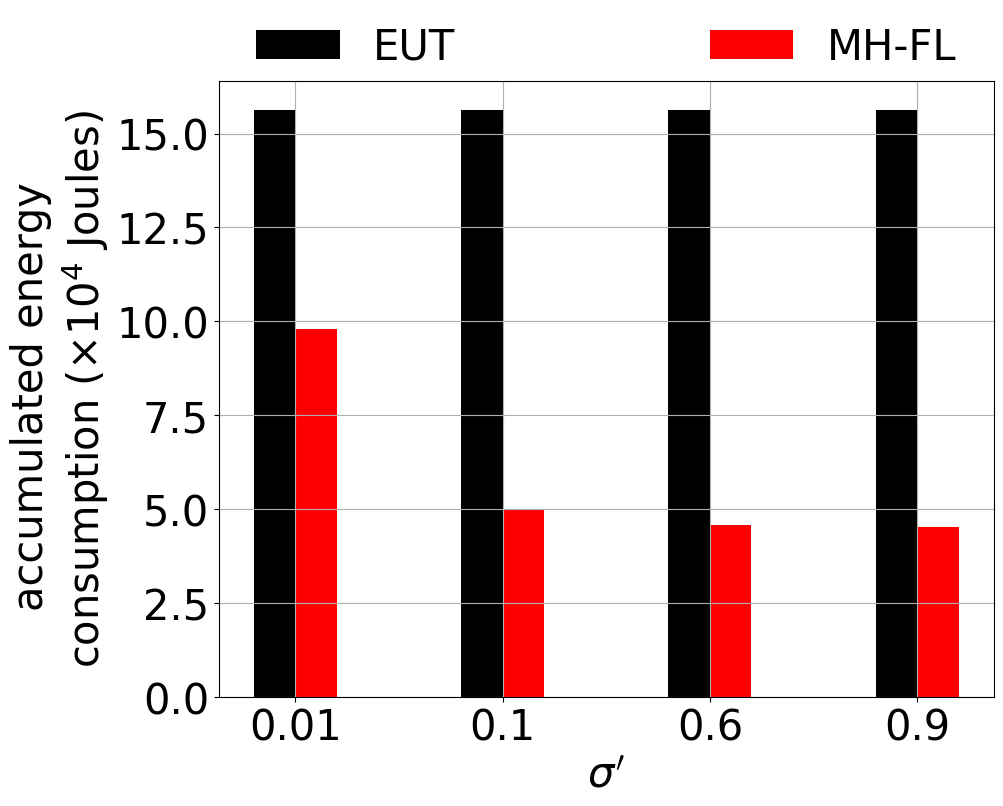}
    \caption{Energy consumption on FMNIST w/ 125 nodes for varying $\sigma'$.}
    \label{fig:pc_fmnist_125_sig}
\end{minipage}\quad\quad
\begin{minipage}{.45\textwidth}
    \centering
    \includegraphics[width=1.06\textwidth]{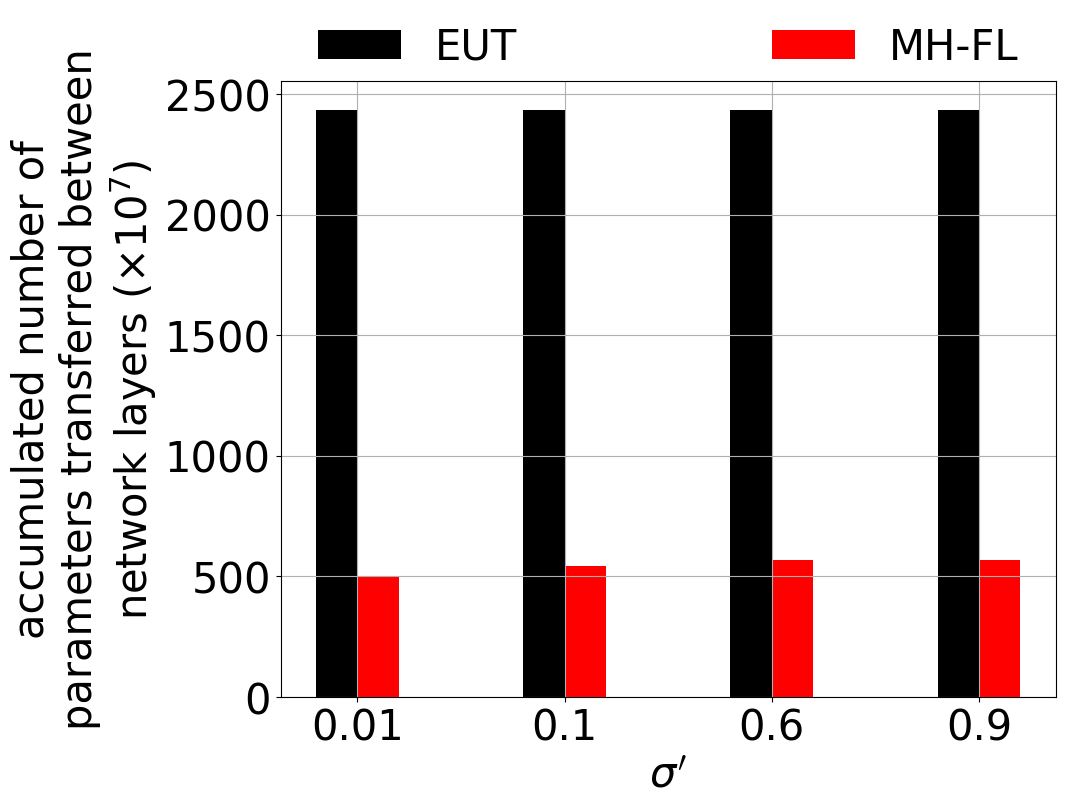}
    \caption{Parameters transferred on FMNIST w/ 125 nodes for varying $\sigma'$.}
    \label{fig:pt_fmnist_125_sig}
\end{minipage}
\end{figure}

\begin{figure}
\centering
\begin{minipage}{.45\textwidth}
    \centering
    \includegraphics[width=\textwidth]{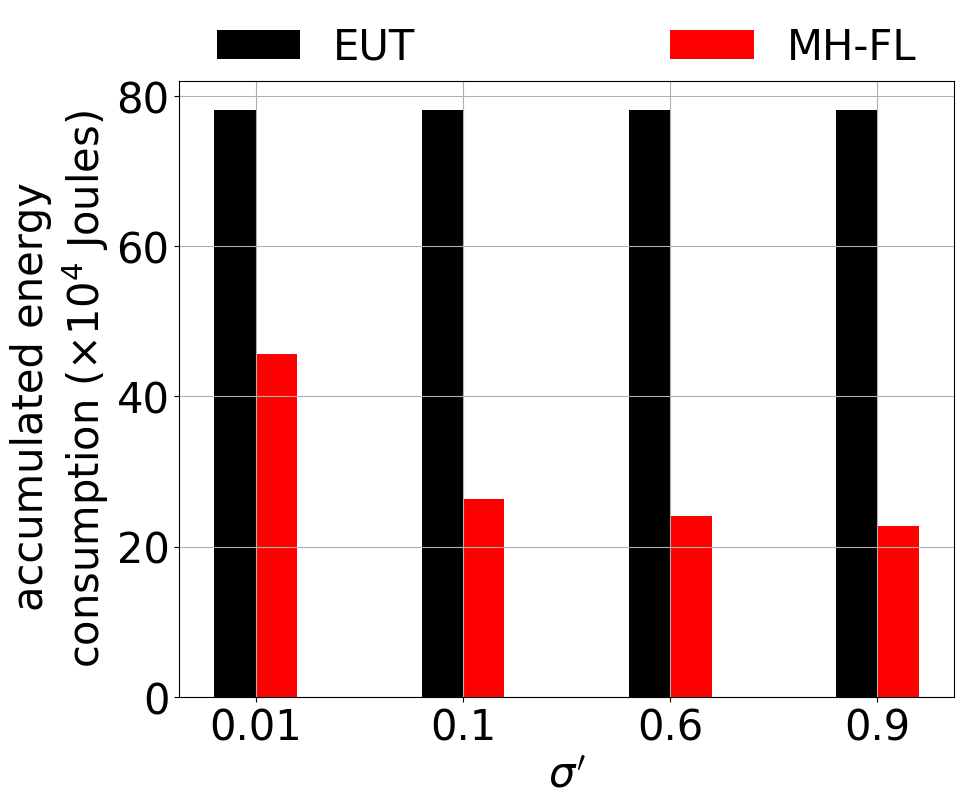}
    \caption{Energy consumption on FMNIST w/ 625 nodes for varying $\sigma'$.}
    \label{fig:pc_fmnist_625_sig}
\end{minipage}\quad\quad
\begin{minipage}{.45\textwidth}
    \centering
    \includegraphics[width=1.12\textwidth]{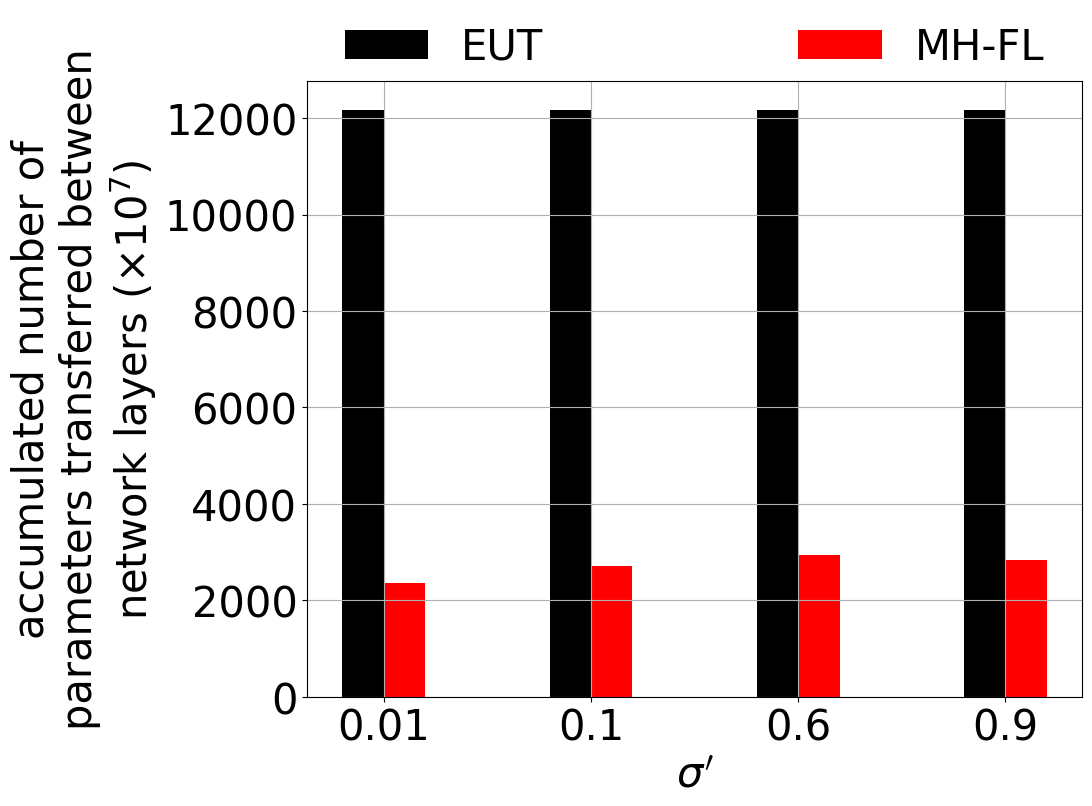}
    \caption{Parameters transferred on FMNIST w/ 625 nodes for varying $\sigma'$.}
    \label{fig:pt_fmnist_625_sig}
\end{minipage}
\end{figure}

%%%%%%%%%%%%%%%%%%%%%%%%%%%%%%%%%%%%%%%%%%%%%%%%%%%%%%%%%%%%%%%%%%%%%%%%%%55

\begin{figure}
\centering
\begin{minipage}{.45\textwidth}
    \centering
    \includegraphics[width=\textwidth]{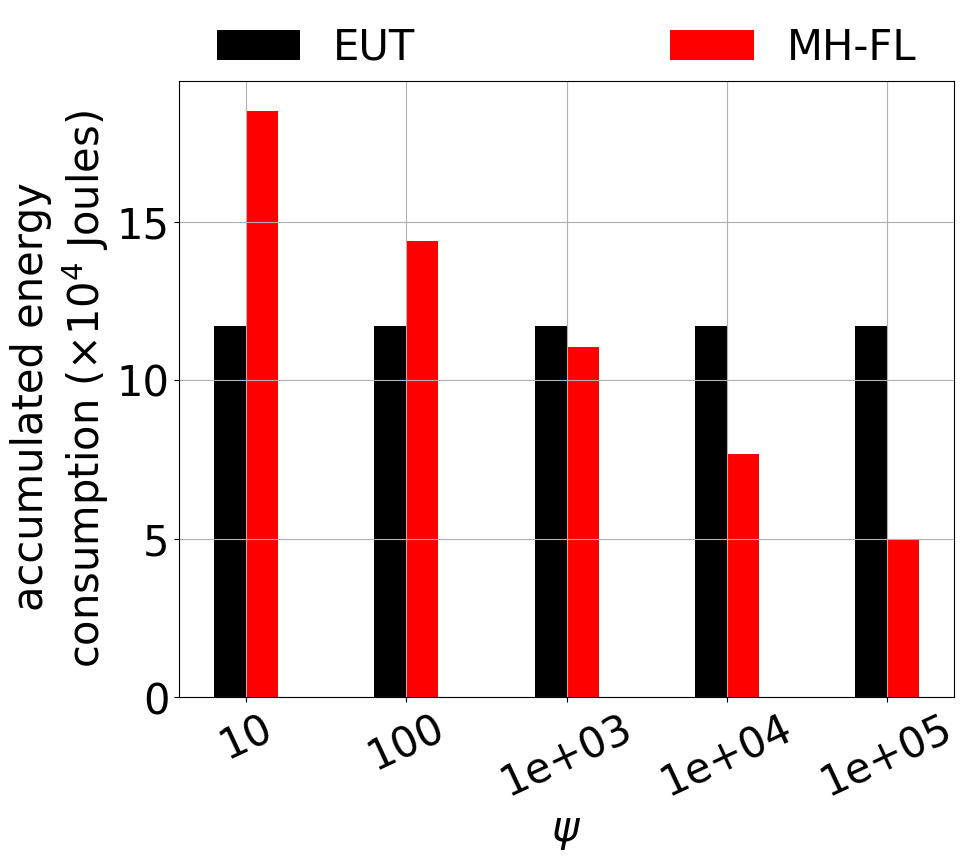}
    \caption{Energy consumption on MNIST w/ 125 nodes for varying $\psi$.}
    \label{fig:pc_mnist_125}
\end{minipage}\quad\quad
\begin{minipage}{.45\textwidth}
    \centering
    \includegraphics[width=1.12\textwidth]{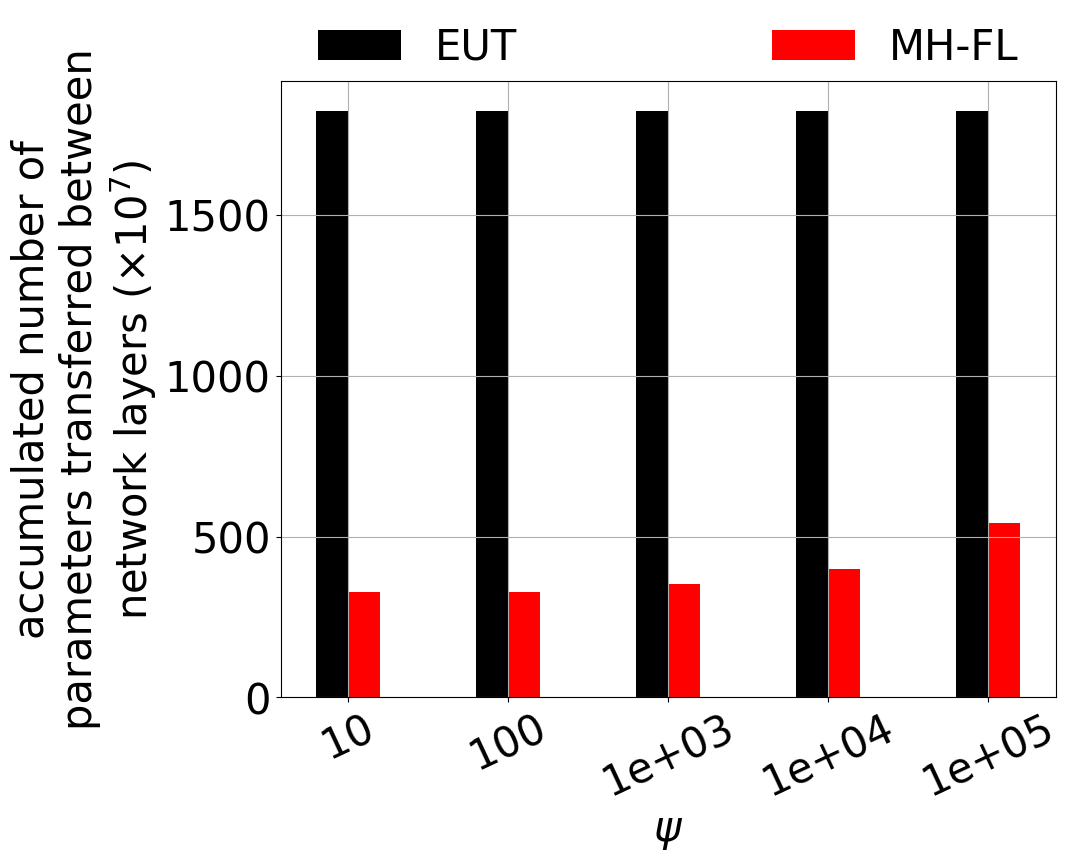}
    \caption{Parameters transferred on MNIST w/ 125 nodes for varying $\psi$.}
    \label{fig:pt_mnist_125}
\end{minipage}
\end{figure}

\begin{figure}
\centering
\begin{minipage}{.45\textwidth}
    \centering
    \includegraphics[width=\textwidth]{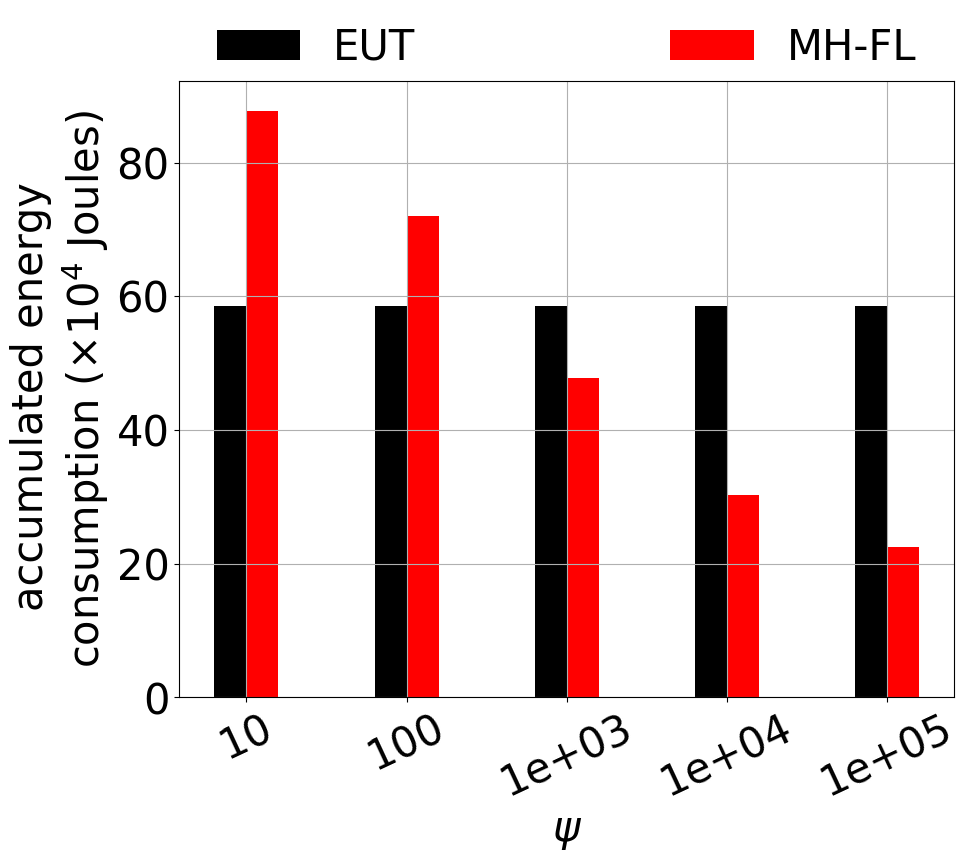}
    \caption{Energy consumption on MNIST w/ 625 nodes for varying $\psi$.}
    \label{fig:pc_mnist_625}
\end{minipage}\quad\quad
\begin{minipage}{.45\textwidth}
    \centering
    \includegraphics[width=1.12\textwidth]{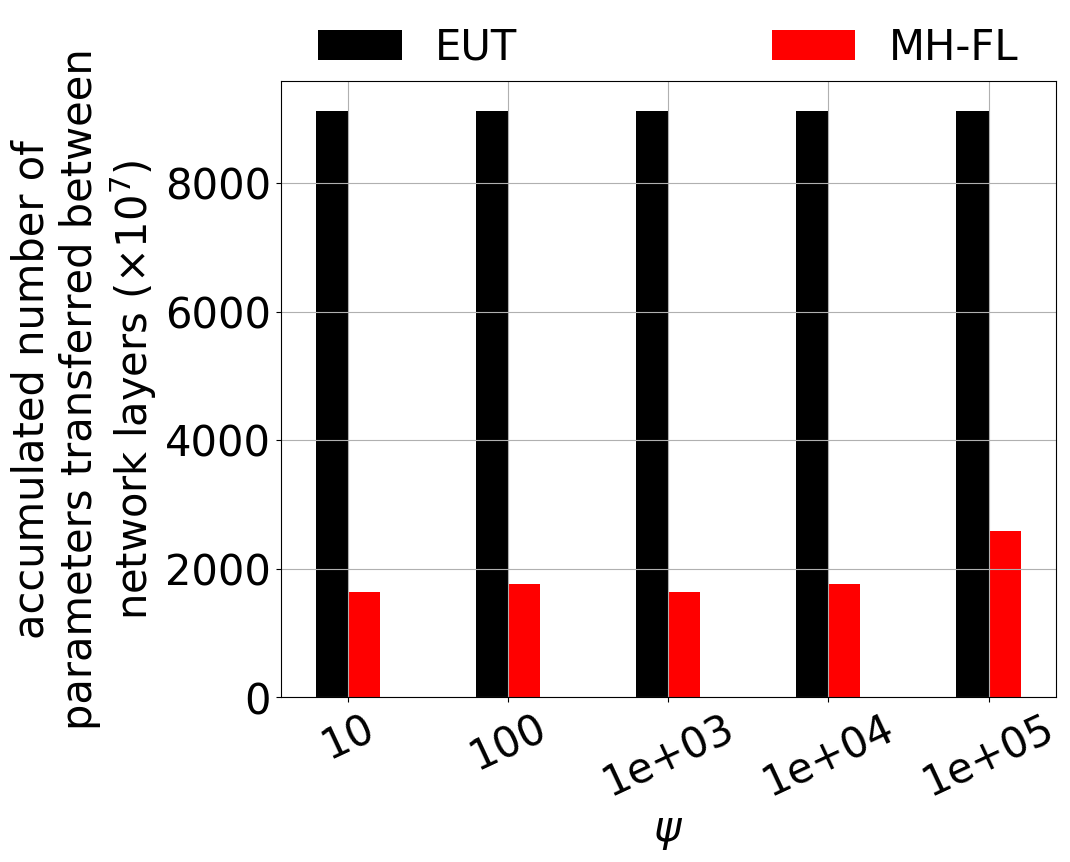}
    \caption{Parameters transferred on MNIST w/ 625 nodes for varying $\psi$.}
    \label{fig:pt_mnist_625}
\end{minipage}
\end{figure}

\begin{figure}
\centering
\begin{minipage}{.45\textwidth}
    \centering
    \includegraphics[width=\textwidth]{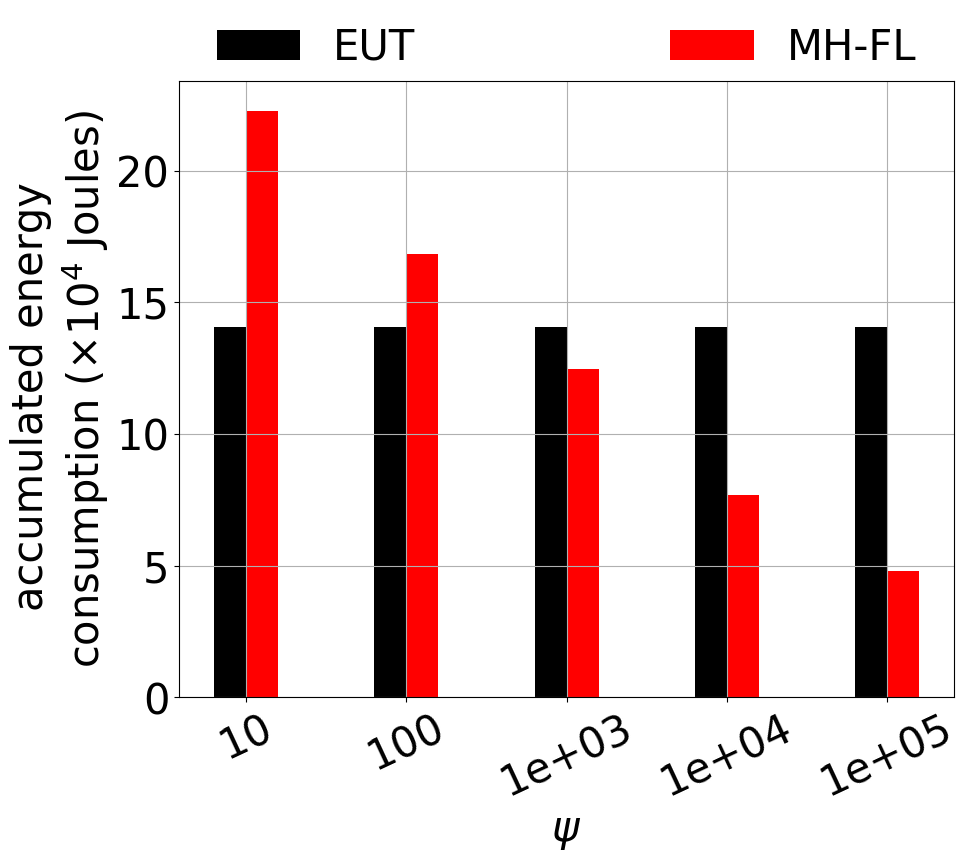}
    \caption{Energy consumption on FMNIST w/ 125 nodes for varying $\psi$.}
    \label{fig:pc_fmnist_125}
\end{minipage}\quad\quad
\begin{minipage}{.45\textwidth}
    \centering
    \includegraphics[width=1.12\textwidth]{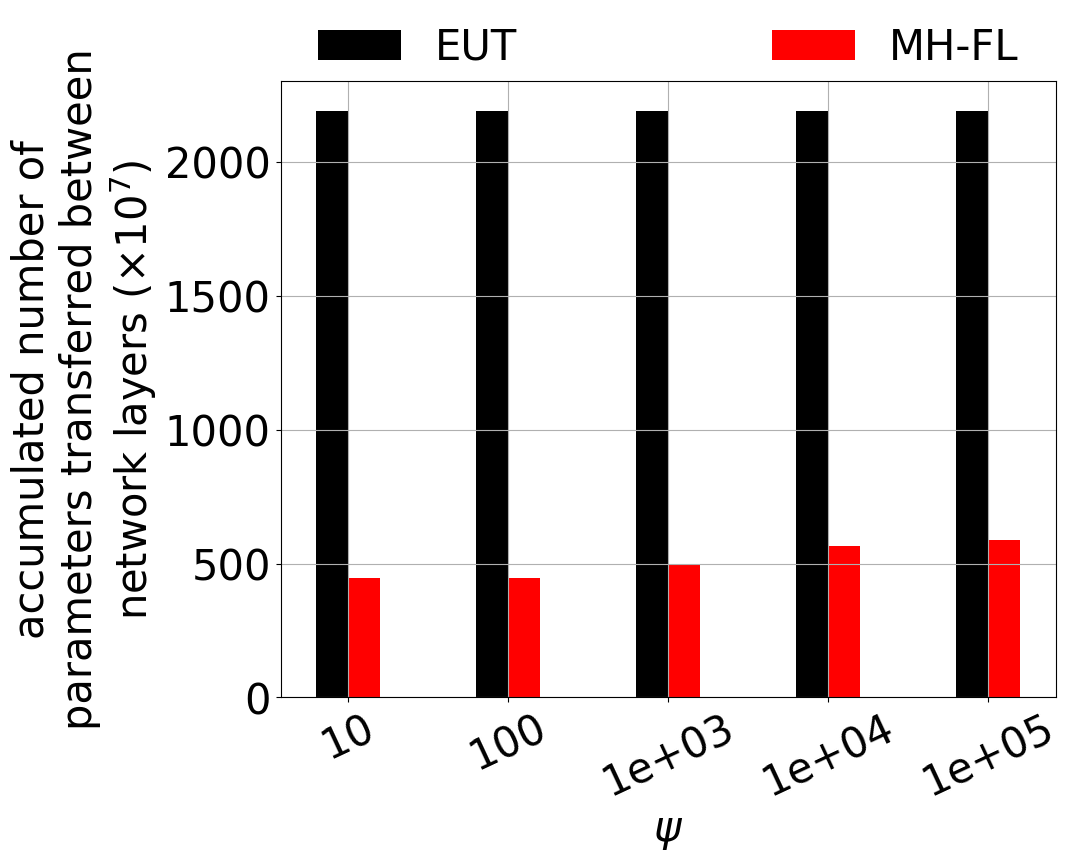}
    \caption{Parameters transferred on FMNIST w/ 125 nodes for varying $\psi$.}
    \label{fig:pt_fmnist_125}
\end{minipage}
\end{figure}

\begin{figure}
\centering
\begin{minipage}{.45\textwidth}
    \centering
    \includegraphics[width=\textwidth]{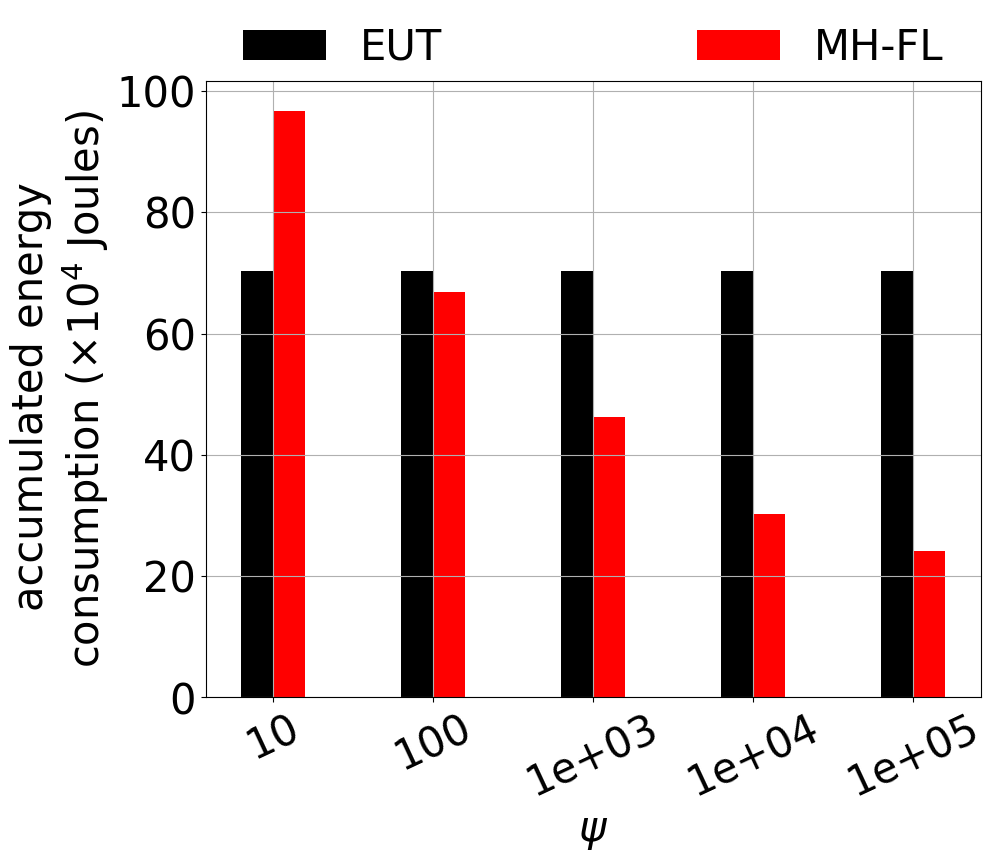}
    \caption{Energy consumption on FMNIST w/ 625 nodes for varying $\psi$.}
    \label{fig:pc_fmnist_625}
\end{minipage}\quad\quad
\begin{minipage}{.45\textwidth}
    \centering
    \includegraphics[width=1.12\textwidth]{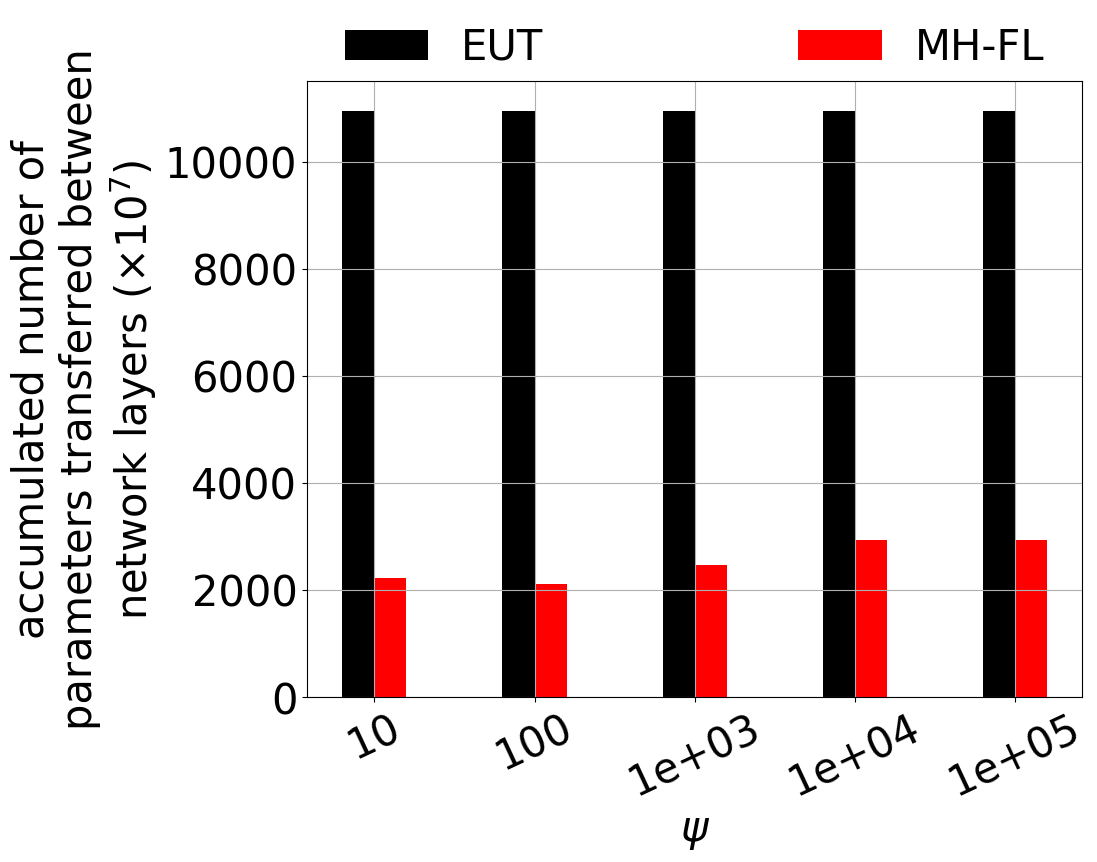}
    \caption{Parameters transferred on FMNIST w/ 625 nodes for varying $\psi$.}
    \label{fig:pt_fmnist_625}
\end{minipage}
\end{figure}

\endgroup %% For Single Column

\end{document}